\documentclass[10pt,journal,twocolumn,a4paper]{IEEEtran}
\usepackage{amsfonts,amsmath,amssymb,amstext,amsthm,cite,color,dsfont,enumerate,hyperref,makeidx,verbatim,xfrac,xr,url}
\usepackage{pict2e}
\usepackage[mathscr]{euscript}
\hypersetup{
pdftitle		={The Sphere Packing Bound via Augustin's Method},
pdfauthor		={Bar{\i}s Nakiboglu},}
\title{The Sphere Packing Bound via Augustin's Method}
\author{Bar\i\c{s} Nakibo\u{g}lu \thanks{e-mail:\href{mailto:bnakib@metu.edu.tr}{bnakib@metu.edu.tr}}}
\IEEEaftertitletext{ 
\flushright
\vspace{-2\baselineskip}
{\it 
Can{\i}m halam Fatma Nakibo\u{g}lu Aydi\c{c}'in an{\i}s{\i}na adanm{\i}{\c{s}}t{\i}r. \hspace{1.6cm}~
\\
Dedicated to the memory of my dear aunt Fatma Nakibo\u{g}lu Aydi\c{c}.}\qquad~\\
~\\
}
\theoremstyle{plain}
\newtheorem{lemma}{Lemma} 
\newtheorem{theorem}{Theorem}

\newtheorem*{conjecture*}{Conjecture}

\theoremstyle{definition}
\newtheorem{definition}{Definition} 
\newtheorem{assumption}{Assumption}

\newtheorem{remark}{Remark}
\newtheorem*{remark*}{Remark}
\definecolor{mygray}{gray}{0.4}
\newcommand{\set} [1]			{{\mathscr{{#1}}}}
\newcommand{\alg}[1]			{{\mathcal{{#1}}}}
\newcommand{\rndv}[1]      {{\mathsf{{#1}}}}

\newcommand{\msr}[1]       {{\it    {{#1}}}}
\newcommand{\cnst}[1]      {{\mathit{{#1}}}}

\newcommand{\cntt}[1]      {{\widetilde{{\mathit{{#1}}}}}}

\newcommand{\sss}[1]		{{\mathit{2}^{{#1}}}}
\newcommand{\integers}[1]	{{\mathbb{Z}}_{^{{#1}}}}

\newcommand{\reals}[1]		{{\mathbb{R}}_{^{{#1}}}}
\newcommand{\bigo} [1]     {{\cnst{O}\left({{#1}}\right)}}

\renewcommand{\vec}[1]     {\overrightarrow{{#1}}}

\newcommand{\dif}[1]       {{\mathrm{d}{#1}}}  
\newcommand{\der}[2]        {\tfrac{\dif{#1}}{\dif{#2}}}  
  
\newcommand{\supp}[1]       {\mathtt{supp}({{#1}})}       

\newcommand{\DEF}[0]			{{\!\!~\triangleq\!~}}  
\newcommand{\mtimes}[0]			{{\circledast}}

\newcommand{\AC}[0]            {{\prec}}

\newcommand{\abs}[1]           {{\left\lvert{{#1}}\right\lvert}}

\newcommand{\lon}[1]           {{{\left\lVert{{#1}}\right\lVert}}} 

\newcommand{\IND}[1]           {{\mathds{1}_{\{#1\}}}}    
\newcommand{\ind}[0]           {{\imath}}
\newcommand{\jnd}[0]           {{\jmath}}
\newcommand{\knd}[0]           {{\kappa}}
\newcommand{\tin}[0]           {{\cnst{t}}}
\newcommand{\blx}[0]           {{\cnst{n}}}
\newcommand{\tlx}[0]           {{\cnst{T}}}

\newcommand{\pint}[0]          {{\cnst{\zeta}}}
\newcommand{\PXS}[2]         {{\bf P}_{{#1}}\!\left[{#2}\right]}
\newcommand{\EXS}[2]         {{\bf E}_{{#1}}\!\left[{#2}\right]}

\newcommand{\PX}[1]          {\PXS{\!}{{#1}}}                      %Probability
\newcommand{\EX}[1]          {\EXS{\!}{{#1}}}                      %Expectation
\newcommand{\Lp}[2]          {{{\alg{L}}}^{{#1}}({#2})}   
\newcommand{\Lon}[1]         {{\Lp{1}{#1}}}
\newcommand{\fX}[0]          {{\cnst{f}}}   
\newcommand{\FX}[0]          {{\cnst{F}}}   

\newcommand{\gX}[0]          {{\cnst{g}}}

\newcommand{\fXS}[0]         {{\set{F}}}

\newcommand{\cm}[1]         {{{\cnst{g}}_{{#1}}}}

\renewcommand{\div}[3]			{{\cnst{d}}_{{#1}}            \!\left({\left.           \! {#2}\right\Vert {#3}}                 \right)}

\newcommand{\RD}[3]				{{\cnst{D}}_{{#1}}            \!\left(\left.            \! {#2}\right\Vert {#3}                  \right)}

\newcommand{\RDF}[4]			{{\cnst{D}}_{{#1}}^{{#2}}     \!\left(\left.            \! {#3}\right\Vert {#4}                  \right)}

\newcommand{\RMI}[3]			{{\cnst{I}}_{{#1}}            \!\left(                  \! {#2};         \!{#3}                \!\right)}

\newcommand{\RC}[2]				{{\cnst{C}}_{{#1},{#2}}}

\newcommand{\CRC}[3]			{{\cnst{C}}_{{#1},{#2},{#3}}}

\newcommand{\RCI}[3]			{{\cntt{C}}_{{#1},{#2}}^{#3}}

\newcommand{\RR}[2]				{{\cnst{S}}_{{#1},{#2}}}
\newcommand{\RRR}[3]			{{\cnst{S}}_{{#1},{#2}}({#3})}

\newcommand{\SC}[1]				{{\cnst{C}}({#1})}

\newcommand{\costc}[0]			{{\cnst{\varrho}}}
\newcommand{\lgm}[0]			{{\cnst{\lambda}}}
    
\newcommand{\rfm}[0]			{{{\msr{\nu}}}}

\newcommand{\cset}[0]			{{\set{A}}}

\newcommand{\spa}[2]			{{\cntt{E}_{sp\!}^{{#1}}}\left({#2}\right)}

\newcommand{\spe}[1]			{{\cnst{E}_{sp\!}}       \left({#1}\right)}
\newcommand{\sce}[1]			{{\cnst{E}_{sc\!}}       \left({#1}\right)}
\newcommand{\rate}[0]			{{\cnst{R}}}
\newcommand{\GCD}[1]			{{{{\cnst{\Phi}}}\left({{#1}}\right)}}

\newcommand{\cln}[1]          {{{\xi}_{{#1}}}}
\newcommand{\cla}[2]          {{{\xi}_{{#1}}^{{#2}}}}
\newcommand{\rens}[0]         {{\cnst{\varphi}}}
\newcommand{\rnb}[0]          {{\cnst{\beta}}}
\newcommand{\rnf}[0]          {{\cnst{\phi}}}

\newcommand{\rng}[0]          {{\cnst{\rho}}}
\newcommand{\rno}[0]          {{\cnst{\alpha}}}
\newcommand{\rnt}[0]          {{\cnst{\eta}}}
\newcommand{\Pe}[0]            {{\it P_{{{\bf e}}}}}      
                  
\newcommand{\Pem}[1]           {{\it P_{{{\bf e}}}^{{#1}}}}         
\newcommand{\tenc}[0]          {\widetilde{\varPsi}} 
\newcommand{\enc}[0]           {{\varPsi}} 
\newcommand{\dec}[0]           {{\varTheta}}    

\newcommand{\brl}[0]           {{\alg{B}}}
\newcommand{\rborel}[1]        {{\brl}({#1})}

\newcommand{\oev}[0]           {{\set{E}}}

\newcommand{\cha}[0]           {{\set{W}}}

\newcommand{\Pcha}[1]          {{\varLambda}^{{{#1}}}}

\newcommand{\smea}[1]          {{{\alg{M}}({#1})}}
\newcommand{\fmea}[1]          {{{\alg{M}}^{^{+}}\!({#1})}}

\newcommand{\pmea}[1]          {{{\alg{P}}({#1})}}

\newcommand{\pdis}[1]          {{{\set{P}}({#1})}}

\newcommand{\dinp}[0]          {{\cnst{x}}}

\newcommand{\inpS}[0]          {{\set{X}}}
\newcommand{\inpA}[0]          {{\alg{X}}}

\newcommand{\dout}[0]          {{\cnst{y}}}
\newcommand{\out}[0]           {{\rndv{Y}}}
\newcommand{\outS}[0]          {{\set{Y}}}
\newcommand{\outA}[0]          {{\alg{Y}}}

\newcommand{\dsta}[0]          {{\cnst{z}}}

\newcommand{\staS}[0]          {{\set{Z}}}
\newcommand{\staA}[0]          {{\alg{Z}}}

\newcommand{\dmes}[0]          {{\cnst{m}}}

\newcommand{\mesS}[0]          {{\set{M}}}

\newcommand{\estS}[0]          {{\widehat{{\set{M}}}}}

\newcommand{\mean}[0]        {{{\msr{\mu}}}}

\newcommand{\mA}[0]				{{\msr{a}}}    
\newcommand{\amn}[1]			{{{\mA}_{{#1}}}}

\newcommand{\mB}[0]				{{\msr{b}}}    
\newcommand{\bmn}[1]			{{{\mB}_{{#1}}}}
\newcommand{\bma}[2]			{{{\mB}_{{#1}}^{{#2}}}}

\newcommand{\Bmn}[1]			{{{\cnst{B}}_{{#1}}}}

\newcommand{\mP}[0]				{{\msr{p}}}

\newcommand{\mQ}[0]				{{\msr{q}}}    
\newcommand{\qmn}[1]			{{{\mQ}_{{#1}}}}
\newcommand{\qma}[2]			{{{\mQ}_{{#1}}^{{#2}}}}
\newcommand{\Qm}[0]				{{{\cnst{Q}}}}

\newcommand{\mS}[0]				{{\msr{s}}}    
\newcommand{\smn}[1]			{{{\mS}_{{#1}}}}

\newcommand{\mU}[0]				{{\msr{u}}}

\newcommand{\Um}[0]				{{{\cnst{U}}}}
\newcommand{\Umn}[1]			{{{\cnst{U}}_{{#1}}}}

\newcommand{\mV}[0]				{{\msr{v}}}    
\newcommand{\vmn}[1]			{{{\mV}_{{#1}}}}
\newcommand{\vma}[2]			{{{\mV}_{{#1}}^{{#2}}}}
\newcommand{\Vm}[0]				{{{\cnst{V}}}}
\newcommand{\Vmn}[1]			{{{\cnst{V}}_{{#1}}}}

\newcommand{\mW}[0]				{{\msr{w}}}    
\newcommand{\wmn}[1]			{{{\mW}_{{#1}}}}
\newcommand{\wma}[2]			{{{\mW}_{{#1}}^{{#2}}}}
\newcommand{\Wm}[0]				{{{\cnst{W}}}}
\newcommand{\Wmn}[1]			{{{\cnst{W}}_{{#1}}}}
\newcommand{\Wma}[2]			{{{\cnst{W}}_{{#1}}^{{#2}}}}

\newcommand{\altug}[0]								{Altu\u{g}~}
\newcommand{\csiszar}[0]							{Csisz\'{a}r~}
\newcommand{\fabregas}[0]							{F\`{a}bregas~}
\newcommand{\harremoes}[0]							{Harremo\"{e}s~}
\newcommand{\korner}[0]								{K\"{o}rner~}
\newcommand{\renyi}[0]								{R\'{e}nyi~}

\makeatletter
\DeclareRobustCommand{\bigplus}{%
	\mathop{\vphantom{\sum}\mathpalette\@bigplus\relax}\slimits@
}
\newcommand{\@bigplus}[2]{\vcenter{\hbox{\make@bigplus{#1}}}}
\newcommand{\make@bigplus}[1]{%
	\sbox\z@{$\m@th#1\sum$}%
	\setlength{\unitlength}{\wd\z@}%
	\begin{picture}(1.4,1.4)
	\linethickness{.17ex}
	\Line(.7,.14)(.7,1.26)
	\Line(.14,.7)(1.26,.7)
	\end{picture}%
}
\DeclareRobustCommand{\bigtimes}{%
	\mathop{\vphantom{\sum}\mathpalette\@bigtimes\relax}\slimits@
}
\newcommand{\@bigtimes}[2]{\vcenter{\hbox{\make@bigtimes{#1}}}}
\newcommand{\make@bigtimes}[1]{%
	\sbox\z@{$\m@th#1\sum$}%
	\setlength{\unitlength}{\wd\z@}%
	\begin{picture}(1,1)
	\linethickness{.17ex}
	\Line(.1,.1)(.9,.9)
	\Line(.1,.9)(.9,.1)
	\end{picture}%
}
\makeatother
\begin{document}
\pagestyle{plain}
\pagenumbering{arabic}
\hypersetup{hidelinks}
\maketitle 
\thispagestyle{empty}
%!TEX root=../main-B.tex
\begin{abstract}
A sphere packing bound (SPB) with a prefactor that is polynomial
in the block length $n$
is established for codes on a length $n$ product channel $W_{[1,n]}$ 
assuming that the maximum order \textonehalf~R\'{e}nyi capacity among 
the component channels,  i.e. $\max_{t\in[1,n]} C_{\sfrac{1}{2},W_{t}}$,
is $\mathit{O}(\ln n)$.
The reliability function of the discrete stationary product channels 
with feedback is bounded from above by the sphere packing exponent.
Both results are proved by first establishing a non-asymptotic SPB.
The latter result continues to hold under a milder stationarity hypothesis.
\end{abstract}
\tableofcontents
%!TEX root=../main-B.tex
\section{Introduction}\label{sec:introduction}
Most proofs establishing the infeasibility of certain performance
for the channel coding problem under fixed rate, fixed error probability, 
or slowly vanishing error probability hypotheses
rely on either a type based expurgation \cite{altugW14C,haroutunian68,shannonGB67A}
or a distinction of cases based on types  \cite{moulin17,polyanskiyPV10,strassen62,tomamichelT13}.
Although similar bounds can, usually, be obtained using the information spectrum approach \cite{hayashi09B} 
with greater generality, one has to give up the initial non-asymptotic bound in order to do so. 
This relative advantage of the method of types \cite{csiszar98,csiszarkorner} over 
the information spectrum approach \cite{han} emerges from four distinct assumptions: 
the product structure of the sample space,
the product structure of the probability measures,
finiteness of the input set, 
and the stationarity of the channel.
The finite input set assumption and the product structure assumptions can be removed and 
the stationarity assumption can be relaxed if one gives up the concept of type for 
the concept of typicality.
The typicality arguments are, usually, employed for deriving asymptotic results,
but they can also be used to obtain non-asymptotic bounds. 

 Augustin's proof of the sphere packing bound (SPB) in \cite{augustin78}  stands out in 
 this high level classification of the techniques for deriving 
 infeasibility results for the channel coding problem.
 It establishes a non-asymptotic bound without assuming the finiteness of the 
 input set or the stationarity of the channel.
 The main aim of this article is to build an understanding of Augustin's method around 
 the concepts of capacity and center.
 We believe such an understanding can guide us when we apply Augustin's method to the
 other information transmission problems.
 To build such an understanding, we derive the SPBs using Augustin's method 
 in a way that makes the role of the \renyi capacity and center more explicit.

Shannon, Gallager, and Berlekamp \cite[Thm. 2]{shannonGB67A} published the 
first rigorous proof of the SPB for arbitrary discrete 
stationary product channels\footnote{Recently, Dalai gave an account 
of the earlier results in \cite[Appendix B]{dalai13A}.} (DSPCs) in 1967.
Haroutunian \cite[Thm. 2]{haroutunian68} published an alternative proof that holds 
for arbitrary stationary product channels (SPCs) with finite input sets
in 1968.
In \cite{haroutunian68}, Haroutunian expressed the sphere packing exponent in 
an alternative form, which he proved to be equal to the one in \cite{shannonGB67A}.
Augustin published yet another proof of 
the SPB that holds for ---possibly non-stationary--- product channels
with arbitrary (i.e. possibly infinite) input sets, \cite[Thm. 4.7]{augustin69}
in 1969. 
Augustin's SPB \cite[Thm. 4.7a]{augustin69} holds even for product channels 
with infinite channel capacity.  
In the same article, Augustin also  established a SPB 
with a polynomial prefactor \cite[Thm. 4.8]{augustin69}, 
under a hypothesis that is satisfied by all DSPCs.

The first two proofs of the SPB for product channels, 
presented in \cite{haroutunian68} and \cite{shannonGB67A}, 
rely on expurgations based on the empirical distribution, i.e. 
the type or the composition, of the input codewords; 
as a result, they are valid only for the SPCs with finite input sets.
Hence, even the non-stationary discrete product channels
are beyond the reach of the results presented in \cite{haroutunian68} 
and \cite{shannonGB67A},  unless the channel has certain 
symmetries or the channel is ---at least approximately--- periodic.
In order to see how the periodicity can be used to overcome
non-stationarity, consider the product of a sequence of channels 
that alternates 
between two distinct channels at odd and even time instances. 
The resulting product channel is formally non-stationary;
yet it can also be interpreted as a stationary product channel 
with larger components.
Thus results of \cite{haroutunian68} and \cite{shannonGB67A} 
are applicable to non-stationary but 
periodic DPCs, as well.
Furthermore, if the channels in the sequence are from 
a finite set \(\alg{W}\) of possible component channels  
and the frequencies 
of elements of \(\alg{W}\)
are asymptotically stable, 
then the results of 
\cite{haroutunian68} or \cite{shannonGB67A} can be applied through 
larger component channels and appropriate worst case 
assumptions.\footnote{The asymptotic stability of frequencies of 
	elements of \(\alg{W}\), rather than the periodicity of the channel, suffices
	because the lack of contiguity for the subcomponents of the component channels 
	is inconsequential for the performance of codes on product channels.}
One can obtain the same result by making minor changes in the proofs 
presented in \cite{haroutunian68}, or in \cite{shannonGB67A}, 
see \cite[\S V.A]{dalaiW17} for one such modification for a related problem.
In fact, with such changes one can handle infinite \(\alg{W}\)'s 
under appropriate finite approximability and asymptotic stability 
assumptions,
albeit with crude approximation error terms and through a rather
complicated proof.

The stationarity of the channel has been assumed even in the proofs of the SPB tailored 
for specific noise models, such as the ones for the Poisson channels
 in \cite{burnashevK99}, \cite{wyner88-b}. 
In their current form, without major changes, neither the approach of Burnashev and Kutoyants 
in \cite{burnashevK99} nor Wyner's approach in \cite{wyner88-b}
---relying on discretization ---
can establish the SPB for a zero dark current Poisson channel whose inputs are 
intensity functions, i.e. \(\fX\)'s, that are bounded as follows:
\begin{align}
\notag
0\leq \fX(\tin)&\leq \gX(\tin)
&
&\forall \tin\in\reals{+} 
\end{align}
where \(\gX\) is a non-periodic function that is integrable on all bounded intervals.
On the other hand, 
this channel satisfies the hypothesis of \cite[Thm. 4.7b]{augustin69}
by \cite[(\ref*{A-eq:poissonchannel-product-capacity})]{nakiboglu19A}
and the SPB for this channel follows from  Augustin's general proof 
for the product channels,
provided that \(\gX\) satisfies rather mild conditions.

The Shannon, Gallager, Berlekamp proof and Haroutunian's proof had greater 
impact on the field than Augustin's proof. 
Variants of Haroutunian's proof can be found in  
\cite{csiszar98,csiszarkorner,dalai17,haroutunianHH07,omura75}.
%\cite[Thm. IV.2]{csiszar98}, \cite[Thm. 10.6]{csiszarkorner}, \cite{omura75}.
For the DSPCs, Haroutunian's method leads to a SPB
with a polynomial prefactor, i.e. a prefactor of the form \(e^{-\bigo{\ln \blx}}\).
The prefactor of the Shannon, Gallager, Berlekamp proof
in \cite{shannonGB67A} is \(e^{-\bigo{\sqrt{\blx}}}\),
which is considerably worse.
In \cite{valemboisF04}, Valembois and Fossorier have improved 
the prefactor of \cite{shannonGB67A}  for moderate block lengths;
the asymptotic behavior of the prefactor, however, is still \(e^{-\bigo{\sqrt{\blx}}}\).
In \cite{wiechmanS08}, Wiechman and Sason improved the prefactor of \cite{valemboisF04}  
for channels with certain symmetries by eliminating the type based expurgation 
step of the derivation and the resulting contribution to the rate back-off term 
in the bound.
This improvement, however, is inconsequential for the asymptotic behavior 
of the prefactor in \cite{wiechmanS08}, which is
\(e^{-\bigo{\sqrt{\blx}}}\), as well.
Augustin derived a SPB with a polynomial prefactor for  certain product channels 
in \cite[Thm. 4.8]{augustin69}; however, the prefactors of his general results
\cite[Thm. 4.7]{augustin69} and \cite[Thm. 31.4]{augustin78}
are \(e^{-\bigo{\sqrt{\blx}}}\).
One of our main contributions is establishing the SPB with a polynomial prefactor
for a large class of product channels. 

Using the list decoding variant of Gallager's bound 
\cite{elias57},\cite[ex 5.20]{gallager}, 
one can see that the exponential decay 
rate of the SPB, i.e. the sphere packing exponent, is tight.
But determining the right prefactor for the SPB is still an open problem 
even for the DSPCs. 
\altug and Wagner \cite{altugW11} considered the DSPCs with positive 
transition probabilities satisfying certain 
symmetry conditions \cite[p. 94]{gallager} and established a SBP
with a prefactor of the form \(\blx^{-\frac{1+\epsilon}{2\rno}}\) 
for any \(\epsilon>0\) for certain \(\rno\) in \((0,1)\).
Their result is tight because later they have proved in \cite{altugW14D} that 
Gallager's bound \cite{gallager65} 
can be improved to have a prefactor \(\blx^{-\frac{1}{2\rno}}\),
for the  aforementioned \(\rno\), for arbitrary DSPCs.
For arbitrary DSPCs, we only have bounds for the constant composition codes that are also
due to \altug and Wagner \cite{altugW14A}.

The SPB has been conjectured to hold for the channel codes on DSPCs with feedback.
Assuming certain symmetries, Dobrushin \cite{dobrushin62A} proved it to be the case.
However, it was challenging to prove the conjecture for arbitrary DSPCs with feedback 
because of the reliance of the standard proofs on the type based expurgations.
In \cite{haroutunian77}, Haroutunian established a lower bound on the error probability 
of codes on arbitrary DSPCs 
with feedback; but the exponent of Haroutunian's bound is equal to the sphere packing exponent only 
for DSPCs with certain symmetries.
Haroutunian points out in \cite{haroutunian77} that his exponent is strictly larger than 
the sphere packing exponent even for the stationary binary input binary output channel 
with the following transition probability matrix
\begin{align}
\notag
\begin{bmatrix}
\sfrac{1}{2} & \sfrac{1}{2} \\
0			 & 1			
\end{bmatrix}.
\end{align}
There are other partial results \cite{comoN10,palaiyanurS10,palaiyanurthesis}
establishing the SPB for certain families of codes 
---rather than all codes--- on the DSPCs with feedback.

Augustin presented a proof sketch establishing the SPB 
for codes on arbitrary DSPCs with feedback
in \cite[Thm. 41.7]{augustin78}. 
A complete proof following Augustin's sketch can be found in 
\cite{nakiboglu19E}. 
One of our main contributions is the new derivation of the SPB for 
codes on DSPCs with feedback. 
Furthermore, our result 
holds for non-stationary and non-periodic DPCs with feedback 
under an appropriate stationarity hypothesis, 
see Assumption \ref{assumption:astationary} and Theorem \ref{thm:fDPCexponent}.

Few years after Augustin's manuscript \cite{augustin78}, Sheverdyaev suggested another proof  in \cite{sheverdyaev82}. 
Sheverdyaev's proof, however, is supported rather weakly at certain critical 
points.
Palaiyanur's thesis \cite[A7]{palaiyanurthesis} includes an in depth discussion of the subtleties
of \cite{sheverdyaev82}.
It is worth mentioning that 
Sheverdyaev has two major claims about DSPCs with feedback in \cite{sheverdyaev82}.
Our reservations are for the claim about the SPB. 
	Sheverdyaev proves the claim about the strong converse satisfactorily
	and demonstrates that the exponential decay rate of 
	the probability of successful transmission is not changed 
	with the availability of the feedback, for rates above capacity. 
	Earlier that year, \csiszar and \korner published \cite{csiszarK82}, which establishes 
	the same result. 
	The result in question was also reported by Augustin in \cite[Thm. 41.3]{augustin78}, as 
	\csiszar and \korner pointed out  in \cite{csiszarK82}.

In the rest of this section, we describe our notation, model, and contributions.
In \S\ref{sec:notation}, we describe the notion we use throughout the article.
In \S\ref{sec:model}, we define the channel coding problem, product channels,
stationarity, memorylessness, and product channels with feedback.
In \S\ref{sec:contributions}, we provide an overview of the article
and our main contributions.

\subsection{The Notation}\label{sec:notation}
We denote the set of all reals by \(\reals{}\), positive reals by \(\reals{+}\), 
non-negative reals by \(\reals{\geq0}\), and integers by \(\integers{}\). 
For any \(\dinp\in\reals{}\), \(\lfloor\dinp\rfloor\) is the greatest integer less than or equal to \(\dinp\) 
and \(\lceil\dinp\rceil\) is the least integer greater than or equal to  \(\dinp\).
We call \((-\infty,\!\infty]\) valued functions continuous  if they satisfy 
the topological definition of continuity for 
the order topology on \((-\infty,\!\infty]\).

For any set \(\outS\), we denote the set of all subsets of \(\outS\),
 i.e. the power set of \(\outS\), by \(\sss{\outS}\)
and the set of all probability mass functions that are non-zero only 
on finitely many members of \(\outS\) by \(\pdis{\outS}\).
We call the set of all \(\dout\)'s 
for which \(\mP(\dout)>0\) the support of \(\mP\) 
and denote it by \(\supp{\mP}\).
Let \(\inpS\) be another set; then we denote the set of all 
functions from \(\inpS\) to \(\outS\) by \(\outS^{\inpS}\).

For any measurable space \((\outS,\outA)\), we denote 
the set of all finite signed measures by \(\smea{\outA}\),
the set of all non-zero finite measures by \(\fmea{\outA}\),
and 
the set of all probability measures by  \(\pmea{\outA}\). 
For any pair of measurable spaces \((\inpS,\inpA)\) and \((\outS,\outA)\),
we denote the set of all \((\inpA,\outA)\)-measurable functions
from \(\inpS\) to \(\outS\) by \(\outA^{\inpA}\)
and 
the set of all transition probabilities  from \((\inpS,\inpA)\) to \((\outS,\outA)\) 
by \(\pmea{\outA|\inpA}\).
The formal definition of a transition probability is as follows.  
\begin{definition}\label{def:transitionprobability}  
Let \((\inpS,\inpA)\) and \((\outS,\outA)\) be measurable  spaces. 
Then a function \(\Wm:\outA\times\inpS\to[0,1]\)
is called a transition probability (stochastic kernel, Markov kernel) 
from \((\inpS,\inpA)\) to \((\outS,\outA)\) if it satisfies the following two constraints: 
\begin{enumerate}[(i)]
\item For all \(\dinp\in\inpS\), the function \(\Wm(\cdot|\dinp):\outA\to[0,1]\) 
is a probability measure on \((\outS,\outA)\).
\item For all \(\oev\in\outA\), the function \(\Wm(\oev|\cdot):\inpS\to[0,1]\) 
is a \((\inpA,\rborel{[0,1]})\)-measurable function. 
\end{enumerate}
\end{definition} 
If \(\inpA\) is the power set of \(\inpS\), then the second constraint, 
i.e. the measurability constraint, is void because it is always satisfied. 
Hence, the above definition is consistent with the customary use of 
the term transition probability in information theory, in which \(\inpS\) and \(\outS\) 
are finite sets and \(\inpA\) and \(\outA\) are their power sets. 
Recall that for any transition probability \(\Wm\),
\(\mP \mtimes \Wm\) defines a joint probability measure with desired properties 
for all probability measures \(\mP\) on \((\inpS,\inpA)\) by
 \cite[Thm. 10.7.2]{bogachev}.

A measure \(\mean\) on \((\outS,\outA)\) is  absolutely continuous with respect to
another measures  \(\rfm\) on \((\outS,\outA)\), i.e. \(\mean\AC \rfm\),
iff \(\mean(\oev)=0\) for any \(\oev \in \outA\) such that \(\rfm(\oev)=0\).

We denote the integral of a measurable function \(\fX\) on \((\outS,\outA)\) 
with respect to a probability measure \(\rfm\!\in\!\pmea{\outA}\), 
i.e. the expected value of \(\fX\) under \(\rfm\),
by \(\EXS{\rfm}{\fX}\) or \(\EXS{\rfm}{\fX(\out)}\).
If the integral is on the real line and with respect to the Lebesgue measure, 
we denote it by \(\int\fX \dif{\dout}\) or \(\int \fX(\dout) \dif{\dout}\), as well.

When discussing the convergence of sequences of functions,
we denote the \(\rfm-\)almost everywhere convergence by \(\xrightarrow{\rfm-a.e.}\),
the convergence in measure for \(\rfm\)    by \(\xrightarrow{\rfm}\) and
the convergence in variation, i.e. \(\Lon{\rfm}\) convergence, by \(\xrightarrow{\Lon{\rfm}}\). 

Our notation will be overloaded for certain symbols; but the relations represented 
by these symbols will be clear from the context.
We denote the product of topologies, 
\(\sigma\)-algebras, and measures  by \(\otimes\).
We denote the Cartesian product  of sets  by \(\times\). 
We use the short hand 
\(\inpS_{\tin}^{\blx}\) for the Cartesian product of sets \(\inpS_{\tin},\ldots,\inpS_{\blx}\)
and 
\(\outA_{\tin}^{\blx}\) for the product of the \(\sigma\)-algebras  \(\outA_{\tin},\ldots,\outA_{\blx}\).
We use \(\abs{\cdot}\) to denote the absolute value of reals and the size of sets. 

The sign \(\leq\) stands for the usual less than or equal to relation for reals
and the corresponding pointwise inequality for functions. 
For \(\mean\) and \(\rfm\) in \(\smea{\outA}\), 
\(\mean \leq \rfm\) iff \(\mean(\oev)\leq \rfm(\oev)\) for all \(\oev\in\outA\).

The minimum of reals  \(\dinp\) and \(\dout\) 
is denoted by \(\dinp\wedge\dout\).
For the real valued functions \(\fX\) and \(\gX\),
\(\fX\wedge\gX\) stands for their pointwise minimum.
We use the symbol \(\vee\) analogously to \(\wedge\); but we represent maxima and 
suprema with it, rather than minima and infima.

\subsection[The Model]{The Channel Model and Channel Coding Problem}\label{sec:model}
A channel code is a strategy to convey from the transmitter at the input of the channel 
to the receiver at the output of the channel, a random choice from a finite message set. 
Once the transmitter and receiver agree on a strategy, the transmitter is given an element 
of the message set, i.e. the message. 
Then the transmitter chooses the channel input, according to the strategy, using the message. 
The channel input determines the probabilistic behavior of the channel output. 
The receiver observes the realization of the channel output and then chooses the decoded 
list based on the channel output, according to the strategy. 
If the message given to the transmitter is in the decoded list determined by the receiver, 
then the transmission is successful, else an error is said to occur. 
Let us proceed with the formal definitions of these concepts.

\begin{definition}\label{def:channel}
A \emph{channel} \(\Wm\) is a function from \emph{the input set} \(\inpS\) 
to the set of all probability measures on \emph{the output space} \((\outS,\outA)\), 
i.e.
\begin{align}
\notag
\Wm:\inpS\to\pmea{\outA}.
\end{align}
\(\outS\) is called the output set and 
\(\outA\) is called the \(\sigma\)-algebra  of the output events. 
A channel \(\Wm\) is a \emph{discrete channel} if both 
\(\inpS\) and \(\outA\) are finite sets.
\end{definition}

We denote the set of all channels with the input set \(\inpS\)
and the output space \((\outS,\outA)\) by \(\pmea{\outA|\inpS}\).
For the purposes of the channel coding problem,  
Definition \ref{def:channel} suffices.
However, while analyzing other information transmission problems
---such as the joint source channel coding problem--- 
one needs to introduce a \(\sigma\)-algebra 
\(\inpA\) on \(\inpS\) and work with the transition 
probabilities, described in Definition \ref{def:transitionprobability}.
Note that every transition probability is a channel,
i.e. \(\pmea{\outA|\inpA}\!\subset\!\pmea{\outA|\inpS}\) for all 
\(\sigma\)-algebras \(\inpA\). 
The converse statement holds only for \(\inpA\!=\!\sss{\inpS}\).

Definition \ref{def:channel} describes the channel as introduced 
in the first paragraph of this subsection accurately and 
it subsumes a diverse collection of channels as special cases.
However, it is not an all-encompassing definition because it might not be possible 
to model the effect of the channel input on the channel output solely by the 
probabilistic rule of the channel output. 
The compound channels and the arbitrarily varying channels fall 
outside of the framework of 
Definition \ref{def:channel}. Those models, however, are beyond 
the scope of this article.

\begin{definition}\label{def:code}
An \((M,L)\) \emph{channel code} on 
\(\Wm:\inpS\to\pmea{\outA}\) is an ordered pair \((\enc,\dec)\) composed 
of  an \emph{encoding function}  \(\enc\) and a \emph{decoding function} \(\dec\):
\begin{itemize}
\item An \emph{encoding function} is a function from the message set \(\mesS\DEF\{1,2,\ldots,M\}\) 
to the input set \(\inpS\).
\item A \emph{decoding function} is a measurable function from the output space \((\outS,\outA)\) 
to \(\estS\!\DEF\!\{\set{L}\!:\!\set{L}\subset\mesS \mbox{~and~}\abs{\set{L}}\leq L\}\)
with its power set  \(\sss{\estS}\) as the \(\sigma\)-algebra.
\end{itemize}
In an \((M,L)\) channel code, \(M\) is called the message set size and \(L\) is called the list size.
\end{definition}
The channel codes are customarily defined with the tacit assumption that their list size is one
and the channel codes with list sizes larger than one are customarily called list codes.
We will neither assume the list size of the codes to be one, nor use the term list code;
instead we will be explicit about the list sizes of the codes throughout the manuscript.

\begin{definition}\label{def:errorprobability}
Given an \((M,L)\) code \((\enc,\dec)\)
on \(\Wm\!:\!\inpS\!\to\!\pmea{\outA}\), for each \(\dmes\!\in\!\mesS\) 
\emph{the conditional error probability} \(\Pem{\dmes}\) is 
\begin{align}
\notag
%\label{eq:def:conditionalerrorprobability}
\Pem{\dmes}
&\DEF \EXS{\Wm(\enc(\dmes))}{\IND{\dmes\notin\dec(\out)}}.
\end{align}
\emph{The average error probability} \(\Pem{av}\) is
\begin{align}
%\label{eq:def:errorprobability}
\notag
\Pem{av} 
&\DEF\tfrac{1}{M} \sum\nolimits_{\dmes\in \mesS} \Pem{\dmes}.
\end{align}
\end{definition}
For a channel \(\Wm\!\), the triplet \((M,L,\Pe)\) is achievable if there 
exists an \((M,L)\) channel code with the average error probability less than or equal to \(\Pe\).
Broadly speaking, the point-to-point channel coding problem aims to characterize the achievable 
\((M,L,\Pe)\) triplets. The abstract formulation given above is general enough to subsume a 
diverse collection of point-to-point channel coding problems as special cases. 
However, it has scant structure to establish 
achievability and infeasibility results that 
are provably close to one another. The product structure, discussed in the following, is 
commonly assumed in order to establish such bounds.

\begin{definition}\label{def:product}
For any \(\blx\in\integers{+}\) and 
\(\Wmn{\tin}\!:\!\inpS_{\tin}\!\to\!\pmea{\outA_{\tin}}\) for \(\tin\) in \(\{1,\ldots,\blx\}\),
the \emph{length \(\blx\) product channel}
\(\Wmn{[1,\blx]}\!:\!\inpS_{1}^{\blx}\!\to\!\pmea{\outA_{1}^{\blx}}\) is defined via the following relation:
\begin{align}
\notag
%\label{eq:def:product}
\Wmn{[1,\blx]}(\dinp_{1}^{\blx})
&=\bigotimes\nolimits_{\tin=1}^{\blx}\Wmn{\tin}(\dinp_{\tin})
&
&\forall \dinp_{1}^{\blx}\in\inpS_{1}^{\blx}.
\end{align}
A product channel is \emph{stationary} iff all 
\(\Wmn{\tin}\)'s are identical.
\end{definition}
\begin{definition}\label{def:memoryless}
A channel \(\Um\!:\!\staS\!\to\!\pmea{\outA_{1}^{\blx}}\)
is a \emph{memoryless channel},
if there exits a product channel \(\Wmn{[1,\blx]}\!:\!\inpS_{1}^{\blx}\!\to\!\pmea{\outA_{1}^{\blx}}\)
satisfying 
\(\Um(\dsta)=\Wm(\dsta)\) for all 
\(\dsta\in\staS\) 
and \(\staS\subset\inpS_{1}^{\blx}\).
\end{definition}

The preceding definition is consistent with the definition of 
the memorylessness used by Cover and Thomas in \cite[p. 184]{coverthomas}:
``The channel is said to be memoryless if the probability distribution of the 
output depends only on the input at that time and is conditionally independent 
of previous channel inputs or outputs.''
The same property is asserted by Gallager in \cite[(4.2.1)]{gallager}  and 
\csiszar and \korner in \cite[p. 84]{csiszarkorner} while describing 
the memorylessness.
However, the authors of these classic texts and the information theory community 
at large use the term  ``the discrete memoryless channel (DMC)'' to describe channels 
that satisfy much more than mere discreteness and memorylessness. 
In particular, customarily the term ``the DMC'' stands 
for the DSPC described in the following paragraph. 

For any discrete channel \(\Wm:\inpS\to\pdis{\outS}\), 
\(\blx\) `independent' uses of it ---denoted by \(\Wmn{[1,\blx]}\)---
is not only a memoryless channel, but also a stationary product channel 
according to Definitions \ref{def:product} and \ref{def:memoryless}.
Thus we call these channels discrete stationary product channels (DSPCs).
If the discrete channels at each time instance are not necessarily the same, i.e. if
\(\Wmn{\tin}\) can be different for different
values of \(\tin\), then we call
\(\Wmn{[1,\blx]}\) a discrete product channel (DPC). 
Furthermore, any  \(\Um\!:\!\staS\!\to\!\pdis{\outS_{1}^{\blx}}\)
satisfying \(\Um(\dsta)\!=\!\Wmn{[1,\blx]}(\dsta)\) for all \(\dsta\) in \(\staS\)
is called a discrete memoryless channel (DMC).
A commonly considered family of DMCs is the one defined via cost constraints. 

Definitions \ref{def:product} and \ref{def:memoryless} can be applied to 
the Poisson channels.
For a duration \(\tlx\) Poisson channel  
the input set \(\fXS^{\tlx}\) is the set of all integrable 
functions of the form \(\fX\!:\!(0,\tlx]\!\to\!\reals{\geq0}\).
The output set is the set of all possible 
sample paths for the arrival process,
i.e. the set of all nondecreasing, right-continuous,
integer valued functions on \((0,\tlx]\).
The \(\sigma\)-algebra of observable events 
is the Borel \(\sigma\)-algebra for the Skorokhod metric on the output set
and \(\Pcha{\tlx}(\fX)\) is the Poisson point process with deterministic intensity function
\(\fX\) for all \(\fX\!\in\!\fXS^{\tlx}\). 
For any  duration \(\tlx\!\in\!\reals{+}\),
intensity levels \(0\!\leq\!\mA\!\leq\!\costc\!\leq\!\mB\!\leq\!\infty\),
and integrable intensity function \(\gX\) satisfying \(\gX(\tin)\!\geq\!\mA\) for all \(\tin\!\in\!(0,\tlx]\),
the Poisson channels 
\(\Pcha{\tlx,\mA,\mB,\costc}\),
\(\Pcha{\tlx,\mA,\mB,\leq\costc}\),
\(\Pcha{\tlx,\mA,\mB,\geq\costc}\),
\(\Pcha{\tlx,\mA,\mB}\),
and
\(\Pcha{\tlx,\mA,\gX(\cdot)}\)
---which are also described in  \cite[\S\ref*{A-sec:examples-poisson}]{nakiboglu19A}---
are obtained by curtailing the input set \(\fXS^{\tlx}\) 
of the Poisson channel \(\Pcha{\tlx}\)
as follows:
\begin{subequations}
\label{eq:def:poissonchannel}
\begin{align}
\label{eq:def:poissonchannel-mean}
\fXS^{\tlx,\mA,\mB,\costc}
&\DEF \{\fX\in\fXS^{\tlx}: 
\mA\leq\!\fX\!\leq \mB 
\mbox{~and~$\int_{0}^{\tlx}$}\fX  \dif{\tin}=\tlx\costc\},
\\
\label{eq:def:poissonchannel-constrained-A}
\fXS^{\tlx,\mA,\mB,\leq\costc}
&\DEF \cup_{\gamma\in[\mA,\costc] }\fXS^{\tlx,\mA,\mB,\gamma},
\\
\label{eq:def:poissonchannel-constrained-B}
\fXS^{\tlx,\mA,\mB,\geq\costc}
&\DEF \cup_{\gamma\in[\costc,\mB] }\fXS^{\tlx,\mA,\mB,\gamma},
\\
\label{eq:def:poissonchannel-bounded}
\fXS^{\tlx,\mA,\mB}
&\DEF \cup_{\gamma\in[\mA,\mB] }\fXS^{\tlx,\mA,\mB,\gamma},
\\
\label{eq:def:poissonchannel-product}
\fXS^{\tlx,\mA,\gX(\cdot)}
&\DEF \{\fX\in\fXS^{\tlx}: 
\mA\leq \fX\leq \gX\}.
\end{align}
\end{subequations}
For any  \(\tlx\!\in\!\reals{+}\) and \(\blx\!\in\!\integers{+}\),  
the Poisson channel \(\Pcha{\tlx,\mA,\mB}\) is 
a length \(\blx\) stationary product channel (SPC),
in particular \(\Pcha{\tlx,\mA,\mB}=\Wmn{[1,\blx]}\)
for \(\Wmn{\tin}=\Pcha{\sfrac{\tlx}{\blx},\mA,\mB}\). 
Whereas, the Poisson channel \(\Pcha{\tlx,\mA,\gX(\cdot)}\) is 
a length \(\blx\) product channel, 
which is stationary if \(\gX(\cdot)\) is periodic  with period \(\sfrac{\tlx}{\blx}\).
The Poisson channels 
\(\Pcha{\tlx,\mA,\mB,\varrho}\), \(\Pcha{\tlx,\mA,\mB,\leq\varrho}\), 
and \(\Pcha{\tlx,\mA,\mB,\geq\varrho}\)
are not product channels, but they are memoryless channels.

In a product channel both the input set and the output space
are products. 
In product channels with feedback, the output space is still a product;
but the input set is enlarged by allowing the channel input at any time 
instance to depend on the previous channel outputs. 
Thus the channel input at time \(\tin\) is a member of 
\({\inpS_{\tin}}^{\outS_{1}^{\tin-1}}\) rather than a member of 
\(\inpS_{\tin}\).
For channels with uncountable input or output sets, 
there are additional measurability requirements and this makes 
the description of the product channels with feedback 
more nuanced. 
Thus we will describe the discrete case first.

In a length \(\blx\) discrete product channel with feedback
each element \(\vec{\dinp_{1}^{\blx}}\) of the input set \(\vec{\inpS_{1}^{\blx}}\) is of the form 
\begin{align}
\notag
\vec{\dinp_{1}^{\blx}}&=(\dinp_{1},\enc_{2},\ldots,\enc_{\blx})
\end{align}
where \(\dinp_{1}\in\inpS_{1}\) and  
\(\enc_{\tin}\in{\inpS_{\tin}}^{\outS_{1}^{\tin-1}}\). 
We use the symbol \(\enc_{\tin}\), rather than 
\(\dinp_{\tin}\), in order reflect in our notation 
the fact that \(\enc_{\tin}\) is a function from
\(\outS_{1}^{\tin-1}\) to  \(\inpS_{\tin}\),
similar to the encoding functions we have discussed
in Definition \ref{def:code}.

\begin{definition}\label{def:Fproduct:discrete}
For any \(\blx\in\integers{+}\) and 
\(\Wmn{\tin}\!:\!\inpS_{\tin}\!\to\!\pdis{\outS_{\tin}}\) for \(\tin\) 
in \(\{1,\ldots,\blx\}\),
the \emph{length \(\blx\) discrete product channel with feedback}
\(\Wmn{\vec{[1,\blx]}}\!:\!\vec{\inpS_{1}^{\blx}}\!\to\!\pdis{\outS_{1}^{\blx}}\) 
is defined via the following relation:
\begin{align}
%\label{eq:def:Fproduct:discrete}
\notag
\Wmn{\vec{[1,\blx]}}(\dout_{1}^{\blx}|\vec{\dinp_{1}^{\blx}}) 
&\!=\!
\Wmn{1}(\dout_{1}|\dinp_{1})
\prod\nolimits_{\tin=2}^{\blx}
\Wmn{\tin}(\dout_{\tin}|\enc_{\tin}(\dout_{1}^{\tin-1})) 
\end{align}
for all \(\vec{\dinp_{1}^{\blx}}\!\in\!\vec{\inpS_{1}^{\blx}}\) 
and \(\dout_{1}^{\blx}\!\in\!\outS_{1}^{\blx}\)
where \(\vec{\inpS_{1}^{\blx}}\!=\!\inpS_{1}\bigtimes 
(\bigtimes_{\tin=2}^{\blx}{\inpS_{\tin}}^{\outS_{1}^{\tin-1}})\).
\end{definition}
For describing the product channels with feedback without assuming the
discreteness, we use the concept of transition probability described in 
Definition \ref{def:transitionprobability}.
Let \(\enc:\staS\to\inpS\) be a \((\staA,\inpA)\)-measurable function,
\(\Wm\) be a transition probability from \((\inpS,\inpA)\) to \((\outS,\outA)\),
and
\(\Wm\circ\enc:\outA\times\staS\to[0,1]\) be
\begin{align}
\notag
\Wm\circ\enc(\oev|\dsta)
&\DEF \Wm(\oev|\enc(\dsta))
&
&\forall \oev\in\outA, \dsta\in\staS.
\end{align}
Then \(\Wm\circ\enc\) is a transition probability from 
\((\staS,\staA)\) to \((\outS,\outA)\) (i.e. \(\Wm\circ\enc\in \pmea{\outA|\staA}\))
as a result of the definitions of the measurability and the transition probabilities. 
On the other hand, 
there exists a unique probability measure \(\mP\mtimes\Um\) 
for any \(\Um\in\pmea{\outA|\staA}\) and \(\mP\in\pmea{\staA}\)
by \cite[Thm. 10.7.2]{bogachev}.
Using these two observation we can generalize Definition \ref{def:Fproduct:discrete}
as follows.

\begin{definition}\label{def:Fproduct}
For any channel \(\Wmn{1}\in\pmea{\outA_{1}|\inpS_{1}}\)
and transition probabilities 
\(\Wmn{\tin}\in\pmea{\outA_{\tin}|\inpA_{\tin}}\)
for \(\tin\!\in\!\{2,\ldots,\blx\}\),
the \emph{length \(\blx\) product channel with feedback}
\(\Wmn{\vec{[1,\blx]}}\) is defined via the following 
relation:
	\begin{align}
%	\label{eq:def:Fproduct}
	\notag
	\Wmn{\vec{[1,\blx]}}(\vec{\dinp_{1}^{\blx}}) 
	&\!=\!
	\Wmn{1}(\dinp_{1}) 
	\mtimes (\Wmn{2}\circ\enc_{2})
	\cdots
	\mtimes (\Wmn{\blx}\circ\enc_{\blx})
	\end{align}
for all \(\vec{\dinp_{1}^{\blx}}\in\vec{\inpA_{1}^{\blx}}\)
where \(\vec{\inpA_{1}^{\blx}}\!=\!\inpS_{1}\bigtimes 
	(\bigtimes_{\tin=1}^{\blx}{\inpA_{\tin}}^{\outA_{1}^{\tin-1}})\).
	A product channel with feedback is \emph{stationary} iff all 
	\(\Wmn{\tin}\)'s are identical.
\end{definition}
\subsection{Overview and Main Contributions}\label{sec:contributions}
In \S\ref{sec:preliminary}, we review \renyi\!\!\!'s information measures
and the sphere packing exponent.

In \S\ref{sec:augustinpreliminary}, we derive preliminary results about 
Augustin's averaging scheme and tilting. 

In \S\ref{sec:product-outerbound}, we establish an asymptotic SPB
with a prefactor that is polynomial in the block length
for product channels, without assuming 
the input sets to be finite
or the channels to be stationary.
This asymptotic SPB, given in Theorem~\ref{thm:productexponent}, is 
derived by using the non-asymptotic SBP for product channels 
given in Lemma \ref{lem:spb-product},
which can be further simplified to \eqref{eq:lem:spb-product:stationary}
for SPCs.

If \(\sup_{\tin\in(0,\tlx]}\!\gX(\tin)\!\) is \(\bigo{\ln\!\tlx}\)
then the Poisson channel \(\Pcha{\tlx,\mA,\gX(\cdot)}\)
satisfies the hypothesis of Theorem \ref{thm:productexponent}.  
Without major changes, 
neither Wyner's approach in \cite{wyner88-b},
nor the approach of Burnashev and Kutoyants  in \cite{burnashevK99}
can establish the SPB for these channels. 
To the best of our knowledge, the SPB has not been proved for any non-stationary 
Poisson channel before ---except for \cite[Thm. 4.7]{augustin69} and \cite[Thm. 31.4]{augustin78},
which imply the SPB for these channels 
in the way that Theorem \ref{thm:productexponent} does, 
albeit with inferior prefactors. 
Augustin's SPBs  in \cite{augustin69} and \cite{augustin78} for product channels
are compared with our results in \S\ref{sec:comparison}.
		
In \S\ref{sec:fproduct-outerbound},
we establish an asymptotic SPB for DSPCs with feedback,
i.e. Theorem \ref{thm:fDSPCexponent}, 
by first deriving a non-asymptotic ---but parametric---
one in Lemma \ref{lem:spb-Fproduct}.
The stationarity hypothesis can be weakened significantly;
Theorem \ref{thm:fDPCexponent} establishes the SPB for 
(possibly non-stationary, non-periodic, and non-symmetric) 
DPCs with feedback 
satisfying Assumption \ref{assumption:astationary}.
Theorem \ref{thm:fDPCexponent} is the first such result to 
best of our knowledge.
Readers who are only interested Theorems \ref{thm:fDSPCexponent} 
and \ref{thm:fDPCexponent} may bypass 
\S\ref{sec:product-outerbound}.  

Proofs of Theorems \ref{thm:fDSPCexponent} and \ref{thm:fDPCexponent}
rely on the averaging and subblock ideas of Augustin \cite{augustin78}, 
Taylor's expansion idea of Sheverdyaev \cite{sheverdyaev82}, and 
the auxiliary channel method of Haroutunian \cite{haroutunian77}.
Nevertheless, they are substantially different from the proofs suggested by 
Augustin  \cite{augustin78} and Sheverdyaev \cite{sheverdyaev82}.
We compare our results with the previous results and
discuss possible extensions in \S\ref{sec:fcomparison}.
Lemmas \ref{lem:taylor} and \ref{lem:tradeoff},
presenting preliminary results,  are new,
to the best of our knowledge.
Lemma \ref{lem:tradeoff} is used to derive SPB 
for DSPCs with feedback from 
Haroutunian's bound in \S\ref{sec:haroutunian}.

In \S\ref{sec:conclusion}, we briefly discuss the novel observation 
underlying Augustin's method and generalizations our results 
to the memoryless channels.

In Appendix \ref{sec:operational}, we determine the channel capacity 
of certain sequences of channels
and provide a sufficient condition for the strong converse
via Theorem \ref{thm:capacity} which answers the question 
``What do Gallager's bound and Arimoto's bound say about the channel capacity and 
the existence of a strong converse?''
The Poisson channels described in \eqref{eq:def:poissonchannel} 
satisfy this sufficient condition.
The strong converses for Poisson channels were reported 
before in \cite{burnashevK99} and \cite{wagnerA04},
but only for the zero dark current case, i.e. \(\mA\!=\!0\) case. 
Note that
characterizing the channel capacity and the conditions for 
the existence of a strong converse, in general, 
is a separate issue that has already been addressed 
by Verd\'{u} and Han \cite{verduH94}.

%!TEX root=../main-B.tex
\section{General Preliminaries}\label{sec:preliminary} 
Our main aim in this section is to introduce the concepts 
that we use in the rest of the article.
We define \renyi\!\!'s information measures 
---i.e. the \renyi divergence, information, mean, capacity, radius, and center---
and review their properties that are relevant for our purposes
in \S\ref{sec:divergence}-\S\ref{sec:capacity}.
All of the propositions in this part of the article,
except Lemma \ref{lem:capacityFproduct} of \S\ref{sec:capacity}, 
are either from \cite{ervenH14} or from \cite{nakiboglu19A}.
We define and analyze sphere packing exponent in \S\ref{sec:SPexponent}.

\subsection{The \renyi Divergence}\label{sec:divergence}
\begin{definition}\label{def:divergence}
For any \(\rno\!\in\!\reals{+}\) and \(\mW,\mQ\!\in\!\fmea{\outA}\)
\emph{the order \(\rno\) \renyi divergence between \(\mW\) and \(\mQ\)} is
	\begin{align}
%	\label{eq:def:divergence}
	\notag
	\!\!\RD{\rno}{\mW}{\mQ}
	&\!\DEF\!\! 
	\begin{cases}
	\tfrac{1}{\rno-1}\!\ln
	\EXS{\rfm}{\left(\!\der{\mW}{\rfm}\!\right)^{\!\rno}
		\left(\!\der{\mQ}{\rfm}\!\right)^{\!1\!-\!\rno}}
	&\rno\!\neq\!1\!
	\\
	\EXS{\rfm}{\der{\mW}{\rfm}\!\left(\!\ln\der{\mW}{\rfm}\!-\!\ln \der{\mQ}{\rfm}\!\right)}
	&\rno\!=\!1
	\end{cases}
	\end{align}
	where \(\rfm\) is any probability measure satisfying \(\mW\AC\rfm\) and \(\mQ\AC\rfm\).
\end{definition}
The \renyi divergence is usually defined for the probability measures; 
the inclusion of the finite measures allows us to invoke 
Lemma \ref{lem:divergence-RM} given in the following.
The propositions derived for the usual definition 
will suffice for our purposes most of the time. 
Thus we borrow them from the recent article of van Erven 
and \harremoes \cite{ervenH14}. 
The equivalence of Definition \ref{def:divergence} and the one used by 
van Erven and \harremoes in \cite{ervenH14} for probability measures 
follows from \cite[Thm. 5]{ervenH14}. 

\begin{lemma}[\!\!\!{\cite[Thm. 3, Thm. 7]{ervenH14}}]\label{lem:divergence-order}
	For all \(\mW,\mQ\in\pmea{\outA}\),
	\(\RD{\rno}{\mW}{\mQ}\) is a
	nondecreasing and lower  semicontinuous function of \(\rno\) on 
	\(\reals{+}\)
	that is continuous on \((0,(1\vee\chi_{\mW,\mQ})]\) where
	\(\chi_{\mW,\mQ}\DEF\sup\{\rno:\RD{\rno}{\mW}{\mQ}<\infty\}\).
\end{lemma}

\begin{lemma}\label{lem:divergence-RM}
	Let \(\mW\), \(\mQ\), \(\mV\) be non-zero finite measures on \((\outS,\outA)\)
	and \(\rno\) be any order in \(\reals{+}\).
	\begin{itemize}
		\item If \(\mV\leq\mQ\), then \(\RD{\rno}{\mW}{\mQ}\leq \RD{\rno}{\mW}{\mV}\).
\item  If \(\mQ=\gamma\mV\) for some \(\gamma\in\reals{+}\)
and 
either \(\mW\) is a probability measure
or \(\rno\neq1\), then 
\(\RD{\rno}{\mW}{\mQ}=\RD{\rno}{\mW}{\mV}-\ln \gamma\).
	\end{itemize} 
\end{lemma}
Lemma \ref{lem:divergence-RM} is an immediate consequence of Definition \ref{def:divergence}.

Let \(\mW\) and \(\mQ\) be two probability measures on the measurable space \((\outS,\outA)\)
and \(\alg{G}\) be a sub-\(\sigma\)-algebra of \(\outA\).
Then the identities  \(\wmn{|\alg{G}}(\oev)=\mW(\oev)\) for all \(\oev\in\alg{G}\) 
and \(\qmn{|\alg{G}}(\oev)=\mQ(\oev)\) for all \(\oev\in\alg{G}\) uniquely define 
probability measures \(\wmn{|\alg{G}}\) and \(\qmn{|\alg{G}}\) on \((\outS,\alg{G})\).
In the following, we denote \(\RD{\rno}{\wmn{|\alg{G}}}{\qmn{|\alg{G}}}\) by 
\(\RDF{\rno}{\alg{G}}{\mW}{\mQ}\). 
\begin{lemma}[\!\!\!{\cite[Thm. 9]{ervenH14}}]\label{lem:divergence-DPI}
	For any order \(\rno\) in \(\reals{+}\), 
	probability measures \(\mW\) and \(\mQ\) on \((\outS,\outA)\),
	and sub-\(\sigma\)-algebra \(\alg{G}\subset\outA\) 
	\begin{align}
	\notag
	\RD{\rno}{\mW}{\mQ}
	&\geq\RDF{\rno}{\alg{G}}{\mW}{\mQ}. 
	\end{align}
\end{lemma}

\begin{lemma}[\!\!\!{\cite[Thm. 3, Thm. 31]{ervenH14}}]\label{lem:divergence-pinsker}
	For any order \(\rno\) in \(\reals{+}\)
	and probability measures \(\mW\) and \(\mQ\) on \((\outS,\outA)\)
	\begin{align}
	\label{eq:lem:divergence-pinsker}
	\RD{\rno}{\mW}{\mQ}
	&\geq\tfrac{1\wedge\rno}{2} \lon{\mW-\mQ}^2.
	\end{align}
\end{lemma}
For orders in \((0,1]\), the bound given in \eqref{eq:lem:divergence-pinsker} is called 
the Pinsker's inequality. For orders in \((0,1)\), it is possible to bound 
\(\RD{\rno}{\mW}{\mQ}\) from above in terms of \(\lon{\mW-\mQ}\) as well,
see \cite[(24) p. 365]{shiryaev}.
We will only need the following identity for \(\rno=\sfrac{1}{2}\) case,
see \cite[(21) p. 364]{shiryaev},
\begin{align}
\label{eq:shiryaev}
\RD{\sfrac{1}{2}}{\mean}{\mQ}
&\leq 2\ln\tfrac{2}{2-\lon{\mean-\mQ}}.
\end{align}    
\begin{lemma}[\!\!\!{\cite[Thm.  12]{ervenH14}}]\label{lem:divergence-convexity}
	For any order \(\rno\) in \(\reals{+}\), 
	the order \(\rno\) \renyi divergence is convex in its second argument for 
	probability measures, i.e.
\vspace{-.1cm}
	\begin{align}
	\notag	
	%	\label{eq:lem:divergence-convexity-second}
	\RD{\rno}{\mW}{\qmn{\beta}} 
	&\leq \beta\RD{\rno}{\mW}{\qmn{1}}+(1-\beta)\RD{\rno}{\mW}{\qmn{0}}
	\end{align}
	for all probability measure \(\mW,\qmn{0},\qmn{1}\) in \(\pmea{\outA}\) 
	and \(\beta\in(0,1)\)
	where \(\qmn{\beta}=\beta\qmn{1}+(1-\beta)\qmn{0}\).
\end{lemma}
\begin{lemma}[\!\!\!{\cite[Thm. 15]{ervenH14}}]\label{lem:divergence:lsc}
	For any \(\rno\) in \(\reals{+}\), \(\RD{\rno}{\mW}{\mQ}\) is a 
	lower semicontinuous function of the pair of probability measures 
	\((\mW,\mQ)\) in the topology of setwise convergence.
\end{lemma}

\subsection{The \renyi Information and Mean}\label{sec:information}
\begin{definition}\label{def:information}
For any \(\rno\!\in\!\reals{+}\), \(\Wm\!:\!\inpS\!\to\!\pmea{\outA}\),
and \(\mP\!\in\!\pdis{\inpS}\),
\emph{the order \(\rno\) \renyi information for the input distribution \(\mP\)} is
\begin{align}
\label{eq:def:information}
\!\!\!\RMI{\rno}{\mP}{\Wm} 
&\!\!\DEF\!\!\!
\begin{cases}
\tfrac{\rno}{\rno-1}\!\ln\! 
\EXS{\rfm}{\!\left[\sum\nolimits_{\dinp}\mP(\dinp)\!\left[\!\der{\Wm(\dinp)}{\rfm}\!\right]^{\rno}\right]^{\sfrac{1}{\rno}}}
&\rno\!\neq\!1
\\
\sum\nolimits_{\dinp} \mP(\dinp) 
\EXS{\rfm}{\der{\Wm(\dinp)}{\rfm}\ln \der{\Wm(\dinp)}{\qmn{1,\mP}}}
&\rno\!=\!1
\end{cases}.
\end{align} 
where \(\rfm\) is any probability measure satisfying \(\qmn{1,\mP}\AC\rfm\)
for \(\qmn{1,\mP}\in\pmea{\outA}\) defined as \(\qmn{1,\mP}\DEF\sum_{\dinp}\mP(\dinp)\Wm(\dinp)\).
\end{definition} 
We call \(\qmn{1,\mP}\) the order one \renyi mean for the input distribution \(\mP\). 
For other positive real orders,
\emph{the order \(\rno\) \renyi mean for the input distribution \(\mP\)}, is defined 
via its Radon-Nikodym derivative as follows:
\begin{align}
\label{eq:def:mean}
\der{\qmn{\rno,\mP}}{\rfm}
&\DEF \tfrac{1}{\kappa}\left(\sum\nolimits_{\dinp}\mP(\dinp)\left(\der{\Wm(\dinp)}{\rfm}\right)^{\rno}\right)^{\sfrac{1}{\rno}}
\end{align}
where \(\qmn{1,\mP}\AC\rfm\) and 
\(\kappa=\EXS{\rfm}{\left(\sum\nolimits_{\dinp}\mP(\dinp)\left(\der{\Wm(\dinp)}{\rfm}\right)^{\rno}\right)^{\sfrac{1}{\rno}}}\).

\begin{remark}
	Gallager's functions \(E_{0}(\rng,\mP)\) can be written in terms of 
	the \renyi information as follows: 
	\begin{align}
	\notag
	E_{0}(\rng,\mP)&=\rng \RMI{\frac{1}{1+\rng}}{\mP}{\Wm}
	&
	&\forall\rng\in (-1,\infty).
	\end{align}
\end{remark}
One can confirm the following identity by substitution 
\begin{align}
\label{eq:sibson}
\!\RD{\rno}{\mP \mtimes \Wm}{\!\mP\!\otimes\!\mQ} 
&\!=\!
\RD{\rno}{\mP \mtimes \Wm}{\!\mP\!\otimes\!\qmn{\rno,\mP}} 
\!+\!\RD{\rno}{\qmn{\rno,\mP}}{\mQ} 
\end{align}
for all \(\rno\!\in\!\reals{+}\), \(\mP\!\in\!\pdis{\inpS}\), \(\mQ\!\in\!\fmea{\outA}\)
where \(\mP \mtimes \Wm\) is the probability measure on \(\sss{\supp{\mP}} \otimes \outA\) whose marginal distribution on 
\(\supp{\mP}\) is \(\mP\) and whose conditional distribution is \(\Wm(\dinp)\).
Using \eqref{eq:sibson} together with Lemma \ref{lem:divergence-pinsker} one 
obtains the following alternative characterization of the \renyi information.
\begin{lemma}[\!\!\!{\cite[Lemma \ref*{A-lem:information:def}]{nakiboglu19A}}]\label{lem:information:def}
For any  \(\rno\!\in\!\reals{+}\), \(\Wm\!:\!\inpS\!\to\!\pmea{\outA}\),
and \(\mP\!\in\!\pdis{\inpS}\)
\begin{align}
\label{eq:lem:information:defA}
\RMI{\rno}{\mP}{\Wm}
&=\RD{\rno}{\mP \mtimes \Wm}{\mP\otimes\qmn{\rno,\mP}}
\\
\label{eq:lem:information:defB}
&=\inf\nolimits_{\mQ\in \pmea{\outA}}\RD{\rno}{\mP \mtimes \Wm}{\mP\otimes\mQ}.
\end{align}
\end{lemma} 
\begin{remark}
	We defined the \renyi information, mean, capacity, radius, and
	center in \cite{nakiboglu19A} for subsets of 
	\(\pmea{\outA}\), rather than functions 
	from some \(\inpS\) to \(\pmea{\outA}\). 
	For functions that are one-to-one, these two approaches are describing
	same quantities with different notation.
	Thus the propositions we are borrowing from \cite{nakiboglu19A}
	are merely restated in an alternative notation.
	The functions we consider, however, are not necessarily one-to-one.
	Nevertheless, one can show easily that 
	each proposition we are borrowing from \cite{nakiboglu19A}
	for subset of \(\pmea{\outA}\) implies the corresponding proposition 
	for functions to \(\pmea{\outA}\).
\end{remark}	

\subsection[The \renyi Capacity, Radius, Center]{The \renyi Capacity, Radius, and Center}\label{sec:capacity}
\begin{definition}\label{def:capacity}
For any \(\rno\) in \(\reals{+}\) and
\(\Wm:\inpS\to\pmea{\outA}\),
\emph{the order \(\rno\) \renyi capacity of \(\Wm\)} is
\begin{align}
%\label{eq:def:capacity}
\notag
\RC{\rno}{\Wm} 
&\DEF \sup\nolimits_{\mP \in \pdis{\inpS}}  \RMI{\rno}{\mP}{\Wm}.
\end{align}
\end{definition}
\begin{remark}
\(E_{0}(\rng,\Wm)=\rng \RC{\frac{1}{1+\rng}}{\Wm}\) for all \(\rng\)
in \((-1,\infty)\) as a result of the corresponding expression for \(E_{0}(\rng,\mP)\). 
\end{remark}
%Lemma \ref{lem:capacityO} describes the \renyi capacity as a function of the order 
%for a given \(\Wm\).
\begin{lemma}[\!\!{\cite[Lemma \ref*{A-lem:capacityO}-(\ref*{A-capacityO-ilsc},\ref*{A-capacityO-zo},\ref*{A-capacityO-zofintieness},\ref*{A-capacityO-continuity})]{nakiboglu19A}}]\label{lem:capacityO}
Let \(\Wm\) be a channel of the form \(\Wm:\inpS\to\pmea{\outA}\). Then
\begin{enumerate}[(a)]
\item\label{capacityO-ilsc} \(\RC{\rno}{\Wm}\) is a nondecreasing lower semicontinuous  function of \(\rno\) on \(\reals{+}\).
\item\label{capacityO-zo} \(\tfrac{1-\rno}{\rno}\RC{\rno}{\Wm}\) is a nonincreasing continuous function of \(\rno\) on \((0,1)\)
and \(\RC{\rno}{\Wm}\) is a continuous function of \(\rno\) on \((0,1]\).
\item\label{capacityO-zofintieness} If \(\RC{\rnt}{\Wm}<\infty\) for an \(\rnt\in(0,1) \), 
then \(\RC{\rno}{\Wm}\) is finite for all \(\rno\in (0,1)\).
\item\label{capacityO-continuity} If \(\RC{\rnt}{\Wm}<\infty\) for an \(\rnt\in\reals{+}\), 
then \(\RC{\rno}{\Wm}\) is a nondecreasing continuous function of \(\rno\) on \((0,\rnt]\).
\end{enumerate}
\end{lemma}

Note that, since \(\RC{\rno}{\Wm}\) is continuous and nondecreasing in \(\rno\) on \((0,1)\)
by Lemma \ref{lem:capacityO}-(\ref{capacityO-ilsc},\ref{capacityO-zo}),
\(\RC{\rno}{\Wm}\) has a limit as \(\rno\) converges to zero from the right.
We denote this limit by \(\RC{0^{_{+}}\!}{\Wm}\):
\begin{align}
\label{eq:def:capacitylimit}
\RC{0^{_{+}}\!}{\Wm}
&\DEF\lim\nolimits_{\rno\downarrow 0} \RC{\rno}{\Wm}. 
\end{align} 
We do not denote this limit by \(\RC{0}{\Wm}\) because \(\RC{0}{\Wm}\) is, customarily, defined 
as the supremum of \(\RMI{0}{\mP}{\Wm}\). 
For the case when the input set is finite, we know that \(\RC{0}{\Wm}=\RC{0^{_{+}}\!}{\Wm}\),
see \cite[Lemma \ref*{A-lem:finitecapacity}-(\ref*{A-finitecapacity-fc})]{nakiboglu19A}.
Unfortunately, we do not have a general result establishing this equality for arbitrary 
channels. 

Note, on the other hand that,
Lemma \ref{lem:capacityO}-(\ref{capacityO-ilsc},\ref{capacityO-zo})
implies 
\begin{align}
\label{eq:orderoneovertwo}
\hspace{-.1cm}\tfrac{\rno\wedge(1-\rno)}{1-\rno}\!\RC{\sfrac{1}{2}}{\Wm}
\!\leq\!\RC{\rno}{\Wm}
&\!\leq\!\tfrac{\rno\vee(1-\rno)}{1-\rno}\!\RC{\sfrac{1}{2}}{\Wm}
&
&\forall\rno\!\in\!(0,1).
\end{align}

For all positive real orders \(\rno\), 
the alternative characterization of the order \(\rno\) \renyi information 
given in Lemma \ref{lem:information:def} implies the following 
alternative expression for the order \(\rno\) \renyi capacity 
\begin{align}
%\label{eq:capacity}
\notag
\RC{\rno}{\Wm}
&=\sup\nolimits_{\mP \in \pdis{\inpS}}\inf\nolimits_{\mQ\in\pmea{\outA}} \RD{\rno}{\mP \mtimes \Wm}{\mP\otimes\mQ}.
\end{align}
In the preceding  expression, the order of the supremum and infimum can be changed without changing 
the value of the expression.
\begin{theorem}[\!\!\!{\cite[Thms. \ref*{A-thm:minimax},\! \ref*{A-thm:Gminimax}]{nakiboglu19A}}]\label{thm:minimax}
\!For any \(\rno\!\in\!\reals{\!}\) and \(\!\Wm\!:\!\inpS\!\to\!\pmea{\outA}\!\)
\begin{align}
\label{eq:thm:minimax:capacity}
\RC{\rno}{\Wm}
&=
\sup\nolimits_{\mP \in \pdis{\inpS}} \inf\nolimits_{\mQ \in \pmea{\outA}} \RD{\rno}{\mP \mtimes \Wm}{\mP \otimes \mQ}
\\
\label{eq:thm:minimax}
&=
\inf\nolimits_{\mQ \in \pmea{\outA}} \sup\nolimits_{\mP \in \pdis{\inpS}} \RD{\rno}{\mP \mtimes \Wm}{\mP \otimes \mQ}
\\
\label{eq:thm:minimaxradius}
&=\inf\nolimits_{\mQ\in\pmea{\outA}}\sup\nolimits_{\dinp \in \inpS} \RD{\rno}{\Wm(\dinp)}{\mQ}.
\end{align}
If \(\RC{\rno}{\Wm}<\infty\), then there exists a unique \(\qmn{\rno,\Wm}\in\pmea{\outA}\),
called the order \(\rno\) \renyi center, such that
\begin{align}
\label{eq:thm:minimaxcenter}
\RC{\rno}{\Wm}
&=\sup\nolimits_{\mP \in \pdis{\inpS}} \RD{\rno}{\mP \mtimes \Wm}{\mP \otimes \qmn{\rno,\Wm}}
\\
\label{eq:thm:minimaxradiuscenter}
&=\sup\nolimits_{\dinp \in \inpS} \RD{\rno}{\Wm(\dinp)}{\qmn{\rno,\Wm}}.
\end{align}
Furthermore, for every countably separated
\(\sigma\)-algebra \(\inpA\) of subsets of \(\inpS\)
satisfying \(\Wm\!\in\!\pmea{\outA|\inpA}\),
the suprema over \(\pdis{\inpS}\) in 
\eqref{eq:thm:minimax:capacity}, \eqref{eq:thm:minimax}, and
\eqref{eq:thm:minimaxcenter}  can be replaced by suprema over 
\(\pmea{\inpA}\).
\end{theorem}

The right hand side of \eqref{eq:thm:minimaxradius} can be interpreted as a radius;
it is, in fact, the definition of the order \(\rno\) \renyi radius:
\begin{definition}
For any \(\rno\!\in\!\reals{+}\), \(\Wm\!:\!\inpS\!\to\!\pmea{\outA}\),
and \(\mQ\!\in\!\pmea{\outA}\)
\emph{the order \(\rno\) \renyi radius of \(\Wm\) relative to \(\mQ\)}, i.e.
\(\RRR{\rno}{\Wm}{\mQ}\), and \emph{the order \(\rno\) \renyi radius of \(\Wm\!\)}, i.e.
\(\RR{\rno}{\Wm}\),  are 
\begin{align}
%\label{eq:def:relativeradius}
\notag
\RRR{\rno}{\Wm}{\mQ} 
&\DEF \sup\nolimits_{\dinp \in \inpS}  \RD{\rno}{\Wm(\dinp)}{\mQ},
\\
%\label{eq:def:radius}
\notag
\RR{\rno}{\Wm} 
&\DEF \inf\nolimits_{\mQ\in\pmea{\outA}} \sup\nolimits_{\dinp \in \inpS}  \RD{\rno}{\Wm(\dinp)}{\mQ}.
\end{align}
\end{definition}
Hence, \eqref{eq:thm:minimaxradiuscenter} allows us to interpret \(\qmn{\rno,\Wm}\) as a center. 
This is why \(\qmn{\rno,\Wm}\) is called the order \(\rno\) \renyi center. 

The following bound on \(\RRR{\rno}{\Wm}{\mQ}\) is called
the van Erven-\harremoes bound.
\begin{lemma}[\!\!\!{\cite[Lemma \ref*{A-lem:EHB}]{nakiboglu19A}}]\label{lem:EHB} 
If \(\RC{\rno}{\Wm}<\infty\) for an \(\rno \in \reals{+}\)
and \(\Wm:\inpS\to\pmea{\outA}\), then 
\begin{align}
%\label{eq:lem:EHB}
\notag
\RC{\rno}{\Wm}+\RD{\rno}{\qmn{\rno,\Wm}}{\mQ}
&\leq \RRR{\rno}{\Wm}{\mQ} 
&
&\forall\mQ \in \pmea{\outA}.
\end{align}
\end{lemma}
The van Erven-\harremoes bound can be used to establish the continuity of 
the \renyi center as a function of the order for the total variation 
topology on \(\pmea{\outA}\).
\begin{lemma}[\!\!\!{\cite[Lemma \ref*{A-lem:centercontinuity}]{nakiboglu19A}}]\label{lem:centercontinuity} 
If \(\RC{\rnt}{\Wm}<\infty\) for an \(\rnt \in \reals{+}\)
and \(\Wm:\inpS\to\pmea{\outA}\), then 
\begin{align}
%\label{eq:lem:centercontinuity}
\notag
\RD{\rno}{\qmn{\rno,\Wm}}{\qmn{\rnf,\Wm}} 
&\leq \RC{\rnf}{\Wm}-\RC{\rno}{\Wm}
\end{align}
for any \(\rno\) and \(\rnf\) 
satisfying \(0<\rno<\rnf\leq\rnt\).
Furthermore, \(\qmn{\rno,\Wm}\) is a continuous function of \(\rno\) on \((0,\rnt]\) 
for the total variation topology on \(\pmea{\outA}\).
\end{lemma}

A well known fact about the DSPCs is that their \renyi capacities  are additive, 
see \cite[Thm. 5]{gallager65}, \cite[(5.6.59)]{gallager}. 
%\cite[pp. 149-150, (5.6.59)]{gallager}. 
In fact, the additivity of the \renyi capacities holds for arbitrary product channels.
Furthermore, whenever it exists, the \renyi center of a product channel is equal to 
the product of the \renyi centers of its component channels.
Lemma \ref{lem:capacityProduct} states these observations formally.
\begin{lemma}[\!\!\!{\cite[Lemma \ref*{A-lem:capacityProduct}]{nakiboglu19A}}]
\label{lem:capacityProduct}
Any length \(\blx\) product channel
\(\Wmn{[1,\blx]}\!:\!\inpS_{1}^{\blx}\!\to\!\pmea{\outA_{1}^{\blx}}\) 
satisfies
\begin{align}
\label{eq:lem:capacityProduct}
\RC{\rno}{\Wmn{[1,\blx]}}
&=\sum\nolimits_{\tin=1}^{\blx}\RC{\rno}{\Wmn{\tin}}
&
&\forall \rno\in \reals{+}.
\end{align}
Furthermore, if \(\RC{\rno}{\Wmn{[1,\blx]}}\) is finite for an order \(\rno\in\reals{+}\), 
then \(\qmn{\rno,\Wmn{[1,\blx]}}=\bigotimes_{\tin=1}^{\blx}\qmn{\rno,\Wmn{\tin}}\).
\end{lemma}
 
Note that the input set of a product channel is a subset of the input set
of the corresponding product channel with feedback.
An immediate consequence of this observation is that 
\(\RC{\rno}{\Wmn{\vec{[1,\blx]}}}\geq\RC{\rno}{\Wmn{[1,\blx]}}\).
More interestingly, the reverse inequality \(\RC{\rno}{\Wmn{\vec{[1,\blx]}}}\leq\RC{\rno}{\Wmn{[1,\blx]}}\) is also true.

\begin{lemma}\label{lem:capacityFproduct}
Any length \(\blx\) product channel with feedback
\(\Wmn{\vec{[1,\blx]}}\!:\!\vec{\inpA_{1}^{\blx}}\!\to\!\pmea{\outA_{1}^{\blx}}\) 
with countably separated \(\sigma-\)algebras \(\inpA_{2},\ldots,\inpA_{\blx}\)
satisfies.
\vspace{-.1cm}
\begin{align}
\label{eq:lem:capacityFproduct}
\RC{\rno}{\Wmn{\vec{[1,\blx]}}}
&=\sum\nolimits_{\tin=1}^{\blx}\RC{\rno}{\Wmn{\tin}}
&
&\forall \rno\in\reals{+}.
\end{align}
\vspace{-.1cm}
Furthermore, if \(\RC{\rno}{\Wmn{\vec{[1,\blx]}}}\) is finite for an order \(\rno\in\reals{+}\), 
then \(\qmn{\rno,\Wmn{\vec{[1,\blx]}}}\!=\!\bigotimes_{\tin=1}^{\blx}\qmn{\rno,\Wmn{\tin}}\).
\end{lemma}
For the case when the component channels are discrete,
Lemma \ref{lem:capacityFproduct} has been common 
knowledge among the researcher working on the error exponents 
with feedback for some time now.
Augustin \cite[pp. 304-306]{augustin78}  mentions the following
equivalent claim
without a proof for the case when input sets of \(\Wmn{\tin}\) are finite:
\begin{quotation}
``\(e^{\frac{\rno-1}{\rno}\RC{\rno}{\Wmn{\vec{[1,\blx]}}}}\qmn{\rno,\Wmn{\vec{[1,\blx]}}}
=e^{\frac{\rno-1}{\rno}\RC{\rno}{\Wmn{[1,\blx]}}}\qmn{\rno,\Wmn{[1,\blx]}}\) 
whenever  \(\qmn{\rno,\Wmn{[1,\blx]}}\) is defined.''
\end{quotation}

\begin{proof}[Proof of Lemma \ref{lem:capacityFproduct}]
We prove the lemma for \(\rno\in \reals{+}\setminus \{1\}\) in the following. 
This implies \(\RC{\rno}{\Wmn{\vec{[1,\blx]}}}=\sum_{\tin=1}^{\blx}\RC{\rno}{\Wmn{\tin}}\) for all 
\(\rno\in\reals{+}\), because 
the \renyi capacity is a nondecreasing lower semicontinuous function of the order by 
Lemma \ref{lem:capacityO}-(\ref{capacityO-ilsc}). 
The claim about the \renyi centers follows from the corresponding claim 
in Lemma \ref{lem:capacityProduct} and the uniqueness of the \renyi centers, 
established in Theorem \ref{thm:minimax}.

Recall that \(\Wmn{[1,\blx]}(\dinp_{1}^{\blx})\!=\!\Wmn{\vec{[1,\blx]}}(\dinp_{1}^{\blx})\)
for all \(\dinp_{1}^{\blx}\) in \(\inpS_{1}^{\blx}\),
which is a subset of \(\vec{\inpA_{1}^{\blx}}\).
Then Lemma \ref{lem:capacityProduct} implies that 
\(\sum_{\tin=1}^{\blx}\RC{\rno}{\Wmn{\tin}}\leq \RC{\rno}{\Wmn{\vec{[1,\blx]}}}\).
Hence, \eqref{eq:lem:capacityFproduct} holds if
\(\RC{\rno}{\Wmn{\tin}}\) is infinite for a \(\tin\in\{1,\ldots,\blx\}\).
Thus, we assume that \(\RC{\rno}{\Wmn{\tin}}\) is finite for all \(\tin\) 
for the rest of the proof. Then for each \(\tin\), 
\(\Wmn{\tin}\) has a unique \renyi center \(\qmn{\rno,\Wmn{\tin}}\) by 
Theorem \ref{thm:minimax}.

On the other hand, 
Definition \ref{def:Fproduct} implies that
for every \(\vec{\dinp_{1}^{\blx}}\) in \(\vec{\inpA_{1}^{\blx}}\) there exists 
an \(\dinp_{1}\in\inpS_{1}\) and \(\enc_{\tin}\in{\inpA_{\tin}}^{\outA_{1}^{\tin-1}}\)
for all  \(\tin\) in \(\{2,\ldots,\blx\}\) such that
\begin{align}
\notag
\Wmn{\vec{[1,\blx]}}(\vec{\dinp_{1}^{\blx}}) 
&\!=\!
\Wmn{1}(\dinp_{1}) 
\mtimes (\Wmn{2}\circ\enc_{2})
\cdots
\mtimes (\Wmn{\blx}\circ\enc_{\blx})
\end{align}
Note that 
\(\Wmn{\tin}\circ\enc_{\tin}\in\pmea{\outA_{\tin}|\outA_{1}^{\tin-1}}\) 
by definition.
Furthermore, Theorem \ref{thm:minimax} implies for all \(\dout_{1}^{\tin-1}\!\in\!\outS_{1}^{\tin-1}\) that
\begin{align}
\notag
\RD{\rno}{\Wmn{1}(\dinp_{1})}{\qmn{\rno,\Wmn{1}}}
&\leq \RC{\rno}{\Wmn{1}}
\\
\notag
\RD{\rno}{\Wmn{\tin}\circ\enc_{\tin}(\cdot|\dout_{1}^{\tin-1})}{\qmn{\rno,\Wmn{\tin}}}
&\leq \RC{\rno}{\Wmn{\tin}}.
\end{align}
Then using \cite[Thm. 10.7.2]{bogachev} we get
\begin{align}
\notag
&\hspace{-.6cm}\RD{\rno}{\!\Wmn{\vec{[1,\blx]}}(\!\vec{\dinp_{1}^{\blx}}\!)}{\!\qmn{\blx}\!}
\\
\notag
&\!\leq\!
\RD{\rno}{\!\Wmn{1}(\!\dinp_{1}\!)\mtimes\!\cdots\!\mtimes(\Wmn{\blx-1}\!\circ\!\enc_{\blx-1})}{\!\qmn{\blx-1}\!}
\!+\!\RC{\rno}{\Wmn{\blx}}
\\
\notag
&\!\leq\!
\RD{\rno}{\!\Wmn{1}(\!\dinp_{1}\!)\mtimes\!\cdots\!\mtimes(\Wmn{\ind}\!\circ\!\enc_{\ind})}{\!\qmn{\ind}\!}
\!+\!\!\sum\nolimits_{\tin=\ind+1}^{\blx}\!\!\RC{\rno}{\Wmn{\tin}}
\\
\notag
&\!\leq\!\sum\nolimits_{\tin=1}^{\blx}\RC{\rno}{\Wmn{\tin}}
\end{align}
where \(\qmn{\ind}=\bigotimes\nolimits_{\tin=1}^{\ind}\qmn{\rno,\Wmn{\tin}}\)
for \(\ind\in\{1,\ldots,\blx\}\).
Hence,
\begin{align}
\notag
\sup\nolimits_{\vec{\dinp_{1}^{\blx}}\in\vec{\inpA_{1}^{\blx}}}
\RD{\rno}{\!\Wmn{\vec{[1,\blx]}}(\!\vec{\dinp_{1}^{\blx}}\!)}{\!\qmn{\blx}\!}
&\!\leq\!\sum\nolimits_{\tin=1}^{\blx}\RC{\rno}{\Wmn{\tin}}.
\end{align}
Then \(\RC{\rno}{\Wmn{\vec{[1,\blx]}}}\leq\sum_{\tin=1}^{\blx}\!\RC{\rno}{\Wmn{\tin}}\) by 
\eqref{eq:thm:minimaxradius}.
Thus  \eqref{eq:lem:capacityFproduct} holds
and
\(\qmn{\rno,\Wmn{\vec{[1,\blx]}}}\!=\!\bigotimes\nolimits_{\tin=1}^{\blx}\qmn{\rno,\Wmn{\tin}}\) 
by Theorem \ref{thm:minimax}.
\end{proof}
\subsection{The Sphere Packing Exponent}\label{sec:SPexponent} 
\begin{definition}\label{def:spherepackingexponent}
For any channel \(\!\Wm\!:\!\inpS\!\to\!\pmea{\outA}\)
and rate \(\rate\!\in\!\reals{\geq0}\!\) 
\emph{the sphere packing exponent} is
\begin{align}
%\label{eq:def:spherepackingexponent}
\notag
\spe{\rate,\Wm}
&\DEF \sup\nolimits_{\rno\in (0,1)} \tfrac{1-\rno}{\rno} \left(\RC{\rno}{\Wm}-\rate\right).
\end{align} 
\end{definition}

\begin{lemma}\label{lem:spherepackingexponent}
	For any channel \(\Wm\!:\!\inpS\!\to\!\pmea{\outA}\), \(\spe{\rate,\Wm}\) 
	is convex and nonincreasing in \(\rate\) on \(\reals{\geq0}\),
	finite on \((\RC{0^{_{+}}\!}{\Wm},\infty)\),
	and continuous on \([\RC{0^{_{+}}\!}{\Wm},\infty)\)
	for \(\RC{0^{_{+}}\!}{\Wm}\) is defined in \eqref{eq:def:capacitylimit}.
	In particular,
\vspace{-.2cm}
\begin{align}
\notag
&\hspace{-.2cm}\spe{\rate,\Wm}
\\
\label{eq:lem:spherepackingexponent}
&\!=\!\begin{cases}
~~~\infty 
&\!\rate\!<\!\RC{0^{_{+}}\!}{\Wm}
\\
\sup\limits_{\rno\in(0,1)}\!\tfrac{1-\rno}{\rno}\!\left(\RC{\rno}{\Wm}\!-\!\rate\right)
&
\!\rate\!=\!\RC{0^{_{+}}\!}{\Wm}
\\
\sup\limits_{\rno\in[\rnf,1)}\!\tfrac{1-\rno}{\rno}\!\left(\RC{\rno}{\Wm}\!-\!\rate\right)
&
\!\rate\!=\!\RC{\rnf}{\Wm}\!\mbox{~for\!~a~}\!\rnf\!\in\!(0,1) 
\\
~~~0
&
\!\rate\!\geq\!\RC{1}{\Wm}
\end{cases}
\end{align}
\end{lemma}

\begin{proof}[Proof of Lemma \ref{lem:spherepackingexponent}]
\(\spe{\rate,\Wm}\) is convex/nonincreasing in \(\rate\),  
because 
\(\tfrac{1-\rno}{\rno}(\RC{\rno}{\Wm}\!-\!\rate)\) 
is convex/nonincreasing in \(\rate\) for any \(\rno\!\in\!(0,1)\)
and
the pointwise supremum of a family of convex/nonincreasing functions 
is convex/nonincreasing.

	Recall that \(\RC{\rno}{\Wm}\) is a nondecreasing in \(\rno\) by 
	Lemma \ref{lem:capacityO}-(\ref{capacityO-ilsc}).
	\begin{itemize} 
		\item If \(\RC{0^{_{+}}\!}{\Wm}=\infty\), 
		then \(\RC{1/2}{\Wm}=\infty\) and \(\spe{\rate,\Wm}=\infty\) for all \(\rate\in\reals{\geq0}\).
		On the other hand \(\rate<\RC{0^{_{+}}\!}{\Wm}\) for all \(\rate\in\reals{\geq0}\). 
		Hence \eqref{eq:lem:spherepackingexponent} holds.
		
		\item If \(\RC{0^{_{+}}\!}{\Wm}\!<\!\infty\) and \(\RC{0^{_{+}}\!}{\Wm}\!=\!\RC{1}{\Wm}\),  
		then \(\spe{\rate,\Wm\!}\!=\!\infty\) for all \(\rate<\RC{1}{\Wm}\) 
		and \(\spe{\rate,\!\Wm\!}\!=\!0\) for all \(\rate\geq \RC{1}{\Wm}\). 
		Thus  \eqref{eq:lem:spherepackingexponent} holds.

		\item If \(\RC{0^{_{+}}\!}{\Wm}\!<\!\infty\) and \(\RC{0^{_{+}}\!}{\Wm}\!\neq\!\RC{1}{\Wm}\), 
		then \(\spe{\rate,\Wm}\!=\!\infty\) for all \(\rate<\RC{0^{_{+}}\!}{\Wm}\). 
		For \(\rate\geq \RC{0^{_{+}}\!}{\Wm}\), 
		the non-negativity of \(\tfrac{1-\rno}{\rno}(\RC{\rno}{\Wm}-\rate)\) 
		implies the restrictions given in  \eqref{eq:lem:spherepackingexponent} for different intervals.
	\end{itemize}	
\(\spe{\rate,\!\Wm}\) is continuous on \((\RC{0^{_{+}}\!}{\Wm},\!\infty)\) by  \cite[Thm. 6.3.3]{dudley}
because it is finite on \((\RC{0^{_{+}}\!}{\Wm},\!\infty)\) by \eqref{eq:lem:spherepackingexponent}
and convex on \(\reals{\geq0}\).
On the other hand \(\spe{\rate,\!\Wm}\) is lower semicontinuous because 
it is the pointwise supremum of continuous functions.
Thus it is continuous from the right because it is nonincreasing. 
Thus \(\spe{\rate,\!\Wm}\) is continuous on \([\RC{0^{_{+}}\!}{\Wm},\!\infty)\), as well.
\end{proof} 
%!TEX root=../main-B.tex
\section{Preliminaries for Augustin's Method}\label{sec:augustinpreliminary}
The propositions proved in this section are used 
in \S\ref{sec:product-outerbound} and \S\ref{sec:fproduct-outerbound} to derive SPBs.
In \S\ref{sec:averaging}, we define the average order \(\rno\) \renyi center 
\(\qma{\rno,\Wm}{\epsilon}\) as the average of the \renyi centers on a specific
length \(\epsilon\) interval around \(\rno\). 
Using the convexity and the monotonicity properties of the \renyi divergence,
we bound the order \(\rno\) \renyi radius of \(\Wm\) relative to  
\(\qma{\rno,\Wm}{\epsilon}\), i.e. \(\RRR{\rno}{\Wm}{\qma{\rno,\Wm}{\epsilon}}\),
from above and call the bound the average \renyi capacity \(\RCI{\rno}{\Wm}{\epsilon}\).
Then we show that both \(\RCI{\rno}{\Wm}{\epsilon}\) and the 
associated sphere packing exponent \(\spa{\epsilon}{\rate,\Wm}\)
differ from the corresponding quantities \(\RC{\rno}{\Wm}\) and \(\spe{\rate,\Wm}\) 
at most by a factor proportional to \(\epsilon\). 
In \S\ref{sec:tilting}, we consider the tilted probability measure between a probability measure 
\(\mW\) and a family of probability measures \(\qmn{\rno}\) that is continuous in 
\(\rno\) for the total variation topology on \(\pmea{\outA}\). 
We show that both the tilted probability measure \(\wma{\rno}{\qmn{\rno}}\)
and the \renyi divergences \(\RD{\rno}{\mW}{\qmn{\rno}}\),
\(\RD{1}{\wma{\rno}{\qmn{\rno}}}{\mW}\), and 
\(\RD{1}{\wma{\rno}{\qmn{\rno}}}{\qmn{\rno}}\) are continuous in \(\rno\)
on \((0,1)\).

\subsection{The Augustin's Averaging}\label{sec:averaging}
If the order \(\rnf\) \renyi capacity of a channel is finite for a 
\(\rnf\!\in\!\reals{+}\),
then the \renyi centers of the channel form a transition probability
from \(((0,1),\rborel{(0,1)})\) to \((\outS,\outA)\).
We define the average \renyi center using this transition probability.
In order to see why such a transition probability structure exists, 
first note that \(\RC{\rno}{\Wm\!}\) is finite for all \(\rno\!\in\!(0,1)\) by  
Lemma \ref{lem:capacityO}-(\ref{capacityO-ilsc},\ref{capacityO-zofintieness})
because \(\RC{\rnf}{\Wm}\) is finite for a \(\rnf\!\in\!\reals{+}\).
This implies the  existence of a unique order \(\rno\) \renyi center 
\(\qmn{\rno,\Wm}\) for each \(\rno\!\in\!(0,1)\) by Theorem \ref{thm:minimax}.
Furthermore, \(\qmn{\rno,\Wm}\) is continuous in \(\rno\) on \((0,1)\) 
for the total variation topology on \(\pmea{\outA}\), by Lemma \ref{lem:centercontinuity}. 
As a result, \(\qmn{\cdot,\Wm}(\oev):(0,1)\to[0,1]\) is a 
continuous and hence a \((\rborel{(0,1)},\rborel{[0,1]})\)-measurable function 
for any \(\oev\in\outA\).
\begin{remark}
The continuity for the topology of setwise convergence 
is sufficient for ensuring the continuity of \(\qmn{\rno,\Wm}(\oev)\) 
in \(\rno\) for all \(\oev\!\in\!\outA\)
and hence for ensuring the existence of the transition probability structure.	
\end{remark}

\begin{definition}\label{def:avcenter}
For any \(\rno,\epsilon\!\in\!(0,1)\) and 
\(\Wm\!:\!\inpS\!\to\!\pmea{\outA}\) satisfying 
\(\RC{\sfrac{1}{2}}{\Wm}\!<\!\infty\),
\emph{the average \renyi  center \(\qma{\rno,\Wm}{\epsilon}\)} is the \(\outS\) marginal of the probability measure  
\(\mU_{\rno,\epsilon} \mtimes \qmn{\cdot,\Wm}\) where \(\mU_{\rno,\epsilon}\) 
is the uniform probability distribution on 
\((\rno-\epsilon\rno,\rno+\epsilon(1-\rno))\):
\begin{align}
%\label{eq:def:avcenter}
\notag
\qma{\rno,\Wm}{\epsilon}
&\DEF \tfrac{1}{\epsilon} \int_{\rno-\epsilon\rno}^{\rno+\epsilon(1-\rno)} \qmn{\rnt,\Wm}  \dif{\rnt}. 
\end{align}
\end{definition}

The order \(\rno\) \renyi radius relative to the order \(\rno\) \renyi center, i.e. 
\(\RRR{\rno}{\Wm}{\qmn{\rno,\Wm}}\), is \(\RC{\rno}{\Wm}\) by Theorem \ref{thm:minimax}.
What can we say about \(\RRR{\rno}{\Wm}{\qma{\rno,\Wm}{\epsilon}}\)?
For channels with certain symmetries such as the ones 
in \cite[Examples \ref*{A-eg:shiftchannel}-\ref*{A-eg:shiftinvariant}]{nakiboglu19A}, 
\(\qmn{\rno,\Wm}\) is the same probability measure for all \(\rno\) for which it exists.
For such channels \(\RRR{\rno}{\Wm}{\qma{\rno,\Wm}{\epsilon}}\!=\!\RC{\rno}{\Wm}\) 
for all \(\rno,\epsilon\in(0,1)\)
because \(\qma{\rno,\Wm}{\epsilon}\!=\!\qmn{\rno,\Wm}\) for all \(\rno,\epsilon\in(0,1)\).
For certain other channels, such as \(\cha\) of \cite[Example \ref*{A-eg:singular-finite}]{nakiboglu19A}, 
\(\qmn{\rno,\Wm}\) is same for all \(\rno\) on an interval
and  \(\RRR{\rno}{\Wm\!}{\qma{\rno,\Wm}{\epsilon}}\!=\!\RC{\rno}{\Wm\!}\)
at least for some \(\rno\) for small enough \(\epsilon\).
However, we cannot assert the equality of \(\qmn{\rno,\Wm\!}\) 
and \(\qma{\rno,\Wm\!}{\epsilon}\) in general and  
\(\RRR{\rno}{\Wm\!}{\qma{\rno,\Wm}{\epsilon}}\!>\!\RC{\rno}{\Wm\!}\) whenever 
\(\qma{\rno,\Wm}{\epsilon}\neq\qmn{\rno,\Wm}\). 
In particular, 
\begin{align}
\notag
\RRR{\rno}{\Wm\!}{\qma{\rno,\Wm}{\epsilon}}
&\!\geq\!\RC{\rno}{\Wm}+\RD{\rno}{\qmn{\rno,\Wm}}{\qma{\rno,\Wm}{\epsilon}}
\end{align}
by Lemma \ref{lem:EHB}.
Lemma \ref{lem:avcapacity} bounds \(\RRR{\rno}{\Wm}{\qma{\rno,\Wm}{\epsilon}}\)
from above in terms  of an integral of the  \renyi capacity, which converges to \(\RC{\rno}{\Wm}\) as \(\epsilon\) 
converges to zero for  any \(\rno\!\in\!(0,1)\). 
\begin{lemma}\label{lem:avcapacity}
For any \(\rno,\epsilon\!\in\!(0,1)\) and 
\(\Wm\!:\!\inpS\!\to\!\pmea{\outA}\) satisfying 
\(\RC{\sfrac{1}{2}}{\Wm}\!<\!\infty\),
\begin{align}
\label{eq:lem:avcapacity}
\sup\nolimits_{\dinp\in \inpS} \RD{\rno}{\Wm(\dinp)}{\qma{\rno,\Wm}{\epsilon}}
&\leq \RCI{\rno}{\Wm}{\epsilon}
\\
\label{eq:lem:avcapacityB}
&\leq \tfrac{\RC{1/2}{\Wm}}{(1-\rno)(1-\epsilon)}
\end{align}
where the average \renyi capacity \(\RCI{\rno}{\Wm}{\epsilon}\) is
\begin{align}
\label{eq:def:avcapacity}
\RCI{\rno}{\Wm}{\epsilon}
&\DEF \tfrac{1}{\epsilon}\int_{\rno-\epsilon\rno}^{\rno+\epsilon(1-\rno)} 
\left[ 1\vee \left(\tfrac{\rno}{1-\rno}\tfrac{1-\rnt}{\rnt}\right)\right] \RC{\rnt}{\Wm} \dif{\rnt}.
\end{align}
\end{lemma}

Before presenting the proof of Lemma \ref{lem:avcapacity}, we point out certain 
properties of \(\RCI{\rno}{\Wm}{\epsilon}\).
As a result of the continuity of  \(\RC{\rno}{\Wm}\) in \(\rno\) on \((0,1)\),
i.e. Lemma \ref{lem:capacityO}-(\ref{capacityO-zo}), we have
\begin{align}
\notag
\lim\nolimits_{\epsilon \downarrow 0}
\RCI{\rno}{\Wm}{\epsilon}
&=\RC{\rno}{\Wm}
&
&\forall \rno\in (0,1).
\end{align}
In fact, we can bound \(\RCI{\rno}{\Wm}{\epsilon}\) from above 
using the monotonicity of \(\RC{\rno}{\Wm}\) and  \(\tfrac{1-\rno}{\rno}\RC{\rno}{\Wm}\)
in \(\rno\), 
i.e. Lemma \ref{lem:capacityO}-(\ref{capacityO-ilsc},\ref{capacityO-zo}),
as follows
\begin{align}
\notag
\RCI{\rno}{\Wm}{\epsilon}
%&\!\leq\!\tfrac{\RC{\rno}{\Wm}}{\epsilon}
%\left[\!\int_{\rno-\epsilon\rno}^{\rno}\! 
%\tfrac{1-\rno(1-\epsilon)}{(1-\rno)(1-\epsilon)} \dif{\rnt}
%\!+\!\int_{\rno}^{\rno+\epsilon(1-\rno)}\!\!
%\tfrac{\rno+(1-\rno)\epsilon}{\rno(1-\epsilon)}  \dif{\rnt}
%\!\right]
%\\
%\notag
&\leq\tfrac{\RC{\rno}{\Wm}}{\epsilon}\int_{\rno-\epsilon\rno}^{\rno} 
\tfrac{1-\rno(1-\epsilon)}{(1-\rno)(1-\epsilon)} \dif{\rnt}
\\
\notag
&\hspace{3cm} 
+\tfrac{\RC{\rno}{\Wm}}{\epsilon}\int_{\rno}^{\rno+\epsilon(1-\rno)} 
\tfrac{\rno+(1-\rno)\epsilon}{\rno(1-\epsilon)}  \dif{\rnt}
\\
\label{eq:avcapacity:bound}
&\leq \RC{\rno}{\Wm}+ \tfrac{\epsilon}{1-\epsilon}\tfrac{\RC{\rno}{\Wm}}{\rno (1-\rno)}.
\end{align}
On the other hand, \(\RCI{\rno}{\Wm\!}{\epsilon}\!\geq\!\RC{\rno}{\!\Wm\!}\)
because \(\RRR{\rno}{\Wm\!}{\qma{\rno,\Wm\!}{\epsilon}}\!\geq\!\RC{\rno}{\!\Wm\!}\)
by Theorem \ref{thm:minimax}.
Thus  \(\RC{\rno}{\Wm}\) can be approximated by \(\RCI{\rno}{\Wm}{\epsilon}\) at any \(\rno\)
in \((0,1)\). 
The expression in \eqref{eq:avcapacity:bound}, however, suggests that it might not 
be possible to do this approximation uniformly on  \((0,1)\). In order to show this formally,
we bound \(\RCI{\rno}{\Wm}{\epsilon}\) from below
using the monotonicity of \(\RC{\rno}{\Wm}\) and \(\tfrac{1-\rno}{\rno}\RC{\rno}{\Wm}\) 
in \(\rno\) as follows, 
\begin{align}
\notag
\RCI{\rno}{\Wm}{\epsilon}
&\geq \tfrac{1}{\epsilon}\int_{\rno-\epsilon\rno}^{\rno-\frac{\epsilon}{2}\rno} 
\tfrac{\rno}{1-\rno}\tfrac{1-\rnt}{\rnt} \RC{\rnt}{\Wm} \dif{\rnt}
\\
\notag
&\geq \tfrac{\rno}{2-\epsilon}(1+\tfrac{\rno \epsilon}{2(1-\rno)}) \RC{\rno-\epsilon\rno}{\Wm}.
\end{align}
Note that this lower bound is true even when \(\RC{1}{\Wm}\) is finite. 
Thus, \(\RC{\rno}{\Wm}\) can be approximated by \(\RCI{\rno}{\Wm}{\epsilon}\) uniformly only on compact subsets 
of \((0,1)\), but not on \((0,1)\) in itself.

The additivity of \renyi capacity for product channels, i.e. Lemma \ref{lem:capacityProduct},
implies the additivity of average \renyi capacity for product channels:
\begin{align}
\label{eq:avcapacity:additivity}
\RCI{\rno}{\Wmn{[1,\blx]}}{\epsilon}=\sum\nolimits_{\tin=1}^{\blx} \RCI{\rno}{\Wmn{\tin}}{\epsilon}.
\end{align} 
Lemma \ref{lem:capacityProduct} also states that 
\(\qmn{\rno,\Wmn{[1,\blx]}}\!=\!\bigotimes_{\tin=1}^{\blx}\qmn{\rno,\Wmn{\tin}}\). 
The average \renyi center \(\qma{\rno,\Wmn{[1,\blx]}}{\epsilon}\), however, does not satisfy such a  product 
structure, in general.  

\begin{proof}[Proof of Lemma \ref{lem:avcapacity}]
The convexity of \(\RD{\rno}{\mW}{\mQ}\) in \(\mQ\), i.e. 
Lemma \ref{lem:divergence-convexity},  and the Jensen's inequality imply
\begin{align}
\label{eq:avcapacity-a}
\RD{\rno}{\Wm(\dinp)}{\qma{\rno,\Wm}{\epsilon}}
&\leq \int_{\rno-\epsilon\rno}^{\rno+\epsilon(1-\rno)}
\tfrac{\RD{\rno}{\Wm(\dinp)}{\qmn{\rnt,\Wm}}}{\epsilon} \dif{\rnt}.
\end{align}
Note that \(\RD{\rno}{\Wm(\dinp)}{\qmn{\rnt,\Wm}}=\tfrac{\rno}{1-\rno}\RD{1-\rno}{\qmn{\rnt,\Wm}}{\Wm(\dinp)}\) 
for any \(\rno\in(0,1)\) by definition and
the \renyi divergence is nondecreasing in its order by Lemma \ref{lem:divergence-order}. Thus 
\begin{align}
\notag
\RD{\rno}{\Wm\!(\dinp)}{\qmn{\rnt,\Wm\!}}
&\!\leq\! 
\IND{\rnt\geq\rno} \RD{\rnt}{\Wm\!(\dinp)}{\qmn{\rnt,\Wm\!}}
\\
\notag
&\qquad+\IND{\rnt<\rno} \tfrac{\rno}{1-\rno}\tfrac{1-\rnt}{\rnt} \RD{\rnt}{\Wm\!(\dinp)}{\qmn{\rnt,\Wm\!}}
\\
\label{eq:avcapacity-b}
&\!=\!\left(1 \vee \tfrac{\rno}{1-\rno}\tfrac{1-\rnt}{\rnt}\right) \RD{\rnt}{\Wm\!(\dinp)}{\qmn{\rnt,\Wm\!}}.
\end{align}
Recall that \(\RD{\rnt}{\Wm(\dinp)}{\qmn{\rnt,\Wm}} \leq \RC{\rnt}{\Wm}\) for all \(\dinp\in\inpS\) 
by Theorem \ref{thm:minimax}. Then \eqref{eq:lem:avcapacity} follows from
using \eqref{eq:def:avcapacity}, \eqref{eq:avcapacity-a}, 
and \eqref{eq:avcapacity-b}.

In order to obtain \eqref{eq:lem:avcapacityB}
from \eqref{eq:lem:avcapacity}
recall that \(\RC{\rno}{\Wm}\) is nondecreasing in \(\rno\) 
by Lemma \ref{lem:capacityO}-(\ref{capacityO-ilsc}) and  
\(\tfrac{1-\rno}{\rno}\RC{\rno}{\Wm}\) is nonincreasing in \(\rno\) on \((0,1)\) by 
Lemma \ref{lem:capacityO}-(\ref{capacityO-zo}). Thus we have,
\begin{align}
\notag 
\RC{\rnt}{\Wm}
&\leq \tfrac{\rnt}{1-\rnt}\RC{\sfrac{1}{2}}{\Wm} \IND{\rnt>\sfrac{1}{2}}+ \RC{\sfrac{1}{2}}{\Wm} \IND{\rnt\leq\sfrac{1}{2}}
&
&\forall \rnt\in(0,1).
\end{align}
Then using the definition of \(\RCI{\rno}{\Wm}{\epsilon}\) given in \eqref{eq:def:avcapacity} we get
\begin{align}
\notag
\RCI{\rno}{\Wm}{\epsilon}
&\leq \tfrac{\RC{\sfrac{1}{2}}{\Wm}}{\epsilon}\int_{\rno-\epsilon\rno}^{\rno+\epsilon(1-\rno)} 
\left(1 \vee \tfrac{\rno}{1-\rno}\tfrac{1-\rnt}{\rnt}\right)
\left(1 \vee \tfrac{\rnt}{1-\rnt}\right)
\dif{\rnt}
\\
\notag
&\leq \tfrac{\RC{\sfrac{1}{2}}{\Wm}}{\epsilon}\int_{\rno-\epsilon\rno}^{\rno+\epsilon(1-\rno)} 
\left(\tfrac{1}{1-\rnt} \vee \tfrac{\rno}{1-\rno}\tfrac{1}{\rnt}\right)
\dif{\rnt}
\\
\notag
&\leq \tfrac{\RC{\sfrac{1}{2}}{\Wm}}{\epsilon}\int_{\rno-\epsilon\rno}^{\rno+\epsilon(1-\rno)} \tfrac{1}{(1-\rno)(1-\epsilon)} 
\dif{\rnt}.
\end{align}
\end{proof}

\begin{definition}\label{def:avspherepacking}
For any \(\epsilon\in (0,1)\),
\(\Wm\!:\!\inpS\!\to\!\pmea{\outA}\) with finite \(\RC{\sfrac{1}{2}}{\Wm\!}\),
and \(\rate\!\in\!\reals{\geq0}\),
\emph{the average sphere packing exponent} 
\(\spa{\epsilon}{\rate,\Wm}\) is 
\begin{align}
\label{eq:def:avspherepacking}
\spa{\epsilon}{\rate,\Wm}
&\DEF \sup\nolimits_{\rno\in (0,1)} \tfrac{1-\rno}{\rno} \left(\RCI{\rno}{\Wm}{\epsilon}-\rate\right).
\end{align}
%%%where \(\RCI{\rno}{\Wm}{\epsilon}\) is defined in equation \eqref{eq:def:avcapacity}.
\end{definition}

\(\spa{\epsilon}{\rate,\Wm}\) is nonincreasing and convex in \(\rate\) on \(\reals{+}\) because it is 
the pointwise supremum of nonincreasing and convex functions of \(\rate\).
One can show that \(\RCI{\rno}{\Wm}{\epsilon}\) is nondecreasing and continuous in \(\rno\) on \((0,1)\) 
for any \(\epsilon\in(0,1)\) using the continuity and monotonicity of \(\RC{\rno}{\Wm}\) in \(\rno\) 
on \((0,1)\). Since we do not need this observation in our analysis, we leave its proof to the interested reader. 
Using the monotonicity of \(\RCI{\rno}{\Wm}{\epsilon}\)
one can also show that \(\spa{\epsilon}{\rate,\Wm}\) is finite and continuous in \(\rate\) on 
\((\lim_{\rno\downarrow0} \RCI{\rno}{\Wm}{\epsilon},\infty)\).

\begin{lemma}\label{lem:avspherepacking}
For any \(\Wm\!:\!\inpS\!\to\!\pmea{\outA}\) with finite \(\RC{\sfrac{1}{2}}{\Wm}\),
\(\rnf\!\in\!(0,1)\),
\(\rate\!\in\![\RC{\rnf}{\Wm},\infty)\), and 
\(\epsilon\in(0,\rnf)\),
\begin{align}
\label{eq:lem:avspherepacking}
0\!\leq\! 
\spa{\epsilon}{\rate,\Wm}\!-\!\spe{\rate,\Wm}
&\!\leq\!\tfrac{\epsilon}{1-\epsilon} 
\tfrac{\rate \vee \spe{\rate,\Wm}}{\rnf}
\\
\label{eq:lem:avspherepackingB}
&\leq \tfrac{\epsilon}{1-\epsilon} \tfrac{\rate}{\rnf^{2}}.
\end{align}
\end{lemma}

\begin{proof}[Proof of Lemma \ref{lem:avspherepacking}]
\(\RC{\rno}{\Wm}\!\leq \RRR{\rno}{\Wm}{\qma{\rno,\Wm}{\epsilon}}\!\leq\!\RCI{\rno}{\Wm}{\epsilon}\) by Lemma \ref{lem:avcapacity}
and Theorem \ref{thm:minimax}. 
Then as a result the definitions of \(\spe{\rate,\Wm}\) and \(\spa{\epsilon}{\rate,\Wm}\) 
we have
\begin{align}
\label{eq:avspherepacking-1}
\spe{\rate,\Wm}&\leq\spa{\epsilon}{\rate,\Wm}
&
&\forall \rate\in \reals{\geq0}. 
\end{align} 
Let us proceed with bounding \(\spa{\epsilon}{\rate,\Wm}\) 
from above for \(\rate\)'s greater than or equal to  \(\RC{\rnf}{\Wm}\).
First note that
\begin{align}
\notag
&\int_{\rno-\epsilon\rno}^{\rno} 
\left(\tfrac{1-\rnt}{\rnt}\RC{\rnt}{\Wm}-\tfrac{1-\rno}{\rno}\rate
\right)\dif{\rnt}
\\
\notag
&\hspace{1.5cm}=\int_{\rno-\epsilon\rno}^{\rno} 
\tfrac{1-\rnt}{\rnt}(\RC{\rnt}{\Wm} -\rate)\dif{\rnt}
+\int_{\rno-\epsilon\rno}^{\rno} 
\tfrac{\rno-\rnt}{\rnt\rno}\rate\dif{\rnt}
\\
\notag
&\hspace{1.5cm}\leq \int_{\rno-\epsilon\rno}^{\rno} 
\tfrac{1-\rnt}{\rnt}(\RC{\rnt}{\Wm} -\rate)\dif{\rnt}
+\tfrac{\epsilon^{2}}{1-\epsilon}\rate.
\end{align}
Then as a result of the definition of \(\RCI{\rno}{\Wm}{\epsilon}\) we have
\begin{align}
\notag
\tfrac{1-\rno}{\rno}\!\left(\RCI{\rno}{\Wm}{\epsilon}\!-\!\rate\right)
&\!\leq\!\tfrac{1}{\epsilon}\int_{\rno-\epsilon\rno}^{\rno}\!
\tfrac{1-\rnt}{\rnt}(\RC{\rnt}{\Wm}\!-\!\rate)\dif{\rnt}
+\tfrac{\epsilon}{1-\epsilon}\rate
\\
\label{eq:avspherepacking-2}
&~~~+\!\tfrac{1}{\epsilon}\int_{\rno}^{\rno+\epsilon(1-\rno)}\! 
\tfrac{1-\rno}{\rno}\!(\RC{\rnt}{\Wm}\!-\!\rate) \dif{\rnt}.
\end{align} 
We bound \(\spa{\epsilon}{\rate,\Wm}\) by bounding the expression  on 
the right hand side of \eqref{eq:avspherepacking-2} separately on two intervals for \(\rno\). 
Note that \(\tfrac{1-\rnt}{\rnt}(\RC{\rnt}{\Wm}\!-\!\rate)\!\leq\!\spe{\rate,\Wm}\) 
for all \(\rnt\!\in\!(0,1)\)
and 
\((1\!-\!\rno)\spe{\rate,\Wm}+\rno \rate\!\leq\!\rate\vee\spe{\rate,\Wm}\) 
for all \(\rno\!\in\!(0,1)\).
Thus for  \(\rno \in [\rnf,1)\), \eqref{eq:avspherepacking-2} implies
\begin{align}
\notag
\tfrac{1-\rno}{\rno}\!\left(\RCI{\rno}{\Wm}{\epsilon}
\!-\!\rate\right)
&\!\leq\!\tfrac{1}{\epsilon}\int_{\rno-\epsilon\rno}^{\rno} 
\spe{\rate,\Wm}\dif{\rnt}
+\tfrac{\epsilon}{1-\epsilon}\rate
\\
\notag
&\hspace{.6cm}+
\!\tfrac{1}{\epsilon}\tfrac{1-\rno}{\rno} \int_{\rno}^{\rno+\epsilon(1-\rno)} 
\!\tfrac{\rnt}{1-\rnt}\spe{\rate,\Wm}\dif{\rnt}
\\
\notag
&\!\leq\!\spe{\rate,\Wm}\!+\!
\tfrac{\epsilon}{1-\epsilon}
\tfrac{1-\rno}{\rno}\spe{\rate,\Wm}
\!+\!\tfrac{\epsilon}{1-\epsilon} \rate
\\
\label{eq:avspherepacking-3new}
&\!\leq\!\spe{\rate,\Wm}\!+\!
\tfrac{\epsilon}{1-\epsilon}
\tfrac{\rate\vee\spe{\rate,\Wm}}{\rnf}.
\end{align}
On the other hand \(\rate\geq\RC{\rnf}{\Wm}\) by the hypothesis 
and the \renyi capacity is nondecreasing in its order by 
Lemma \ref{lem:capacityO}-(\ref{capacityO-ilsc}).
Thus  for \(\rno \in (0,\rnf]\), \eqref{eq:avspherepacking-2} implies 
\begin{align}
\notag
\!\!\tfrac{1-\rno}{\rno}\!&\left(\RCI{\rno}{\Wm}{\epsilon}-\rate\right)
\\
\notag
&\!\leq\!\tfrac{\epsilon}{1-\epsilon}\rate
\!+\!\tfrac{1}{\epsilon}
\!\int_{\rnf}^{\rno+\epsilon(1-\rno)} 
\!\!\tfrac{(1-\rno)\rnt}{\rno(1-\rnt)}\spe{\rate,\Wm}\dif{\rnt}
\IND{\rno\in[\frac{\rnf-\epsilon}{1-\epsilon},\rnf]}
\\
\notag
&\!\leq\!\tfrac{\epsilon}{1-\epsilon}\rate
\!+\!\tfrac{1}{\epsilon}\tfrac{\rno(1-\epsilon)+\epsilon}{\rno(1-\epsilon)}
\!\!\int_{\rnf}^{\rno(1-\epsilon)+\epsilon}\!\!\!\spe{\rate,\Wm}\dif{\rnt}
\IND{\rno\in[\frac{\rnf-\epsilon}{1-\epsilon},\rnf]}
\\
\notag
&\!=\!\tfrac{\epsilon}{1-\epsilon}\rate
\!+\!
\left[
\tfrac{\rno(1-\epsilon) +2\epsilon-\rnf}{\epsilon}
\!-\!\tfrac{\rnf-\epsilon}{\rno(1-\epsilon)}
\right]\!
\spe{\rate,\Wm}\!
\IND{\rno\in[\frac{\rnf-\epsilon}{1-\epsilon},\rnf]}
\\
\label{eq:avspherepacking-4new}
&\!\leq
\tfrac{\epsilon}{1-\epsilon}
\tfrac{\rnf\rate+(1-\rnf)\spe{\rate,\Wm}}{\rnf}
+\!(1-\rnf)\spe{\rate,\Wm}.
\end{align}

%\clearpage
%On the other hand, \(\rate\geq\RC{\rnf}{\Wm}\) by the hypothesis 
%and the \renyi capacity is nondecreasing in its order by 
%Lemma \ref{lem:capacityO}-(\ref{capacityO-ilsc}).
%Thus  for \(\rno \in (0,\rnf]\), \eqref{eq:avspherepacking-2} 
%implies 
%\begin{align}
%\notag
%\!\!\tfrac{1-\rno}{\rno}\!&\left(\RCI{\rno}{\Wm}{\epsilon}-\rate\right)
%\\
%\notag
%&\!\leq\!\tfrac{\epsilon}{1-\epsilon}\rate
%\!+\!\tfrac{1}{\epsilon}
%\!\int_{\rnf}^{\rno+\epsilon(1-\rno)} 
%\!\!\tfrac{(1-\rno)\rnt}{\rno(1-\rnt)}\spe{\rate,\Wm}\dif{\rnt}
%\IND{\rno\in[\frac{\rnf-\epsilon}{1-\epsilon},\rnf]}
%\\
%\notag
%&\!\leq\!\tfrac{\epsilon}{1-\epsilon}\rate
%\!+\!\tfrac{1}{\epsilon}\tfrac{\rno(1-\epsilon)+\epsilon}{\rno(1-\epsilon)}
%\!\!\int_{\rnf}^{\rno(1-\epsilon)+\epsilon}\!\!\!\spe{\rate,\Wm}\dif{\rnt}
%\IND{\rno\in[\frac{\rnf-\epsilon}{1-\epsilon},\rnf]}
%\\
%\notag
%&\!=\!\tfrac{\epsilon}{1-\epsilon}\rate
%\!+\!
%\left[
%\tfrac{\rno(1-\epsilon) +2\epsilon-\rnf}{\epsilon}
%\!-\!\tfrac{\rnf-\epsilon}{\rno(1-\epsilon)}
%\right]\!
%\spe{\rate,\Wm}\!
%\IND{\rno\in[\frac{\rnf-\epsilon}{1-\epsilon},\rnf]}
%\\
%&\!\leq
%\tfrac{\epsilon}{1-\epsilon}
%\tfrac{\rnf\rate+(1-\rnf)\spe{\rate,\Wm}}{\rnf}
%+\!(1-\rnf)\spe{\rate,\Wm}.
%\end{align}
%
%\clearpage

Note that
\eqref{eq:lem:avspherepacking} follows from \eqref{eq:avspherepacking-1},
\eqref{eq:avspherepacking-3new}, and \eqref{eq:avspherepacking-4new}.
In order to obtain \eqref{eq:lem:avspherepackingB} from \eqref{eq:lem:avspherepacking},
recall that \(\RC{\rno}{\Wm}\) is nondecreasing in \(\rno\) by 
Lemma \ref{lem:capacityO}-(\ref{capacityO-ilsc}) and 
\(\tfrac{1-\rno}{\rno}\RC{\rno}{\Wm}\) is nonincreasing in  \(\rno\) on \((0,1)\) by 
Lemma \ref{lem:capacityO}-(\ref{capacityO-zo}).
Then \(\spe{\rate,\Wm}\leq \tfrac{1-\rnf}{\rnf}\rate\) and hence
\(\rate\vee \spe{\rate,\Wm}\leq \sfrac{\rate}{\rnf}\)
for all \(\rate\geq\RC{\rnf}{\Wm}\)
by Lemma \ref{lem:spherepackingexponent}.
\end{proof}

\vspace{-.2cm}
\subsection{Tilting with a Family of Measures}\label{sec:tilting}
\begin{definition}\label{def:tiltedprobabilitymeasure}
For any \(\rno\in\reals{+}\) and \(\mW,\mQ\) in \(\pmea{\outA}\) 
satisfying \(\RD{\rno}{\mW}{\mQ}<\infty\),
\emph{the order \(\rno\) tilted probability measure} 
\(\wma{\rno}{\mQ}\) is 
\begin{align}
\label{eq:def:tiltedprobabilitymeasure}
\der{\wma{\rno}{\mQ}}{\rfm}
&\DEF e^{(1-\rno)\RD{\rno}{\mW}{\mQ}}
(\der{\mW}{\rfm})^{\rno} (\der{\mQ}{\rfm})^{1-\rno}
\end{align}
where \(\rfm\in\pmea{\outA}\) satisfies 
\(\mW\AC\rfm\) and \(\mQ\AC\rfm\).
\end{definition}
In many applications, the tilted probability measure between two fixed probability measures 
is of interest for orders in \((0,1)\).
In our analysis we need to allow one of those probability measures to change with the order, as well.
Lemma \ref{lem:tilting}, in the following,
considers the tilted probability measure \(\wma{\rno}{\qmn{\rno}}\) 
as a function of the order \(\rno\) for the case when \(\qmn{\rno}\) is a continuous 
function of \(\rno\) from \((0,1)\) to \(\pmea{\outA}\) for the total
variation topology on \(\pmea{\outA}\). 
\begin{lemma}\label{lem:tilting}
Let \(\qmn{\rno}\) be 
a continuous function of \(\rno\) from \((0,1)\)
to \(\pmea{\outA}\) for the total variation topology
and \(\mW\in\pmea{\outA}\) satisfy
\(\RD{\rno}{\mW}{\qmn{\rno}}<\infty\) for all \(\rno\in(0,1)\). Then 
\begin{enumerate}[(a)]
\item\label{tilting-vma} 
\(\wma{\rno}{\qmn{\rno}}\) is a continuous function of \(\rno\) from \((0,1)\)
to \(\pmea{\outA}\) for the total variation topology,
\item\label{tilting-divergence}
\(\RD{\rno}{\mW}{\qmn{\rno}}\),
\(\RD{1}{\wma{\rno}{\qmn{\rno}}}{\mW}\), and 
\(\RD{1}{\wma{\rno}{\qmn{\rno}}}{\qmn{\rno}}\) are continuous functions of \(\rno\)
from \((0,1)\) to \(\reals{\geq0}\).
\end{enumerate}
\end{lemma}

\begin{remark}
The continuity of \(\qmn{\rno}\) in the order \(\rno\) for the total variation topology
on \(\pmea{\outA}\) does not imply 
the continuity of the corresponding Radon-Nikodym derivatives \(\der{\qmn{\rno}}{\rfm}\) 
for some reference measure \(\rfm\), 
see \cite[Remark \ref*{A-rem:continuityintvt}]{nakiboglu19A}.
Thus Lemma \ref{lem:tilting} is not a corollary of
standard results on the continuity of  integrals, such as \cite[Cor. 2.8.7]{bogachev}. 
\end{remark}

Lemma \ref{lem:tilting} does not assume \(\qmn{\rno}\) to be any particular family 
of probability measures, such as \renyi centers; \(\qmn{\rno}\) is unspecified 
except for its continuity in \(\rno\) and finiteness of \(\RD{\rno}{\mW}{\qmn{\rno}}\)
for all \(\rno\) in \((0,1)\).
However, for any channel \(\Wm\) with finite \(\RC{\sfrac{1}{2}}{\Wm}\), 
the \renyi center \(\qmn{\rno,\Wm}\) satisfies the hypothesis of Lemma \ref{lem:tilting} 
for \(\mW\!=\!\Wm(\dinp)\) for all \(\dinp\!\in\!\inpS\) because 
\(\RC{\rno}{\Wm\!}\!<\!\infty\)  for all \(\rno\!\in\!(0,1)\) by 
Lemma \ref{lem:capacityO}-(\ref{capacityO-zofintieness}),
\(\RD{\rno}{\Wm(\dinp)}{\qmn{\rno,\Wm}}\leq \RC{\rno}{\Wm}\) by Theorem \ref{thm:minimax},
and
\(\qmn{\rno,\Wm}\) is continuous in \(\rno\) for the total variation 
topology on \(\pmea{\outA}\) by 
Lemma \ref{lem:centercontinuity}.
\begin{remark}
We believe \(\qmn{\rno,\Wm}\) satisfies the monotonicity described 
in \cite[Conjecture \ref*{A-con:renyicentermonotonicity}]{nakiboglu19A}. 
If that is the case, we do not need Lemma \ref{lem:tilting},
we can establish the continuity of 
\(\wma{\rno}{\qmn{\rno,\Wm}}\), \(\RD{\rno}{\mW}{\qmn{\rno,\Wm}}\),
\(\RD{1}{\wma{\rno}{\qmn{\rno,\Wm}}}{\mW}\), and 
\(\RD{1}{\wma{\rno}{\qmn{\rno,\Wm}}}{\qmn{\rno,\Wm}}\)
for \(\mW\!=\!\Wm(\dinp)\) for all \(\dinp\!\in\!\inpS\)
using standard results on the continuity of integrals, 
such as \cite[Cor. 2.8.7]{bogachev}.
\end{remark}
\begin{proof}[Proof of Lemma \ref{lem:tilting}]
	Any function \(\gX\) on \((0,1)\) is continuous iff 
	for every convergent sequence \(\rno_{\blx}\to\rno\) in \((0,1)\), 
	the sequence \(\gX(\rno_{\blx})\) converges  to \(\gX(\rno)\)
	by \cite[Thm. 21.3]{munkres} because \((0,1)\) is metrizable.
	Let \(\{\rno_{\blx}\}_{\blx\in\integers{+}}\) be a convergent 
	sequence such that  \(\lim\nolimits_{\blx \to \infty} \rno_{\blx}=\rno\)
	and \(\rfm\) be
	\begin{align}
	\notag
	\rfm=\tfrac{\mW}{4}+\tfrac{\qmn{\rno}}{4}+\tfrac{1}{2}\sum\nolimits_{\blx\in\integers{+}} \tfrac{\qmn{\rno_{\blx}}}{2^{\blx}}.
	\end{align}
	
	Instead of working with measures as members of \(\smea{\outA}\) for the total variation topology,
	we work with the corresponding Radon-Nikodym  derivatives with respect to \(\rfm\) as members of \(\Lon{\rfm }\).
	We can do so because all of the measures we are considering are absolutely continuous with respect to \(\rfm\)
	and
	for any sequence 
	\(\{\cln{\blx}\}_{\blx\in\integers{+}}\!\subset\!\Lon{\rfm}\), 
	\(\cln{\blx}\!\xrightarrow{\Lon{\rfm}}\!\cln{}\) iff the corresponding sequence of 
	measures \(\{\cln{\blx}\rfm\}_{\blx\in\integers{+}}\) converges to \(\cln{}\rfm\) in \(\smea{\outA}\) for the total variation topology. 
	
	For any finite signed measure \(\mean\) such that \(\mean\AC\rfm\), we denote its Radon-Nikodym derivative with respect to \(\rfm\) by  \(\cln{\mean}\):
	\begin{align}
	\notag
	\cln{\mean}
	&=\der{\mean}{\rfm}
	&
	&\forall \mean \in \smea{\outA} \mbox{~such that~} \mean\AC\rfm.
	\end{align}
	We make an exception and denote the Radon-Nikodym derivative of \(\wma{\rno}{\qmn{\rno}}\) 
	by  \(\cln{\rno}\) rather than \(\cln{\wma{\rno}{\qmn{\rno}}}\).
	
\begin{enumerate}
	\item[(\ref{tilting-vma})] 
	For any  \(\rno\in (0,1)\) let \(\cln{\smn{\rno}}\) be 
	\(\cln{\smn{\rno}}\DEF \cla{\mW}{\rno} \cla{\qmn{\rno}}{1-\rno}\).
	Using the triangle equality we get:
	\begin{align}
	\hspace{-.6cm}\abs{\cln{\smn{\rno}}\!-\!\cln{\smn{\rno_{\blx}}}\!}
	\notag
	&\!=\!
	\abs{\cln{\smn{\rno}}\!-\!\cla{\mW}{\rno_{\blx}}\cla{\qmn{\rno_{\blx}}}{1-\rno_{\blx}}}
	\\
	\label{eq:tilting-smn-1}
	&\!\leq\!\abs{\cln{\smn{\rno}}\!-\!\cla{\mW}{\rno_{\blx}}\cla{\qmn{\rno}}{1-\rno_{\blx}}}
	+\cla{\mW}{\rno_{\blx}}\abs{\cla{\qmn{\rno}}{1-\rno_{\blx}}\!-\!\cla{\qmn{\rno_{\blx}}}{1-\rno_{\blx}}}.
	\end{align}
\begin{itemize}
	\item \(\{\cla{\mW}{\rno_{\blx}}\cla{\qmn{\rno}}{1-\rno_{\blx}}\}_{\blx\in\integers{+}}\) is uniformly integrable
	because 
	\(\EXS{\rfm}{\IND{\oev}\cla{\mW}{\rno_{\blx}}\!\cla{\qmn{\rno}}{1-\rno_{\blx}}}
	\!\leq\!\mW(\oev)^{\rno_{\blx}}\!\qmn{\rno}(\oev)^{1-\rno_{\blx}}\) 
	by the H\"{o}lder's inequality and \(\mW(\oev)^{\rno_{\blx}}\!\qmn{\rno}(\oev)^{1-\rno_{\blx}}\) 
	is bounded above by  \(\mW(\oev)\!+\!\qmn{\rno}(\oev)\). 
	
	\item \(\cla{\mW}{\rno_{\blx}}\cla{\qmn{\rno}}{1-\rno_{\blx}}\!\xrightarrow{\rfm}\!\cln{\smn{\rno}}\)
	because almost everywhere convergence implies convergence in measure for finite measures by
	\cite[Thm. 2.2.3]{bogachev} and  
	\(\cla{\mW}{\rno_{\blx}}\cla{\qmn{\rno}}{1-\rno_{\blx}}\!\xrightarrow{\rfm-a.e.}\!\cln{\smn{\rno}}\)
	by definition.
\end{itemize}		
	Then 
	\begin{align}
	\label{eq:tilting-smn-2}
	\cla{\mW}{\rno_{\blx}}\cla{\qmn{\rno}}{1-\rno_{\blx}}
	&\xrightarrow{\Lon{\rfm}}
	\cln{\smn{\rno}}
	\end{align}
	by the Lebesgue-Vitali convergence theorem \cite[4.5.4]{bogachev}.
	
	Using the derivative test one can confirm that \((\dsta+\tau)^\beta\!-\!\dsta^{\beta}\) 
	is a nonincreasing function of \(\dsta\) for any 
	\(\dsta\!\geq\!0\), \(\tau\!\geq\!0\), and \(\beta\!\in\!(0,1)\). 
	Thus \((\dsta+\tau)^\beta\!-\!\dsta^{\beta}\!\leq\!\tau^{\beta}\) for
	any \(\dsta\!\geq\!0\), \(\tau\!\geq\!0\), and \(\beta\!\in\!(0,1)\). Then using the H\"{o}lder's inequality we get,
	\begin{align}
	\notag
	\EXS{\rfm}{\cla{\mW}{\rno_{\blx}}\abs{\cla{\qmn{\rno}}{1-\rno_{\blx}}\!-\!\cla{\qmn{\rno_{\blx}}}{1-\rno_{\blx}}}}
	&\!\leq\! 
	\EXS{\rfm}{\cla{\mW}{\rno_{\blx}}\abs{\cln{\qmn{\rno}}\!-\!\cln{\qmn{\rno_{\blx}}}}^{1-\rno_{\blx}}}
	\\
	\notag
	&\!\leq\!\EXS{\rfm}{\cln{\mW}}^{\rno_{\blx}}
	\EXS{\rfm}{\abs{\cln{\qmn{\rno}}\!-\!\cln{\qmn{\rno_{\blx}}}}}^{1-\rno_{\blx}}
	\end{align}
	Then using 
	\(\!\cln{\qmn{\rno_{\blx}}}\!\!\xrightarrow{\!\Lon{\rfm}\!}\!\cln{\qmn{\rno}}\),
	\(\rno_{\blx}\!\xrightarrow{}\!\rno\),
	and \(\rno\!\in\!(0,1)\) we get
	\begin{align}
	\label{eq:tilting-smn-3}
	\lim\nolimits_{\blx\to\infty}
	\EXS{\rfm}{\cla{\mW}{\rno_{\blx}}\abs{\cla{\qmn{\rno}}{1-\rno_{\blx}}\!-\!\cla{\qmn{\rno_{\blx}}}{1-\rno_{\blx}}}}
	&\!=\!0.
	\end{align}
	\eqref{eq:tilting-smn-1}, \eqref{eq:tilting-smn-2}, and \eqref{eq:tilting-smn-3}
	imply \(\cln{\smn{\rno_{\blx}}}\xrightarrow{\Lon{\rfm}}\cln{\smn{\rno}}\). 
	Thus \(\smn{\rno}\) is continuous in \(\rno\) on \((0,1)\) 
	for the total variation topology on \(\fmea{\outA}\).
	Then \(\lon{\smn{\rno}}=\EXS{\rfm}{ \cln{\smn{\rno}}}\) is continuous in \(\rno\), as well. 
	Furthermore,  \(\lon{\smn{\rno}}\) is positive because \(\lon{\smn{\rno}}=e^{(\rno-1)\RD{\rno}{\mW}{\qmn{\rno}}}\) and \(\RD{\rno}{\mW}{\qmn{\rno}}\) is finite
	by the hypothesis of the lemma.
	On the other hand, \(\cln{\rno}=\sfrac{\cln{\smn{\rno}}}{\lon{\smn{\rno}}}\)
	and the triangle inequality implies 
	\begin{align}
	\notag
	\abs{\cln{\rno}-\cln{\rno_{\blx}}}
	&\leq  \tfrac{1}{\lon{\smn{\rno}}}\abs{\cln{\smn{\rno}}-\cln{\smn{\rno_{\blx}}}}
	+\abs{\tfrac{\lon{\smn{\rno}}-\lon{\smn{\rno_{\blx}}}}{\lon{\smn{\rno}}}}. 
	\end{align}
	Since \(\lon{\smn{\rno}}\) is positive for all \(\rno\in(0,1)\),
	the continuity of \(\lon{\smn{\rno}}\) in \(\rno\)
	and
	\(\cln{\smn{\rno_{\blx}}}\!\!\xrightarrow{\Lon{\rfm}}\!\cln{\smn{\rno}}\)
	imply \(\cln{\rno_{\blx}}\xrightarrow{\Lon{\rfm}}\cln{\rno}\).
	Thus \(\wma{\rno}{\qmn{\rno}}\) is continuous in \(\rno\) on \((0,1)\) 
	for the total variation topology on \(\pmea{\outA}\).
	
	\item[(\ref{tilting-divergence})]
\(\lon{\smn{\rno}}\) is positive for all \(\rno\!\in\!(0,1)\) 
by the hypothesis because \(\RD{\rno}{\mW}{\qmn{\rno}}\!=\!\tfrac{\ln \lon{\smn{\rno}}}{\rno-1}\).
	Furthermore,  \(\RD{\rno}{\mW}{\qmn{\rno}}\) is continuous in \(\rno\) because 
	product and composition of continuous functions are continuous.
	
	\(\RD{1}{\wma{\rno}{\qmn{\rno}}}{\mW}\) and \(\RD{1}{\wma{\rno}{\qmn{\rno}}}{\qmn{\rno}}\) 
	are both lower semicontinuous in \(\rno\) because the \renyi divergence is jointly lower semicontinuous in its arguments 
	for the topology of setwise convergence by Lemma \ref{lem:divergence:lsc}
	and \(\wma{\rno}{\qmn{\rno}}\) and \(\qmn{\rno}\) are continuous in the topology of setwise convergence. 
	
	\(\RD{1}{\wma{\rno}{\qmn{\rno}}}{\mW}\) is upper semicontinuous in \(\rno\) because
	\(\RD{\rno}{\mW}{\qmn{\rno}}\) is continuous in \(\rno\), \(\RD{1}{\wma{\rno}{\qmn{\rno}}}{\qmn{\rno}}\)
	is lower semicontinuous in \(\rno\), and \(\RD{1}{\wma{\rno}{\qmn{\rno}}}{\mW}\) satisfies 
	\begin{align}
	\notag
	\RD{1}{\wma{\rno}{\qmn{\rno}}}{\mW}
	&=\tfrac{1-\rno}{\rno}\RD{\rno}{\mW}{\qmn{\rno}}-\tfrac{1-\rno}{\rno}\RD{1}{\wma{\rno}{\qmn{\rno}}}{\qmn{\rno}}.
	\end{align}

	Then \(\RD{1}{\wma{\rno}{\qmn{\rno}}}{\mW}\) is continuous in \(\rno\) 
	because it is both lower semicontinuous and upper semicontinuous in  \(\rno\).
	
	Expressing \(\RD{1}{\wma{\rno}{\qmn{\rno}}}{\qmn{\rno}}\) in terms of \(\RD{1}{\wma{\rno}{\qmn{\rno}}}{\mW}\) and 
	following a similar reasoning, 
	we deduce that \(\RD{1}{\wma{\rno}{\qmn{\rno}}}{\qmn{\rno}}\) is continuous in \(\rno\),
	as well. 
\end{enumerate}

\end{proof} 
%!TEX root=../main-B.tex
\section{The SPB for Product Channels}\label{sec:product-outerbound}
\begin{assumption}\label{assumption:individual-ologn}
\(\{\Wmn{\tin}\}_{\tin\in\integers{+}}\) is a sequence of channels such that 
the maximum \(\RC{\sfrac{1}{2}}{\Wmn{\tin}}\) among the first \(\blx\) 
\(\Wmn{\tin}\)'s is \(\bigo{\ln \blx}\):
there exists \(\blx_{0}\in \integers{+}\) and \(K\in\reals{+}\) such that 
\begin{align}
%\label{eq:assumption:individual-ologn}
\notag
\max\nolimits_{\tin\in [1,\blx]} \RC{\sfrac{1}{2}}{\Wmn{\tin}}
&\leq K \ln\blx
&
&\forall \blx\geq\blx_{0}.
\end{align}
\end{assumption}

\begin{theorem}\label{thm:productexponent}
Let \(\{\Wmn{\tin}\}_{\tin\in\integers{+}}\) be a sequence of channels  satisfying  Assumption \ref{assumption:individual-ologn},
\(\varepsilon\) be a positive real number,
and \(\rno_{0}\), \(\rno_{1}\) be orders satisfying \(0\!<\!\rno_{0}\!<\!\rno_{1}\!<\!1\).
Then for any sequence of codes on the product channels \(\{\Wmn{[1,\blx]}\}_{\blx\in\integers{+}}\) satisfying
\begin{align}
\label{eq:thm:productexponent-hypothesis}
\RC{\rno_{1}}{\Wmn{[1,\blx]}}\!\geq\!\ln\!\tfrac{M_{\blx}}{L_{\blx}}
&\!\geq\!\RC{\rno_{0}}{\Wmn{[1,\blx]}}\!+\!\varepsilon(\ln\blx)^{2}
&\forall\blx\!\geq\!\blx_{0}
\end{align}
there exists a \(\tau\!\in\!\reals{+} \) and an \(\blx_{1}\!\geq\!\blx_{0}\) such that 
\begin{align}
\label{eq:thm:productexponent}
\Pem{av}^{(\blx)}
&\geq \blx^{-\tau}  e^{-\spe{\ln\frac{M_{\blx}}{L_{\blx}},\Wmn{[1,\blx]}}} 
&
&\forall \blx\geq\blx_{1}.
\end{align}
\end{theorem}
The main aim of this section is to prove the asymptotic SPB given in 
Theorem \ref{thm:productexponent}; we do so following the program 
put forward by Augustin in \cite{augustin69}.
In \S\ref{sec:momentbound}, we bound the moments of certain zero mean random variables 
related to the tilted probability measures.
In \S\ref{sec:smalldeviations}, we bound the small deviation probability for sums of independent random 
variables using the Berry-Esseen theorem.
In \S\ref{sec:pouterbound}, we first derive non-asymptotic, but parametric, SPBs for codes 
on arbitrary product channels and on certain product channels with feedback;
then 
we prove Theorem \ref{thm:productexponent} using the bound for codes 
on arbitrary product channels.
In \S\ref{sec:comparison}, we compare our SPBs with the SBPs derived by
Augustin in \cite{augustin69} and \cite{augustin78} for 
the product channels.

We make a brief digression to discuss the implications of Theorem \ref{thm:productexponent},
before starting its proof. 
Theorem \ref{thm:productexponent} and
the list decoding variant of Gallager's bound, i.e. Lemma \ref{lem:gallager},
determine the 
optimal \(\Pem{av}^{(\blx)}\) up to a polynomial factor for all sequences of product 
channels satisfying Assumption \ref{assumption:individual-ologn}.
In order to see why, note that 
if there exists an \(\rno_{0}\!\in\![\tfrac{1}{1+L_{\blx}},1]\)
satisfying \(\ln \tfrac{M_{\blx}}{L_{\blx}}\geq \RC{\rno_{0}}{\Wmn{[1,\blx]}}\)
then
\begin{align}
\Pem{av}^{(\blx)}&\leq e^{\frac{1-\rno_{0}}{\rno_{0}}}e^{-\spe{\ln\frac{M_{\blx}}{L_{\blx}},\Wmn{[1,\blx]}}}
\end{align}
by Lemma \ref{lem:gallager} because
\(\spe{\rate,\!\Wm\!}\!\leq\!\tfrac{1-\rno}{\rno}\!+\!\spe{\rate\!+\!1,\!\Wm\!}\)
for any \(\rno\in(0,1)\) satisfying 
\(\rate\!\geq\!\RC{\rno}{\Wm\!}\)
by Lemma \ref{lem:spherepackingexponent}.

If the sequence of channels satisfying Assumption \ref{assumption:individual-ologn}
is composed of channels with bounded order \(\sfrac{1}{2}\) \renyi capacity,  i.e. 
if \(\sup_{\tin\in\integers{+}}\RC{\sfrac{1}{2}}{\Wmn{\tin}}\!\leq\!K\)
for some  \(K\!\in\!\reals{+}\), 
then we can bound \(\tau\) in Theorem \ref{thm:productexponent} from above, as well.
But, our bounds are too crude to recover the right polynomial 
prefactor. 

Assumption \ref{assumption:individual-ologn} holds for all stationary product channels
and many non-stationary product channels. 
As an example consider the Poisson channel \(\Pcha{\tlx,\mA,\gX(\cdot)}\) 
whose input set is described in \eqref{eq:def:poissonchannel-product}.
The \renyi capacity of \(\Pcha{\tlx,\mA,\gX(\cdot)}\) is
determined in \cite[(\ref*{A-eq:poissonchannel-product-capacity})]{nakiboglu19A}
to be: 
\begin{align}
\notag
\RC{\rno}{\Pcha{\tlx,\mA,\gX(\cdot)}}
&=\!\int_{0}^{\tlx}\!\left[\!\left(\rno\tfrac{\gX-\mA}{\gX^{\rno}-\mA^{\rno}}\right)^{\frac{1}{1-\rno}}
\!-\!\tfrac{\rno}{\rno-1}\!\tfrac{\mA\gX^{\rno}-\gX \mA^{\rno}}{\gX^{\rno}-\mA^{\rno}}
\!\right]\!\dif{\tin}\!
\end{align}
Then 
\(\RC{\frac{1}{2}}{\Pcha{\tlx,\mA,\gX(\cdot)}}
=\tfrac{1}{4}\int\nolimits_{0}^{\tlx}(\sqrt{\gX(\tin)}-\sqrt{\mA})^{2}\dif{\tin}\)
and the Poisson channels \(\Wmn{[1,\blx]}\!=\!\Pcha{\blx,\mA,\gX(\cdot)}\) 
satisfy Assumption \ref{assumption:individual-ologn}  
provided that  \(\sup_{\tin\in(0,\tlx]}\gX(\tin)\) is \(\bigo{\ln \tlx}\).
Thus Theorem \ref{thm:productexponent} implies the SPB 
for the Poisson channel \(\Pcha{\tlx,\mA,\gX(\cdot)}\) asymptotically,
provided that  \(\sup_{\tin\in(0,\tlx]}\gX(\tin)\) is \(\bigo{\ln \tlx}\).

\subsection{Moment Bounds for Tilting}\label{sec:momentbound}
The tilted probability measures arise naturally in the trade off between  the exponents of the false alarm 
and the missed detection probabilities in the binary hypothesis testing problem with independent samples. 
We use them in the same vein with the help of the following bound.
\begin{lemma}\label{lem:MomentBound}
	Let \(\mW\) and \(\mQ\) be two probability measures on the measurable space \((\outS,\outA)\) such that 
	\(\RD{\sfrac{1}{2}}{\mW}{\mQ}<\infty\).
	Then
	\begin{align}
	\label{eq:lem:MomentBound}
	\EXS{\wma{\rno}{\mQ}}{\abs{\cln{\rno}}^{\knd}}^{\sfrac{1}{\knd}}
	&\leq 3^{\frac{1}{\knd}} \tfrac{((1-\rno)\RD{\rno}{\mW}{\mQ})\vee \knd}{\rno (1-\rno)}
	\end{align}
	for all \(\knd\in\reals{+}\) and \(\rno \in(0,1)\)
	where \(\wma{\rno}{\mQ}\) is the tilted probability measure given in \eqref{eq:def:tiltedprobabilitymeasure}
	and \(\cln{\rno}\) is defined using the Radon-Nikodym derivative 
	of \(\wmn{ac}\), i.e. the 
	component\footnote{Any \(\mW\) has a unique decomposition 
		\(\mW\!=\!\wmn{ac}\!+\!\wmn{s}\) such that \(\wmn{ac}\AC\mQ\) and  \(\wmn{s}\perp\mQ\)
		by the Lebesgue decomposition theorem, \cite[5.5.3]{dudley}.} 
	of \(\mW\) that is absolutely continuous in \(\mQ\), as follows 
	\begin{align}
	\notag
	\cln{\rno}
	&\DEF\ln \der{\wmn{ac}}{\mQ}-\EXS{\wma{\rno}{\mQ}}{\ln \der{\wmn{ac}}{\mQ}}.
	\end{align}
\end{lemma}

\begin{proof}
	Note that for any \(\gamma>0\),
	\begin{align}
\notag
	\EXS{\wma{\rno}{\mQ}}{\abs{\cln{\rno}}^{\knd}}
	&=\EXS{\wma{\rno}{\mQ}}{\IND{\cln{\rno}\!>\!\gamma} \abs{\cln{\rno}}^{\knd}}
	+\EXS{\wma{\rno}{\mQ}}{\IND{\abs{\cln{\rno}}\leq\gamma} \abs{\cln{\rno}}^{\knd}}
	\\
	\label{eq:MomentBound-0}
	&\hspace{1cm}+\EXS{\wma{\rno}{\mQ}}{\IND{\cln{\rno}\!<\!-\gamma} \abs{\cln{\rno}}^{\knd}}.
	\end{align}
	In the following,  we bound the expectations in the preceding expression 
	from above for an arbitrary \(\gamma\)
	and show that these bounds are not larger than \((\gamma_{0})^{\knd}\)
	for \(\gamma=\gamma_{0}\), see \eqref{eq:MomentBound-4} and \eqref{eq:MomentBound-6}, where
	\begin{align}
	%\label{eq:MomentBound-1}
	\notag
	\gamma_{0}&\DEF\tfrac{((1-\rno)\RD{\rno}{\mW}{\mQ}) \vee \knd}{\rno(1-\rno)}.
	\end{align}
	This will imply \(\EXS{\wma{\rno}{\mQ}}{\abs{\cln{\rno}}^{\knd}}\leq 3 (\gamma_{0})^{\knd}\)
	and thus \eqref{eq:lem:MomentBound}.

	Using the identity \(\der{\wma{\rno}{\mQ}}{\mW}=e^{(\rno-1) \cln{\rno}+\RD{1}{\wma{\rno}{\mQ}}{\mW}}\),
	we can bound the first expectation in \eqref{eq:MomentBound-0} for all  
	\(\gamma\geq 0\) and  \(\knd\geq 0\) as follows
	\begin{align}
	\notag
	\EXS{\wma{\rno}{\mQ}}{\IND{\cln{\rno}\!>\gamma}\!\abs{\cln{\rno}}^{\knd}} 
	&\!=\!\EXS{\mW}{\IND{\cln{\rno}\!>\gamma}\!\abs{\cln{\rno}}^{\knd} e^{(\rno-1)\cln{\rno}+\RD{1}{\wma{\rno}{\mQ}}{\mW}}}
	\\
	\label{eq:MomentBound-2}
	&\!\leq\! e^{\RD{1}{\wma{\rno}{\mQ}}{\mW}}
	\sup\nolimits_{\dsta>\gamma}e^{-(1-\rno)\dsta} \dsta^{\knd}.
	\end{align}
	On the other hand, for any  \(\beta>0\), \(\knd\geq 0\), and \(\gamma\geq 0\) we have
	\begin{align}
	\label{eq:MomentBound-3}
	\sup\nolimits_{\dsta > \gamma}   e^{-\beta \dsta} \dsta^{\knd}
	&=
	\begin{cases}
	(\tfrac{\knd}{e \beta})^{\knd}
	&\gamma\leq\tfrac{\knd}{\beta}\\
	e^{-\beta \gamma} \gamma^{\knd}
	&\gamma>\tfrac{\knd}{\beta}
	\end{cases}.
	\end{align}
	Using \eqref{eq:MomentBound-2} and \eqref{eq:MomentBound-3}
	for \(\beta=(1-\rno)\) and \(\gamma=\gamma_{0}\) and invoking 
	\((1-\rno)\RD{\rno}{\mW}{\mQ}=\rno \RD{1}{\wma{\rno}{\mQ}}{\mW}+(1-\rno)\RD{1}{\wma{\rno}{\mQ}}{\mQ}\)
	we get
	\begin{align}
	\notag
	\EXS{\wma{\rno}{\mQ}}{\IND{\cln{\rno}\!>\gamma}\abs{\cln{\rno}}^{\knd}} 
	&\!\leq\! e^{\RD{1}{\wma{\rno}{\mQ}}{\mW}-\frac{1-\rno}{\rno}\RD{\rno}{\mW}{\mQ}}
	(\gamma_{0})^{\knd}
	\\
	\label{eq:MomentBound-4}
	&=e^{-\frac{1-\rno}{\rno}\RD{1}{\wma{\rno}{\mQ}}{\mQ}} (\gamma_{0})^{\knd}.
	\end{align}
	Using the identity \(\der{\wma{\rno}{\mQ}}{\mQ}=e^{\rno \cln{\rno}+\RD{1}{\wma{\rno}{\mQ}}{\mQ}}\),
	we can bound the third expectation in \eqref{eq:MomentBound-0} for all  \(\gamma\geq 0\) and  \(\knd\geq 0\) as follows
	\begin{align}
	\notag
	\EXS{\wma{\rno}{\mQ}}{\IND{\cln{\rno}\!<-\gamma}\!\abs{\cln{\rno}}^{\knd}}
	&\!=\!\EXS{\mQ}{\IND{\cln{\rno}<-\gamma}\abs{\cln{\rno}}^{\knd}e^{\rno \cln{\rno}+\RD{1}{\wma{\rno}{\mQ}}{\mQ}}}
	\\
	\label{eq:MomentBound-5}
	&\!\leq\! e^{\RD{1}{\wma{\rno}{\mQ}}{\mQ}}
	\sup\nolimits_{\dsta>\gamma} e^{-\rno\dsta} \dsta^{\knd}.
	\end{align}
	Using \eqref{eq:MomentBound-3} and \eqref{eq:MomentBound-5}
	for \(\beta=\rno\) and \(\gamma=\gamma_{0}\) and  invoking 
	\((1-\rno)\RD{\rno}{\mW}{\mQ}=\rno \RD{1}{\wma{\rno}{\mQ}}{\mW}+(1-\rno)\RD{1}{\wma{\rno}{\mQ}}{\mQ}\)
	we get
	\begin{align}
	\notag
	\EXS{\wma{\rno}{\mQ}}{\IND{\cln{\rno}\!<-\gamma}\!\abs{\cln{\rno}}^{\knd}}
	&\!\leq\! e^{\RD{1}{\wma{\rno}{\mQ}}{\mQ}-\RD{\rno}{\mW}{\mQ}}
	(\gamma_{0})^{\knd}
	\\
	\label{eq:MomentBound-6}
	&\!=\! e^{-\frac{\rno}{1-\rno}\RD{1}{\wma{\rno}{\mQ}}{\mW}}(\gamma_{0})^{\knd}.
	\end{align}
\end{proof}

\subsection{A Corollary of the Berry-Esseen Theorem}\label{sec:smalldeviations}
In this subsection we derive a lower bound to the probability of having a small deviation 
from the mean for sums of independent random variables using the Berry-Esseen theorem. 
Let us start with recalling the Berry-Esseen theorem.

\begin{lemma}[Berry-Esseen Theorem \cite{berry41}, \cite{esseen42}, \cite{shevtsova10}]\label{lem:berryesseen}
Let \(\{\cln{\tin}\}_{\tin\in \integers{+}}\) be independent  random variables satisfying
\begin{align}
\notag
\EX{\cln{\tin}}
&=0
&
&\forall\tin
&
&\mbox{and}
&
\cm{2}
&\in\reals{+}
\end{align}
where \(\cm{\knd}=\left(\sum\nolimits_{\tin=1}^{\blx} \EX{\abs{\cln{\tin}}^{\knd}}\right)^{\sfrac{1}{\knd}}\).
Then there exists an absolute constant \(\omega\leq 0.5600\) such that 
\begin{align}
\notag
\abs{\PX{\sum\nolimits_{\tin=1}^{\blx}\cln{\tin}< \tau \cm{2}}-\GCD{\tau}}
&\leq \omega \left(\tfrac{\cm{3}}{\cm{2}}\right)^{3}
\end{align}
where \(\GCD{\tau} =\tfrac{1}{\sqrt{2\pi}} \int_{-\infty}^{\tau} e^{-\frac{\dsta^2}{2}} \dif{\dsta}.\)
\end{lemma}
The following lemma is obtained by applying the Berry-Esseen theorem for appropriately 
chosen values of \(\tau\); thus it is merely a corollary of the Berry-Esseen theorem.

\begin{lemma}\label{lem:berryesseenN}
Let \(\{\cln{\tin}\}_{\tin\in \integers{+}}\) be independent  zero mean random variables. Then
\begin{align}
\label{eq:lem:berryesseenN}
\PX{\abs{\sum\nolimits_{\tin=1}^{\blx}\cln{\tin}}<3\cm{\knd}}
&\geq  \tfrac{1}{2\sqrt{\blx}}
&~&&
&\forall \blx \in \integers{+},\knd\geq 3.
\end{align}
\end{lemma}
Augustin derived a similar bound, i.e. \cite[Thm. 18.2]{augustin78};
the proof of Lemma \ref{lem:berryesseenN} is similar to the proof of 
that bound.

\begin{proof}[Proof of Lemma \ref{lem:berryesseenN}]
If \(\sfrac{\cm{\knd}}{\cm{2}}\!\geq\!\sfrac{\sqrt{2}}{3}\), then using the Markov inequality we get
\begin{align}
\notag
\PX{\abs{\sum\nolimits_{\tin=1}^{\blx}\cln{\tin}}\leq 3 \cm{\knd}}
&\geq 1 -\PX{\abs{\sum\nolimits_{\tin=1}^{\blx}\cln{\tin}}>\cm{2}\sqrt{2} }
\\
\notag
&\geq \tfrac{1}{2}.
\end{align}
Thus \eqref{eq:lem:berryesseenN} holds. Hence, we assume that 
\(\sfrac{\cm{\knd}}{\cm{2}}\!<\!\sfrac{\sqrt{2}}{3}\)
for the rest of the proof. 
By the Berry-Esseen theorem we have 
\begin{align}
\notag
&\hspace{-.3cm}\PX{\abs{\sum\nolimits_{\tin=1}^{\blx}\cln{\tin}}\leq 3 \cm{\knd}}
\\
\notag
&\geq \left(\GCD{\tfrac{3\cm{\knd}}{\cm{2}}} -\omega \left(\tfrac{\cm{3}}{\cm{2}}\right)^{3}\right)
-     \left(\GCD{\tfrac{-3\cm{\knd}}{\cm{2}}}+\omega \left(\tfrac{\cm{3}}{\cm{2}}\right)^{3}\right)
\\
\notag
&=2 \left[\int_{0}^{\frac{3\cm{\knd}}{\cm{2}}} \tfrac{e^{-\frac{\dsta^2}{2}}}{\sqrt{2\pi}}  \dif{\dsta}
-\omega \left(\tfrac{\cm{3}}{\cm{2}}\right)^{3} \right].
\end{align}
On the other hand
\(\int_{0}^{\tau} e^{-\frac{\dsta^2}{2}}\dif{\dsta}
%\geq  \tau e^{-\frac{1}{\tau}\int_{0}^{\tau} \frac{\dsta^2}{2} \dif{\dsta}}
\geq  \tau e^{-\sfrac{\tau^{2}}{6}} 
\) 
by the Jensen's inequality because 
the exponential function is convex.
Thus
\begin{align}
\label{eq:berryesseenN-A}
\!\PX{\abs{\sum\nolimits_{\tin=1}^{\blx}\cln{\tin}}\!\leq\!3 \cm{\knd}}
&\!\geq\!2\!\left[
\tfrac{\sfrac{3\cm{\knd}}{\cm{2}}}{\sqrt{2\pi}}
e^{-\frac{(\sfrac{3\cm{\knd}}{\cm{2}})^{2}}{6}}
\!-\!\omega \left(\tfrac{\cm{3}}{\cm{2}}\right)^{3} \right].
\end{align}

Since \(\knd\geq3\), the H\"{o}lder's inequality implies
\begin{align}
\notag
\EX{\sum\nolimits_{\tin=1}^{\blx}\abs{\cln{\tin}}^{3}}
&\leq 
\EX{\sum\nolimits_{\tin=1}^{\blx}\abs{\cln{\tin}}^{\knd}}^{\frac{1}{\knd-2}}
\EX{\sum\nolimits_{\tin=1}^{\blx}\abs{\cln{\tin}}^{2}}^{\frac{\knd-3}{\knd-2}}.
\end{align}
Then \((\tfrac{\cm{3}}{\cm{2}})^{3}\!\leq\!(\tfrac{\cm{\knd}}{\cm{2}})^{\frac{\knd}{\knd-2}}\).
Thus using  \(\tfrac{\cm{\knd}}{\cm{2}}\!<\!\tfrac{\sqrt{2}}{3}\)
and \eqref{eq:berryesseenN-A} we get
\begin{align}
\notag
\PX{\abs{\sum\nolimits_{\tin=1}^{\blx}\cln{\tin}}\leq 3 \cm{\knd}}
&\geq  2 \left[\tfrac{e^{-\sfrac{1}{3}}}{\sqrt{2\pi}}3-0.56 \right] \tfrac{\cm{\knd}}{\cm{2}}
\\
\label{eq:berryesseenN-E}
&\geq 0.595 \tfrac{\cm{\knd}}{\cm{2}}.
\end{align}
In order to bound \(\sfrac{\cm{\knd}}{\cm{2}}\) we use the  
Jensen's inequality and the concavity of the functions \(\dsta^{\rno}\)
in \(\dsta\) for \(\rno\in (0,1]\).
\begin{align}
\notag
\EX{\sum\nolimits_{\tin=1}^{\blx}\tfrac{1}{\blx}  \abs{\cln{\tin}}^{2}}^{\sfrac{1}{2}}
%&=\EX{\sum\nolimits_{\tin=1}^{\blx}\tfrac{1}{\blx}  \abs{\cln{\tin}}^{\knd \frac{2}{\knd}}}^{\sfrac{1}{2}}
%\\
%\notag
&\!\leq\!\EX{\sum\nolimits_{\tin=1}^{\blx}\tfrac{1}{\blx}  \abs{\cln{\tin}}^{\knd}}^{\sfrac{1}{\knd}}
&~&&
&\forall \knd\geq 2.
\end{align}
Then \(\sfrac{\cm{\knd}}{\cm{2}}\geq 
\blx^{\frac{1}{\knd}-\frac{1}{2}}\geq \sfrac{1}{\sqrt{\blx}}\)
and
\eqref{eq:lem:berryesseenN} follows from \eqref{eq:berryesseenN-E}.
\end{proof}

\subsection[Non-asymptotic SPBs]{Non-asymptotic SPBs for Product Channels}\label{sec:pouterbound}
The ultimate aim of this subsection is to establish the asymptotic SPB
given in Theorem \ref{thm:productexponent}.
First we establish a non-asymptotic SPB
for product channels, i.e. Lemma \ref{lem:spb-product},
using Lemmas \ref{lem:MomentBound} and \ref{lem:berryesseenN}, 
the intermediate value theorem, and pigeon hole arguments.
We prove Theorem \ref{thm:productexponent} using Lemma \ref{lem:spb-product}
at the end of this subsection. We make a brief digression
before that proof and point out three variants of Lemma \ref{lem:spb-product}
that are proved without invoking the averaging scheme described in \S \ref{sec:averaging}.
Lemma \ref{lem:spb-monotone-center} is for product channels 
satisfying
\begin{align}
\label{eq:lem:spb-monotone-center:hypothesis}
e^{\frac{\rno-1}{\rno}\RC{\rno}{\Wmn{[1,\blx]}}}\qmn{\rno,\Wmn{[1,\blx]}}
&\leq e^{\frac{\dsta-1}{\dsta}\RC{\dsta}{\Wmn{[1,\blx]}}}\qmn{\dsta,\Wmn{[1,\blx]}}
\end{align}
for all \(\rno\!\leq\!\dsta\)  in \((0,1)\).
Lemma \ref{lem:spb-constant-center} is for product channels whose \renyi centers do not 
change with the order.
Lemma \ref{lem:spb-fixed-density} establishes the SPB given in Lemma \ref{lem:spb-constant-center}
for  product channels with feedback, under a stronger hypothesis.

\begin{lemma}\label{lem:spb-product}
Let \(\Wmn{[1,\blx]}\) be a product channel
for an \(\blx\!\in\!\integers{+}\),
\(\rnf\in(0,1)\), \(\epsilon\in(0,\tfrac{\blx}{\blx+1})\), \(\knd\geq 3\),
and \(\gamma\) be
\vspace{-.1cm}
\begin{align}
\label{eq:lem:spb-product:gamma}
\gamma
&\DEF\tfrac{3\sqrt[\knd]{3}}{1-\epsilon}\left(\sum\nolimits_{\tin=1}^{\blx} \left(\RC{1/2}{\Wmn{\tin}}\vee \knd \right)^{\knd} \right)^{\sfrac{1}{\knd}}.
%%%\\ 
%%%\notag
%%%&=
%%%\left[\tfrac{3\sqrt[\knd]{3}}{1-\epsilon}\left(\sum\nolimits_{\tin=1}^{\blx} \left(\RC{1/2}{\Wmn{\tin}}\vee \knd \right)^{\knd} \right)^{\frac{1}{\knd}}\right]
%%%\wedge
%%%\left[3\sqrt[\knd]{2 \knd!} \left(\sum\nolimits_{\tin=1}^{\blx} e^{\frac{\RC{1/2}{\Wmn{\tin}}}{1-\epsilon}}\right)^{\frac{1}{\knd}}\right].
\end{align}
\vspace{-.1cm}
If  \(M\) and \(L\) are integers such that
\(\tfrac{M}{L}> 16 \sqrt{\blx} e^{\RCI{\rnf}{\Wmn{[1,\blx]}}{\epsilon}\!+\!\frac{\gamma}{1-\rnf}}\),
then any \((M,L)\) channel code on \(\Wmn{[1,\blx]}\) satisfies
\vspace{-.1cm}
\begin{align}
\label{eq:lem:spb-product}
\Pem{av}
&\geq \left(\tfrac{\epsilon e^{-2\gamma}}{16 e^{2}\blx^{3/2}}\right)^{\sfrac{1}{\rnf}}  e^{-\spa{\epsilon}{\ln\frac{M}{L},\Wmn{[1,\blx]}}}  
\end{align}
\vspace{-.1cm}
for \(\spa{\epsilon}{\rate,\Wm}\)  defined in \eqref{eq:def:avspherepacking}.
\end{lemma} 
We have presented Lemma \ref{lem:spb-product} via \eqref{eq:lem:spb-product}
in order to emphasize its similarity to the Gallager's bound, 
i.e. Lemma \ref{lem:gallager}. 
However, the expression on the right hand side of \eqref{eq:lem:spb-product} converges to zero as 
\(\rnf\) converges to zero because \(\tfrac{\epsilon e^{-2\gamma}}{16 e^{2}\blx^{3/2}}<1\).
By changing the analysis slightly it is possible to obtain the following alternative bound
which does not have that problem:
\begin{align}
\label{eq:lem:spb-product-alt}
\Pem{av}
&\geq \tfrac{\epsilon e^{-2\gamma}}{16\blx^{3/2}}  e^{-\spa{\epsilon}{\rate,\Wmn{[1,\blx]}}}  
\end{align}
where \(\rate=\ln \tfrac{M}{L}-2\gamma-\ln \tfrac{16 e^{2}\blx^{3/2}}{\epsilon}\).
The bound given in \eqref{eq:lem:spb-product-alt} is preferable especially 
for codes with low rates on  channels satisfying 
\(\lim_{\rate\downarrow 0} \spe{\rate,\Wmn{[1,\blx]}}<\infty\). 

We can make \(\RCI{\rnf}{\Wmn{[1,\blx]}}{\epsilon}\) and \(\spa{\epsilon}{\ln\tfrac{M}{L},\Wmn{[1,\blx]}}\)
as close as we please to \(\RC{\rnf}{\Wmn{[1,\blx]}}\) and \(\spe{\ln\tfrac{M}{L},\Wmn{[1,\blx]}}\) by choosing 
\(\epsilon\) small enough. 
But as we decrease \(\epsilon\), the first term of the lower bound in  \eqref{eq:lem:spb-product} 
also decreases. The appropriate choice of \(\epsilon\) balances these two effects. 
The choice of \(\kappa\) influences the constraint on \(\ln \tfrac{M}{L}\) and 
the lower bound in  \eqref{eq:lem:spb-product}  only through the value of \(\gamma\). 
The appropriate choice of \(\kappa\) minimizes  the value of \(\gamma\).
The constraint on \(\ln \tfrac{M}{L}\) becomes easier to satisfy as \(\rnf\) decreases;
however, the smaller values of \(\rnf\) also lead to smaller, i.e. worse, 
lower bounds in \eqref{eq:lem:spb-product}. 
Thus for a given \(\ln \tfrac{M}{L}\) we desire to have the greatest possible
value for \(\rnf\) to have the best bound in \eqref{eq:lem:spb-product}.
As an example consider the stationary case when
\(\Wmn{\tin}=\Wm\) for all \(\tin\)
and set  \(\knd\!=\!\ln\blx\) and \(\epsilon\!=\!\tfrac{1}{\blx}\).
Invoking \eqref{eq:avcapacity:bound} to bound \(\RCI{\rnf}{\Wmn{[1,\blx]}}{\epsilon}\)
and
Lemma \ref{lem:avspherepacking} to bound \(\spa{\epsilon}{\rate,\Wmn{[1,\blx]}}\),
we get the following:
For any \(\blx\geq 10\), if 
\begin{align}
\notag
\tfrac{1}{\blx}\ln\tfrac{M}{L}
&\geq 
\tfrac{\ln 16\sqrt{\blx}}{\blx}
+\RC{\rnf}{\Wm}
+\tfrac{\RC{\rnf}{\Wm}+  13.2 \rnf (\RC{\sfrac{1}{2}}{\Wm}\vee \ln\blx)}{(\blx-1)\rnf(1-\rnf)}
\end{align}
for a \(\rnf\!\in\!(\sfrac{1}{\blx},1)\) then any \((M,L)\) channel code on \(\Wmn{[1,\blx]}\) satisfies
\begin{align}
\label{eq:lem:spb-product:stationary}
\Pe
&\geq \left(\tfrac{(\sfrac{L}{M})^{\frac{1}{(\blx-1)\rnf}}}{16 e^{2}(\blx\vee e^{\RC{\sfrac{1}{2}}{\Wm}})^{29}}\right)^{\sfrac{1}{\rnf}}
e^{-\blx\spe{\frac{1}{\blx}\ln\frac{M}{L},\Wm}}.
\end{align}

\begin{proof}[Proof of Lemma \ref{lem:spb-product} and  \eqref{eq:lem:spb-product-alt}]
Let \((\enc,\dec)\) be an \((M,L)\) channel code on \(\Wmn{[1,\blx]}\). 
In order to avoid lengthy expressions we denote \(\Wmn{[1,\blx]}(\enc(\dmes))\)
by \(\wma{}{\dmes}\) and its marginal in \(\pmea{\outA_{\tin}}\)
by \(\wma{\tin}{\dmes}\).
	Let us describe \(\vma{\rno,\tin}{\dmes}\in\pmea{\outA_{\tin}}\) 
	through its Radon-Nikodym derivative:
	\begin{align}
	\notag
	\der{\vma{\rno,\tin}{\dmes}}{\rfm}
	&\!\DEF\!e^{(1-\rno)\RD{\rno}{\wma{\tin}{\dmes}}{\qma{\rno,\Wmn{\tin}}{\epsilon}}} 
	\!\left(\!\der{\wma{\tin}{\dmes}}{\rfm}\!\right)^{\!\rno}
	\left(\!\der{\qma{\rno,\Wmn{\tin}}{\epsilon}}{\rfm}\!\right)^{\!1-\rno}
	\end{align}
	where \(\qma{\rno,\Wmn{\tin}}{\epsilon}\) is the average \renyi center of \(\Wmn{\tin}\)
	and \(\rfm\) is any probability measure satisfying  both \(\wma{\tin}{\dmes}\AC\rfm\) and 
	\(\qma{\rno,\Wmn{\tin}}{\epsilon}\AC\rfm\).
	
	Let \(\cla{\rno,\tin}{\dmes}\) be a random variable defined as
	\begin{align}
	\notag
	\cla{\rno,\tin}{\dmes}
	&\DEF \ln \der{(\wma{\tin}{\dmes})_{ac}}{\qma{\rno,\Wmn{\tin}}{\epsilon}}-\EXS{\vma{\rno,\tin}{\dmes}}{\ln \der{(\wma{\tin}{\dmes})_{ac}}{\qma{\rno,\Wmn{\tin}}{\epsilon}}}
	\end{align}
	where \((\wma{\tin}{\dmes})_{ac}\) is the component of \(\wma{\tin}{\dmes}\) 
	that is absolutely continuous in \(\qma{\rno,\Wmn{\tin}}{\epsilon}\). 
	Note that \(\cla{\rno,\tin}{\dmes}\) can also be written as follows:
	\begin{align}
	\label{eq:spb-product-01}
	\cla{\rno,\tin}{\dmes}
	&=
	\tfrac{1}{\rno-1}
	\left(\ln \der{\vma{\rno,\tin}{\dmes}}{\wma{\tin}{\dmes}}-\RD{1}{\vma{\rno,\tin}{\dmes}}{\wma{\tin}{\dmes}}\right)
	\\
	\label{eq:spb-product-02}
	&=
	\tfrac{1}{\rno}
	\left(\ln \der{\vma{\rno,\tin}{\dmes}}{\qma{\rno,\Wmn{\tin}}{\epsilon}}-\RD{1}{\vma{\rno,\tin}{\dmes}}{\qma{\rno,\Wmn{\tin}}{\epsilon}}\right).
	\end{align}
	Let \(\qmn{\rno}\) and \(\vma{\rno}{\dmes}\) be the probability measures defied as
	\begin{align}
	\notag
	\qmn{\rno}
	&\DEF \bigotimes\nolimits_{\tin=1}^{\blx} \qma{\rno,\Wmn{\tin}}{\epsilon},
	\\
	\notag
	\vma{\rno}{\dmes} 
	&\DEF \bigotimes\nolimits_{\tin=1}^{\blx} \vma{\rno,\tin}{\dmes}.
	\end{align}
	Let \(\cla{\rno}{\dmes}\) be a random variable in 
	\((\outS_{1}^{\blx},\outA_{1}^{\blx},\vma{\rno}{\dmes})\)
	\begin{align}
	\notag
	\cla{\rno}{\dmes}
	&\DEF \sum\nolimits_{\tin=1}^{\blx} \cla{\rno,\tin}{\dmes}.
	\end{align}
	As a result of \eqref{eq:spb-product-01}, \eqref{eq:spb-product-02}, and  the product structure of 
	\(\qmn{\rno}\), \(\vma{\rno}{\dmes}\), and \(\wma{}{\dmes}\) we have
	\begin{align}
	\label{eq:spb-product-03}
	\cla{\rno}{\dmes}
	&\!=\!
	\tfrac{1}{\rno-1}
	\left[\ln \der{\vma{\rno}{\dmes}}{\wma{}{\dmes}}\!-\!\RD{1}{\vma{\rno}{\dmes}}{\!\wma{}{\dmes}}\right],
	\\
	\label{eq:spb-product-04}
	&\!=\!
	\tfrac{1}{\rno}
	\left[\ln \der{\vma{\rno}{\dmes}}{\qmn{\rno}}\!-\!\RD{1}{\vma{\rno}{\dmes}}{\qmn{\rno}}\right].
	\end{align}
	Let \(\oev_{\dmes}\!\in\!\outA_{1}^{\blx}\)  be  \(\oev_{\dmes}\!\DEF\!\{\dout\in\outS_{1}^{\blx}:\dmes\in\dec(\dout)\}\). Then 
	for any \(\rno\!\in\!(0,1)\) and real numbers \(\tau_{1}\) and \(\tau_{2}\) we have  
	\begin{align}
	\notag
	\Pem{\dmes}
	&\!\geq\!e^{-\RD{1}{\vma{\rno}{\dmes}}{\wma{}{\dmes}}-\tau_{1}}
	\EXS{\vma{\rno}{\dmes}}{\IND{\outS_{1}^{\blx}\setminus \oev_{\dmes}}\!
	\IND{\cla{\rno}{\dmes}\geq \frac{\tau_{1}}{\rno-1}}},
	\\
	\notag
	\qmn{\rno}(\!\oev_{\dmes}\!)
	&\!\geq\!e^{-\RD{1}{\vma{\rno}{\dmes}}{\qmn{\rno}}-\tau_{2}}
	\EXS{\vma{\rno}{\dmes}}{\IND{\oev_{\dmes}}\! 
	\IND{\cla{\rno}{\dmes}\leq \frac{\tau_{2}}{\rno}}}.
	\end{align}
	Then for \(\tau_{1}=\tfrac{\gamma}{\rno}\) and
	\(\tau_{2}=\tfrac{\gamma}{1-\rno}\) using \eqref{eq:spb-product-03} and \eqref{eq:spb-product-04} we get
	\begin{align}
	\notag
	&
	\Pem{\dmes}
	e^{\RD{1}{\vma{\rno}{\dmes}}{\wma{}{\dmes}}+\frac{\gamma}{\rno}}
	+
	\qmn{\rno}(\oev_{\dmes})
	e^{\RD{1}{\vma{\rno}{\dmes}}{\qmn{\rno}}+\frac{\gamma}{1-\rno}}
	\\
	\label{eq:spb-product-05}
	&
	\hspace{2.5cm}\geq
	\EXS{\vma{\rno}{\dmes}}{\IND{\frac{-\gamma}{(1-\rno)\rno}\leq \cla{\rno}{\dmes}\leq \frac{\gamma}{(1-\rno)\rno}}}.
	\end{align}
	The random variables \(\cla{\rno,\tin}{\dmes}\) are zero mean in the probability space 
	\((\outS_{1}^{\blx},\outA_{1}^{\blx},\vma{\rno}{\dmes})\) by construction.
	Furthermore, they are jointly independent because of the product structure of the probability measures 
	\(\wma{}{\dmes}\), \(\qmn{\rno}\), and \(\vma{\rno}{\dmes}\).
	Thus we can apply Lemma \ref{lem:berryesseenN} to bound the right hand side of \eqref{eq:spb-product-05} 
	from below, if we can show that \(\gamma\) defined in \eqref{eq:lem:spb-product:gamma}
	is large enough. To that end, first note that 
	\(\EXS{\vma{\rno}{\dmes}}{\abs{\cla{\rno,\tin}{\dmes}}^{\knd}}=\EXS{\vma{\rno,\tin}{\dmes}}{\abs{\cla{\rno,\tin}{\dmes}}^{\knd}}\)
	by construction. Then Lemma \ref{lem:MomentBound} implies
	\begin{align}
	\label{eq:spb-product-06}
	\EXS{\vma{\rno}{\dmes}}{\abs{\cla{\rno,\tin}{\dmes}}^{\knd}}
	&\leq 3 \left[\tfrac{\left((1-\rno)\RD{\rno}{\wma{\tin}{\dmes}}{\qma{\rno,\Wmn{\tin}}{\epsilon}} \right)\vee \knd}{\rno (1-\rno)}\right]^{\knd}.
	\end{align}
	We can bound \(\RD{\rno}{\wma{\tin}{\dmes}}{\qma{\rno,\Wmn{\tin}}{\epsilon}}\) using 
	\eqref{eq:lem:avcapacityB} of Lemma \ref{lem:avcapacity}
	\begin{align}
	\label{eq:spb-product-07}
	(1-\rno)\RD{\rno}{\wma{\tin}{\dmes}}{\qma{\rno,\Wmn{\tin}}{\epsilon}}
	&\leq  \tfrac{\RC{1/2}{\Wmn{\tin}}}{1-\epsilon}.
	\end{align}
	Using the definition of \(\gamma\) given in 
	\eqref{eq:lem:spb-product:gamma} and 
	together with \eqref{eq:spb-product-06} and \eqref{eq:spb-product-07} we get
	\begin{align}
	\notag
	3\left(\sum\nolimits_{\tin=1}^{\blx}\EXS{\vma{\rno}{\dmes}}{\abs{\cla{\rno,\tin}{\dmes}}^{\knd}} \right)^{\sfrac{1}{\knd}}
	&\leq \tfrac{\gamma}{\rno(1-\rno)}. 
	\end{align}
	Then \eqref{eq:spb-product-05} and Lemma \ref{lem:berryesseenN} implies
	\begin{align}
	\label{eq:spb-product-08}
	\Pem{\dmes}
	e^{\RD{1}{\vma{\rno}{\dmes}}{\wma{}{\dmes}}+\frac{\gamma}{\rno}}
	+
	\qmn{\rno}(\oev_{\dmes})
	e^{\RD{1}{\vma{\rno}{\dmes}}{\qmn{\rno}}+\frac{\gamma}{1-\rno}}
	&\geq \tfrac{1}{2\sqrt{\blx}}.
	\end{align}
	On the other hand,  the product structure of the probability measures \(\wma{}{\dmes}\) and \(\qmn{\rno}\)
	implies
	\begin{align}
	\notag
	\RD{\rno}{\wma{}{\dmes}}{\qmn{\rno}}
	&=\sum\nolimits_{\tin=1}^{\blx} 
	\RD{\rno}{\wma{\tin}{\dmes}}{\qma{\rno,\Wmn{\tin}}{\epsilon}}.
	\end{align}
	Bounding each term in the sum using  \eqref{eq:lem:avcapacity} of Lemma \ref{lem:avcapacity} 
	and then invoking \eqref{eq:avcapacity:additivity} we get
	\begin{align}
	\label{eq:spb-product-09}
	\RD{\rno}{\wma{}{\dmes}}{\qmn{\rno}}
	&\leq \RCI{\rno}{\Wmn{[1,\blx]}}{\epsilon}.
	\end{align}
	
	In the following, we show that the message set has a size \(\approx\tfrac{M \epsilon}{\blx}\) 
	subset in which all the messages has a conditional error probability greater than
	\(\approx (\tfrac{e^{-2\gamma}}{\sqrt{\blx}})^{\frac{1}{\rnf}}
	(\tfrac{\epsilon}{\blx})^{\frac{1-\rnf}{\rnf}}
	e^{-\spa{\epsilon}{\rate,\Wmn{[1,\blx]}}}\).
	The existence of such a subset will imply \eqref{eq:lem:spb-product}.
	We prove the existence of such a subset using \eqref{eq:spb-product-08}, \eqref{eq:spb-product-09},
	the intermediate value theorem, and pigeon hole arguments.
	Let us consider the subset of the message set, \(\mesS_1\) defined as follows:
	\begin{align}
	\notag
	\mesS_1
	&\!\DEF\!\!\left\{\!\dmes\!:\!\!\inf\limits_{\rno\in(\rnf,1)}
	\!\left[\!(\qmn{\rno}(\!\oev_{\dmes}\!)\!+\!\tfrac{L}{M})
	e^{\RD{1}{\vma{\rno}{\dmes}}{\qmn{\rno}}+\frac{\gamma}{1-\rno}}\!\right]\!
	\geq\!\tfrac{1}{4\sqrt{\blx}}\!\right\}\!.
	\end{align}
	First, we  bound the size of \(\mesS_{1}\) from above. We can bound \(\RD{1}{\vma{\rno}{\dmes}}{\qmn{\rno}}\)
	using the definitions of \(\vma{\rno,\tin}{\dmes}\), \(\vma{\rno}{\dmes}\), \(\qmn{\rno}\),
	the non-negativity of the \renyi divergence for probability measures,
	which is implied by Lemma \ref{lem:divergence-pinsker},
	and  \eqref{eq:spb-product-09}, as follows
	\begin{align}
	\label{eq:spb-product-010}
	\RD{1}{\vma{\rno}{\dmes}}{\qmn{\rno}}
	&=\RD{\rno}{\wma{}{\dmes}}{\qmn{\rno}}
	-\tfrac{\rno}{1-\rno}\RD{1}{\vma{\rno}{\dmes}}{\wma{}{\dmes}}
	\\
	\label{eq:spb-product-011}
	&\leq \RCI{\rno}{\Wmn{[1,\blx]}}{\epsilon}
	\end{align}
	for all \(\dmes \in \mesS,\rno\in (0,1)\).
	Then summing the inequality in the condition for membership of \(\mesS_{1}\) over the members of 
	\(\mesS_{1}\) we get 
	\begin{align}
	\notag
	2 L e^{\RCI{\rno}{\Wmn{[1,\blx]}}{\epsilon}+\frac{\gamma}{1-\rno}}
	&\geq \abs{\mesS_{1}}  \tfrac{1}{4\sqrt{\blx}}
	&~&&
	&\forall \rno\in(\rnf,1).
	\end{align}
	Then \(\tfrac{\abs{\mesS_{1}}  }{L}
	\leq 8\sqrt{\blx} e^{\RCI{\rnf}{\Wmn{[1,\blx]}}{\epsilon}+\frac{\gamma}{1-\rnf}}\).
	Consequently \(\abs{\mesS_{1}}<\tfrac{M}{2}\) because 
	\(\tfrac{M}{L}> 16 \sqrt{\blx} e^{\RCI{\rnf}{\Wmn{[1,\blx]}}\epsilon+\frac{\gamma}{1-\rnf}}\) 
	by the hypothesis. 
	
	On the other hand, as a result of the definition of \(\mesS_{1}\), for each \(\dmes\in \mesS\setminus\mesS_{1}\) there is 
	an \(\rno\in(\rnf,1)\) satisfying
	\begin{align}
	\notag
	(\qmn{\rno}(\oev_{\dmes})+\tfrac{L}{M})e^{\RD{1}{\vma{\rno}{\dmes}}{\qmn{\rno}}+\frac{\gamma}{1-\rno}}
	&<\tfrac{1}{4\sqrt{\blx}}.
	\end{align}
	Furthermore, \(\qmn{\rno}\) is continuous in \(\rno\) for the total variation topology 
	on \(\pmea{\outA_{1}^{\blx}}\)
	by construction.\footnote{In particular 
		\(\lon{\qmn{\rno}-\qmn{\rnt}}\!\leq\! \sqrt{8 \blx \ln \tfrac{\epsilon}{\epsilon-\abs{\rno-\rnt}}}\)
		for all \(\rno\) and \(\rnt\) in \((0,1)\) such that \(\abs{\rnt-\rno}<\epsilon\) by 
		\eqref{eq:shiryaev}, the product structure of \(\qmn{\rno}\),
		and the definition of the average \renyi center,
		which implies
		\(\lon{\qma{\rno,\Wmn{\tin}}{\epsilon}-\qma{\rnt,\Wmn{\tin}}{\epsilon}}\!\leq\!\tfrac{2\abs{\rno-\rnt}}{\epsilon}\)
		for all \(\Wmn{\tin}\).}
	Then \(\RD{1}{\vma{\rno}{\dmes}}{\qmn{\rno}}\) is continuous in \(\rno\) by Lemma \ref{lem:tilting}
	and  \(\qmn{\rno}\!(\oev_{\dmes})\) is continuous in \(\rno\), as well.
	Thus \((\qmn{\rno}(\oev_{\dmes})\!+\!\tfrac{L}{M})e^{\RD{1}{\vma{\rno}{\dmes}}{\qmn{\rno}}+\frac{\gamma}{1-\rno}}\)
	is continuous in \(\rno\).
	Since \(\lim_{\rno\uparrow1}(\qmn{\rno}(\oev_{\dmes})+\tfrac{L}{M})e^{\RD{1}{\vma{\rno}{\dmes}}{\qmn{\rno}}+\frac{\gamma}{1-\rno}}=\infty\),
	using the intermediate value theorem \cite[4.23]{rudin}
	we can conclude that for each \(\dmes\in \mesS\setminus\mesS_{1}\) 
	there exists an \(\rno_{\dmes}\!\in\!(\rnf,1)\) such that 
	\begin{align}
	\label{eq:spb-product-012}
	(\qmn{\rno_{\dmes}}(\oev_{\dmes})+\tfrac{L}{M})e^{\RD{1}{\vma{\rno_{\dmes}}{\dmes}}{\qmn{\rno_{\dmes}}}+\frac{\gamma}{1-\rno_{\dmes}}}
	&=\tfrac{1}{4\sqrt{\blx}}.
	\end{align}
	Then for any positive integer \(K\), there exists a length \(\tfrac{1}{K}\) interval with  
	\(\lceil\tfrac{M}{2 K}\rceil\) or more \(\rno_{\dmes}\)'s. 
	Let \([\rnt,\rnt+\tfrac{1}{K}]\) be the aforementioned interval,
	\(\tilde{\epsilon}\) and \(\tilde{\rno}\) be real numbers in 
	\((0,1)\) 
	\begin{align}
	\notag
	\tilde{\epsilon}
	&\DEF\tfrac{1}{K}+\epsilon(1-\tfrac{1}{K}),
	&
	\notag
	\tilde{\rno}
	&\DEF\tfrac{1-\epsilon}{1-\tilde{\epsilon}}\rnt.
	\end{align}
	Then \(\qma{\rno,\Wmn{\tin}}{\epsilon}\leq\tfrac{\tilde{\epsilon}}{\epsilon}\qma{\tilde{\rno},\Wmn{\tin}}{\tilde{\epsilon}}\)
	for all \(\rno\in[\rnt,\rnt+\tfrac{1}{K}]\),	
	by the definition of the average \renyi center.
	Thus 
	\begin{align}
	\notag
	\qmn{\rno}
	&\leq
	(\tfrac{\tilde{\epsilon}}{\epsilon})^{\blx}\tilde{\mQ}
	&
	&\forall \rno\in[\rnt,\rnt+\tfrac{1}{K}]
	\end{align}
	where \(\tilde{\mQ}\in\pmea{\outA_{1}^{\blx}}\) is defined as follows
	\begin{align}
	\notag
	\tilde{\mQ}
	&\DEF\bigotimes\nolimits_{\tin=1}^{\blx} \qma{\tilde{\rno},\Wmn{\tin}}{\tilde{\epsilon}}.
	\end{align}
	On the other hand, at least half of the messages with \(\rno_{\dmes}\)'s in \([\rnt,\rnt+\tfrac{1}{K}]\)  
	satisfy \(\tilde{\mQ}(\oev_{\dmes})\leq 2 \tfrac{L}{\lceil \sfrac{M}{2 K} \rceil}\).
	Then at least \(\lceil\tfrac{1}{2}\lceil\tfrac{M}{2 K}\rceil\rceil\) messages
	with \(\rno_{\dmes}\)'s in \([\rnt,\rnt+\tfrac{1}{K}]\) satisfy
	\begin{align}
	\notag
	\qmn{\rno_{\dmes}}(\oev_{\dmes})
	&\leq \tfrac{4L}{M} K \left(1+\tfrac{1}{K}\tfrac{1-\epsilon}{\epsilon}\right)^{\blx}
	\end{align}
	Note that \(\tfrac{\blx(1-\epsilon)}{\epsilon}\!>\!1\) because 
	\(\epsilon\!<\!\tfrac{\blx}{\blx+1}\). Then we can set
	\(K\) to \(\!\lfloor\tfrac{\blx(1-\epsilon)}{\epsilon}\rfloor\) and use the identity 
	\((1+\dsta)^{\sfrac{1}{\dsta}}<e\) to get
	\begin{align}
	\label{eq:spb-product-013}
	\qmn{\rno_{\dmes}}(\oev_{\dmes})
	&\leq \tfrac{4L}{M}\tfrac{\blx(1-\epsilon)}{\epsilon} e^{2}.
	\end{align}
	Then using  \eqref{eq:spb-product-012}, 
	we can bound \(\RD{1}{\vma{\rno_{\dmes}}{\dmes}}{\qmn{\rno_{\dmes}}}\) for all 
	\(\dmes\) satisfying \eqref{eq:spb-product-013} as follows
	\begin{align}
	\label{eq:spb-product-014}
	e^{\RD{1}{\vma{\rno_{\dmes}}{\dmes}}{\qmn{\rno_{\dmes}}}+\frac{\gamma}{1-\rno_{\dmes}}}
	&\geq  \tfrac{1}{4\sqrt{\blx}} \tfrac{\epsilon}{4 e^{2} \blx}\tfrac{M}{L}.
	\end{align} 
	On the other hand  we can bound \(\Pem{\dmes}\) using \eqref{eq:spb-product-08} and \eqref{eq:spb-product-012} 
	\begin{align}
	\label{eq:spb-product-015}
	\Pem{\dmes}
	e^{\RD{1}{\vma{\rno_{\dmes}}{\dmes}}{\enc(\dmes)}+\frac{\gamma}{\rno_{\dmes}}}
	&\geq \tfrac{1}{4\sqrt{\blx}}.
	\end{align}
	Using  \eqref{eq:spb-product-09}, \eqref{eq:spb-product-010}, \eqref{eq:spb-product-014}, and \eqref{eq:spb-product-015} we get
	\begin{align}
	\label{eq:spb-product-016}
	\Pem{\dmes}
	e^{\frac{1-\rno_{\dmes}}{\rno_{\dmes}}\RCI{\rno_{\dmes}}{\Wmn{[1,\blx]}}{\epsilon}+2\frac{\gamma}{\rno_{\dmes}}}
	&\geq (\tfrac{1}{4\sqrt{\blx}})^{\frac{1}{\rno_{\dmes}}} \left(\tfrac{\epsilon}{4 e^{2}\blx}\tfrac{M}{L}\right)^{\frac{1-\rno_{\dmes}}{\rno_{\dmes}}}.
	\end{align}
	Hence, for all  \(\dmes\) satisfying \eqref{eq:spb-product-013} as a result of 
	the definition of \(\spa{\epsilon}{\rate,\Wm}\) given in \eqref{eq:def:avspherepacking}  we have
	\begin{align}
	\notag
	\Pem{\dmes}
	&\geq
	e^{-2\frac{\gamma}{\rnf}}
	(\tfrac{1}{4\sqrt{\blx}})^{\frac{1}{\rnf}} 
	\left(\tfrac{\epsilon}{4 e^{2}\blx}\right)^{\frac{1-\rnf}{\rnf}}
	e^{-\spa{\epsilon}{\rate,\Wmn{[1,\blx]}}}
	\end{align}
	where \(\rate=\ln \tfrac{M}{L}\).
	Since there are at least \(\lceil \tfrac{M}{4 K} \rceil\) such messages we get the inequality 
	given in \eqref{eq:lem:spb-product}.
	
	In order to obtain \eqref{eq:lem:spb-product-alt},
	we change the analysis after \eqref{eq:spb-product-016}.
%	and introduce an approximation error term to the rate of the averaged sphere packing exponent term:
	For all \(\dmes\) satisfying \eqref{eq:spb-product-013}, as a result of \eqref{eq:spb-product-016} and 
	the definition of \(\spa{\epsilon}{\rate,\Wm}\) given 
	in \eqref{eq:def:avspherepacking}  we have
	\begin{align}
	\notag
	\Pem{\dmes}
	&\geq (\tfrac{1}{4\sqrt{\blx}}) e^{-2\gamma} e^{-\spa{\epsilon}{\rate,\Wmn{[1,\blx]}}}  
	\end{align}
	where \(\rate=\ln \tfrac{M}{L}-2\gamma-\ln \tfrac{16 e^{2}\blx^{3/2}}{\epsilon}\).	

	Since there are at least \(\lceil \tfrac{M}{4 K} \rceil\) such messages we get the 
	bound given in \eqref{eq:lem:spb-product-alt}.
\end{proof}

The \renyi centers and capacities of certain channels satisfy 
\eqref{eq:lem:spb-monotone-center:hypothesis}  
for all \(\rno\leq\dsta\)  in \((0,1)\), 
e.g.
the Poisson channel \(\Pcha{\tlx,\mA,\mB,\costc}\)
whose input set is described in \eqref{eq:def:poissonchannel-mean}. 
The \renyi capacity and center 
of \(\Pcha{\tlx,\mA,\mB,\costc}\) are determined in
\cite[Example \ref*{A-eg:poissonchannel-mean}]{nakiboglu19A}:
\begin{align}
\notag
\RC{\rno}{\Pcha{\tlx,\mA,\mB,\costc}}
&\!=\!\begin{cases}
\tfrac{\rno}{\rno-1}
\left[\left(\tfrac{\costc-\mA}{\mB-\mA} \mB^{\rno}+\tfrac{\mB-\costc}{\mB-\mA} \mA^{\rno}\right)^{\sfrac{1}{\rno}}-\costc
\right]\tlx
&\rno\!\neq\!1
\\
\left[
\tfrac{\costc-\mA}{\mB-\mA} \mB\ln \tfrac{\mB}{\costc}
+\tfrac{\mB-\costc}{\mB-\mA} \mA\ln \tfrac{\mA}{\costc}
\right]\tlx
&\rno\!=\!1
\end{cases}.
\end{align}
The order \(\rno\) \renyi center of \(\Pcha{\tlx,\mA,\mB,\costc}\) is the stationary 
Poisson processes with intensity \((\tfrac{\costc-\mA}{\mB-\mA} \mB^{\rno}+\tfrac{\mB-\costc}{\mB-\mA} \mA^{\rno})^{\sfrac{1}{\rno}}\).
For channels satisfying 
\eqref{eq:lem:spb-monotone-center:hypothesis}
for all \(\rno\leq\dsta\)  in \((0,1)\), 
the averaging scheme is not needed to establish the SPB.
In addition, the resulting bound is sharper than the one 
given in Lemma \ref{lem:spb-product}.

\begin{lemma}\label{lem:spb-monotone-center}
Let \(\Wmn{[1,\blx]}\) be a product channel 
for an \(\blx\in\integers{+}\) satisfying 
\(\RC{1}{\Wmn{[1,\blx]}}\geq\tfrac{\rnf^{2}}{2}\)
for a \(\rnf\in(0,1)\) and
\eqref{eq:lem:spb-monotone-center:hypothesis}
for all \(\rno,\dsta\) satisfying \(\rnf\!\leq\!\rno\!\leq\!\dsta\!<1\),
\(\knd\) satisfy \(\knd\!\geq\!3\), and \(\gamma\) be
\begin{align}
\label{eq:lem:spb-monotone-center:gamma}
\gamma
&\DEF 3\sqrt[\knd]{3}\left(\sum\nolimits_{\tin=1}^{\blx} \left(\RC{1/2}{\Wmn{\tin}}\vee \knd \right)^{\knd} \right)^{\sfrac{1}{\knd}}.
\end{align}
If \(M\), \(L\) are integers satisfying 
\(\tfrac{M}{L}\!>\!16 \sqrt{\blx} e^{\RC{\rnf}{\Wmn{[1,\blx]}}+\frac{\gamma}{1-\rnf}}\), 
then any \((M,L)\) channel code on \(\Wmn{[1,\blx]}\) 
satisfies 
\begin{align}
\label{eq:lem:spb-monotone-center}
\Pem{av}
&\geq \tfrac{\rnf^{2} e^{-2\gamma}}{32\blx^{1/2}\RC{1}{\Wmn{[1,\blx]}}} e^{-\spe{\rate,\Wm}}  
\end{align}
where \(\rate=\ln \tfrac{M}{L}-\ln \tfrac{95\blx^{1/2}\RC{1}{\Wmn{[1,\blx]}}}{\rnf^{2}e^{-2\gamma}}\).
\end{lemma}

\begin{proof}[Proof of Lemma \ref{lem:spb-monotone-center}]
We use \(\qmn{\rno,\Wmn{\tin}}\)'s rather than \(\qma{\rno,\Wmn{\tin}}{\epsilon}\)'s to define 
\(\qmn{\rno}\); thus \(\qmn{\rno}\) is equal to \(\qmn{\rno,\Wmn{[1,\blx]}}\).
We repeat the analysis of Lemma \ref{lem:spb-product} up to \eqref{eq:spb-product-012}:
There are at least \(\lceil\!\tfrac{M}{2}\!\rceil\) messages satisfying the following identity for some 
\(\rno_{\dmes}\!\in\!(\rnf,1)\)
\begin{align}
%\label{eq:spb-monotone-center-012-short}
\notag
(\qmn{\rno_{\dmes}}(\oev_{\dmes})+\tfrac{L}{M})e^{\RD{1}{\vma{\rno_{\dmes}}{\dmes}}{\qmn{\rno_{\dmes}}}+\frac{\gamma}{1-\rno_{\dmes}}}
&=\tfrac{1}{4\sqrt{\blx}}.
\end{align}
Let \(K\) be \(\lfloor \tfrac{2\RC{1}{\Wmn{[1,\blx]}}}{\rnf^{2}} \rfloor\).
Note that there exists a length \(\tfrac{1}{K}\) interval with  
\(\lceil\tfrac{M}{2 K}\rceil\) or more \(\rno_{\dmes}\)'s. 
Let \([\rnt\!-\!\tfrac{1}{K},\rnt]\) be the aforementioned interval;
then for all \(\rno\in[\rnt\!-\!\tfrac{1}{K},\rnt]\)
we have
\(\tfrac{1-\rno}{\rno}\RC{\rno}{\Wmn{[1,\blx]}}-\tfrac{1-\rnt}{\rnt}\RC{\rnt}{\Wmn{[1,\blx]}}\leq 1\)
by the monotonicity of \(\RC{\rno}{\Wm}\) in \(\rno\), 
i.e. Lemma \ref{lem:capacityO}-(\ref{capacityO-ilsc}).
Then as a result of the hypothesis of the lemma we have
\(\qmn{\rno}\!\leq\!e \qmn{\rnt}\) for all \(\rno\) in \([\rnt\!-\!\tfrac{1}{K},\rnt]\).
On the other hand at least half of the messages with \(\rno_{\dmes}\)'s in \([\rnt\!-\!\tfrac{1}{K},\rnt]\),  
satisfy \(\qmn{\rnt}(\oev_{\dmes})\!\leq\!2 \tfrac{L}{\lceil \sfrac{M}{2 K} \rceil}\).
Then at least \(\lceil\tfrac{1}{2}\lceil\tfrac{M}{2 K}\rceil\rceil\) messages
with \(\rno_{\dmes}\)'s in \([\rnt\!-\!\tfrac{1}{K},\rnt]\) satisfy
\begin{align}
\label{eq:spb-monotone-center-013-short}
\qmn{\rno_{\dmes}}(\oev_{\dmes})
&\leq \tfrac{4eL K}{M}.
\end{align}
Using \eqref{eq:spb-monotone-center-013-short} instead of \eqref{eq:spb-product-013} 
and repeating the rest of the analysis 
we get  \eqref{eq:lem:spb-monotone-center} using \(8(4e+1)\leq95\).
\end{proof}
 
For certain channels the \renyi center does not change with 
the order on the interval that it exits,
e.g. \cite[Example \ref*{A-eg:shiftinvariant}]{nakiboglu19A},
the binary symmetric channels.
The hypothesis of Lemma \ref{lem:spb-monotone-center}, 
described in \eqref{eq:lem:spb-monotone-center:hypothesis}, is satisfied 
for these channels as a result of the monotonicity of 
\(\tfrac{1-\rno}{\rno}\RC{\rno}{\Wm}\), i.e. 
Lemma \ref{lem:capacityO}-(\ref{capacityO-zo}).
But it is possible to derive the following sharper bound
for these channels.

\begin{lemma}\label{lem:spb-constant-center}
Let \(\Wmn{[1,\blx]}\) be a product channel for an \(\blx\in\integers{+}\) satisfying 
\vspace{-.1cm}
\begin{align}
\label{eq:lem:spb-constant-center:hypothesis}
\qmn{\rno,\Wmn{[1,\blx]}}=\qmn{\rnf,\Wmn{[1,\blx]}}
&
&\forall \rno\in(\rnf,1)
\end{align}
\vspace{-.1cm}
for a \(\rnf\in(0,1)\) and \(\knd\!\geq\!3\).
If \(M\), \(L\) are integers satisfying 
\(\tfrac{M}{L}\!>\!16 \sqrt{\blx} e^{\RC{\rnf}{\Wmn{[1,\blx]}}+\frac{\gamma}{1-\rnf}}\)
for \(\gamma\) described in \eqref{eq:lem:spb-monotone-center:gamma},
then any \((M,L)\) channel code on \(\Wmn{[1,\blx]}\) 
satisfies 
\vspace{-.1cm}
\begin{align}
\label{eq:lem:spb-constant-center}
\Pem{av}
&\geq \tfrac{e^{-2\gamma}}{16\blx^{1/2}} e^{-\spe{\rate,\Wmn{[1,\blx]}}}  
\end{align}
\vspace{-.1cm}
where \(\rate=\ln \tfrac{M}{L}-\ln \tfrac{20 \blx^{\sfrac{1}{2}}}{e^{-2\gamma}} \).
\end{lemma}

\begin{proof}[Proof of Lemma \ref{lem:spb-constant-center}]
\(\qmn{\rno,\Wmn{\tin}}\!=\!\qmn{\rnf,\Wmn{\tin}}\) for all \(\rno\!\in\![\rnf,1)\)
by  the hypothesis of the lemma and Lemma \ref{lem:capacityProduct}.
We use \(\qmn{\rnf,\Wmn{\tin}}\)'s rather than \(\qma{\rno,\Wmn{\tin}}{\epsilon}\)'s 
to define the probability measure \(\qmn{\rno}\).
Since \(\qmn{\rno}\) is the same probability measure for all \(\rno\in[\rnf,1)\), we denote it by \(\mQ\).
We repeat the analysis of Lemma \ref{lem:spb-product} up to \eqref{eq:spb-product-012}:
There are at least \(\lceil \tfrac{M}{2} \rceil\) messages satisfying
\begin{align}
%\label{eq:spb-constant-center-1}
\notag
(\mQ(\oev_{\dmes})+\tfrac{L}{M})e^{\RD{1}{\vma{\rno_{\dmes}}{\dmes}}{\mQ}+\frac{\gamma}{1-\rno_{\dmes}}}
&=\tfrac{1}{4\sqrt{\blx}}
\end{align}
for some \(\rno_{\dmes}\in(\rnf,1)\).
Among these \(\lceil \tfrac{M}{2} \rceil\) messages, there exists at least \(\lceil \tfrac{M}{4} \rceil\) messages 
satisfying
\begin{align}
\label{eq:spb-constant-center-2}
\mQ(\oev_{\dmes})
&\leq  \tfrac{4L}{M}. 
\end{align}
Using \eqref{eq:spb-constant-center-2} instead of \eqref{eq:spb-product-013} and repeating the rest of the analysis 
we get  \eqref{eq:lem:spb-constant-center}.
\end{proof}

Lemma \ref{lem:berryesseenN} is a key ingredient 
of the proof of Lemma \ref{lem:spb-product}.
The independence hypothesis of Lemma \ref{lem:berryesseenN} is implied by
 the product structure of each \(\Wmn{[1,\blx]}(\enc(\dmes))\) and \(\qmn{\rno}\). 
However, the product structure is not necessary for the independence, provided that 
the channel has certain symmetries.
\begin{lemma}\label{lem:spb-fixed-density}
Let \(\Wmn{\vec{[1,\blx]}}\) be a product channel with feedback 
for an \(\blx\in\integers{+}\) 
satisfying\footnote{\((\Wmn{\tin}(\dinp))_{ac}\) stands for 
	the component  that is absolutely continuous in  \(\qmn{\tin}\).}
for each \(\tin\in\{1,\ldots,\blx\}\)
\begin{align}
\label{eq:lem:spb-fixed-density:hypothesis-A}
\qmn{\rno,\Wmn{\tin}}
&=\qmn{\tin}
&
&\forall \rno\in(0,1),
\\
\label{eq:lem:spb-fixed-density:hypothesis-B}
\qmn{\tin}(\der{(\Wmn{\tin}(\dinp))_{ac}}{\qmn{\tin}}\leq \dsta)
&=\gX_{\tin}(\dsta)
&
&\forall \dinp\in \inpS_{\tin},\dsta\in\reals{\geq0} 
\end{align}
for a \(\qmn{\tin}\!\in\!\pmea{\outA_{\tin}}\) and 
a cumulative distribution function \(\gX_{\tin}\),
\(\rnf\!\in\!(0,1)\), \(\knd\!\geq\!3\),  
and \(\gamma\) be the constant defined in  \eqref{eq:lem:spb-monotone-center:gamma}. 
If \(M\), \(L\) are integers  satisfying
\(\tfrac{M}{L}> 16 \sqrt{\blx} e^{\RC{\rnf}{\Wmn{[1,\blx]}}+\frac{\gamma}{1-\rnf}}\),
then any \((M,L)\) channel code on \(\Wmn{\vec{[1,\blx]}}\) 
satisfies \eqref{eq:lem:spb-constant-center}.
\end{lemma}

Any product channel whose component channels are modular shift channels described in 
\cite[Example \ref*{A-eg:shiftchannel}]{nakiboglu19A}, satisfy the 
constraints given in \eqref{eq:lem:spb-fixed-density:hypothesis-A} and \eqref{eq:lem:spb-fixed-density:hypothesis-B}. 
Products of more general shift invariant channels described in 
\cite[Example \ref*{A-eg:shiftinvariant}]{nakiboglu19A}, do 
satisfy the constraint given in \eqref{eq:lem:spb-fixed-density:hypothesis-A} but they  may or may not 
satisfy the constraint given in \eqref{eq:lem:spb-fixed-density:hypothesis-B} depending on \(\fXS\).  

\begin{proof}[Proof of Lemma \ref{lem:spb-fixed-density}]
As we have done for Lemma \ref{lem:spb-constant-center} we use 
\(\qmn{\tin}\)'s, rather than \(\qma{\rno,\Wmn{\tin}}{\epsilon}\)'s,
to define \(\mQ\).  
Although \(\Wmn{\vec{[1,\blx]}}(\enc(\dmes))\) 
is not necessarily a product measure,
\(\cla{\rno,\tin}{\dmes}\)'s 
are jointly independent random variables in the probability space \((\outS_{1}^{\blx},\outA_{1}^{\blx},\vma{\rno}{\dmes})\) 
for any \(\rno\in(0,1)\) and \(\dmes\in\mesS\), as a
result of the hypothesis of the lemma given in \eqref{eq:lem:spb-fixed-density:hypothesis-B}.
The rest of the proof is identical to the proof of Lemma \ref{lem:spb-constant-center}.
\end{proof}

\begin{proof}[Proof of Theorem \ref{thm:productexponent}]
We prove Theorem \ref{thm:productexponent} using Lemmas \ref{lem:avspherepacking} and \ref{lem:spb-product}.
Note that we are free to choose different values for \(\epsilon\) and \(\knd\) for different values of \(\blx\),
provided that the  hypotheses of Lemmas \ref{lem:avspherepacking} and \ref{lem:spb-product} are satisfied. 

As a result of Assumption \ref{assumption:individual-ologn} there exists a \(K\in [1,\infty)\)
and an \(\blx_{0}\in \integers{+}\) such that 
\(\max\nolimits_{\tin \in [1,\blx]} \RC{\sfrac{1}{2}}{\Wmn{\tin}} \leq K \ln \blx\) for all
\(\blx\geq\blx_{0}\).
Let \(\knd_{\blx}\) be \(K \ln\blx\) and \(\epsilon_{\blx}\) be \(\sfrac{1}{\blx}\).
Then
\begin{align}
\label{eq:productexponent-2}
\gamma_{\blx}
&\leq 4 e K  \ln\blx
\end{align}
for all \(\blx\) large enough.
Furthermore,
\eqref{eq:avcapacity:bound} and \eqref{eq:productexponent-2}  
imply
\begin{align}
%\label{eq:productexponent-3}
\notag
16 \sqrt{\blx} e^{\RCI{\rno_{0}}{\Wmn{[1,\blx]} }\epsilon+\frac{\gamma_{\blx}}{1-\rno_{0}}}
&\!\leq\!16 e^{\RC{\rno_{0}}{\Wmn{[1,\blx]} }+(\frac{1}{2}+\frac{4eK}{1-\rno_{0}}+\frac{K}{\rno_{0}(1-\rno_{0})})\ln \blx}
\end{align}
for all \(\blx\) large enough.
Thus as a result of the hypothesis of the theorem, 
hypotheses of Lemma \ref{lem:spb-product} is satisfied for all \(\blx\)
large enough. Thus using \eqref{eq:productexponent-2} we can conclude that 
\begin{align}
\label{eq:productexponent-4}
\Pem{av}
&\geq \left(\tfrac{\blx^{-1-8eK}}{16 e^{2}\blx^{3/2}}\right)^{\frac{1}{\rno_{0}}} 
e^{-\spa{\sfrac{1}{\blx}}{\ln \frac{M_{\blx}}{L_{\blx}},\Wmn{[1,\blx]}}}
\end{align}
for all \(\blx\) large enough.

On the other hand  Lemma \ref{lem:avspherepacking}, 
the hypothesis given in \eqref{eq:thm:productexponent-hypothesis},
and the monotonicity of \(\RC{\rno}{\Wm}\) in \(\rno\) imply that 
\begin{align}
\notag
\spa{\sfrac{1}{\blx}}{\ln \tfrac{M_{\blx}}{L_{\blx}},\Wmn{[1,\blx]}}
&\leq 
\spe{\ln \tfrac{M_{\blx}}{L_{\blx}},\Wmn{[1,\blx]}}+\tfrac{\RC{\rno_{1}}{\Wmn{[1,\blx]}}}{ (\blx-1)\rno_{0}^{2}}
\end{align}
for all \(\blx\) large enough.
Then using the monotonicity of \(\RC{\rno}{\Wm}\) and \(\tfrac{1-\rno}{\rno}\RC{\rno}{\Wm}\) 
in \(\rno\)  we can conclude that 
\begin{align}
\notag
%\label{eq:productexponent-5}
\!\!\spa{\sfrac{1}{\blx}}{\ln\tfrac{M_{\blx}}{L_{\blx}},\!\Wmn{[1,\blx]}}
&\!\leq\! 
\spe{\ln\tfrac{M_{\blx}}{L_{\blx}},\!\Wmn{[1,\blx]}}
\!+\!\tfrac{(\frac{\rno_{1}}{1-\rno_{1}} \vee 1)\blx K \ln \blx}{(\blx-1)\rno_{0}^{2}} 
\end{align}
for all \(\blx\) large enough. 
Then  \eqref{eq:thm:productexponent} follows from \eqref{eq:productexponent-4}.
\end{proof}

\subsection{Augustin's SPBs for Product Channels}\label{sec:comparison}
In the following, we compare our results with the SPBs derived by 
Augustin in \cite{augustin69} and \cite{augustin78} for the product channels.
Augustin works with the maximum 
error probability, \(\Pem{\max} \DEF \max\nolimits_{\dmes\in \mesS} \Pem{\dmes}\),
rather than the average error probability. 
This, however, is inconsequential for our purposes because 
any SPB
for \(\Pem{av}\) holds for \(\Pem{\max}\) as is and 
any SPB for \(\Pem{\max}\) can be converted into a 
SPB for \(\Pem{av}\) through a standard application of Markov inequality 
for channel codes,
with definite and essentially inconsequential correction terms. 

The main advantage of Theorem \ref{thm:productexponent} over the 
SPBs in \cite{augustin69} and \cite{augustin78}, is its polynomial prefactor. 
Augustin did establish a SPB with a polynomial prefactor,
but only under considerably stronger hypotheses, \cite[Thm. 4.8]{augustin69}.
In addition all of the asymptotic SPBs in \cite{augustin69} and \cite{augustin78}
assume the uniform continuity  condition described 
in Assumption \ref{assumption:cumulative-equicontinuity},
given in the following.
Theorem \ref{thm:productexponent}, on the other hand, does not 
have such a hypothesis.
\begin{assumption}\label{assumption:cumulative-equicontinuity}
	\(\{\Wmn{\tin}\}_{\tin\in\integers{+}}\) 
	and \(\RC{0^{_{+}}\!}{\Wm}\) defined in \eqref{eq:def:capacitylimit}
	satisfy
	\begin{align}
	\notag
	\lim\nolimits_{\rno \downarrow 0}\sup\nolimits_{\blx\in \integers{+}}
	\tfrac{1}{\blx}\left[\RC{\rno}{\Wmn{[1,\blx]}}-\RC{0^{_{+}}\!}{\Wmn{[1,\blx]}}\right]
	&=0.
	\end{align} 
\end{assumption}
\begin{remark}
	This assumption is given as equation (7) in \cite{augustin69} and Condition 31.3a  in \cite{augustin78}.
	In \cite{augustin78}, the condition is stated without \(\sfrac{1}{\blx}\) factor; we believe that 
	it is a typo.
\end{remark}
After this general overview, let us continue with a discussion of the individual results.

\cite[Thm. 4.7b]{augustin69} bounds \(\Pem{\max}^{(\blx)}\) from below by 
	\(e^{-e^{K}\sqrt{32\blx}-\spe{\ln\frac{M_{\blx}}{L_{\blx}},\Wmn{[1,\blx]}}}\)
	for large enough \(\blx\) for any sequence of channels satisfying
	\(\sup_{\tin\in \integers{+}} \RC{1}{\Wmn{\tin}}<K\) for some \(K\in\reals{+}\)
	and
	Assumption \ref{assumption:cumulative-equicontinuity}.
	Thus \cite[Thm. 4.7b]{augustin69} proves a claim weaker than 
	Theorem \ref{thm:productexponent} 
	under a hypothesis stronger than Assumption \ref{assumption:individual-ologn}.

\cite[Thm. 4.8]{augustin69} is a SPB with a polynomial prefactor
	for product channels satisfying Assumptions \ref{assumption:cumulative-equicontinuity} and \ref{assumption:individual-density}.
		\begin{assumption}\label{assumption:individual-density}
\(\exists K\in\reals{+},\{\rfm_{\tin}\}_{\tin\in\integers{+}}\) satisfying 
		\begin{align}
		\notag
		\tfrac{1}{K}\leq \der{\Wmn{\tin}(\dinp)}{\rfm_{\tin}}
		&\leq K 
		&
		&\Wmn{\tin}(\dinp)-\mbox{a.e.}
		&
		&\forall \dinp\in\inpS_{\tin}.    
		\end{align}
	\end{assumption}
	Assumption \ref{assumption:individual-density} implies Assumption \ref{assumption:individual-ologn},
	but the converse is not true, 
	e.g. if \(\Wmn{\tin}\!=\!\Pcha{\tlx,\mA,\mB}\) then Assumption \ref{assumption:individual-ologn}
	holds but Assumption \ref{assumption:individual-density} does not hold.
	Thus \cite[Thm. 4.8]{augustin69} is weaker than Theorem \ref{thm:productexponent} because  it
	establishes the same claim under a stronger hypothesis.

\cite[Thm. 4.7a]{augustin69} and \cite[Thm. 31.4]{augustin78} bound \(\Pem{\max}^{(\blx)}\)
	from below by \(e^{-\bigo{\sqrt{\blx}}-\spe{\ln\frac{M_{\blx}}{L_{\blx}},\Wmn{[1,\blx]}}}\)
	for large enough \(\blx\).
	These SPBs are not comparable with 
	Theorem \ref{thm:productexponent} because their hypotheses are not comparable 
	with the hypotheses of Theorem \ref{thm:productexponent}.
	However, these SPBs can be proved  without relying on Assumption \ref{assumption:cumulative-equicontinuity},
	using a variant of  Lemma \ref{lem:spb-product}, 
	which is obtained by applying Chebyshev's inequality instead of Lemma \ref{lem:berryesseenN}.
	
\cite[Lem. 31.2]{augustin78} is quite similar to Lemma \ref{lem:spb-product};
the main difference is the infimum taken over \((0,1)\).  
In order to remove this infimum and obtain a bound 
similar to Lemma \ref{lem:spb-product}, 
one needs to assume an equicontinuity condition similar to the one in
Assumption \ref{assumption:cumulative-equicontinuity}.
  
%!TEX root=../main-B.tex
\section{The SPB for Product Channels with Feedback}\label{sec:fproduct-outerbound}
\begin{theorem}\label{thm:fDSPCexponent}
Let \(\{\Wmn{\tin}\}_{\tin\in\integers{+}}\) be a sequence of discrete channels
satisfying \(\Wmn{\tin}\!=\!\Wm\) for all \(\tin\!\in\!\integers{+}\)
and \(\rno_{0}\), \(\rno_{1}\) be orders satisfying \(0\!<\!\rno_{0}\!<\!\rno_{1}\!<\!1\).
Then for any sequence of codes on the discrete stationary product channels with feedback 
\(\{\Wmn{\vec{[1,\blx]}}\}_{\blx\in\integers{+}}\) satisfying 
\begin{align}
\label{eq:thm:fDSPCexponent-hypothesis}
\RC{\rno_{1}}{\Wm}\geq \tfrac{1}{\blx}\ln \tfrac{M_{\blx}}{L_{\blx}}
&\geq  \RC{\rno_{0}}{\Wm}+\tfrac{\ln \blx}{\blx^{\sfrac{1}{4}}}
&
&\forall \blx\geq\blx_{0} 
\end{align}
there exists an \(\blx_{1}\geq\blx_{0}\) such that 
\begin{align}
\label{eq:thm:fDSPCexponent}
\Pem{av(\blx)}
&\geq
e^{-\blx\left[\spe{\frac{1}{\blx}\ln\frac{M_{\blx}}{L_{\blx}},\Wm}+\frac{1}{\rno_{0}}
	\frac{\ln \blx}{\blx^{\sfrac{1}{4}}}\right]}
&
&\forall \blx\geq\blx_{1}.
\end{align}
\end{theorem}

We prove Theorem \ref{thm:fDSPCexponent}
using ideas from Sheverdyaev \cite{sheverdyaev82},
Haroutunian \cite{haroutunian68}, \cite{haroutunian77},
and Augustin \cite{augustin69}, \cite{augustin78}. 
In \S\ref{sec:taylor},  we establish a Taylor's expansion for \(\RD{\rno}{\mW}{\mQ}\) 
around \(\rno\!=\!1\) assuming \(\RD{\lgm}{\mW}{\mQ}\) is finite for a \(\lgm\!>\!1\). 
In \S\ref{sec:tradeoff}, we recall the auxiliary channel method and 
prove that for any channel \(\Wm\) satisfying 
\(\lim_{\rno\uparrow 1}\!\frac{1-\rno}{\rno}\RC{\rno}{\Wm}\!=\!0\)
and rate \(\rate\!\in\!(\RC{0^{_{+}}\!}{\Wm},\RC{1}{\Wm})\) 
there exists a channel \(\Vm\) and a constant \(\rnb\!>\!1\) satisfying both 
\(\RC{\rnb}{\Vm}\!\lesssim\!\rate\) and 
\(\sup_{\dinp\in\inpS}\RD{1}{\Vm(\dinp)}{\Wm(\dinp)}\!\lesssim\!\spe{\rate,\Wm}\).
In \S\ref{sec:fpouterbound},  we first introduce the concept of subblocks 
and derive a non-asymptotic SPB using it; then we prove Theorem \ref{thm:fDSPCexponent}
using this SPB.
In \S\ref{sec:fcomparison}, we provide an asymptotic SPB for (possibly non-stationary) 
DPCs, i.e. Theorem \ref{thm:fDPCexponent}, and compare our results with the earlier ones.
In \S\ref{sec:haroutunian}, we show that Haroutunian's bound,
the results of \S\ref{sec:tradeoff}, and the concept of subblocks
imply an asymptotic SPB for DSPCs with feedback,
as well.

%The main aim of this section is to prove the asymptotic SPB given in 
%Theorem \ref{thm:fDSPCexponent}; for that, we use Augustin's method \cite{augustin69}, 
%\cite{augustin78} together with ideas from Sheverdyaev \cite{sheverdyaev82} and 
%Haroutunian \cite{haroutunian68}, \cite{haroutunian77}. 
%In \S\ref{sec:taylor},  we establish a Taylor's expansion for \(\RD{\rno}{\mW}{\mQ}\) 
%around \(\rno=1\) assuming \(\RD{\lgm}{\mW}{\mQ}\) is finite for a \(\lgm>1\). 
%In \S\ref{sec:tradeoff}, we first recall the auxiliary channel method and then 
%prove that for every \(\rate\!\in\!(\RC{0^{_{+}}\!}{\Wm},\RC{1}{\Wm})\) there exists 
%a channel \(\Vm\) and a constant \(\rnb>1\) satisfying both 
%\(\RC{\rnb}{\Vm}\lesssim\rate\) and 
%\(\sup_{\dinp\in\inpS}\RD{1}{\Vm(\dinp)}{\Wm(\dinp)}\lesssim \spe{\rate,\Wm}\)
%provided that \(\lim_{\rno\uparrow 1} \frac{1-\rno}{\rno}\RC{\rno}{\Wm}=0\).
%In \S\ref{sec:fpouterbound},  we first derive a non-asymptotic SPB and 
%then use this non-asymptotic SPB to prove Theorem \ref{thm:fDSPCexponent}.
%In \S\ref{sec:fcomparison}, we provide an asymptotic SPB for (possibly non-stationary) 
%DPCs, i.e. Theorem \ref{thm:fDPCexponent} ,  discuss its generalization to non-discrete channels,
%and compare our results with some of the earlier results.

\subsection[A Taylor's Expansion for \(\RD{\rno}{\mW}{\mQ}\)]{A Taylor's Expansion for the \renyi Divergence}\label{sec:taylor}
Sheverdyaev employed the Taylor's expansion 
---albeit with approximation error terms that are not explicit--- 
for his attempt to prove the SPB for the codes on 
the DSPCs with feedback in \cite{sheverdyaev82}.
Recently, Fong and Tan \cite[Prop. 11]{fongT16} bounded \(\RD{\rnb}{\mW}{\mQ}\) for \(\rnb\in [1,\tfrac{5}{4}]\) 
using Taylor's expansion for the case when \(\outS\) is a finite set and \(\outA\) is \(\sss{\outS}\). 
The bound by Fong and Tan, however, is not appropriate for our purposes because
its approximation error term is proportional to \(\abs{\outS}\). 
Assuming \(\der{\mW}{\mQ}\) to be 
bounded Sason and Verd\'{u}\label{reference:sasonverdu} derived a similar 
bound\footnote{\cite[Thm. 35-(b)]{sasonV16} is obtained by expressing
	\(\RD{\rno}{\mW}{\mQ}\), which is not an \(\fX\)-divergence, 
	as a monotonically increasing function of the order \(\rno\) Hellinger divergence 
	between \(\mW\) and \(\mQ\), which is  an \(\fX\)-divergence. 
	Guntuboyina, Saha, and Schiebinger \cite{guntuboyinaSS14} 
	presented a general method for establishing sharp bounds among \(\fX\)-divergences, 
	without assuming either \(\der{\mW}{\mQ}\) or \(\der{\mQ}{\mW}\) to be bounded. 
	Yet such conditions can easily be included in the framework proposed in \cite{guntuboyinaSS14}.}
\cite[Thm. 35-(b), (469)]{sasonV16}.
In Lemma \ref{lem:taylor} we bound \(\RD{\rnb}{\mW}{\mQ}\) for \(\rnb\in(1,\lgm)\) 
using Taylor's expansion assuming only \(\RD{\lgm}{\mW}{\mQ}\) to be  finite.
\begin{lemma}\label{lem:taylor}
Let \(\mW\) and \(\mQ\) be two probability measures on the measurable space  \((\outS,\outA)\) satisfying 
\(\RD{\lgm}{\mW}{\mQ}\leq \gamma\) for a \(\gamma\in \reals{+}\) and a \(\lgm\in(1,\infty)\). 
Then for any \(\rnb\in (1,\lgm)\) 
\begin{align}
\label{eq:lem:taylor}
\hspace{-.15cm}0\!\leq\!
\RD{\rnb}{\mW}{\mQ}\!-\!\RD{1}{\mW}{\mQ}\!
&\leq\!\tfrac{2(\rnb-1)}{e^{2}}\!\left[\!1\!+\!e^{(\rnb-1)\gamma}(\tfrac{\gamma e^{\tau}}{2\tau})^{2}\!\right]
\end{align}
where \(\tau=\tfrac{(\lgm-\rnb)\gamma }{2} \wedge 1\).
\end{lemma}
The \renyi divergences with orders greater than one are not customarily used for
establishing the SPB;
Sheverdyaev's proof in \cite{sheverdyaev82}, is an exception in this 
respect.\footnote{
To be precise Sheverdyaev does not explicitly use \renyi 
divergences in \cite{sheverdyaev82}, but his analysis can be easily expressed 
via \renyi divergences.}
In \S\ref{sec:tradeoff}, we use Lemma \ref{lem:taylor} to construct an auxiliary channel
with certain properties desirable for our purposes,
see Lemma \ref{lem:tradeoff}.

\begin{proof}[Proof of Lemma \ref{lem:taylor}]
\(\RD{\rnb}{\mW}{\mQ}-\RD{1}{\mW}{\mQ}\) is non-negative
 because the \renyi divergence is a nondecreasing 
function of the order by Lemma \ref{lem:divergence-order}. 
In order to bound \(\RD{\rnb}{\mW}{\mQ}-\RD{1}{\mW}{\mQ}\) from above we use Taylor's theorem
on the function \(\gX(\rno)\) defined as follows:
\begin{align}
\notag
\gX(\rno)
&\DEF\EXS{\mQ}{(\der{\mW}{\mQ})^{\rno}}.
\end{align}
\(\gX(\rno)\) is continuous in \(\rno\) on \((0,\lgm)\) by \cite[Cor. 2.8.7]{bogachev} because 
\(\EXS{\mQ}{(\der{\mW}{\mQ})^{\lgm}}=e^{(\lgm-1) \RD{\lgm}{\mW}{\mQ}}<\infty\) by the hypothesis
and \((\der{\mW}{\mQ})^{\rno}\leq 1+(\der{\mW}{\mQ})^{\lgm}\).
In order to apply Taylor's theorem to \(\gX(\rno)\), we show that \(\gX(\rno)\) is twice 
differentiable and bound its second derivative. 
To that end, first note that we can bound \(\dinp^{\rno}\abs{\ln\dinp}^{\knd}\) for any \(\rno\!\in\!(0,\lgm)\) and 
\(\knd\!\in\!\{1,2\}\), using the derivative test as follows
\begin{align}
\notag
\dinp^{\rno}\abs{\ln\dinp}^{\knd}
&\leq  (\tfrac{\knd}{e \rno })^{\knd} \IND{\dinp\in [0,1]}+(\tfrac{\knd}{e(\lgm-\rno)})^{\knd}\dinp^{\lgm}\IND{\dinp\in (1,\infty)}.
\end{align}
Hence, for all \(\rno\in(0,\lgm)\) and \(\knd\in\{1,2\}\) we have
\begin{align}
\notag
\abs{\der{^{\knd}}{\rno^{\knd}} (\der{\mW}{\mQ})^{\rno}}
&=(\der{\mW}{\mQ})^{\rno}\abs{\ln \der{\mW}{\mQ}}^{\knd}
\\
\label{eq:taylor-A}
&\leq  (\tfrac{\knd}{e \rno })^{\knd}+(\tfrac{\knd}{e(\lgm-\rno)})^{\knd}(\der{\mW}{\mQ})^{\lgm}.
\end{align} 
The expression on the right hand side has finite expectation under \(\mQ\)  for any \(\rno\in(0,\lgm)\) and
\(\knd\in\{1,2\}\). Thus \(\gX(\rno)\) is  twice differentiable in \(\rno\) on \((0,\lgm)\) 
by \cite[Cor. 2.8.7]{bogachev}. Furthermore, for  \(\rno\) in \((0,\lgm)\) and \(\knd\in\{1,2\}\)
we have
\begin{align}
\label{eq:taylor-B}
\der{^{\knd}}{\rno^{\knd}}\gX(\rno)
&=\EXS{\mQ}{(\der{\mW}{\mQ})^{\rno} (\ln \der{\mW}{\mQ})^{\knd}}.
\end{align}
Since \(\gX(\rno)\) is twice differentiable applying Taylor's theorem \cite[Appendix B4]{dudley}  
around \(\rno=1\) we get
\begin{align}
\label{eq:taylor-C}
\hspace{-.3cm}\gX(\rnb)
&\!\leq\!1\!+\!(\rnb\!-\!1)\!\left.\der{}{\rno}\gX(\rno)\right\vert_{\rno=1}
\!+\!\tfrac{(\rnb-1)^2}{2!}\!\!\sup\limits_{\rno\in (1,\rnb)}\!\der{^{2}}{\rno^{2}}\gX(\rno).
\end{align} 
On the other hand using \eqref{eq:taylor-A} and \eqref{eq:taylor-B} we get
\begin{align}
\label{eq:taylor-D}
\der{^2}{\rno^{2}}\gX(\rno)
&\leq \left(\tfrac{2}{e}\right)^{2}+\left(\tfrac{2}{e(\lgm-\rnb)}\right)^{2} \gX(\lgm)
&
&\forall \rno\in (1,\rnb).
\end{align}
Then using the identity \(\ln \dsta \leq \dsta-1\) together with
\eqref{eq:taylor-B}, \eqref{eq:taylor-C}, and \eqref{eq:taylor-D} we get
the following inequality \(\rnb\in (1,\lgm)\)
\begin{align}
\notag
\ln \gX(\rnb)
&\leq (\rnb-1) \RD{1}{\mW}{\mQ}+(\rnb-1)^2\tfrac{2}{e^{2}}
\left(1+\tfrac{\gX(\lgm)}{(\lgm-\rnb)^{2}} \right).
\end{align}
On the other hand \(\gX(\rno)=e^{(\rno-1)\RD{\rno}{\mW}{\mQ}}\) by definition and
\(\RD{\lgm}{\mW}{\mQ}\leq\gamma\) by the hypothesis. Thus
\begin{align}
\label{eq:taylor-E}
\RD{\rnb}{\mW}{\mQ}-\RD{1}{\mW}{\mQ}
&\leq(\rnb-1)\tfrac{2}{e^{2}}\left(1+\tfrac{e^{(\lgm-1)\gamma}}{(\lgm-\rnb)^{2}}\right).
\end{align}
Note that \(\RD{\rno}{\mW}{\mQ}\leq \gamma\) for any \(\rno\in(\rnb,\lgm)\) because 
\(\RD{\lgm}{\mW}{\mQ}\leq \gamma\) and the \renyi divergence is a nondecreasing 
in its order by Lemma \ref{lem:divergence-order}.
Thus using the analysis leading to \eqref{eq:taylor-E}, we get
the following inequality \(\forall \rno \in (\rnb,\lgm]\)
\begin{align}
\label{eq:taylor-F}
\RD{\rnb}{\mW}{\mQ}-\RD{1}{\mW}{\mQ}
&\leq(\rnb-1)\tfrac{2}{e^{2}}\left(1+\tfrac{e^{(\rno-1)\gamma}}{(\rno-\rnb)^{2}}\right).
\end{align}
Using the derivative test we can confirm that the least upper bound among the upper bounds
given in \eqref{eq:taylor-F} is the one at \(\rno=\lgm\wedge (\tfrac{2}{\gamma}+\rnb)\) 
and the resulting upper bound is the one given in  \eqref{eq:lem:taylor}. 
As a side note, let us point out that the least upper bound is strictly less than the upper 
bound at \(\rno=\lgm\) iff \(\gamma(\lgm-\rnb)>2\).
\end{proof}
\subsection{The Auxiliary Channel Method}\label{sec:tradeoff}
Haroutunian's seminal paper \cite{haroutunian68}, establishing the SPB 
for the stationary product channels with finite input sets, used the performance of a 
code on an auxiliary channel as an anchor to bound its performance on the actual channel.
To the best of our knowledge, this is the first explicit use of the auxiliary channel method.
In a nutshell, auxiliary channel method can be described as follows: 
\begin{enumerate}[(i)]
\item Choose an auxiliary channel \(\Vm\!:\!\inpS\!\to\!\pmea{\outA}\) based on 
the actual channel \(\Wm\!:\!\inpS\!\to\!\pmea{\outA}\)
and the code \((\enc,\dec)\).
\item Bound the performance of \((\enc,\dec)\) on \(\Vm\).
\item Bound the performance of \((\enc,\dec)\) on \(\Wm\) 
using the bound derived in part (ii) and the features of \(\Vm\).
\end{enumerate}
Many infeasibility results that are derived without using the auxiliary channel method, 
can be interpreted as an implicit application of the auxiliary channel method,
as well. As an example, let us consider a version of Arimoto's  bound,
due to Augustin \cite[Thm. 27.2-(ii)]{augustin78}, given in the following:
If \(M\) and \(L\) are positive integers satisfying \(\ln\tfrac{M}{L}>\RC{1}{\Wm}\) for a 
channel \(\Wm\),
then the average error probability \(\Pem{av}\) of any  \((M,L)\) channel code on 
\(\Wm\) satisfies
\begin{align}
\label{eq:lem:augustin}
\Pem{av}
&\geq 1- e^{\frac{\rno-1}{\rno}(\RC{\rno}{\Wm}-\ln\frac{M}{L})}
&
&\forall \rno>1.
\end{align}
Augustin obtained \eqref{eq:lem:augustin} by a convexity argument in \cite{augustin78};
but it can be derived using the auxiliary channel method as well:
Let \(\Vm\!:\!\inpS\!\to\!\pmea{\outA}\) be  such that \(\Vm(\dinp)=\qmn{\rno,\Wm}\) for all \(\dinp\in\inpS\) 
and \(\Pem{\Vm}\) be the average error probability of \((\enc,\dec)\) on \(\Vm\), 
then \(\Pem{\Vm}\!\geq\!1\!-\!\tfrac{L}{M}\) for  any \((M,L)\) channel code \((\enc,\dec)\).
Furthermore,
\(\RC{\rno}{\Wm}\!\geq\!\RD{\rno}{\mP \mtimes \Wm}{\mP \mtimes \Vm}\)
by Theorem \ref{thm:minimax}
and 
 \(\RD{\rno}{\mP \mtimes\!\Wm\!}{\mP \mtimes\!\Vm\!}\!\geq\!
 \tfrac{\ln [(\Pem{av})^{\rno}(\Pem{\Vm})^{1-\rno}
 	+(1-\Pem{av})^{\rno}(1-\Pem{\Vm})^{1-\rno}]}{\rno-1}\) by
Lemma \ref{lem:divergence-DPI}.
Thus
\(\RC{\rno}{\Wm}\!\geq\tfrac{\ln [(1-\Pem{av})^{\rno}(1-\Pem{\Vm})^{1-\rno}]}{\rno-1}\)
and \eqref{eq:lem:augustin} follows.

In \cite{haroutunian77}, Haroutunian applied the auxiliary channel method to 
bound the error probability of codes on DSPCs with feedback from below.
The exponential decay rate of Haroutunian's bound with block length, however,
is greater than the sphere packing exponent for most channels.   
In \S\ref{sec:fpouterbound}, we use the auxiliary channel method
---via subblocks --- to establish a SPB.  
To do that, 
we employ the auxiliary channel described in Lemma \ref{lem:tradeoff}-(\ref{tradeoff-B},\ref{tradeoff-C}), 
given in the following.  
In \S\ref{sec:haroutunian}, we demonstrate that
one can establish a SPB for codes on DSPCs with feedback
by applying Lemma \ref{lem:tradeoff}-(\ref{tradeoff-B},\ref{tradeoff-C}) to subblocks
and invoking Haroutunian's bound in \cite{haroutunian77},
as well.

Lemma \ref{lem:tradeoff} describes its auxiliary channels using 
the order \(\rno\) \renyi center \(\qmn{\rno,\Wm}\) described in Theorem \ref{thm:minimax}, 
the average \renyi center \(\qma{\rno,\Wm}{\epsilon}\) described in  Definition \ref{def:avcenter},
and the tilted channel defined in the following.
\begin{definition}\label{def:tiltedchannel}
	For any \(\rno\!\in\!\reals{+}\), \(\!\Wm\!:\!\inpS\!\to\!\pmea{\outA}\), 
	and \(\mQ\!\in\!\pmea{\outA}\) such that
	\(\sup_{\dinp\in\inpS}\RD{\rno}{\Wm(\dinp)}{\mQ}\!<\!\infty\), \emph{the order \(\rno\) tilted channel} 
	\(\Wma{\rno}{\mQ}\!:\!\inpS\!\to\!\pmea{\outA}\)
	is defined  as 
	\begin{align}
	\label{eq:def:tiltedchannel}
	\der{\Wma{\rno}{\mQ}(\dinp)}{\rfm}
	&\DEF e^{(1-\rno)\RD{\rno}{\Wm(\dinp)}{\mQ}}
	(\der{\Wm(\dinp)}{\rfm})^{\rno} (\der{\mQ}{\rfm})^{1-\rno}
	\end{align}
	for all \(\dinp\in\inpS\)
	where \(\rfm\in\pmea{\outA}\) satisfies 
	\(\Wm(\dinp)\AC\rfm\) and \(\mQ\AC\rfm\).
\end{definition}
\begin{lemma}\label{lem:tradeoff}
	For any channel \(\Wm\!:\!\inpS\!\to\!\pmea{\outA}\) satisfying
	both \(\RC{0^{_{+}}\!}{\Wm}\!\neq\!\RC{1}{\Wm}\) and 
	\(\lim_{\rno\uparrow 1}\frac{1-\rno}{\rno}\RC{\rno}{\Wm}=0\)
	and rate \(\rate\) in \((\RC{0^{_{+}}\!}{\Wm},\RC{1}{\Wm})\)
	there exist a \(\rnf\in(0,1)\) such that 
	\begin{align}
	\notag
	\rate
	&=\RC{\rnf}{\Wm}
	\end{align}
	and an \(\rnt\in(\rnf,1)\) such that
	\begin{align}
	\notag
	\spe{\rate,\Wm}
	&=\tfrac{1-\rnt}{\rnt}\RC{\rnt}{\Wm}.
	\end{align}
	Furthermore \(\Wm\), \(\rate\), \(\rnf\), \(\rnt\) satisfy the following assertions.
	\begin{enumerate}[(a)]
		\item\label{tradeoff-A}
There exists an \(\fX\!:\!\inpS\!\to\![\rnf,\rnt]\) satisfying 
both \eqref{eq:lem:tradeoff-A-1} and \eqref{eq:lem:tradeoff-A-2} 
for all \(\dinp\!\in\!\inpS\).
\begin{align}
\label{eq:lem:tradeoff-A-1}
\RD{1}{\Wma{\fX(\dinp)}{\qmn{\fX(\dinp),\Wm}}(\dinp)}{\qmn{\fX(\dinp),\Wm}}
&\leq \rate
\\
\label{eq:lem:tradeoff-A-2}
\RD{1}{\Wma{\fX(\dinp)}{\qmn{\fX(\dinp),\Wm}}(\dinp)}{\Wm(\dinp)}
&\leq \spe{\rate,\Wm}
\end{align}
\item\label{tradeoff-B}
For all \(\epsilon\in (0,\sfrac{\rnf}{2})\) 
there exists an \(\fX_{\epsilon}\!:\!\inpS\!\to\![\rnf,\rnt]\) satisfying 
both \eqref{eq:lem:tradeoff-B-1} and \eqref{eq:lem:tradeoff-B-2} 
for all \(\dinp\!\in\!\inpS\).
\begin{align}
\label{eq:lem:tradeoff-B-1}
\hspace{-.5cm}\!\RD{1}{\!\Wma{\fX_{\epsilon}(\dinp)}{\qma{\fX_{\epsilon}(\dinp),\Wm}{\epsilon}}\!(\dinp)}{\qma{\fX_{\epsilon}(\dinp),\Wm}{\epsilon}\!}
&\!\leq\!\rate+\tfrac{2\epsilon\RC{\sfrac{1}{2}}{\Wm}}{\rnf(1-\rnf)^{2}}
\\
\label{eq:lem:tradeoff-B-2}
\!\RD{1}{\!\Wma{\fX_{\epsilon}(\dinp)}{\qma{\fX_{\epsilon}(\dinp),\Wm}{\epsilon}}\!(\dinp)}{\!\Wm(\dinp)}
&\!\leq\!\spe{\rate,\Wm}\!+\!\tfrac{2\epsilon\RC{\sfrac{1}{2}}{\Wm}}{\rnf^{2}(1-\rnt)}\!
\end{align}

		\item\label{tradeoff-C}
If \(\epsilon\in (0,\sfrac{\rnf}{2})\) then 
		for all \(\rnb\in(1,\tfrac{1+\rnt}{2\rnt})\) we have
		\begin{align}
		\notag
		\hspace{-.3cm}
		\RC{\rnb}{\Wma{\fX_{\epsilon}}{\qma{\fX_{\epsilon},\Wm}{\epsilon}}}
		&\!\leq\!\rate\!+\!\tfrac{2\epsilon\RC{\sfrac{1}{2}}{\Wm}}{\rnf(1-\rnf)^{2}}\!+\!\ln \tfrac{1}{\epsilon}
		\\
		\label{eq:lem:tradeoff-C}
		&\!\qquad\!+\!
(\rnb\!-\!1)e^{(\rnb-1)\frac{2\RC{\sfrac{1}{2}}{\Wm}}{1-\rnt}}
\!\left[\!\tfrac{4\vee 2\RC{\sfrac{1}{2}}{\Wm}}{1-\rnt}\right]^{2}\!.
		\end{align}
	\end{enumerate}
\end{lemma}
\begin{remark}
There is a slight abuse of notation in the symbol 
\(\Wma{\fX_{\epsilon}}{\qma{\fX_{\epsilon},\Wm}{\epsilon}}\)
in Lemma \ref{lem:tradeoff}-(\ref{tradeoff-C}).
It  stands for a \(\Vm:\inpS\to\pmea{\outA}\) satisfying
\(\Vm(\dinp)=\Wma{\fX_{\epsilon}(\dinp)}{\qma{\fX_{\epsilon}(\dinp),\Wm}{\epsilon}}(\dinp)\)
for all \(\dinp\in\inpS\)
\end{remark}

\begin{proof}[Proof of Lemma \ref{lem:tradeoff}]
\(\RC{\rno}{\Wm}\) is continuous in \(\rno\) on  \((0,1]\) by  
Lemma \ref{lem:capacityO}-(\ref{capacityO-zo}).
Then for any \(\rate\in(\RC{0^{_{+}}\!}{\Wm},\RC{1}{\Wm})\) there exists a \(\rnf\in (0,1)\) 
such that \(\rate=\RC{\rnf}{\Wm}\) by the intermediate value theorem \cite[4.23]{rudin}. 
Furthermore, \(\spe{\rate,\Wm}\leq \tfrac{1-\rnf}{\rnf}\RC{\rnf}{\Wm}\) by
the expression for \(\spe{\rate,\Wm}\) given in Lemma \ref{lem:spherepackingexponent} and 
the monotonicity of \(\tfrac{1-\rno}{\rno}\RC{\rno}{\Wm}\) in \(\rno\) established in
Lemma \ref{lem:capacityO}-(\ref{capacityO-zo}).
Then the existence of \(\rnt\) follows from  the intermediate value theorem, \cite[4.23]{rudin},
and the hypothesis of the lemma because \(\tfrac{1-\rno}{\rno}\RC{\rno}{\Wm}\) is continuous 
in \(\rno\) by Lemma \ref{lem:capacityO}-(\ref{capacityO-zo}).
\begin{enumerate}
\item[(\ref{tradeoff-A})] \(\qmn{\rno,\Wm}\) is continuous in \(\rno\) by 
Lemma \ref{lem:centercontinuity}.
Thus we can replace \(\qma{\rno,\Wm}{\epsilon}\) with \(\qmn{\rno,\Wm}\) in the proof of part (\ref{tradeoff-B}) 
to prove this part. 

\item[(\ref{tradeoff-B})]
We prove the existence of the function \(\fX_{\epsilon}\) by showing that 
\eqref{eq:lem:tradeoff-B-1} and \eqref{eq:lem:tradeoff-B-2}
are satisfied for some \(\rno\in[\rnf,\rnt]\) for each \(\dinp\in\inpS\).
We denote 
\(\Wma{\rno}{\qma{\rno,\Wm}{\epsilon}}(\dinp)\)
---which is \(\Wma{\rno}{\mQ}(\dinp)\) defined in \eqref{eq:def:tiltedchannel}
for \(\mQ=\qma{\rno,\Wm}{\epsilon}\)---
by \(\vmn{\rno}\) in the proof of this part.
Note that \(\vmn{\rno}\) satisfies
\begin{align}
\notag
\hspace{-.2cm}
\RD{1}{\vmn{\rno}}{\qma{\rno,\Wm}{\epsilon}}\!+\!\tfrac{\rno}{1-\rno} \RD{1}{\vmn{\rno}}{\!\Wm(\dinp)}
\!=\!\RD{\rno}{\!\Wm(\dinp)}{\qma{\rno,\Wm}{\epsilon}}
\end{align}
for all \(\rno\!\in\!(0,1)\). 
Then \eqref{eq:lem:avcapacity} of Lemma \ref{lem:avcapacity} implies that
\begin{align}
\label{eq:tradeoff-B-1}
\RD{1}{\vmn{\rno}}{\qma{\rno,\Wm}{\epsilon}}+\tfrac{\rno}{1-\rno} \RD{1}{\vmn{\rno}}{\!\Wm(\dinp)}
&\leq \RCI{\rno}{\Wm}{\epsilon}
\end{align}
for all \(\rno\!\in\!(0,1)\).
Then using the non-negativity of the \renyi divergence, 
which is implied by Lemma \ref{lem:divergence-pinsker}, we get
\begin{align}
\label{eq:tradeoff-B-2}
\RD{1}{\vmn{\rnf}}{\qma{\rnf,\Wm}{\epsilon}}
&\leq \RCI{\rnf}{\Wm}{\epsilon},
&
&
\\
\label{eq:tradeoff-B-3}
\RD{1}{\vmn{\rnt}}{\!\Wm(\dinp)}
&\leq \tfrac{1-\rnt}{\rnt} \RCI{\rnt}{\Wm}{\epsilon}.
&
&
\end{align}
As a result of \eqref{eq:tradeoff-B-2}, \(\RD{1}{\vmn{\rnf}}{\qma{\rnf,\Wm}{\epsilon}}\) and  
\(\RD{1}{\vmn{\rnt}}{\qma{\rnt,\Wm}{\epsilon}}\) satisfy one of the following three cases:
\begin{enumerate}[(i)]
\item If \(\RD{1}{\vmn{\rnf}}{\qma{\rnf,\Wm}{\epsilon}}\!=\!\RCI{\rnf}{\Wm}{\epsilon}\), 
then \(\RD{1}{\vmn{\rnf}}{\Wm(\dinp)}\!=\!0\) by \eqref{eq:tradeoff-B-1}.
Then \eqref{eq:lem:tradeoff-B-1} and \eqref{eq:lem:tradeoff-B-2} 
hold for \(\rno\!=\!\rnf\) 
as a result of  \eqref{eq:orderoneovertwo} and  \eqref{eq:avcapacity:bound}.
\item If \(\RD{1}{\vmn{\rnt}}{\qma{\rnt,\Wm}{\epsilon}}\!\leq\!\RCI{\rnf}{\Wm}{\epsilon}\), 
then \eqref{eq:lem:tradeoff-B-1}  and \eqref{eq:lem:tradeoff-B-2} 
hold for \(\rno\!=\!\rnt\) 
as a result of \eqref{eq:orderoneovertwo}, \eqref{eq:avcapacity:bound}, and \eqref{eq:tradeoff-B-3}.

\item If \(\!\RD{1}{\vmn{\rnf}}{\!\qma{\rnf,\Wm}{\epsilon}\!}\!<\!\RCI{\rnf}{\Wm}{\epsilon}\) 
and
\(\!\RD{1}{\vmn{\rnt}}{\!\qma{\rnt,\Wm}{\epsilon}\!}\!>\!\RCI{\rnf}{\Wm}{\epsilon}\), 
then \(\RD{1}{\vmn{\rno}}{\qma{\rno,\Wm}{\epsilon}}\!=\!\RCI{\rnf}{\Wm}{\epsilon}\)
for some \(\rno\!\in\!(\rnf,\rnt)\) by the intermediate value theorem 
\cite[4.23]{rudin} 
provided that \(\RD{1}{\vmn{\rno}}{\qma{\rno,\Wm}{\epsilon}}\) is continuous in \(\rno\). 
The continuity of \(\RD{1}{\vmn{\rno}}{\qma{\rno,\Wm}{\epsilon}}\),
on the other hand, follows from 
\(\lon{\qma{\rno,\Wm}{\epsilon}-\qma{\rno',\Wm}{\epsilon}}\!\leq\!\tfrac{1-\epsilon}{\epsilon}\!\abs{\rno\!-\!\rno'}\),
which holds for all \(\rno,\rno'\!\in\!(0,1)\),
and
Lemma \ref{lem:tilting}-(\ref{tilting-divergence}).
The \(\rno\) satisfying \(\RD{1}{\vmn{\rno}}{\qma{\rno,\Wm}{\epsilon}}=\RCI{\rnf}{\Wm}{\epsilon}\)
satisfies \eqref{eq:lem:tradeoff-B-1} 
as a result of \eqref{eq:orderoneovertwo} and  \eqref{eq:avcapacity:bound}.
Furthermore, 
\(\RD{1}{\vmn{\rno}}{\Wm(\dinp)}\!\leq\!\spa{\epsilon}{\RCI{\rnf}{\Wm}{\epsilon},\Wm}\)
for the same \(\rno\) 
by \eqref{eq:tradeoff-B-1}
and the definition of the average sphere packing exponent given in \eqref{eq:def:avspherepacking}. 
On the other hand, \(\spa{\epsilon}{\rate,\Wm}\) is a nonincreasing in \(\rate\)
because  it is the pointwise supremum of such functions.
Then 
\(\spa{\epsilon}{\RCI{\rnf}{\!\Wm}{\epsilon},\Wm}\!\leq\!\spe{\rate,\!\Wm}\!+\!\tfrac{2\epsilon}{\rnf^{2}}\rate\)
by Lemma \ref{lem:avspherepacking}.
Thus  \eqref{eq:lem:tradeoff-B-2} holds for \(\rno\) satisfying 
\(\RD{1}{\vmn{\rno}}{\qma{\rno,\Wm}{\epsilon}}=\RCI{\rnf}{\Wm}{\epsilon}\) 
by \eqref{eq:orderoneovertwo}.
\end{enumerate}
\item[(\ref{tradeoff-C})] 
We introduce two shorthands for notational brevity:
\begin{align}
\notag
\Vm(\dinp)
&=\Wma{\fX_{\epsilon}(\!\dinp\!)}{\qma{\fX_{\epsilon}(\!\dinp\!),\Wm}{\epsilon}}(\dinp),
\\
\notag
\Qm(\dinp)
&=\qma{\fX_{\epsilon}(\!\dinp\!),\!\Wm\!}{\epsilon}.
\end{align}
The \renyi divergence is a nondecreasing function of the order by 
Lemma \ref{lem:divergence-order} and \(\fX_{\epsilon}(\!\dinp\!)\in[\rnf,\rnt]\) for 
all \(\dinp\in\inpS\) by part (\ref{tradeoff-B}); then
\begin{align}
\label{eq:tradeoff-C-2}
\RD{\frac{1}{\rnt}}{\Vm(\dinp)}{\Qm(\dinp)}
&\leq \RD{\frac{1}{\fX_{\epsilon}(\!\dinp\!)}}{\Vm(\dinp)}{\Qm(\dinp)}.
\end{align}
The definitions of the \renyi divergence, \(\Vm\), and \(\Qm\) imply
\begin{align}
\notag
\hspace{-.3cm}\RD{\frac{1}{\fX_{\epsilon}(\!\dinp\!)}}{\Vm(\dinp)}{\Qm(\dinp)}
&\!=\!\RD{\fX_{\epsilon}(\!\dinp\!)}{\Wm(\dinp)}{\Qm(\dinp)}\\
\notag
&~~~~+\!\tfrac{1}{\frac{1}{\fX_{\epsilon}(\!\dinp\!)}-1}
\ln \EXS{\rfm}{\der{\Wm(\dinp)}{\rfm} \IND{\frac{\dif{\Qm(\dinp)}}{\dif{\rfm}}>0}}
\\
\label{eq:tradeoff-C-3}
&\leq  \RD{\fX_{\epsilon}(\!\dinp\!)}{\Wm(\dinp)}{\Qm(\dinp)}.
\end{align}
Using \eqref{eq:lem:avcapacityB} of Lemma \ref{lem:avcapacity}, together with 
\(\fX_{\epsilon}(\!\dinp\!)\leq \rnt\) and
\(\epsilon\leq\sfrac{\rnf}{2}\leq \sfrac{1}{2}\), we get 
\begin{align}
\label{eq:tradeoff-C-4}
\RD{\fX_{\epsilon}(\!\dinp\!)}{\Wm(\dinp)}{\Qm(\dinp)}
&\leq  \tfrac{2\RC{\sfrac{1}{2}}{\Wm}}{1-\rnt}.
\end{align}
First bounding \(\RD{\frac{1}{\rnt}}{\Vm(\dinp)}{\Qm(\dinp)}\) 
using \eqref{eq:tradeoff-C-2}, \eqref{eq:tradeoff-C-3}, \eqref{eq:tradeoff-C-4}
and then applying  Lemma \ref{lem:taylor} we get,
\begin{align}
\notag
&\RD{\rnb}{\Vm(\dinp)}{\Qm(\dinp)}-\RD{1}{\Vm(\dinp)}{\Qm(\dinp)}
\\
\notag
&\hspace{1.4cm}\leq
\tfrac{2(\rnb-1)}{e^{2}}\left(1+e^{(\rnb-1)\frac{2\RC{\sfrac{1}{2}}{\Wm}}{1-\rnt}}(\tfrac{2\RC{\sfrac{1}{2}}{\Wm}}{1-\rnt}\tfrac{e^{\tau_{\rnb}}}{2\tau_{\rnb}})^{2}\right)
\end{align}
for all \(\rnb\!\in\!(1,\tfrac{1}{\rnt})\), \(\dinp\in\inpS\)
where \(\tau_{\rnb}\!=\!(\tfrac{1}{\rnt}-\rnb)\frac{\RC{\sfrac{1}{2}}{\Wm}}{1-\rnt}\wedge 1\).
Since \(\sfrac{e^{\dsta}}{\dsta}\) is a decreasing function of \(\dsta\) on \((0,1)\), 
\begin{align}
\notag
&\RD{\rnb}{\Vm(\dinp)}{\Qm(\dinp)}-\RD{1}{\Vm(\dinp)}{\Qm(\dinp)}
\\
\notag
&\hspace{1.4cm}\leq
\tfrac{2(\rnb-1)}{e^{2}}\left(1+e^{(\rnb-1)\frac{2\RC{\sfrac{1}{2}}{\Wm}}{1-\rnt}}(\tfrac{\RC{\sfrac{1}{2}}{\Wm}}{1-\rnt}\tfrac{e^{\tau}}{\tau})^{2}\right)
\end{align}
for all \(\rnb\!\in\!(1,\tfrac{1+\rnt}{2\rnt})\), \(\dinp\!\in\!\inpS\) 
where \(\tau=\frac{\RC{\sfrac{1}{2}}{\Wm}}{2\rnt}\wedge 1\).
Thus
\begin{align}
\notag
&\RD{\rnb}{\Vm(\dinp)}{\Qm(\dinp)}-\RD{1}{\Vm(\dinp)}{\Qm(\dinp)}
\\
\notag
&\hspace{1.4cm}\leq
2(\rnb-1)
\left(1+e^{(\rnb-1)\frac{2\RC{\sfrac{1}{2}}{\Wm}}{1-\rnt}}
(\tfrac{2\vee \RC{\sfrac{1}{2}}{\Wm}}{1-\rnt})^{2}\right)
\\
\label{eq:tradeoff-C-5}
&\hspace{1.4cm}\leq
(\rnb\!-\!1)e^{(\rnb-1)\frac{2\RC{\sfrac{1}{2}}{\Wm}}{1-\rnt}}
\left[\tfrac{4\vee 2\RC{\sfrac{1}{2}}{\Wm}}{1-\rnt}\right]^{2}
\end{align}
On the other hand \(\Qm(\dinp)\leq \tfrac{1}{\epsilon} \qmn{\Wm}\) 
for all  \(\dinp\!\in\!\inpS\) by the definition of \(\Qm(\dinp)\)
where
\(\qmn{\Wm}\DEF\int_{0}^{1} \qmn{\dsta,\Wm}  \dif{\dsta}\). 
Then as a result of Lemma \ref{lem:divergence-RM} we have
\begin{align}
\label{eq:tradeoff-C-6}
\RD{\rnb}{\Vm(\dinp)}{\qmn{\Wm}}
&\leq \RD{\rnb}{\Vm(\dinp)}{\Qm(\dinp)}+\ln \tfrac{1}{\epsilon}.
\end{align}
Since \(\RR{\rnb}{\Vm}\leq \RRR{\rnb}{\Vm}{\qmn{\Wm}}\)
by definition,
\eqref{eq:lem:tradeoff-B-1},
\eqref{eq:tradeoff-C-5} and \eqref{eq:tradeoff-C-6} imply
\begin{align}
\notag
\hspace{-.4cm}\RR{\rnb}{\Vm}
&\!\leq\!
\rate\!+\!\tfrac{2\epsilon\rate}{\rnf(1-\rnf)^{2}}\!+\!\ln \tfrac{1}{\epsilon}
\\ 
\notag
&\qquad~\quad
\!+\!(\rnb\!-\!1)e^{(\rnb-1)\frac{2\RC{\sfrac{1}{2}}{\Wm}}{1-\rnt}}
\!\left[\!\tfrac{4\vee 2\RC{\sfrac{1}{2}}{\Wm}}{1-\rnt}\right]^{2}\!\!.
\end{align}
Then \eqref{eq:lem:tradeoff-C} follows from \(\RC{\rnb}{\Vm}=\RR{\rnb}{\Vm}\)
established in Theorem \ref{thm:minimax}.
\end{enumerate}
\vspace{-.3cm}
\end{proof} 

\subsection[A Non-asymptotic SPB]{\!A Non-asymptotic SPB for Product Channels with Feedback}\label{sec:fpouterbound}
The ultimate aim of this subsection is to prove Theorem \ref{thm:fDSPCexponent}.
To that end we first derive the following parametric bound on 
the error probability of codes on DSPCs with feedback.
\begin{lemma}\label{lem:spb-Fproduct} 
	Let \(\blx\) be a positive integer, \(\Wmn{\vec{[1,\blx]}}\) be a DSPC with feedback satisfying 
	\(\Wmn{\tin}\!=\!\Wm\) for all \(\tin\) for a \(\Wm\) for which  \(\RC{0^{_{+}}\!}{\Wm}\!\neq\!\RC{1}{\Wm}\),
	\(\rno_{0}\!<\!\rno_{1}\!<\!\dsta\) be orders in \((0,1)\) 
	satisfying\footnote{The existence of such a \(\dsta\) is established in Lemma \ref{lem:tradeoff}.}
	\(\tfrac{1-\dsta}{\dsta}\RC{\dsta}{\Wm}\!=\!\spe{\RC{\rno_{1}}{\Wm},\!\Wm\!}\), 
	and \(M\), \(L\), \(\knd\) be positive integers satisfying
	\(\lfloor\tfrac{\blx}{\knd}\rfloor\RC{\sfrac{1}{2}}{\Wm}\geq 2\) and
	\begin{align}
	\label{eq:lem:spb-Fproduct-hypothesis}
	\hspace{-.2cm}
	\RC{\rno_{1}}{\!\Wm}\!\geq\! 
	\tfrac{1}{\blx}\!\ln\!\tfrac{M}{L}
	&\!\geq\! 
	\RC{\rno_{0}}{\!\Wm}\!+\!\tfrac{\RC{\sfrac{1}{2}}{\!\Wm}}{1-\dsta}\!
	\!\left[\!\tfrac{2\epsilon}{\rno_{0}(1-\dsta)}\!+\!\tfrac{14}{\sqrt[3]{\knd}}\!\right]\!
	\!+\!\tfrac{\knd}{\blx}\!\ln\!\tfrac{1}{\epsilon}\!\!
	\end{align}
	for an \(\epsilon\!\in\!(0,\tfrac{\rno_{0}}{2})\).  Then any \((M,L)\) channel code on  
	\(\Wmn{\vec{[1,\blx]}}\) satisfies
	\begin{align}
	\label{eq:lem:spb-Fproduct}
	\hspace{-.3cm}
	\Pem{av}
	&\!\geq\!\tfrac{1}{4}
	e^{-\blx\left[\!\spe{\frac{1}{\blx}\ln\frac{M}{L},\Wm}
		+\frac{\RC{\sfrac{1}{2}}{\Wm}}{\rno_{0}(1-\dsta)}\left[\!\frac{6\epsilon}{\rno_{0}(1-\dsta)}+\frac{15}{\sqrt[3]{\knd}}\right]
		-\frac{\knd \ln\epsilon}{\blx\rno_{0}}\right]}\!.
	\end{align}
\end{lemma}

Lemma \ref{lem:spb-Fproduct} is proved using the auxiliary channel method:
\begin{enumerate}[(i)]
\item Apply Lemma \ref{lem:tradeoff} on subblocks
to choose \(\Vm\)
\item Use \eqref{eq:lem:augustin} to bound the error probability on \(\Vm\),
i.e. \(\Pem{\Vm}\).
\item Use Lemma \ref{lem:divergence-DPI} to bound \(\Pem{av}\)
in terms of \(\Pem{\Vm}\). 
\end{enumerate}
We have described all ingredients of the proof strategy given above,
except the concept of subblocks. 
Before the proof of Lemma \ref{lem:spb-Fproduct}, let us
revisit the DPC with feedback  and introduce the concept of subblocks.

Any DPC with feedback can be reinterpreted as a shorter
DPC with feedback with larger component channels,
which we call subblocks. 
Consider for example a length \(\blx\) DPC with feedback 
\(\Wmn{\vec{[1,\blx]}}\).
Recall that the input set of \(\Wmn{\vec{[1,\blx]}}\)
can be written in terms of the input and output sets of
the component channels as 
\begin{align}
\notag
&\bigtimes\nolimits_{\tin=1}^{\blx}{\inpS_{\tin}}^{\outS_{1}^{\tin-1}}
\end{align} 
where \(\set{A}^{\set{B}}\) is the set of all functions from
the set \(\set{B}\) to the set \(\set{A}\),
\(\set{A}^{\emptyset}=\set{A}\),
\(\outS_{\ind}^{\jnd}=\bigtimes\nolimits_{\tin=\ind}^{\jnd}\outS_{\tin}\)
for all integers \(\ind\leq\jnd\),
and \(\outS_{\ind}^{\jnd}=\emptyset\) 
for all integers \(\ind>\jnd\).
Furthermore, the output set of \(\Wmn{\vec{[1,\blx]}}\) is
\(\outS_{1}^{\blx}\)
and the transition probabilities of 
\(\Wmn{\vec{[1,\blx]}}\) can be written as
\begin{align}
\notag
\Wmn{\vec{[1,\blx]}}(\dout_{1}^{\blx}|\enc_{1}^{\blx})
&\!=\!\prod\nolimits_{\tin=1}^{\blx}\Wmn{\tin}(\dout_{\tin}|\enc_{\tin}(\dout_{1}^{\tin-1})).
\end{align}
where \(\enc_{\tin}\in {\inpS_{\tin}}^{\outS_{1}^{\tin-1}}\).

The preceding description can be modified 
to define a subblock \(\Wmn{\vec{[\tau,\tin]}}\) 
for any \(\tin>\tau\), analogously.
Furthermore, these subblocks can be used to construct alternative
descriptions of the DPC with feedback.
Let \(\tin_{0},\ldots,\tin_{\knd}\) a sequence of integers satisfying  
\(\tin_{0}\!=\!0\), \(\tin_{\knd}\!=\!\blx\), and 
\(\tin_{\jnd}\!<\!\tin_{\ind}\) for all \(\jnd\!<\!\ind\)
and \(\Umn{\ind}\!:\!\set{A}_{\ind}\to\pdis{\set{B}_{\ind}}\) 
be \(\Wmn{\vec{[1+\tin_{\ind-1},\tin_{\ind} ]}}\)
for each \(\ind\!\in\!\{1,\ldots,\knd\}\):
\begin{align}
\notag
\set{A}_{\ind}
&\!=\!\bigtimes\nolimits_{\jnd=1+\tin_{\ind-1}}^{\tin_{\ind}} {\inpS_{\jnd}}^{\outS_{1+\tin_{\ind-1}}^{\jnd-1}}
\\
\notag
\set{B}_{\ind}
&\!=\!\outS_{1+\tin_{\ind-1}}^{\tin_{\ind}}
\\
\notag
\Umn{\ind}(\bmn{\ind}|\amn{\ind})
&\!=\!\prod\nolimits_{\jnd=1+\tin_{\ind-1}}^{\tin_{\ind}}\Wmn{\jnd}(\dout_{\jnd}|\enc_{\jnd}(\dout_{1+\tin_{\ind-1}}^{\jnd-1}))
&
&
\end{align}
where 
\(\enc_{\jnd}\!\in\!{\inpS_{\jnd}}^{\outS_{1+\tin_{\ind-1}}^{\jnd-1}}\),
\(\amn{\ind}\!=\!\enc_{1+\tin_{\ind-1}}^{\tin_{\ind}}\),
and
\(\bmn{\ind}\!=\!\dout_{1+\tin_{\ind-1}}^{\tin_{\ind}}\).

Then the length \(\blx\) DPC with feedback \(\Wmn{\vec{[1,\blx]}}\) and 
the length \(\knd\) DPC with feedback \(\Umn{\vec{[1,\knd]}}\) are 
representing the same channel:
\begin{align}
\notag
\bigtimes\nolimits_{\tin=1}^{\blx}{\inpS_{\tin}}^{\outS_{1}^{\tin-1}}
&=\bigtimes\nolimits_{\ind=1}^{\knd}{\set{A}_{\ind}}^{\set{B}_{1}^{\ind-1}}
\\
\notag
\outS_{1}^{\blx}
&=\set{B}_{1}^{\knd}
\\
\notag
\Wmn{\vec{[1,\blx]}}(\dout_{1}^{\blx}|\enc_{1}^{\blx})
&=\Umn{\vec{[1,\knd]}}(\bma{1}{\knd}|\tenc_{1}^{\knd})
\end{align}
where \(\tenc_{\ind}(\bma{1}{\ind-1})
\!=\!(\enc_{1+\tin_{\ind-1}}(\bma{1}{\ind-1}),\ldots,\enc_{\tin_{\ind}}(\cdot,\bma{1}{\ind-1}))\)
and \(\bmn{\ind}\!=\!\dout_{1+\tin_{\ind-1}}^{\tin_{\ind}}\).
This observation plays a crucial role in the proof of Lemma \ref{lem:spb-Fproduct}
and hence in establishing the SPB for codes on the DPCs with feedback.

\begin{proof}[Proof of Lemma \ref{lem:spb-Fproduct}]
We divide the interval \([1,\blx]\) into \(\knd\) subintervals of,
approximately, equal length: 
we set \(\tin_{0}\) to zero and define \(\ell_{\ind}\) and \(\tin_{\ind}\)
for \(\ind \in\{1,\ldots,\knd\}\) as follows
\begin{align}
\notag
\ell_{\ind}
&\DEF \lceil  \sfrac{\blx}{\knd} \rceil  \IND{\ind\leq \blx-\lfloor \sfrac{\blx}{\knd} \rfloor \knd}+
\lfloor \sfrac{\blx}{\knd} \rfloor \IND{\ind >   \blx-\lfloor \sfrac{\blx}{\knd} \rfloor \knd},
\\
\notag
\tin_{\ind}
&\DEF \tin_{\ind-1}+\ell_{\ind}.
\end{align}
The length \(\blx\) DSPC with feedback \(\Wmn{\vec{[1,\blx]}}\) 
can be interpreted as a length \(\knd\) 
DPC\footnote{\(\Umn{\vec{[1,\blx]}}\) is stationary iff \(\ell_{\ind}\) 
	is same for all \(\ind\), i.e. iff \(\sfrac{\blx}{\knd}\) is an integer.}
with feedback \(\Umn{\vec{[1,\knd]}}\) for \(\Umn{\ind}\!:\!\set{A}_{\ind}\to\set{B}_{\ind} \)
defined as follows
\begin{align}
\notag
\Umn{\ind}
&\DEF \Wmn{\vec{[1+\tin_{\ind-1},\tin_{\ind}]}}
&
&\forall \ind\in \{1,\ldots,\knd\}.
\end{align}
As a result any \((M,L)\) channel code \((\enc,\dec)\) on the channel 
\(\Wmn{\vec{[1,\blx]}}\!:\!(\bigtimes\nolimits_{\tin=1}^{\blx}{\inpS_{\tin}}^{\outS_{1}^{\tin-1}})
\!\to\!\pdis{\outS_{1}^{\blx}}\) 
is also an \((M,L)\) channel code on 
\(\Umn{\vec{[1,\knd]}}\!:\!(\bigtimes\nolimits_{\ind=1}^{\knd}{\set{A}_{\ind}}^{\set{B}_{1}^{\ind-1}})
\!\to\!\pdis{\set{B}_{1}^{\knd}}\) 
with exactly the same error probability. 
In the rest of the proof we work with the latter interpretation.

Since \(\Wmn{\tin}\!=\!\Wm\) for all \(\tin\),  
Lemma \ref{lem:capacityFproduct} and the definition of the sphere packing exponent imply 
\begin{align}
\label{eq:spb-Fproduct-01}
\RC{\rno}{\Umn{\ind}}
&\!=\!\ell_{\ind}\RC{\rno}{\Wm}
\\
\label{eq:spb-Fproduct-02}
\spe{\RC{\rno}{\Umn{\ind}},\Umn{\ind}}
&\!=\!\ell_{\ind} \spe{\RC{\rno}{\Wm},\Wm}
\end{align}
for all \(\ind\in\{1,\ldots,\knd\}\) and \(\rno\in(0,1)\).

We define \(\rnf\) and \(\rnt\) by applying Lemma \ref{lem:tradeoff} to 
\(\Wm\) for \(\rate\) defined as
\begin{align}
\label{eq:spb-Fproduct-03}
\rate
&\!\DEF \tfrac{1}{\blx}\!\ln\!\tfrac{M}{L}
\!-\!\tfrac{\RC{\sfrac{1}{2}}{\!\Wm}}{1-\dsta}\!\!\left[\!\tfrac{2\epsilon}{\rno_{0}(1-\dsta)}\!+\!\tfrac{14}{\sqrt[3]{\knd}}\!\right]\!-\!\tfrac{\knd}{\blx}\!\ln\!\tfrac{1}{\epsilon}.
\end{align}
Then \(\rnf\!\in\![\rno_{0},\rno_{1}]\) by \eqref{eq:lem:spb-Fproduct-hypothesis}
and the monotonicity of \(\RC{\rno}{\Wm}\) in \(\rno\), i.e. Lemma \ref{lem:capacityO}-(\ref{capacityO-ilsc}).
Hence, the definition of \(\dsta\),
the monotonicity of \(\tfrac{1-\rno}{\rno}\RC{\rno}{\Wm}\) in \(\rno\), i.e. Lemma \ref{lem:capacityO}-(\ref{capacityO-zo}),
and
the monotonicity of \(\spe{\rate,\Wm}\) in \(\rate\), i.e. Lemma \ref{lem:spherepackingexponent},
imply \(\rnt\in[\rno_{0},\dsta]\).

For each \(\ind\!\in\!\{1,\ldots,\knd\}\), we define  \(\rnf_{\ind}\) and \(\rnt_{\ind}\) by applying Lemma \ref{lem:tradeoff} 
to \(\Umn{\ind}\) for \(\ell_{\ind} \rate\).
Then  \(\rnf_{\ind}\!=\!\rnf\) and \(\rnt_{\ind}\!=\!\rnt\) for all \(\ind\) by \eqref{eq:spb-Fproduct-01} and  \eqref{eq:spb-Fproduct-02}.
We denote \(\Wma{\fX_{\epsilon}}{\qma{\fX_{\epsilon},\Wm}{\epsilon}}\) 
resulting  from applying Lemma \ref{lem:tradeoff}-(\ref{tradeoff-B},\ref{tradeoff-C})
to \(\Umn{\ind}\) by \(\Vmn{\ind}:\set{A}_{\ind}\to\pdis{\set{B}_{\ind}}\), i.e.
\begin{align}
\label{eq:spb-Fproduct-04}
\Vmn{\ind}(\mA)
&={\Umn{\ind}}_{\fX_{\epsilon}(\mA)}^{\qma{\fX_{\epsilon}(\mA),\Umn{\ind}}{\epsilon}}(\mA)
&
&\forall \mA\in \set{A}_{\ind}. 
\end{align}
Then \(\rnf\in[\rno_{0},\rno_{1}]\), \(\rnt\in[\rno_{0},\dsta]\),
and  Lemma \ref{lem:tradeoff}-(\ref{tradeoff-B}) imply
\begin{align}
\label{eq:spb-Fproduct-05}
\hspace{-.2cm}
\!\RD{1}{\!\Vmn{\ind}\!(\mA)}{\!\Umn{\ind}\!(\mA)}
&\!\leq\!\spe{\ell_{\ind}\rate,\Umn{\ind}}\!+\!\tfrac{2\epsilon\RC{\sfrac{1}{2}}{\Umn{\ind}}}{\rno_{0}^{2}(1-\dsta)}\!
&
&\forall\mA\!\in\!\set{A}_{\ind}
\end{align}
On the other hand, \(\RC{\sfrac{1}{2}}{\Umn{\ind}}\geq 2\) by \eqref{eq:spb-Fproduct-01}
and the hypothesis \(\lfloor\tfrac{\blx}{\knd}\rfloor\RC{\sfrac{1}{2}}{\Wm}\!\geq\!2\). 
Thus  Lemma \ref{lem:tradeoff}-(\ref{tradeoff-C}) implies
\begin{align}
\notag
\!\!\RC{\rnb}{\Vmn{\ind}}
&\!\leq\!\ell_{\ind}\rate\!+\!\tfrac{2\epsilon\RC{\sfrac{1}{2}}{\Umn{\ind}}}{\rno_{0}(1-\dsta)^{2}}\!+\!\ln \tfrac{1}{\epsilon}
\!+\!
(\rnb\!-\!1)e^{(\rnb-1)\!\frac{2\RC{\sfrac{1}{2}}{\Umn{\ind}}}{1-\dsta}}
\!\left[\!\tfrac{2\RC{\sfrac{1}{2}}{\Umn{\ind}}}{1-\dsta}\!\right]^{2}
\end{align}
for all \(\rnb\!\in\!(1,\tfrac{1+\dsta}{2\dsta})\). 
Furthermore,
\(\lfloor\tfrac{\blx}{\knd}\rfloor\RC{\sfrac{1}{2}}{\Wm}\!\geq\!2\) and \(\knd\!\geq\!1\)
imply
\(1\!+\!\tfrac{\knd^{\sfrac{2}{3}}(1-\dsta)}{4\blx\RC{\sfrac{1}{2}}{\Wm}}\!\leq\!\tfrac{1+\dsta}{2\dsta}\).
Hence
\begin{align}
\notag
\RC{\rnb}{\Vmn{\ind}}
&\!\leq\ell_{\ind}\rate\!+\!\tfrac{2\epsilon\RC{\sfrac{1}{2}}{\Umn{\ind}}}{\rno_{0}(1-\dsta)^{2}}\!+\!
\ln \tfrac{1}{\epsilon}
+\tfrac{e^{\sfrac{1}{\sqrt[3]{\knd}}}}{\sqrt[3]{\knd}}\tfrac{2\RC{\sfrac{1}{2}}{\Umn{\ind}}}{1-\dsta}
\\
\label{eq:spb-Fproduct-06}
&\!\leq\ell_{\ind}\rate\!+\!\tfrac{2\epsilon\RC{\sfrac{1}{2}}{\Umn{\ind}}}{\rno_{0}(1-\dsta)^{2}}\!+\!
\ln \tfrac{1}{\epsilon}
+\tfrac{6}{\sqrt[3]{\knd}}\tfrac{\RC{\sfrac{1}{2}}{\Umn{\ind}}}{1-\dsta}
\end{align}
for \(\rnb=1+\tfrac{\knd^{\sfrac{2}{3}}(1-\dsta)}{4\blx \RC{\sfrac{1}{2}}{\Wm}}\).

We use \(\Vmn{\ind}\)'s described in \eqref{eq:spb-Fproduct-04} to define the
length \(\knd\)  DPC with 
feedback  
\(\Vmn{\vec{[1,\knd]}}\!:\!\left(\bigtimes\nolimits_{\ind=1}^{\knd}{\set{A}_{\ind}}^{\set{B}_{1}^{\ind-1}}\right)
\!\to\!\pdis{\set{B}_{1}^{\knd}}\).
Then using Lemma \ref{lem:capacityFproduct},  \eqref{eq:spb-Fproduct-01},  
\eqref{eq:spb-Fproduct-03}, and \eqref{eq:spb-Fproduct-06} we get
\begin{align}
\notag
\RC{\rnb}{\Vmn{\vec{[1,\blx]}}}
&\!\leq\!\ln\!\tfrac{M}{L}\!-\!\blx\tfrac{\RC{\sfrac{1}{2}}{\!\Wm}}{1-\dsta}\tfrac{8}{\sqrt[3]{\knd}}
\end{align}
for \(\rnb=1+\tfrac{\knd^{\sfrac{2}{3}} (1-\dsta)}{4\blx\RC{\sfrac{1}{2}}{\Wm}}\).
We bound the average error probability of \(\!(\enc,\dec)\!\) 
on \(\!\Vmn{\vec{[1,\blx]}}\!\), i.e. \(\!\Pem{\!\Vmn{\vec{[1,\blx]}}}\!\),
using \eqref{eq:lem:augustin}
and \(\tau\geq\ln(1+\tau)\): 
\begin{align}
\notag
\Pem{\Vmn{\vec{[1,\blx]}}}
&\geq 1-e^{-\sqrt[3]{\knd}}
\\
\label{eq:spb-Fproduct-07}
&\geq \tfrac{\sqrt[3]{\knd}}{1+\sqrt[3]{\knd}}.
\end{align}

On the other hand  \eqref{eq:spb-Fproduct-01}, \eqref{eq:spb-Fproduct-02}, and \eqref{eq:spb-Fproduct-05}
imply
\begin{align}
\label{eq:spb-Fproduct-08}
\hspace{-.2cm}
\RD{1}{\!\Vmn{\vec{[1,\knd]}}(\dinp)}{\!\Umn{\vec{[1,\knd]}}(\dinp)}
&\!\leq\!\blx\spe{\rate,\Wm}
\!+\!\blx\tfrac{2\epsilon  \RC{\sfrac{1}{2}}{\Wm}}{\rno_{0}^{2}(1-\dsta)}
\end{align}
for all \(\dinp\!\in\!\left(\bigtimes\nolimits_{\ind=1}^{\knd}{\set{A}_{\ind}}^{\set{B}_{1}^{\ind-1}}\right)\).

Let \(\mP\) be the probability distribution generated by the encoder \(\enc\) on the input set of \(\Umn{\vec{[1,\knd]}}\),
i.e. on \(\!\left(\!\bigtimes\nolimits_{\ind=1}^{\knd}{\set{A}_{\ind}}^{\set{B}_{1}^{\ind-1}}\!\right)\!\), 
for the uniform distribution over the message set. 
Then
Lemma \ref{lem:divergence-DPI}
and the identity 
\(\tau \ln \tau\!+\!(1-\tau)\ln(1-\tau)\!\geq\!\ln\sfrac{1}{2}\), 
which holds for all \(\tau\in[0,1]\),
imply
\begin{align}
\label{eq:spb-Fproduct-09}
\RD{1}{\mP \mtimes  \Vmn{\vec{[1,\knd]}}}{\mP \mtimes  \Umn{\vec{[1,\knd]}}}
&\geq \ln\sfrac{1}{2}-\Pem{\Vmn{\vec{[1,\blx]}}} \ln \Pem{av}.
\end{align}

Note that \(\RD{1}{\!\mP\mtimes\!\Vmn{\vec{[1,\knd]}}}{\!\mP\mtimes\!\Umn{\vec{[1,\knd]}}}\)
is bounded from above by  the supremum of 
\(\RD{1}{\!\Vmn{\vec{[1,\knd]}}(\dinp)}{\!\Umn{\vec{[1,\knd]}}(\dinp)}\)
over the common input set of \(\Vmn{\vec{[1,\knd]}}\) and \(\Umn{\vec{[1,\knd]}}\), i.e. \(\bigtimes\nolimits_{\ind=1}^{\knd}{\set{A}_{\ind}}^{\set{B}_{1}^{\ind-1}}\).
Then using 
\eqref{eq:spb-Fproduct-07}, 
\eqref{eq:spb-Fproduct-08}, 
and  
\eqref{eq:spb-Fproduct-09}
we get
\begin{align}
\notag
\ln \Pem{av}
&\geq -\tfrac{1+\sqrt[3]{\knd}}{\sqrt[3]{\knd}}
\left(\blx \spe{\rate,\Wm}
\!+\!\blx\tfrac{2\epsilon\RC{\sfrac{1}{2}}{\Wm}}{\rno_{0}^{2}(1-\dsta)}
\!+\!\ln 2\right).
\end{align}
Then using the identity \(\!\spe{\RC{\rnf}{\Wm},\Wm}\!\leq\!\tfrac{(1\!-\!\rnf)\RC{\rnf}{\Wm}}{\rnf}\),
which is implied by Lemma \ref{lem:spherepackingexponent},
together with \eqref{eq:orderoneovertwo}, \eqref{eq:lem:spb-Fproduct-hypothesis}, 
and \eqref{eq:spb-Fproduct-03} we get
\begin{align}
\notag
\ln \Pem{av}
&\geq -\blx\spe{\rate,\Wm}
\!+\!\blx\tfrac{\RC{\sfrac{1}{2}}{\Wm}}{\rno_{0}}\!
\left[\tfrac{4\epsilon}{\rno_{0}(1-\dsta)}\!+\!\tfrac{1}{\sqrt[3]{\knd}}\right]
\!+\!2\ln 2.
\end{align}
Then \eqref{eq:lem:spb-Fproduct} is implied by
\eqref{eq:spb-Fproduct-03} 
via the following consequence
of  Lemma \ref{lem:spherepackingexponent}:
If \(\rate\!=\!\RC{\rno}{\Wm}\)  for an \(\rno\!\in\![\rno_{0},\rno_{1}]\),
\begin{align}
\notag
\spe{\rate,\Wm}
\leq \spe{\rate+\delta,\Wm}+\tfrac{(1-\rno_{0})}{\rno_{0}}\delta
&
&\forall\delta\geq0.
\end{align}
~\vspace{-.9cm}\\
\end{proof}

\begin{remark}
The input sets of the subblocks grow rapidly with their length; 
in particular  
\begin{align}
\notag
\ln\abs{\set{A}_{\ind}}
&=\left(\sum\nolimits_{\jnd=0}^{\ell_{\ind}-1}\abs{\outS}^{\jnd}\right)\ln\abs{\inpS}.
\end{align}
This rapid growth would have made our bounds useless, at least for establishing a 
result in the spirit Lemma \ref{lem:spb-Fproduct},
if the approximation error terms in Lemma \ref{lem:tradeoff}-(\ref{tradeoff-B},\ref{tradeoff-C}) 
were in terms of  \(\ln \abs{\inpS}\) rather than \(\RC{\sfrac{1}{2}}{\Wm}\),
which grows only linearly with the length of the subblock.
\end{remark}

One is initially inclined to use Lemma \ref{lem:spb-Fproduct} either 
for \(\knd\!=\!\blx\) or for \(\knd\!=\!1\), i.e.
apply Lemma \ref{lem:tradeoff} either 
to \(\Wmn{\vec{[1,\blx]}}\) or to the component channel \(\Wm\).
Both of these choices, however, lead to poor approximation error terms. 
Instead we use Lemma \ref{lem:spb-Fproduct} for 
\(\knd\approx \blx^{\sfrac{3}{4}}\) to prove Theorem \ref{thm:fDSPCexponent}. 
In \cite{augustin78}, while proving a statement similar to Theorem \ref{thm:fDSPCexponent}, Augustin 
used subblocks in a similar fashion; other ingredients of Augustin's analysis, however, are quite different. 
Palaiyanur discussed Augustin's proof sketch in more detail in his thesis \cite[A.8]{palaiyanurthesis}.
A complete proof following Augustin's sketch can be found in \cite{nakiboglu19E}.

\begin{proof}[Proof of Theorem \ref{thm:fDSPCexponent}]
We prove Theorem \ref{thm:fDSPCexponent} by applying
Lemma \ref{lem:spb-Fproduct} 
for appropriately chosen \(\epsilon_{\blx}\)
and \(\knd_{\blx}\).
Note that \(\epsilon\) and \(\knd\)  can take any value 
as long as the hypothesis of Lemma \ref{lem:spb-Fproduct}
is satisfied.
For \(\epsilon_{\blx}\!=\!\tfrac{\rno_{0}(1-\rnt)}{\sqrt[4]{\blx}}\)
and \(\knd_{\blx}\!=\!\lfloor \blx^{\sfrac{3}{4}} \rfloor\),
the hypothesis  \(\lfloor\tfrac{\knd}{\blx}\rfloor\RC{\sfrac{1}{2}}{\Wm}\geq 2\)
holds for all \(\blx\) large enough. 
Furthermore, the other hypothesis of Lemma \ref{lem:spb-Fproduct} given in
\eqref{eq:lem:spb-Fproduct-hypothesis} 
is satisfied for \(\blx\) large enough because of \eqref{eq:thm:fDSPCexponent-hypothesis}.
Thus we can apply Lemma \ref{lem:spb-Fproduct} 
with \(\epsilon_{\blx}=\tfrac{\rno_{0}(1-\rnt)}{\sqrt[4]{\blx}}\)
and \(\knd_{\blx}=\lfloor \blx^{\sfrac{3}{4}} \rfloor\) 
for all \(\blx\) large enough. 
Then \eqref{eq:lem:spb-Fproduct} implies \eqref{eq:thm:fDSPCexponent}
for \(\blx\) large enough.
\end{proof}

\subsection{Extensions and Comparisons}\label{sec:fcomparison}
Theorem \ref{thm:fDSPCexponent} is stated for stationary sequences of channels,
but it holds for periodic sequences of channels too.
In other words, Theorem \ref{thm:fDSPCexponent} assumed \(\Wmn{\tin}\!=\!\Wm\) for all 
\(\tin\!\in\!\integers{+}\); but its assertions hold whenever there exists a
\(\tau\!\in\!\integers{+}\) satisfying \(\Wmn{\tin}\!=\!\Wmn{\tin+\tau}\) for all 
\(\tin\!\in\!\integers{+}\).
Thus the SPB holds for codes on the periodic discrete product channels, as well.

It is possible establish similar results under weaker stationarity hypotheses.
In order to prove the SPB for codes on the DPCs with feedback using the approach
employed for proving Theorem \ref{thm:fDSPCexponent}, 
we need the \renyi capacity of the subblocks to be approximately
equal to one another as functions, i.e.
\begin{align}
\notag
\RC{\rno}{\Umn{\ind}}
&\approx\tfrac{\ell_{\ind}}{\blx}\RC{\rno}{\Wmn{[1,\blx]}}
\end{align}
uniformly over \(\ind\) and \(\rno\).
This condition is a stationarity hypotheses too; but it is considerably
weaker than assuming all \(\Wmn{\tin}\)'s to be identical.
There is not just one but many precise ways to impose this condition
and each one of them leads to a slightly different result. 
Assumption \ref{assumption:astationary} and Theorem \ref{thm:fDPCexponent}
are provided as examples.
In order to prove Theorem \ref{thm:fDPCexponent} we need to modify
Lemmas \ref{lem:tradeoff} and \ref{lem:spb-Fproduct}, slightly. 
We present those modifications, their proofs, and the proof of 
Theorem \ref{thm:fDPCexponent} in 
Appendix \ref{sec:fDPCexponent}.
%\cite[Appendix \ref*{B-sec:fDPCexponent}]{nakiboglu18B}.
\begin{assumption}\label{assumption:astationary}
\(\{\Wmn{\tin}\}_{\tin\in\integers{+}}\) is a sequence of channels 
satisfying the following three conditions for some \(\rens\!:\!(0,1)\to\reals{+}\)
\begin{enumerate}[i.]
	\item \(\lim_{\blx\to\infty}\!\tfrac{1}{\blx}\!\RC{\rno}{\Wmn{[1,\blx]}}\!=\!\rens(\rno)\)
	for all \(\rno\!\in\!(0,1)\).
	\item \(\lim_{\rno\uparrow1}\tfrac{1-\rno}{\rno}\rens(\rno)=0\).
	\item There exists 
	\(K\!\in\!\reals{+}\) and \(\blx_{0}\!\in\!\integers{+}\)  such that
	\begin{align}
	\notag
	\sup\limits_{\rno\in(0,1)}\sup\limits_{\tin\in\integers{+}}
	\abs{\RC{\rno}{\Wmn{[\tin,\tin+\blx-1]}}-\blx\rens(\rno)}
	&\!\leq\!K\!\ln\!\blx
	&
	&\forall\blx\geq\blx_{0}.
	\end{align}
\end{enumerate}
\end{assumption}

\begin{theorem}\label{thm:fDPCexponent}
	Let \(\{\Wmn{\tin}\}_{\tin\in\integers{+}}\) be a sequence of discrete channels
	satisfying Assumption \ref{assumption:astationary}
	and \(\rno_{0}\), \(\rno_{1}\) be orders satisfying \(0\!<\!\rno_{0}\!<\!\rno_{1}\!<\!1\).
	Then for any sequence of codes on the discrete product channels with feedback 
	\(\{\Wmn{\vec{[1,\blx]}}\}_{\blx\in\integers{+}}\) satisfying 
	\begin{align}
\label{eq:thm:fDPCexponent-hypothesis}
\RC{\rno_{1}}{\Wmn{[1,\blx]}}\!\geq\!
\ln \tfrac{M_{\blx}}{L_{\blx}}
&\!\geq\!\RC{\rno_{0}}{\Wmn{[1,\blx]}}\!+\!(K\!+\!1)\!\blx^{\sfrac{3}{4}}\!\ln\!\blx
&
&\forall \blx\geq\blx_{0} 
\end{align}
	there exists an \(\blx_{1}\geq\blx_{0}\) such that 
	\begin{align}
	\label{eq:thm:fDPCexponent}
	\Pem{av}^{(\blx)}
	&\!\geq\!e^{-\spe{\ln\frac{M_{\blx}}{L_{\blx}},\Wmn{[1,\blx]}}-\frac{6K+1}{\rno_{0}}\blx^{\sfrac{3}{4}}\ln\blx}
	&
	&\forall \blx\!\geq\!\blx_{1}.
	\end{align}
\end{theorem}

We have confined the claims of Theorem \ref{thm:fDPCexponent} to discrete channels in order avoid 
certain measurability issues. 
We believe, however, it should be possible to resolve those issues and to extend Theorem \ref{thm:fDPCexponent}
to  any sequence of channels  satisfying Assumption \ref{assumption:astationary}.
Augustin makes the same conjecture for the stationary channels in \cite[Cor. 41.9]{augustin78}.

Augustin sketches a derivation of the SPB for codes on 
finite input set SPCs with feedback in \cite[\S41]{augustin78}.
The approximation error terms in Augustin's asymptotic SPB
\cite[Thm. 41.7]{augustin78} 
are \(\bigo{\blx^{-\sfrac{1}{3}}\ln \blx}\) 
rather than \(\bigo{\blx^{-\sfrac{1}{4}}\ln\blx}\).
A complete proof of SPB for codes on DSPCs with feedback 
following Augustin's sketch can be found in \cite{nakiboglu19E}.

Throughout this section, we have refrained from making any assumptions 
on the \renyi centers of the component channels or 
their relation to the output distributions of the component channels.
Such assumptions may lead to sharper bounds under milder stationarity hypotheses. 
For example, Lemma \ref{lem:spb-fixed-density} establishes a non-asymptotic bound for certain
product channels with feedback that can be used, in place of Lemma \ref{lem:spb-product}, 
to prove the asymptotic SPB given in Theorem \ref{thm:productexponent}
under Assumption \ref{assumption:individual-ologn}.
Thus if the sequence \(\{\Wmn{\tin}\}_{\tin\in\integers{+}}\) satisfies 
Assumption \ref{assumption:individual-ologn}
and
every channel in \(\{\Wmn{\tin}\}_{\tin\in\integers{+}}\) satisfies 
the hypothesis of Lemma \ref{lem:spb-fixed-density}, then the SPB
holds with a polynomial prefactor for codes on \(\Wmn{\vec{[1,\blx]}}\).
First Dobrushin \cite{dobrushin62A} and then Haroutunian \cite{haroutunian77} 
employed similar observations to establish the SPB for codes on 
certain DSPCs with feedback.
Later, Augustin \cite[p. 318]{augustin78} did the same for codes on certain product channels 
with feedback.

\subsection{Haroutunian's Bound and Subblocks}\label{sec:haroutunian}
Haroutunian's article \cite{haroutunian77} is probably 
the most celebrated work on the exponential lower bounds 
to the error probability of channel codes on DSPCs with feedback. 
In the rest of this section we discuss \cite{haroutunian77}
in light of Lemma \ref{lem:tradeoff} and the concept of subblocks.

In \cite{haroutunian77}, Haroutunian considers  \((M,L)\)
channel codes satisfying \(\rate\!=\!\tfrac{1}{\blx}\ln \tfrac{M}{L}\) 
on DSPCs with feedback \(\Wmn{\vec{[1,\blx]}}\) 
satisfying \(\Wmn{\tin}=\Wm\)
for a \(\Wm\!:\!\inpS\!\to\!\pdis{\outS}\) to prove that for
any rate \(\rate\geq0\) and \(\varepsilon\!>\!0\)
the following bound holds for large enough \(\blx\)
\begin{align}
\label{eq:haroutunianbound}
\Pem{av}^{(\blx)}\geq (1-\varepsilon) e^{-\blx (E_{h}(\rate-\varepsilon,\Wm)+\varepsilon)}
\end{align}
where \(E_{h}(\rate,\Wm)\), which is customarily called Haroutunian's exponent, 
is defined as, \cite[p. 180]{csiszarkorner}, \cite[(15)]{haroutunian77},  
\begin{align}
\label{eq:def:haroutunianexponent}
\hspace{-.2cm}E_{h}(\rate,\Wm)
&\!\DEF\!\inf\nolimits_{\Vm:\RC{1}{\Vm}\leq\rate} \sup\nolimits_{\dinp\in\inpS}\RD{1}{\Vm(\dinp)}{\Wm(\dinp)}. 
\end{align}
Haroutunian points out not only that \(E_{h}(\rate,\Wm)\) is greater than or equal to 
\(\spe{\rate,\Wm}\) for all \(\rate\),
but also that the inequality is strict on \((\RC{0}{\Wm},\RC{1}{\Wm})\) 
even for most of the binary input binary output 
channels, see \cite[Thm. 3.1]{haroutunian77}.
Thus, for certain \(\Wm\)'s there does not exist any \(\Vm\) satisfying both  \(\RC{1}{\Vm}\!\leq\!\rate\) 
and  \(\sup\nolimits_{\dinp\in\inpS}\RD{1}{\Vm(\dinp)}{\Wm(\dinp)}\!\leq\!\spe{\rate,\Wm}\) at the same time. 

On the other hand, both inequalities are satisfied approximately for 
\(\Vm\!=\!\Wma{\fX_{\epsilon}}{\qma{\fX_{\epsilon},\Wm}{\epsilon}}\)
by Lemma \ref{lem:tradeoff}-(\ref{tradeoff-B},\ref{tradeoff-C}).
In particular,
\begin{align}
\label{eq:haroutunian1}
\sup\nolimits_{\dinp\in\inpS}\!\RD{1}{\!\Vm(\dinp)}{\!\Wm(\dinp)}
&\!\leq\!\spe{\rate,\Wm}\!+\!\tfrac{2\epsilon\RC{\sfrac{1}{2}}{\Wm}}{\rnf^{2}(1-\rnt)}\!
\\
\label{eq:haroutunian2}
\RC{1}{\Vm}
&\!\leq\!\rate\!+\!\tfrac{2\epsilon\RC{\sfrac{1}{2}}{\Wm}}{\rnf(1-\rnf)^{2}}\!+\!\ln \tfrac{1}{\epsilon}
\end{align}
for any \(\epsilon\!\in\!(0,\sfrac{\rnf}{2})\) where \(\rnf\) and \(\rnt\) are determined
uniquely by  
\(\RC{\rnf}{\Wm}\!=\!\rate\) and \(\tfrac{1-\rnt}{\rnt}\!=\!\spe{\rate,\Wm}\).

If we apply \eqref{eq:haroutunian1} and \eqref{eq:haroutunian2}
to \(\Wmn{\vec{[1,\blx]}}\) 
for \(\rate_{\blx}\!=\!\blx(\rate\!-\!\varepsilon)\!\)
and
\(\epsilon_{\blx}\!=\!\sfrac{1}{\blx}\),
then the additivity of the \renyi capacity for the product channels with feedback,
i.e. Lemma \ref{lem:capacityFproduct},
and monotonicity of \(E_{h}(\rate,\Wm)\) in \(\rate\),
i.e. \cite[Thm. 3.5]{haroutunian77}, 
imply 
\begin{align}
\notag
\limsup\nolimits_{\blx\to\infty} \tfrac{1}{\blx} 
E_{h}(\blx\rate,\Wmn{\vec{[1,\blx]}})
&\leq \spe{\rate-\varepsilon,\Wm}
&
&\forall\varepsilon>0.
\end{align}
Then the continuity of \(\spe{\rate,\Wm}\) in \(\rate\), i.e Lemma \ref{lem:spherepackingexponent},
and the identity \(\spe{\rate,\Wm}\leq E_{h}(\rate,\Wm)\)
imply
\begin{align}
\label{eq:haroutunian3}
\lim\nolimits_{\blx\to\infty} \tfrac{1}{\blx} 
E_{h}(\blx\rate,\Wmn{\vec{[1,\blx]}})
&=\spe{\rate,\Wm}.
\end{align}
Recall that any channel code on \(\Wmn{\vec{[1,\blx\ell]}}\)
is also a channel code on \(\Umn{\vec{[1,\blx]}}\)
where \(\Umn{\tin}\!=\!\Wmn{\vec{[1,\ell]}}\) for all \(\tin\),
see the discussion of the concept of subblocks in the beginning of \S\ref{sec:fpouterbound}.
Thus \eqref{eq:haroutunianbound} implies that
for any \(\varepsilon\!>\!0\) for large enough \(\blx\)
\begin{align}
\notag
\Pem{av}^{(\ell\blx)}\geq (1-\varepsilon) e^{-\blx (E_{h}(\ell\rate-\varepsilon,\Wmn{\vec{[1,\ell]}})+\varepsilon)}.
\end{align}
Then \eqref{eq:haroutunian3} implies
\begin{align}
\notag
\limsup\nolimits_{\blx\to\infty}
\tfrac{1}{\blx}\ln \tfrac{1}{\Pem{av}^{(\blx)}}
\leq \spe{\rate,\Wm}.
\end{align}
Thus Lemma \ref{lem:tradeoff}, the concept of subblocks, and Haroutunian's bound
imply the most important asymptotic conclusion of Theorem \ref{thm:fDSPCexponent}, 
i.e. the reliability function of the DSPC with feedback is bounded from above
by the sphere packing exponent.

\begin{remark}
In \cite{palaiyanurS15}, Palaiyanur and Sahai used the method of types 
to establish the following relation for all discrete channels \(\Wm\),
\begin{align}
\label{eq:palaiyanursahai}
\lim\nolimits_{\blx\to\infty}\tfrac{1}{\blx}E_{h}(\blx\rate,\Wmn{[1,\blx]})&=\spe{\rate,\Wm}.
\end{align}
This was first reported in Palaiyanur's thesis \cite[Lemma 7]{palaiyanurthesis}.
Note that \eqref{eq:palaiyanursahai} cannot be used in the preceding
argument to establish 
the sphere packing exponent as an upper bound to the reliability function 
of the DSPCs with feedback because 
\(\Wmn{\vec{[1,\blx\ell]}}\) is not equivalent to \(\Bmn{\vec{[1,\blx]}}\)
for \(\Bmn{\tin}=\Wmn{[1,\ell]}\).

It is worth mentioning that  \eqref{eq:haroutunian3} implies \eqref{eq:palaiyanursahai}
by the definition of Haroutunian's exponent.
\end{remark}  
%!TEX root=../main-B.tex
\section{Discussion}\label{sec:conclusion}
We have established SPBs with approximation error terms that are polynomial 
in the block length for a class of product channels,
which includes all stationary product channels. 
Our results hold for a large class of non-stationary product channels, 
which might have infinite channel capacity. 

We have presented a new proof of the SPB for the codes 
on DSPCs with feedback that can be applied to the codes 
on DPCs with feedback satisfying a milder stationarity hypothesis,
see \S\ref{sec:fcomparison} and 
Appendix \ref{sec:fDPCexponent}.
%\cite[Appendix \ref*{B-sec:fDPCexponent}]{nakiboglu18B}.
The validity of SPB for codes on DSPCs with feedback implies 
improvements in the bounds
for codes with errors-and-erasures decoding on DSPCs with feedback 
that were previously derived using Haroutunian's bound  
in \cite[{\S}V,{\S}IV]{nakibogluZ12} and \cite[\S2.4,\S2.5]{nakiboglu11}.

In our judgment, the averaging described in \S\ref{sec:averaging} is one 
way of employing the following more fundamental observation
\begin{align}
\notag
\lim\nolimits_{\rnf\to\rno} \RRR{\rno}{\Wm}{\qmn{\rnf,\Wm}}
&=\RR{\rno}{\Wm}
&
&\forall \rno\in(0,1).
\end{align}
The preceding observation and Theorem \ref{thm:minimax}, are at the heart of Augustin's method. 
However, only the preceding observation can be  interpreted as a novelty of Augustin's 
method because Theorem \ref{thm:minimax} is employed by Shannon, Gallager, and 
Berlekamp in \cite{shannonGB67A}, albeit in an indirect way and for discrete
channels
only.\footnote{The equality of the \renyi capacity to the \renyi radius and the existence of a 
	\renyi center is invoked via \cite[(4.22)]{shannonGB67A}. The uniqueness of the \renyi center 
	is implicit in the analysis of \(\mathtt{f_{s}}\) as a function of \(\mathtt{s}\); 
	it is established in the discussion between \cite[(A27) and (A28)]{shannonGB67A}.}

In \S \ref{sec:product-outerbound} and \S\ref{sec:fproduct-outerbound}, we have confined our 
discussion of the SPB to the product channels.
The \renyi capacity and center, as defined in \S\ref{sec:capacity}, served our purposes 
satisfactorily.
For studying the SPB on the memoryless channels, however, the Augustin capacity and center, 
described below, are better suited. 
The \renyi information has multiple non-equivalent definitions.
The following definition was proposed and analyzed by Augustin \cite[\S34]{augustin78}
and later popularized by \csiszar \cite{csiszar95}:
\begin{align}
%\label{eq:def:augustininformation}
\notag
{\cnst{I}}_{{\rno}}^{{\scriptscriptstyle c}}\!\left(\!\mP;\!\Wm\!\right)
&\DEF\inf\nolimits_{\mQ\in\pmea{\outA}} \sum\nolimits_{\dinp}\mP(\dinp) \RD{\rno}{\Wm(\dinp)}{\mQ}.
\end{align}
We have called this quantity Augustin information in \cite{nakiboglu17,nakiboglu18C}.
The Augustin capacity and center are defined analogously to 
the \renyi capacity and center.\footnote{The constrained \renyi capacity 
\(\CRC{\rno}{\Wm}{\cset}\) is defined by taking the supremum of the 
\renyi information over the priors in a
subset \(\cset\) of \(\pdis{\inpS}\), rather than \(\pdis{\inpS}\) itself,
see \cite[Appendix \ref*{A-sec:constrainedcapacity}]{nakiboglu19A}.
The unconstrained \renyi and Augustin capacities are 
equal, see \cite[Prop. 1]{csiszar95} or
\cite[Thms. \ref*{C-thm:Lminimax} and \ref*{C-thm:Gminimax}]{nakiboglu18C}; 
however, this is not the case 
in general for the constrained capacities.\label{footnotenumber}}
Using these concepts and assuming a bounded cost function Augustin derived
a SPB for channel codes on  the cost constrained memoryless channels in 
\cite[Ch. VII]{augustin78}.
Augustin's framework is general enough to subsume the Poisson channels described 
in \eqref{eq:def:poissonchannel} 
as special cases in the way that the framework of Theorem \ref{thm:productexponent} 
subsumed the Poisson channels described  in 
\eqref{eq:def:poissonchannel-bounded} and \eqref{eq:def:poissonchannel-product} as 
special cases.  
The Gaussian channels studied by Shannon \cite{shannon59}, Ebert \cite{ebert66},
and Richters \cite{richters67},
however, are not subsumed by Augustin's framework because the quadratic cost function 
used for these channels is not bounded. 
To remedy this situation, we have recently derived a SPB with 
a polynomial prefactor for codes on the cost constrained (possibly non-stationary) 
memoryless channels, \cite{nakiboglu17,nakiboglu18D}, 
without assuming the boundedness of the cost function.
We have also derived the SPB for codes 
on the stationary memoryless channels with 
convex composition constraints on the codewords
in \cite{nakiboglu17,nakiboglu18D}.
It seems extending the results to the channels with memory is the pressing issue
in this line of work; but that is likely to be more challenging than 
the case of memoryless channels.

\numberwithin{equation}{section}
\appendices
%!TEX root=../main-B.tex
\section{Proof of Theorem \ref{thm:fDPCexponent}}\label{sec:fDPCexponent}
The proof of Theorem \ref{thm:fDPCexponent} relies on 
the variants of Lemmas \ref{lem:tradeoff}
and \ref{lem:spb-Fproduct} given in 
Lemmas \ref{lem:tradeoff-gamma} and 
\ref{lem:spb-Fproduct-gamma} presented in the following.
\begin{lemma}\label{lem:tradeoff-gamma}
	Let \(\rens\!:\!(0,1)\!\to\!\reals{+}\) be an increasing function for which
	\(\tfrac{1-\rno}{\rno}\!\rens(\rno)\) is decreasing 
	and \(\lim_{\rno\uparrow 1}\frac{1-\rno}{\rno}\rens(\rno)\!=\!0\),
	\(\gX\!:\!(0,1)\!\to\!\reals{+}\) be 
	\(\gX(\dsta)\!\DEF\!\sup_{\rno\in[\dsta,1)}\tfrac{1-\rno}{\rno}(\rens(\rno)-\rens(\dsta))\),
	\(\rnf\) be in \((0,1)\),
	\(\rnt\!\in\!(\rnf,1)\) be such that \(\gX(\rnf)\!=\!\tfrac{1-\rnt}{\rnt}\rens(\rnt)\),
	and \(\Wm\!:\!\inpS\!\to\!\pmea{\outA}\) be a channel satisfying 
	\begin{align}
	\label{eq:lem:tradeoff-gamma-hypothesis}
	\RC{\rno}{\Wm}
	&\leq \rens(\rno)+\gamma
	&
	&\forall \rno\in[\rnf,\rnt]
	\end{align}
	for a \(\gamma\in\reals{+}\).
	Then for all \(\epsilon\in (0,\sfrac{\rnf}{2})\)
	there exists a channel \(\Vm\!:\!\inpS\!\to\!\pmea{\outA}\) satisfying
	both
	\begin{align}
	\label{eq:lem:tradeoff-gamma-B}
	\!\RD{1}{\!\Vm(\dinp)}{\!\Wm(\dinp)}
	&\!\leq\!\gX(\rnf)\!+\!\tfrac{\gamma}{\rnf}
	\!+\!\tfrac{2\epsilon(\rens(\sfrac{1}{2})+\gamma)}{\rnf^{2}(1-\rnt)}
	\end{align}
	\(\dinp\in\inpS\) and
	\begin{align}
	\notag
	\hspace{-.3cm}
	\RC{\rnb}{\Vm}
	&\!\leq\!\rens(\rnf)\!+\!\gamma\!+\!\tfrac{2\epsilon(\rens(\sfrac{1}{2})+\gamma)}{\rnf(1-\rnt)^{2}}\!+\!\ln \tfrac{1}{\epsilon}
	\\
	\label{eq:lem:tradeoff-gamma-C}
	&\qquad+
	(\rnb\!-\!1)e^{(\rnb-1)\frac{2(\rens(\sfrac{1}{2})+\gamma)}{\rnf(1-\rnt)^{2}}}
	\!\left[\!\tfrac{4\vee 2(\rens(\sfrac{1}{2})+\gamma)}{\rnf (1-\rnt)^{2}}\right]^{2}
	\end{align}
	for all \(\rnb\in(1,\tfrac{1+\rnt}{2\rnt})\), simultaneously.
\end{lemma}
\begin{proof}[Proof of Lemma \ref{lem:tradeoff-gamma}]
	The existence of a unique \(\rnf\) and the corresponding \(\rnt\) follows from
	the monotonicity properties of \(\rens\),
	the continuity of \(\rens\) implied by them,
	the intermediate value theorem \cite[4.23]{rudin},
	and \(\lim_{\rno\uparrow 1}\frac{1-\rno}{\rno}\rens(\rno)=0\).
	
	The monotonicity of \(\rens(\rno)\) and \(\tfrac{1-\rno}{\rno}\rens(\rno)\)
	in \(\rno\) also imply
	\begin{align}
	\label{eq:orderoneovertwo-gamma}
	\rens(\rno)&\leq \tfrac{\rens(\sfrac{1}{2})}{1-\rno}.
	\end{align}	
	In the following, we prove that for all \(\epsilon\in(0,\sfrac{\rnf}{2})\) 
	there exists an \(\fX\!:\!\inpS\!\to\![\rnf,\rnt]\) satisfying 
	both \eqref{eq:lem:tradeoff-gamma-B-1} and \eqref{eq:lem:tradeoff-gamma-B-2} 
	for all \(\dinp\!\in\!\inpS\).
	\begin{align}
	\label{eq:lem:tradeoff-gamma-B-1}
	\hspace{-.5cm}\!\RD{1}{\!\Wma{\fX(\dinp)}{\qma{\fX(\dinp),\Wm}{\epsilon}}\!(\dinp)}{\qma{\fX(\dinp),\Wm}{\epsilon}\!}
	&\!\leq\!\rens(\rnf)+\gamma+\tfrac{\epsilon(\rens(\sfrac{1}{2})+\gamma)}{(1-\epsilon)\rnf(1-\rnt)^{2}}
	\\
	\label{eq:lem:tradeoff-gamma-B-2}
	\!\RD{1}{\!\Wma{\fX(\dinp)}{\qma{\fX(\dinp),\Wm}{\epsilon}}\!(\dinp)}{\!\Wm(\dinp)}
	&\!\leq\!\gX(\rnf)+\tfrac{\gamma}{\rnf}\!+\!\tfrac{\epsilon(\rens(\sfrac{1}{2})+\gamma)}{(1-\epsilon)\rnf^{2}(1-\rnt)}\!
	\end{align}
	Note that 
	\eqref{eq:orderoneovertwo-gamma}, \eqref{eq:lem:tradeoff-gamma-B-1},
	\(\tfrac{1}{1-\epsilon}\leq 2\),  and \(\rnf<\rnt<1\) imply
	\begin{align}
	\label{eq:tradeoff-C-4-gamma}
	\RD{1}{\!\Wma{\fX(\dinp)}{\qma{\fX(\dinp),\Wm}{\epsilon}}\!(\dinp)}{\qma{\fX(\dinp),\Wm}{\epsilon}\!}
	&\leq \tfrac{2(\rens(\sfrac{1}{2})+\gamma)}{\rnf (1-\rnt)^{2}}
	&
	&\forall \dinp\in\inpS.
	\end{align}
	Following an analysis analogous to the one deriving  
	\eqref{eq:lem:tradeoff-C} from 
	\eqref{eq:lem:tradeoff-B-1}
	and invoking \eqref{eq:tradeoff-C-4-gamma}
	instead of \eqref{eq:tradeoff-C-4} 
	we conclude that 
	\eqref{eq:lem:tradeoff-gamma-B-1},
	\eqref{eq:lem:tradeoff-gamma-B-2},
	and
	\(\tfrac{1}{1-\epsilon}\!\leq\!2\)
	imply 
	\eqref{eq:lem:tradeoff-gamma-B}
	and
	\eqref{eq:lem:tradeoff-gamma-C}
	for \(\Vm(\dinp)\!=\!\Wma{\fX(\dinp)}{\qma{\fX(\dinp),\Wm}{\epsilon}}\!(\dinp)\).
	Thus we are left with establishing \eqref{eq:lem:tradeoff-gamma-B-1}
	and
	\eqref{eq:lem:tradeoff-gamma-B-2}.
	
	We denote  \(\Wma{\rno}{\qma{\rno,\Wm}{\epsilon}}(\dinp)\) by \(\vmn{\rno}\). Note that \(\vmn{\rno}\) satisfies
	\begin{align}
\notag
	\RD{1}{\vmn{\rno}}{\qma{\rno,\Wm}{\epsilon}}\!+\!\tfrac{\rno}{1-\rno} \RD{1}{\vmn{\rno}}{\!\Wm(\dinp)}
	\!=\!\RD{\rno}{\!\Wm(\dinp)}{\qma{\rno,\Wm}{\epsilon}}
	\end{align}
	for all \(\rno\!\in\!(0,1)\). 
	Then \eqref{eq:lem:avcapacity} of Lemma \ref{lem:avcapacity},
	\eqref{eq:avcapacity:bound}, \eqref{eq:lem:tradeoff-gamma-hypothesis},
	and \eqref{eq:orderoneovertwo-gamma} imply that
	\begin{align}
	\label{eq:tradeoff-gamma-B-1}
	\RD{1}{\vmn{\rno}}{\qma{\rno,\Wm}{\epsilon}}
	\!+\!\tfrac{\rno}{1-\rno} \RD{1}{\vmn{\rno}}{\!\Wm(\dinp)}
	&\!\leq\!\rens(\rno)\!+\!\gamma\!+\!\tfrac{\epsilon(\rens(\sfrac{1}{2})+\gamma)}{(1-\epsilon)\rnf(1-\rnt)^{2}} 
	\end{align}
	for all \(\rno\!\in\![\rnf,\rnt]\).
	Then using the non-negativity of the \renyi divergence and \(\gX(\rnf)=\tfrac{1-\rnt}{\rnt}\rens(\rnt)\) 
	we get
	\begin{align}
	\label{eq:tradeoff-gamma-B-2}
	\RD{1}{\vmn{\rnf}}{\qma{\rnf,\Wm}{\epsilon}}
	&\!\leq\!\rens(\rnf)\!+\!\gamma\!+\!\tfrac{\epsilon(\rens(\sfrac{1}{2})+\gamma)}{(1-\epsilon)\rnf(1-\rnt)^{2}},
	&
	&
	\\
	\label{eq:tradeoff-gamma-B-3}
	\RD{1}{\vmn{\rnt}}{\!\Wm(\dinp)}
	&\!\leq\!\gX(\rnf)\!+\!\tfrac{\gamma}{\rnf}\!+\!\tfrac{\epsilon(\rens(\sfrac{1}{2})+\gamma)}{(1-\epsilon)\rnf^{2}(1-\rnt)}.
	&
	&
	\end{align}
	As a result of \eqref{eq:tradeoff-gamma-B-2}, \(\RD{1}{\vmn{\rnf}}{\qma{\rnf,\Wm}{\epsilon}}\) and  
	\(\RD{1}{\vmn{\rnt}}{\qma{\rnt,\Wm}{\epsilon}}\) satisfy one of the following three cases:
	\begin{enumerate}[(i)]
		\item If \(\RD{1}{\vmn{\rnf}}{\qma{\rnf,\Wm}{\epsilon}}\!=\!\rens(\rnf)\!+\!\gamma\!+\!\tfrac{\epsilon(\rens(\sfrac{1}{2})+\gamma)}{(1-\epsilon)\rnf(1-\rnt)^{2}}\), 
		then \eqref{eq:tradeoff-gamma-B-1} implies \(\RD{1}{\vmn{\rnf}}{\Wm(\dinp)}\!=\!0\).
		Thus \eqref{eq:lem:tradeoff-gamma-B-1}  and \eqref{eq:lem:tradeoff-gamma-B-2} 
		hold for \(\fX(\dinp)\!=\!\rnf\).

		\item If \(\RD{1}{\vmn{\rnt}}{\qma{\rnt,\Wm}{\epsilon}}\!\leq\!
		\rens(\rnf)\!+\!\gamma\!+\!\tfrac{\epsilon(\rens(\sfrac{1}{2})+\gamma)}{(1-\epsilon)\rnf(1-\rnt)^{2}}\), 
		then \eqref{eq:lem:tradeoff-gamma-B-1}  and \eqref{eq:lem:tradeoff-gamma-B-2} 
		hold for \(\fX(\dinp)\!=\!\rnt\) by \eqref{eq:tradeoff-gamma-B-3}.
		
		\item If \(\!\RD{1}{\vmn{\rnf}}{\!\qma{\rnf,\Wm}{\epsilon}\!}\!<\!
		\rens(\rnf)\!+\!\gamma\!+\!\tfrac{2\epsilon(\rens(\sfrac{1}{2})+\gamma)}{\rnf(1-\rnt)^{2}}
		\!<\!\RD{1}{\vmn{\rnt}}{\!\qma{\rnt,\Wm}{\epsilon}\!}\), 
		then 
		\(\RD{1}{\vmn{\rno}}{\qma{\rno,\Wm}{\epsilon}}\!=\!
		\rens(\rnf)\!+\!\gamma\!+\!\tfrac{\epsilon(\rens(\sfrac{1}{2})+\gamma)}{(1-\epsilon)\rnf(1-\rnt)^{2}}\)
		for some \(\rno\!\in\!(\rnf,\rnt)\) by the intermediate value theorem \cite[4.23]{rudin}
		because \(\RD{1}{\vmn{\rno}}{\qma{\rno,\Wm}{\epsilon}}\) is continuous in \(\rno\)
		by Lemma \ref{lem:tilting}-(\ref{tilting-divergence}).
		On the other hand
		\(\RD{1}{\vmn{\rno}}{\Wm(\dinp)}\!\leq\!\gX(\rnf)\) 
		for the \(\rno\) satisfying
		\(\RD{1}{\vmn{\rno}}{\qma{\rno,\Wm}{\epsilon}}\!=\!
		\rens(\rnf)\!+\!\gamma\!+\!\tfrac{\epsilon(\rens(\sfrac{1}{2})+\gamma)}{(1-\epsilon)\rnf(1-\rnt)^{2}}\)
		by \eqref{eq:tradeoff-gamma-B-1}.
		Thus \eqref{eq:lem:tradeoff-gamma-B-1} and \eqref{eq:lem:tradeoff-gamma-B-2}
		holds for \(\fX(\dinp)=\rno\).
	\end{enumerate}
\end{proof} 
\begin{lemma}\label{lem:spb-Fproduct-gamma} 
	Let \(\Wmn{\vec{[1,\blx]}}\) be a DPC with feedback satisfying
	\(\RC{0^{_{+}}\!}{\Wmn{[1,\blx]}}\!\neq\!\RC{1}{\Wmn{[1,\blx]}}\),
	orders \(\rno_{0}\!<\!\rno_{1}\!<\!\dsta\) in \((0,1)\) satisfy
	\(\!\tfrac{1-\dsta}{\dsta}\!\RC{\rnt}{\!\Wmn{[1,\blx]}}\!=\!\spe{\!\RC{\rno_{1}}{\!\Wmn{[1,\blx]}},\!\Wmn{[1,\blx]}\!}\),
	\(\knd\in\integers{+} \) be less than \(\blx\).
	If \(\{\Wmn{\tin}\}\) satisfy both 
	\(\lfloor\tfrac{\blx}{\knd}\rfloor\tfrac{\RC{\sfrac{1}{2}}{\Wmn{[1,\blx]}}}{\blx}+\gamma\geq 2\)
	and 
	\begin{align}
	\label{eq:lem:spb-Fproduct-gamma-hypothesis-new}
	\!\!\!\!
	\sup_{{\rno\in[\rno_{0},\dsta]}}\max_{\tin\in[1,\blx-\ell+1]}
	\!\left(\RC{\rno}{\Wmn{[\tin,\tin+\ell-1]}}-\tfrac{\ell}{\blx}\RC{\rno}{\Wmn{[1,\blx]}}\right)
	&\!\leq\!\gamma\!
	&
	&
	\end{align}
	for all \(\ell=\{\lfloor\tfrac{\blx}{\knd}\rfloor,\lceil\tfrac{\blx}{\knd}\rceil\}\)
	for the same \(\gamma\in\reals{+}\)
	and positive integers \(M\), \(L\) satisfy
	\begin{align}
	\label{eq:lem:spb-Fproduct-gamma-hypothesis-old}
	\RC{\rno_{1}}{\!\Wmn{[1,\blx]}}\!\geq\!\ln\!\tfrac{M}{L}
	&\!\geq\! 
	\RC{\rno_{0}}{\!\Wmn{[1,\blx]}}\!+\!\tfrac{\RC{\sfrac{1}{2}}{\!\Wmn{[1,\blx]}}}{\rno_{0}(1-\dsta)^{2}}\!
	\!\left[\!2\epsilon\!+\!\tfrac{14}{\sqrt[3]{\knd}}\!\right]\!
	\!+\!\knd(\gamma\!-\!\ln\epsilon),
	\end{align}
	for some \(\epsilon\!\in\!(0,\sfrac{\rno_{0}}{2})\),	
	then  any \((M,L)\) code on  \(\Wmn{\vec{[1,\blx]}}\) satisfies
	\begin{align}
	\label{eq:lem:spb-Fproduct-gamma}
	\Pem{av}
	&\!\geq\!\tfrac{1}{4}
	e^{-\spe{\ln\frac{M}{L},\Wmn{[1,\blx]}}
		-\frac{\RC{\sfrac{1}{2}}{\Wmn{[1,\blx]}}+\kappa\gamma}{\rno_{0}^{2}(1-\dsta)^{2}}
		\left[\!6\epsilon+\frac{15}{\sqrt[3]{\knd}}\right]
		-\frac{\knd (3\gamma-\ln\epsilon)}{\rno_{0}}}.
	\end{align}
\end{lemma} 
\begin{proof}[Proof of Lemma \ref{lem:spb-Fproduct-gamma}]
	We divide the interval \([1,\blx]\) into \(\knd\) subintervals of,
	approximately, equal length.
	In particular, 
	we set \(\tin_{0}\) to zero and define \(\ell_{\ind}\) and \(\tin_{\ind}\)
	for \(\ind \in \{1,\ldots,\knd\}\) as follows
	\begin{align}
	\notag
	\ell_{\ind}
	&\DEF \lceil  \sfrac{\blx}{\knd} \rceil  \IND{\ind\leq \blx-\lfloor \sfrac{\blx}{\knd} \rfloor \knd}+
	\lfloor \sfrac{\blx}{\knd} \rfloor \IND{\ind >   \blx-\lfloor \sfrac{\blx}{\knd} \rfloor \knd},
	\\
	\notag
	\tin_{\ind}
	&\DEF \tin_{\ind-1}+\ell_{\ind}.
	\end{align}
	If \(\Umn{\ind}\DEF \Wmn{\vec{[1+\tin_{\ind-1},\tin_{\ind}]}}\)
	for  all \(\ind\in\{1,\ldots,\knd\}\),
	then the length \(\blx\) DPC with feedback \(\Wmn{\vec{[1,\blx]}}\) 
	can be interpreted as a length \(\knd\) 
	DPC  with feedback \(\Umn{\vec{[1,\knd]}}\). 
	As a result any \((M,L)\) channel code \((\enc,\dec)\) on the channel 
	\(\Wmn{\vec{[1,\blx]}}\) is also an \((M,L)\) channel code on 
	\(\Umn{\vec{[1,\knd]}}\) with exactly the same error probability. 
	In the rest of the proof we work with the latter interpretation.
	
	Let \(\rnf\in(0,1)\) be the order satisfying
	\(\RC{\rnf}{\Wmn{[1,\blx]}}=\rate\)
	where
	\begin{align}
	\label{eq:spb-Fproduct-gamma-01}
	\rate
	&\!=\!\tfrac{M}{L}
	\!-\!\tfrac{\RC{\sfrac{1}{2}}{\Wmn{[1,\blx]}}+\knd\gamma}{\rno_{0}(1-\dsta)^{2}}\!\!\left[\!2\epsilon\!+\!\tfrac{14}{\sqrt[3]{\knd}}\!\right]\!-\!\knd(\gamma\!-\!\ln\epsilon).
	\end{align}
	Then \(\rnf\!\in\![\rno_{0},\rno_{1}]\) by \eqref{eq:lem:spb-Fproduct-gamma-hypothesis-old}
	and the monotonicity of \(\RC{\rno}{\Wm}\) in \(\rno\), i.e. Lemma \ref{lem:capacityO}-(\ref{capacityO-ilsc}).
	
	Let \(\rnt\in(0,1)\) be the order satisfying
	\begin{align}
	\notag
	\tfrac{1-\rnt}{\rnt}\RC{\rnt}{\Wmn{[1,\blx]}}
	&\!=\!\spe{\rate,\Wmn{[1,\blx]}}.
	\end{align}
	Then the definition of \(\dsta\),
	the monotonicity of \(\tfrac{1-\rno}{\rno}\RC{\rno}{\Wm}\) in \(\rno\), i.e. Lemma \ref{lem:capacityO}-(\ref{capacityO-zo}),
	and
	the monotonicity of \(\spe{\rate,\Wm}\) in \(\rate\), i.e. Lemma \ref{lem:spherepackingexponent},
	imply \(\rnt\in[\rno_{0},\dsta]\).
	
	For all \(\ind\in\{1,\ldots,\knd\}\) and \(\rno\in(0,1)\), 
	Lemma \ref{lem:capacityFproduct} implies
	\begin{align}
	\label{eq:spb-Fproduct-gamma-02}
	\RC{\rno}{\Umn{\ind}}
	&\!=\!\RC{\rno}{\Wmn{[1+\tin_{\ind-1},\tin_{\ind}]}}.
	\end{align}
	Then for each \(\ind\!\in\!\{1,\ldots,\knd\}\),	
	\(\Wm\!=\!\Umn{\ind}\) and \(\rens(\rno)\!=\!\tfrac{\ell_{\ind}}{\blx}\RC{\rno}{\Wmn{[1,\blx]}}\)
	satisfy the hypothesis of Lemma \ref{lem:tradeoff-gamma} as a result of 
	\eqref{eq:lem:spb-Fproduct-gamma-hypothesis-new}.
	We denote the channel resulting  from applying Lemma \ref{lem:tradeoff-gamma} to 
	\(\Umn{\ind}\) by \(\Vmn{\ind}:\set{A}_{\ind}\to\pdis{\set{B}_{\ind}}\). 
	Then	
	\begin{align}
	\label{eq:spb-Fproduct-gamma-03}
	\hspace{-.1cm}
	\!\RD{1}{\!\Vmn{\ind}\!(\mA)}{\!\Umn{\ind}\!(\mA)}
	&\!\leq\!\tfrac{\ell_{\ind}}{\blx}\spe{\rate,\Wmn{[1,\blx]}}\!+\!\tfrac{\gamma}{\rno_{0}}
	\!+\!\tfrac{2\epsilon(\frac{\ell_{\ind}}{\blx}\RC{\sfrac{1}{2}}{\Wmn{[1,\blx]}}+\gamma)}{\rno_{0}^{2}(1-\dsta)}\!
	\end{align}
	for all \(\mA\!\in\!\set{A}_{\ind}\) and \(\ind\in\{1,\ldots,\knd\}\).
	
	Lemma \ref{lem:tradeoff-gamma} also implies
	\begin{align}
	\notag
	\!\!\RC{\rnb}{\Vmn{\ind}}
	&\!\leq\!\tfrac{\ell_{\ind}}{\blx}\rate\!+\!\gamma
	\!+\!\tfrac{2\epsilon(\frac{\ell_{\ind}}{\blx}\RC{\sfrac{1}{2}}{\Wmn{[1,\blx]}}+\gamma)}{\rno_{0}(1-\dsta)^{2}}\!
	\!+\!\ln \tfrac{1}{\epsilon}
	\\
	\notag
	&\qquad\!+\!
	(\rnb\!-\!1)e^{(\rnb-1)\frac{2(\frac{\ell_{\ind}}{\blx}\RC{\sfrac{1}{2}}{\Wmn{[1,\blx]}}+\gamma)}{\rno_{0}(1-\dsta)^{2}}}
	\!\left[\!
	\tfrac{4\vee 2(\frac{\ell_{\ind}}{\blx}\RC{\sfrac{1}{2}}{\Wmn{[1,\blx]}}+\gamma)}{\rno_{0}(1-\dsta)^{2}}\!\right]^{2}
	\end{align}
	for all \(\rnb\!\in\!(1,\tfrac{1+\dsta}{2\dsta})\). 
	Furthermore,
\(1\!+\!\tfrac{\knd^{\sfrac{2}{3}}\rnf(1-\dsta)^{2}}{4(\RC{\sfrac{1}{2}}{\Wmn{[1,\blx]}}+\knd\gamma)}\!\leq\!\tfrac{1+\dsta}{2\dsta}\)
because
	\(\tfrac{\ell_{\ind}}{\blx}\!\RC{\sfrac{1}{2}}{\Wmn{[1,\blx]}}\!+\!\gamma\!\geq\!2\) 
	and \(\knd\!\geq\!1\). 	
	Hence,
	\begin{align}
	\notag
	\RC{\rnb}{\Vmn{\ind}}
	&\!\leq\!\tfrac{\ell_{\ind}}{\blx}\rate\!+\!\gamma	\!-\!\ln \epsilon
	\!+\!\tfrac{\frac{\ell_{\ind}}{\blx}\RC{\sfrac{1}{2}}{\Wmn{[1,\blx]}}+\gamma}{\rno_{0}(1-\dsta)^{2}}
	\left[2\epsilon+2\tfrac{e^{\sfrac{1}{\sqrt[3]{\knd}}}}{\sqrt[3]{\knd}}\right]\!
	\\
	\notag
	&\!\leq\!\tfrac{\ell_{\ind}}{\blx}\rate\!+\!\gamma	\!-\!\ln \epsilon
	\!+\!\tfrac{\frac{\ell_{\ind}}{\blx}\RC{\sfrac{1}{2}}{\Wmn{[1,\blx]}}+\gamma}{\rno_{0}(1-\dsta)^{2}}
	\left[2\epsilon+\tfrac{6}{\sqrt[3]{\knd}}\right]\!
	\end{align}
	for \(\rnb=1\!+\!\tfrac{\knd^{\sfrac{2}{3}}\rnf(1-\dsta)^{2}}{4(\RC{\sfrac{1}{2}}{\Wmn{[1,\blx]}}+\knd\gamma)}\).
	Then using Lemma \ref{lem:capacityFproduct},  \eqref{eq:spb-Fproduct-gamma-01},  
	and  \eqref{eq:spb-Fproduct-gamma-02} we get
	\begin{align}
	\notag
	\RC{\rnb}{\Vmn{\vec{[1,\blx]}}}
	&\!\leq\!\ln\!\tfrac{M}{L}\!-\!\tfrac{\RC{\sfrac{1}{2}}{\Wmn{[1,\blx]}}+\kappa\gamma}{\rno_{0}(1-\dsta)^{2}}
	\tfrac{8}{\sqrt[3]{\knd}}
	\end{align}
	for \(\rnb=1\!+\!\tfrac{\knd^{\sfrac{2}{3}}\rnf(1-\dsta)^{2}}{4(\RC{\sfrac{1}{2}}{\Wmn{[1,\blx]}}+\knd\gamma)}\).
	We bound the average error probability of \(\!(\enc,\dec)\!\) 
	on \(\!\Vmn{\vec{[1,\blx]}}\!\)
	using \eqref{eq:lem:augustin} and \(\tau\geq\ln(1+\tau)\):
	\begin{align}
	\notag
	\Pem{\Vmn{\vec{[1,\blx]}}}
	&\geq 1-e^{-\sqrt[3]{\knd}}
	\\
	\label{eq:spb-Fproduct-gamma-04}
	&\geq \tfrac{\sqrt[3]{\knd}}{1+\sqrt[3]{\knd}}.
	\end{align}
	
	On the other hand  \eqref{eq:spb-Fproduct-gamma-03} imply
	\begin{align}
	\notag
	\hspace{-.2cm}
	\RD{1}{\!\Vmn{\vec{[1,\knd]}}(\dinp)}{\!\Umn{\vec{[1,\knd]}}(\dinp)}
	&\!\leq\!\spe{\rate,\Wmn{[1,\blx]}}\!+\!\tfrac{\knd\gamma}{\rno_{0}}
	\\
	\label{eq:spb-Fproduct-gamma-05}
	&\qquad\!+\!\tfrac{2\epsilon(\RC{\sfrac{1}{2}}{\Wmn{[1,\blx]}}+\kappa\gamma)}{\rno_{0}^{2}(1-\dsta)}\!
	\end{align}
	for all \(\dinp\!\in\!\left(\bigtimes\nolimits_{\ind=1}^{\knd}{\set{A}_{\ind}}^{\set{B}_{1}^{\ind-1}}\right)\).
	
	Let \(\mP\) be the probability distribution generated by the encoder \(\enc\) on the input set of \(\Umn{\vec{[1,\knd]}}\),
	i.e. on \(\!\left(\!\bigtimes\nolimits_{\ind=1}^{\knd}{\set{A}_{\ind}}^{\set{B}_{1}^{\ind-1}}\!\right)\!\), 
	for the uniform distribution over the message set. 
	Then Lemma \ref{lem:divergence-DPI}
	and the identity \(\tau \ln \tau\!+\!(1-\tau)\ln(1-\tau)\!\geq\!\ln\sfrac{1}{2}\), 
	which holds for all \(\tau\in[0,1]\), imply
	\begin{align}
	\label{eq:spb-Fproduct-gamma-06}
	\RD{1}{\mP \mtimes  \Vmn{\vec{[1,\knd]}}}{\mP \mtimes  \Umn{\vec{[1,\knd]}}}
	&\!\geq\!\ln\sfrac{1}{2}-\Pem{\Vmn{\vec{[1,\blx]}}} \ln \Pem{av}.
	\end{align}
	
	Note that \(\RD{1}{\!\mP\mtimes\!\Vmn{\vec{[1,\knd]}}}{\!\mP\mtimes\!\Umn{\vec{[1,\knd]}}}\)
	is bounded from above by  the supremum of 
	\(\RD{1}{\!\Vmn{\vec{[1,\knd]}}(\dinp)}{\!\Umn{\vec{[1,\knd]}}(\dinp)}\)
	over the common input set of \(\Vmn{\vec{[1,\knd]}}\) and \(\Umn{\vec{[1,\knd]}}\), i.e. \(\bigtimes\nolimits_{\ind=1}^{\knd}{\set{A}_{\ind}}^{\set{B}_{1}^{\ind-1}}\).
	Then using 
	\eqref{eq:spb-Fproduct-gamma-04}, 
	\eqref{eq:spb-Fproduct-gamma-05}, 
	and  
	\eqref{eq:spb-Fproduct-gamma-06}
	we get,
	\begin{align}
	\notag
	\Pem{av}
	&\!\geq\!
	e^{-\frac{1+\sqrt[3]{\knd}}{\sqrt[3]{\knd}}\left(\spe{\rate,\Wmn{[1,\blx]}}+\frac{\knd\gamma}{\rno_{0}}
		+\frac{2\epsilon(\RC{\sfrac{1}{2}}{\Wmn{[1,\blx]}}+\kappa\gamma)}{\rno_{0}^{2}(1-\dsta)}\!
		+\ln 2\right)}.
	\end{align}
	Then using the identity \(\!\spe{\RC{\rnf}{\Wm},\Wm}\!\leq\!\tfrac{(1\!-\!\rnf)\RC{\rnf}{\Wm}}{\rnf}\),
	which is implied by Lemma \ref{lem:spherepackingexponent},
	together with \eqref{eq:orderoneovertwo}, \eqref{eq:lem:spb-Fproduct-gamma-hypothesis-old}, 
	and \eqref{eq:spb-Fproduct-gamma-01} we get
		\begin{align}
	\notag
	\Pem{av}
	&\!\geq\!\tfrac{1}{4}
	e^{-\spe{\rate,\Wmn{[1,\blx]}}-\frac{2\knd\gamma}{\rno_{0}}
		-\frac{\RC{\sfrac{1}{2}}{\Wmn{[1,\blx]}}+\kappa\gamma}{\rno_{0}^{2}(1-\dsta)}[4\epsilon+\frac{1}{\sqrt[3]{\knd}}]}.
	\end{align}
	On the other hand,
	\(\spe{\rate,\Wm}
	\leq \spe{\rate+\delta,\Wm}+\tfrac{(1-\rno_{0})}{\rno_{0}}\delta\) 
	for all \(\delta\geq0\)
	provided that \(\rate\!=\!\RC{\rno}{\Wm}\)  for an \(\rno\!\in\![\rno_{0},\rno_{1}]\)
	as a result of Lemma \ref{lem:spherepackingexponent}.
	Then \eqref{eq:lem:spb-Fproduct-gamma} follows from \eqref{eq:spb-Fproduct-gamma-01}.
\end{proof}

\begin{proof}[Proof of Theorem \ref{thm:fDPCexponent}]
	Note that \(\rno_{0}\), \(\rno_{1}\), \(\epsilon\),  \(\knd\), and \(\gamma\)
	are free parameters in Lemma \ref{lem:spb-Fproduct-gamma}.
	Thus we can choose a different value for each \(\blx\).
	In particular, if 
	\((\rno_{0})_{\blx}=\rno_{0}\),
	\((\rno_{1})_{\blx}=\rno_{1}\),
	\(\epsilon_{\blx}=\tfrac{1}{\blx}\), 
	\(\knd_{\blx}=\lfloor \blx^{\sfrac{3}{4}} \rfloor\),
	and 
	\(\gamma_{\blx}=2K \ln \blx\),
	then Assumption \ref{assumption:astationary} implies 
	the existence of \(\delta>0\) satisfying
	\(\dsta_{\blx}<1-\delta\) for large enough \(\blx\).
	Furthermore, the hypotheses of Lemma \ref{lem:spb-Fproduct-gamma}
	are satisfied for the above defined values of the parameters 
	for large enough \(\blx\) as a result of Assumption \ref{assumption:astationary}.
	Then  Theorem \ref{thm:fDPCexponent} follows from 
	Lemma \ref{lem:spb-Fproduct-gamma}
	and \(\dsta_{\blx}<1-\delta\). 
\end{proof}

\section{Operational Significance of \(\RC{\rno}{\Wm}\)}\label{sec:operational} 
The information transmission problems did not play any role in the definition 
the \renyi capacity, presented in \S\ref{sec:preliminary}; neither did they play 
any role in the analysis of it in \cite{nakiboglu19A}.
In information theorists' parlance: the \renyi capacity does not 
have any operational significance because of its definition. 
Channel coding theorems and their converses establish the operational significance of the \renyi capacity. 
In this section we review two well known results that quantify this operational significance. 

For the DSPCs, the order one \renyi capacity is equal to the channel capacity \cite{shannon48}, 
i.e. \(\RC{1}{\Wm}\) is the threshold for the rates below which reliable 
communication is possible and above which reliable communication is impossible
for any DSPC.
The \renyi capacities of other orders bound the performance of the channel codes through the 
sphere packing exponent defined in \S\ref{sec:SPexponent} 
and the strong converse exponent that will be defined in \S\ref{sec:arimoto}:
via Gallager's bound \cite[Thm. 1]{gallager65} for rates less than \(\RC{1}{\Wm}\) and
via Arimoto's bound \cite[Thm. 1]{arimoto73} for rates greater than \(\RC{1}{\Wm}\).
We derive Gallager's bound  in \S\ref{sec:gallager},  Arimoto's bound  in \S\ref{sec:arimoto}.
Then in \S\ref{sec:Ccapacity} we define the channel capacity formally and demonstrate 
that Gallager's bound and Arimoto's bound determine 
the channel capacity for a class of channels much broader than the DSPCs.

\subsection{Gallager's Bound}\label{sec:gallager} 
For rates less than \(\RC{1}{\Wm}\), the sphere packing exponent confines 
the optimal performance for the channel coding problem 
through  Gallager's bound, \cite[Thm. 1]{gallager65}, \cite[p. 538]{gallager}.
Let us start with deriving Gallager's  bound for 
the list decoding.
\begin{lemma}\label{lem:gallager}
For any \(\Wm\!:\!\inpS\!\to\!\pmea{\outA}\), 
 \(\mP\!\in\!\pdis{\inpS}\), message set size \(M\!\in\!\integers{+}\), 
list size \(L\!\in\!\{\!1,\ldots,M\!-\!1\!\}\), and 
\(\rno\!\in\![\tfrac{1}{1+L},1)\), there exists an \((M,L)\) channel code satisfying
\begin{align}
\label{eq:lem:gallager}
\ln \Pem{av}&\leq \tfrac{\rno-1}{\rno} \left[\RMI{\rno}{\mP}{\Wm}-\tfrac{1}{L}\ln \tbinom{M-1}{L}\right].
\end{align}
Furthermore, for any  
\(M,L\) satisfying \(\ln\!\tfrac{e M}{L}\!\in\![\RC{\!\frac{1}{1+L}\!}{\Wm\!},\!\RC{1}{\Wm\!})\)
and \(\epsilon>0\), there exists an \((M,L)\) 
channel code  such that
\begin{align}
\label{eq:lem:gallager-sp}
\Pem{av}&\leq (1+\epsilon)e^{-\spe{\ln\frac{e M}{L},\Wm}}.
\end{align}
If \(\lim_{\rno \uparrow 1} \tfrac{1-\rno}{\rno}\RC{\rno}{\Wm}=0\), 
then \eqref{eq:lem:gallager-sp} holds for \(\epsilon=0\), 
as well.
\end{lemma}
Gallager's bound is proved using a standard random coding argument for a maximum likelihood decoder,
without invoking probabilistic results, 
such as the law of large numbers or the central limit theorem. 
It is possible to strengthen the result by considering channels satisfying additional hypotheses. 
\altug and Wagner \cite{altugW12B}, \cite{altugW14D} have shown that for the  DSPCs when \(L=1\) 
for large enough rates it is possible to replace the \((1+\epsilon)\) term with an \(\bigo{\blx^{-\sfrac{1}{2\rno}}}\) term for certain 
\(\rno\in(\sfrac{1}{2},1)\), where \(\blx\) is the block length.
Later  Scarlett, Martinez, and \fabregas \cite{scarlettMF13} presented an alternative derivation of the result.
First Scarlett, Martinez, and \fabregas \cite{scarlettMF14} and
then Honda \cite{honda15}, derived approximations to random coding union bound which can be used to characterize the 
\(\bigo{\blx^{-\sfrac{1}{2\rno}}}\) term explicitly.

Shannon, Gallager, and Berlekamp \cite[Thm. 2]{shannonGB67A} proved the 
SPB for arbitrary DSPCs. This implies that Gallager's bound is tight for any DSPC
and at any fixed rate less than \(\RC{1}{\Wm}\) and greater than \(\RC{0^{_{+}}\!}{\Wm}\) 
for large enough list size \(L\), in terms of the exponential decay rate of the error 
probability with the block length.
Theorem \ref{thm:productexponent} presented in \S\ref{sec:product-outerbound} establishes
the SPB and hence the tightness of Gallager's bound ---in the above 
discussed sense--- more generally.

\begin{proof}[Proof of Lemma \ref{lem:gallager}] 
Instead of bounding the error probability of a code, we bound the average error probability of an ensemble of codes. 
We assume that the messages are assigned to the members of \(\inpS\), 
independently of one another and according to a \(\mP\) in \(\pdis{\inpS}\). 
We use a maximum likelihood decoder: for each \(\dout\in\outS\), decoded list \(\dec(\dout)\) is composed of 
\(L\) messages  with the greatest \(\der{\Wm(\enc(\dmes))}{\qmn{\rno,\mP}}\) values. 
If there is a tie, decoder chooses the messages with the lower indices.  

In order to bound the expected value of the average error probability of the code over the ensemble, let us consider the expected value 
of  the conditional error probability of the message with the greatest index. An error will occur only when Radon-Nikodym derivatives 
of at least \(L\) other messages are at least as large as the Radon-Nikodym derivative of the transmitted message. We bound this probability using 
an auxiliary threshold \(\gamma\): 
\begin{align}
\notag
\hspace{-.1cm}\EX{\Pem{av}} 
&\!\leq\!
\sum\nolimits_{\dinp}\!\mP(\!\dinp\!)
\EXS{\qmn{\rno,\mP}\!}{\IND{\fX_{\dinp}\leq \gamma}\fX_{\dinp}}
\\
\notag
&~~~+\!\tbinom{M-1}{L}\!\sum\nolimits_{\dinp}\!\mP(\!\dinp\!)
\EXS{\qmn{\rno,\mP}\!}{\!
\left[\sum\nolimits_{\dsta}\!\mP(\!\dsta\!)\IND{\fX_{\dsta}\geq\fX_{\dinp}}\!\right]^{L}
\!\!\IND{\fX_{\dinp}>\gamma}\fX_{\dinp}\!}
\end{align}
where \(\fX_{\dinp}=\der{\Wm(\dinp)}{\qmn{\rno,\mP}}\) for all \(\dinp\in\inpS\).

We bound the fist and the second terms separately: 
\begin{align}
\notag
\EXS{\qmn{\rno,\mP}}{\IND{\fX_{\dinp}\leq \gamma}\fX_{\dinp}}
&\leq \gamma^{1-\rno} \EXS{\qmn{\rno,\mP}}{\IND{\fX_{\dinp}\leq \gamma}(\fX_{\dinp})^{\rno}}.
\end{align}
In order to bound the second term we use the identity
\(\sum\nolimits_{\dsta} \mP(\dsta)(\fX_{\dsta})^{\rno}=e^{(\rno-1)\RMI{\rno}{\mP}{\Wm}}\),
which is a consequence of \eqref{eq:def:information} and \eqref{eq:def:mean},
and the fact that \(\rno\geq 1- L\rno\), which follows from
\(\tfrac{1}{1+L}\leq \rno\).
\begin{align}
\notag
&\EXS{\qmn{\rno,\mP}}{
	\left[\sum\nolimits_{\dsta} \mP(\dsta)\IND{\fX_{\dsta}\geq\fX_{\dinp}} \right]^{L}
	\IND{\fX_{\dinp}>\gamma}\fX_{\dinp}}
\\
\notag
&\hspace{1.5cm}\leq
\EXS{\qmn{\rno,\mP}}{
	\left[\sum\nolimits_{\dsta} \mP(\dsta)  (\tfrac{\fX_{\dsta}}{\fX_{\dinp}})^{\rno}\right]^{L}
	\IND{\fX_{\dinp}>\gamma}\fX_{\dinp}}
\\
\notag
&\hspace{1.5cm}
=e^{(\rno-1)\RMI{\rno}{\mP}{\Wm}L}
\EXS{\qmn{\rno,\mP}}{\IND{\fX_{\dinp}>\gamma}(\fX_{\dinp})^{1-\rno L}}
\\
\notag
&\hspace{1.5cm}
\leq e^{(\rno-1)\RMI{\rno}{\mP}{\Wm}L} 
\gamma^{1-\rno-\rno L}
\EXS{\qmn{\rno,\mP}}{\IND{\fX_{\dinp}>\gamma}(\fX_{\dinp})^{\rno}}
\end{align}

If we set \(\gamma=\left[\tbinom{M-1}{L}\right]^{\frac{1}{L\rno}} e^{\frac{\rno-1}{\rno}\RMI{\rno}{\mP}{\Wm}}\) 
and apply the identity 
\(\sum\nolimits_{\dinp} \mP(\dinp)(\fX_{\dinp})^{\rno}=e^{(\rno-1)\RMI{\rno}{\mP}{\Wm}}\)
again we get 
\begin{align}
\notag
\EX{\Pem{av}}
&\leq 
\gamma^{1-\rno} \sum\nolimits_{\dinp} \mP(\dinp) \EXS{\qmn{\rno,\mP}}{(\fX_{\dinp})^{\rno}}
\\
\notag
&=\gamma^{1-\rno}e^{(\rno-1)\RMI{\rno}{\mP}{\Wm}}
\\
\notag
&=\left[\tbinom{M-1}{L}\right]^{\frac{1-\rno}{L\rno}} e^{\frac{\rno-1}{\rno}\RMI{\rno}{\mP}{\Wm}}.
\end{align}
Since there exists a code with \(\Pem{av}\) less than or equal to \(\EX{\Pem{av}}\), there exists a code
satisfying \eqref{eq:lem:gallager}.

Using the Stirling's approximation for factorials,   
i.e. \(\sqrt{2\pi \blx} (\sfrac{\blx}{e})^{\blx}\leq  \blx!  \leq  e \sqrt{\blx} (\sfrac{\blx}{e})^{\blx}\),
and the identity \(\ln \tau \leq \tau-1\) we get
\begin{align}
\notag
\tfrac{1}{L}\ln \tbinom{M-1}{L}
&\leq \tfrac{1}{L}\ln \tfrac{e\sqrt{M-1}}{2\pi \sqrt{L(M-1-L)}}  +\ln \tfrac{M-1}{L}
\\
\notag
&\qquad~\qquad~\qquad
+\tfrac{M-1-L}{L}  \ln \left(1+\tfrac{L}{M - 1-L}\right)   
\\
\label{eq:stirlingbound}
&\leq \ln \tfrac{M-1}{L}+1.
\end{align}
For any \(\rno\in [\tfrac{1}{L+1},1]\),
\eqref{eq:lem:gallager} and \eqref{eq:stirlingbound}
implies
\begin{align}
\notag
\ln \Pem{av}
&\leq \tfrac{\rno-1}{\rno}
[\RMI{\rno}{\mP}{\Wm}+\ln\tfrac{M}{M-1}-\ln \tfrac{M e}{L}].
\end{align}
In order to obtain \eqref{eq:lem:gallager-sp}, first note that
there exists an 
\(\rno\) in \([\tfrac{1}{1+L},1)\) such that 
\(\tfrac{1-\rno}{\rno}(\RC{\rno}{\Wm}-\rate)>\spe{\rate,\Wm}-\ln(1+\epsilon)\)
for any 
\(\rate\!\in\![\RC{\frac{1}{1+L}}{\Wm},\RC{1}{\Wm})\) and \(\epsilon\!>\!0\)
by  Lemma \ref{lem:spherepackingexponent}.
On the other hand,  for any \(M\!\geq\!2\)  and  \(\rno\) there exists a
\(\mP\in\pdis{\inpS}\) such that 
\(\RMI{\rno}{\mP}{\Wm}\!+\!\ln\tfrac{M}{M-1}\!\geq\!\RC{\rno}{\Wm}\)
by Definition \ref{def:capacity}. 
Thus, for any \(\epsilon\!>\!0\) and \(L,M\!\in\!\integers{+}\) such that 
\(\ln \tfrac{e M}{L}\!\in\![\RC{\frac{1}{1+L}}{\Wm},\RC{1}{\Wm})\) there exists 
an \(\rno\!\in\![\frac{1}{1+L},1)\) and a  \(\mP\!\in\!\pdis{\inpS}\) such that
\begin{align}
\notag
\spe{\tfrac{M e}{L},\!\Wm}\!-\!\ln (1+\epsilon)
&\!\leq\!\tfrac{1-\rno}{\rno}\!
[\RMI{\rno}{\mP}{\Wm}\!+\!\ln\tfrac{M}{M-1}\!-\!\ln\tfrac{M e}{L}].
\end{align}
If \(\lim_{\rno\uparrow 1}\!\frac{1-\rno}{\rno}\!\RC{\rno}{\Wm}\!=\!0\)
and \(\rate\!=\!\RC{\rnf}{\Wm}\) for a \(\rnf\!\in\!(0,1)\)
then \(\spe{\rate,\Wm}\!=\!\sup_{\rno\in [\rnf,1]}\!\tfrac{1-\rno}{\rno}\!\left(\RC{\rno}{\Wm}\!-\!\rate\right)\)
by Lemma \ref{lem:spherepackingexponent}.
Then as a result of the extreme value theorem \cite[27.4]{munkres},
for any \(\rate\!\in\![\RC{\frac{1}{1+L}}{\Wm},\RC{1}{\Wm})\) there exists an 
\(\rno\!\in\![\tfrac{1}{1+L},1]\) such that 
\(\tfrac{1-\rno}{\rno}(\RC{\rno}{\Wm}\!-\!\rate)\!=\!\spe{\rate\!,\!\Wm\!}\). 
Thus \eqref{eq:lem:gallager-sp} holds for \(\epsilon\!=\!0\).
\end{proof}
\begin{remark}
	We can choose a constant threshold because we work with \(\der{\Wm(\dinp)}{\qmn{\rno,\mP}}\).
	If we had worked with \(\der{\Wm(\dinp)}{\rfm}\) for another measure \(\rfm\), our threshold 
	would have been  \(\gamma \der{\qmn{\rno,\mP}}{\rfm}\), which does not necessarily have the same value 
	for all \(\dout\in\outS\).
\end{remark}
\begin{remark}
	In \cite[Prob. 5.20, p. 538]{gallager}, \(\tfrac{1}{L}\ln \tbinom{M-1}{L}\) 
	is bounded from above by \(\ln (M-1)\).
	For sequences of codes with bounded list size, this bound and the bound in 
	\eqref{eq:stirlingbound} are asymptotically equivalent. 
	However, if \(L\) grows exponentially with the block length,
	then the bound \(\tfrac{1}{L}\ln \tbinom{M-1}{L}\leq\ln (M-1)\) is inferior
	and only bound in \eqref{eq:stirlingbound} leads to an achievability 
	bound that has a matching converse.
\end{remark}

\subsection{Arimoto's Bound}\label{sec:arimoto} 
For rates greater than \(\RC{1}{\Wm}\) the strong converse exponent 
confines the optimal performance for the channel 
coding problem through Arimoto's bound \cite[Thm. 1]{arimoto73}.
Let us start our discussion by recalling the definition of the
strong converse exponent.

\begin{definition}\label{def:strongconverseexponent}
	For any channel \(\!\Wm\!:\!\inpS\!\to\!\pmea{\outA}\)
	satisfying  \(\RC{1}{\Wm}\!<\!\infty\)
	and rate \(\rate\!\in\!\reals{\geq0}\) 
	\emph{the strong converse exponent} is 
	\begin{align}
	%\label{eq:def:strongconverseexponent}
	\notag
	\sce{\rate,\Wm}
	&\DEF \sup\nolimits_{\rno\in [1,\infty)} \tfrac{1-\rno}{\rno} \left(\RC{\rno}{\Wm}-\rate\right).
	\end{align} 
\end{definition}
%The strong 
%converse exponent is denoted by \(C_{sp}(\rate)\) 
%in \cite{omura75} by Omura.
\begin{lemma}\label{lem:strongconverseexponent}
	For any channel \(\!\Wm\!:\!\inpS\!\to\!\pmea{\outA}\)
	satisfying  \(\RC{1}{\Wm}\!<\!\infty\),
	\(\sce{\rate,\Wm}\) is convex, nondecreasing, continuous, and finite 
	function of \(\rate\) on \(\reals{\geq0}\).
	In particular,
	\vspace{-.2cm}
	\begin{align}
	\notag
	&\hspace{-.2cm}\sce{\rate,\Wm}
	\\
	\label{eq:lem:strongconverseexponent}
	&\!=\!\begin{cases}
	~~~0
	&\!\rate\!\leq\!\RC{1}{\Wm}
	\\
	\sup\limits_{\rno\in [1,\rnf]}\!\tfrac{1-\rno}{\rno}\!\left(\RC{\rno}{\Wm}\!-\!\rate\right)
	&
	\!\rate\!=\!\RC{\rnf}{\Wm}\!\mbox{~for\!~a~}\!\rnf\!\in\!(1,\chi) 
	\\
	\sup\limits_{\rno\in [1,\chi]}\!\tfrac{1-\rno}{\rno}\!\left(\RC{\rno}{\Wm}\!-\!\rate\right)
	&
	\!\rate\!\geq\!\RC{\chi}{\Wm}
	\end{cases}
	\end{align}
	where \(\chi\!=\!\sup\{\rno\!:\!\RC{\rno}{\Wm}\!<\!\infty\}\) 
\end{lemma}
\begin{remark}
	There is a slight abuse of notation in \eqref{eq:lem:strongconverseexponent}
	for channels satisfying \(\lim\nolimits_{\rno\uparrow\infty}\RC{\rno}{\Wm}<\infty\).
	For such channels the term \(\sup\nolimits_{\rno\in [1,\chi]}\) in the last line of
	\eqref{eq:lem:strongconverseexponent} stands for 
	\(\sup\nolimits_{\rno\in [1,\infty)}\). 
\end{remark} 

Note that if \(\chi=1\), then \(\sce{\rate,\Wm}\) is zero for all \(\rate\)'s.

\begin{proof}[Proof of Lemma \ref{lem:strongconverseexponent}]
	\(\sce{\rate,\Wm}\) is convex/nondecreasing in \(\rate\),  
	because \(\tfrac{1-\rno}{\rno}(\RC{\rno}{\Wm}\!-\!\rate)\) is 
	convex/nondecreasing in \(\rate\) for any \(\rno\geq1\)
	and
	the pointwise 
	supremum of a family of convex/nondecreasing functions 
	is convex/nondecreasing.
	Note on the other hand \(\spe{\rate,\Wm}\leq \rate\)
	by definition.
	
	If \(\rate\leq\RC{1}{\Wm}\), then
	the expression \(\tfrac{1-\rno}{\rno}(\RC{\rno}{\Wm}\!-\!\rate)\) 
	can not be positive for \(\rno>1\)
	because
	\(\RC{\rno}{\Wm\!}\) is a nondecreasing in \(\rno\) by 
	Lemma \ref{lem:capacityO}-(\ref{capacityO-ilsc}).
	Thus \(\sce{\rate,\Wm}\) is zero 
	for \(\rate\leq\RC{1}{\Wm}\).
	
	If \(\rate=\RC{\rnf}{\Wm}\) for an \(\rnf\in(1,\chi)\)
	then \(\tfrac{1-\rno}{\rno}(\RC{\rno}{\Wm}\!-\!\rate)\) 
	is not positive for \(\rno\)'s  in \((\rnf,\infty)\). 
	Thus \eqref{eq:lem:strongconverseexponent} holds.
	
	For \(\rate\geq\RC{\chi}{\Wm}\) case, note that
	\(\tfrac{1-\rno}{\rno}(\RC{\rno}{\Wm}\!-\!\rate)\) is
	negative infinity for all \(\rno>\chi\). 
\end{proof}

\begin{lemma}\label{lem:arimoto}
Any \((M,L)\) channel code \((\enc,\dec)\)
on a channel \(\Wm:\inpS\to\pmea{\outA}\)
satisfying \(\abs{\dec(\dout)}=L\) for all \(\dout\in\outS\)
satisfies
\begin{align}
\label{eq:lem:arimoto}
\div{\rno}{\Pem{av}}{1-\tfrac{L}{M}}
&\leq \RMI{\rno}{\mP}{\Wm} 
&
&\forall \rno \in \reals{+} 
\end{align}
where \(\mP\) is the probability mass function generated by the encoder 
\(\enc\)  on \(\inpS\) when each message has equal probability mass  
and 
the function \(\div{\rno}{\cdot}{\cdot}\!:\![0,1]\!\times\![0,1]\!\to\![0,\infty]\)
is defined as 
\begin{align}
\notag
\!\!\div{\rno}{\mA}{\mB}
&\!\DEF\!
\begin{cases}
\tfrac{\ln \left[\mA^{\rno}\mB^{1-\rno}+(1-\mA)^{\rno}(1-\mB)^{1-\rno}\right]}{\rno-1}
&\rno\!\in\!\reals{+}\!\setminus\!\{1\}
\\
\mA\ln \tfrac{\mA}{\mB}+(1-\mA)\ln \tfrac{1-\mA}{1-\mB}
&\rno\!=\!1
\end{cases}.
\end{align}
If \(\RC{1}{\Wm}<\infty\), 
then any \((M,L)\) channel code on \(\Wm\)
such that \(\ln \tfrac{M}{L}\geq \RC{1}{\Wm}\) satisfies
\begin{align}
\label{eq:lem:arimoto-sp}
\Pem{av}\geq 1- e^{-\sce{\ln\frac{M}{L},\Wm}}.
\end{align}
\end{lemma}

We derive Arimoto's bound given in \eqref{eq:lem:arimoto-sp} from
\eqref{eq:lem:arimoto}, 
which is a result of the monotonicity of the \renyi divergence in the underlying 
\(\sigma\)-algebra, i.e.  Lemma \ref{lem:divergence-DPI},
and the expression for the \renyi information given in 
\eqref{eq:lem:information:defA}. 
For \(\rno\!=\!1\), \eqref{eq:lem:arimoto} is equivalent to 
Fano's inequality, \cite[Thm. 3.8]{csiszarkorner}. 
For \(\rno\in(0,1)\) and \(\rno\in(1,\infty)\) the monotonicity of the \renyi divergence 
in the underlying  \(\sigma\)-algebra is equivalent to  the H\"{o}lder's inequality and 
the reverse H\"{o}lder's inequality, respectively.
Sheverdyaev \cite{sheverdyaev82} is the first one to use the reverse H\"{o}lder's inequality to obtain 
a bound equivalent to \eqref{eq:lem:arimoto} for \(\rno\in(1,\infty)\).
Augustin had obtained a similar bound, \cite[Thm. 27.2-(ii)]{augustin78}, using the Jensen's inequality. 
Later, Nagaoka \cite{nagaoka01} derived an analogous bound in the context of quantum information theory.
More recently, Polyanskiy and Verd\'{u} \cite{polyanskiyV10} obtained a bound equivalent to \eqref{eq:lem:arimoto} 
for \(\rno\in(0,1)\cup(1,\infty)\).

First Omura \cite{omura75} and then Dueck and \korner \cite{dueckK79} have shown that 
Arimoto's bound is tight for the DSPCs, 
in terms of the exponential decay rate of the probability of correct decoding with block length.
If \(\RC{\rno}{\Wm}\!=\!\infty\) for all \(\rno\!>\!1\), 
then \(\sce{\rate,\Wm}\!=\!0\) for all \(\rate\!\in\![\RC{1}{\Wm},\infty)\),
see \cite[Example \ref*{A-eg:discontinuity}]{nakiboglu19A} for such a channel. 
Furthermore, \(\sce{\RC{1}{\Wm},\Wm}\!=\!0\) for any \(\Wm\) with finite \(\RC{1}{\Wm}\).
Thus, Arimoto's bound is non-trivial only when \(\ln\frac{M}{L}>\RC{1}{\Wm}\) and \(\RC{\rno}{\Wm}<\infty\) 
for an \(\rno\in (1,\infty)\).

\begin{proof}[Proof of Lemma \ref{lem:arimoto}]
For any sub-\(\sigma\)-algebra \(\alg{G}\) of \(\sss{\mesS}\otimes\outA\)
as a result of \eqref{eq:lem:information:defA} and 
Lemma \ref{lem:divergence-DPI} we have
\begin{align}
\notag
\RDF{\rno}{\alg{G}}{\mP \mtimes \Wm}{\mP\otimes\qmn{\rno,\mP}}
&\leq\RMI{\rno}{\mP}{\Wm}.
\end{align}
Then \eqref{eq:lem:arimoto} holds because 
\(\RDF{\rno}{\alg{G}}{\mP \mtimes \Wm}{\mP\!\otimes\!\qmn{\rno,\mP}}\)
is equal to 
\(\div{\rno}{1\!-\!\Pem{av}}{\!\tfrac{L}{M}\!}\)
when  \(\alg{G}\) is the \(\sigma\)-algebra generated by 
the set \(\{(\dmes,\dout):\dmes\!\in\!\dec(\dout)\}\). 

On the other hand, \(\mA^{\rno}\mB^{1-\rno}\!\leq\!e^{(\rno-1)\div{\rno}{\mA}{\mB}}\) 
for any \(\rno\) in \((1,\infty)\) and \(\mA,\mB\) in \([0,1]\). 
Thus using \(\RMI{\rno}{\mP}{\Wm}\leq \RC{\rno}{\Wm}\) we get
\begin{align}
\notag
(1-\Pem{av})^{\rno} (\tfrac{L}{M})^{1-\rno}
&\leq e^{(\rno-1)\RC{\rno}{\Wm}}
\end{align}
for all \(\rno\in (1,\infty)\). Then  \eqref{eq:lem:arimoto-sp} follows from 
Lemma \ref{lem:strongconverseexponent}.
Note that reducing the size of \(\dec(\dout)\) from \(L\) to a
smaller value for some \(\dout\) can only increase \(\Pem{av}\).
Thus \eqref{eq:lem:arimoto-sp} holds for all \((M,L)\) codes
not just the ones satisfying \(\abs{\dec(\dout)}\!=\!L\) for all
\(\dout\!\in\!\outS\).
\end{proof}

\subsection{The Channel Capacity}\label{sec:Ccapacity} 
All of the concepts we have considered thus far have been defined for a given channel. 
However, some of the most important concepts in information theory, such as the channel capacity, 
are asymptotic concepts that are defined for a given sequence of channels, rather than a given channel.
The sequence of channels in consideration, i.e. 
\(\{\Wm^{(\blx)}\}_{\blx\in\integers{+}}\), 
can usually be described in terms of 
an increasing sequence of input sets satisfying 
\(\inpS^{(\blx)}\!\subset\!\inpS^{(\knd)}\) whenever \(\blx\!\leq\!\knd\)
and a canonical channel \(\Wm\!:\!\inpS\!\to\!\pmea{\outA}\) satisfying
\begin{align}
\notag
\Wm^{(\blx)}(\dinp)
&=\Wm(\dinp)
&
&\dinp\in\inpS^{(\blx)}.
\end{align} 

As an example, consider the Poisson channels \(\Pcha{\tlx,\mA,\mB,\varrho}\) described in \eqref{eq:def:poissonchannel-mean}
and let \(\Wm^{(\blx)}\) be \(\Pcha{\blx,\mA,\mB,\varrho}\). 
When the problem is posed in this form,
\(\outS^{(\blx)}\) is the set of all possible sample paths for the arrival process on the interval \((0,\blx]\)
and hence \(\outS^{(\blx)}\) is different for each value of \(\blx\). 
If we extend the intensity function in time horizon to infinity at fixed intensity level 
\(\varrho\) for each \(\Wm(\dinp)\) for \(\dinp\in\inpS^{(\blx)}\) and let \(\outS^{(\blx)}\) be 
the set of all possible sample paths for the arrival process on \(\reals{+}\),
we get a model that is equivalent to the original one.
This equivalent model, however, satisfies the above mentioned conditions.  
A similar modification works, if  we are given a sequence of channels 
\(\{\Wmn{\ind}\}_{\ind\in\integers{+}}\) and  \(\Wm^{(\blx)}\) is  
the product channel \(\Wmn{[1,\blx]}\), 
or the product channel with feedback \(\Wmn{\vec{[1,\blx]}}\).

\begin{definition}\label{def:Ccapacity}
Given a sequence of channels \(\{\Wm^{(\blx)}\}_{\blx\in\integers{+}}\) 
and  a sequence of scaling factors \(\{\tin^{(\blx)}\}_{\blx\in\integers{+}}\), 
a sequence of codes \(\{(\enc^{(\blx)},\dec^{(\blx)})\}_{\blx\in\integers{+}}\) is reliable iff 
\begin{align}
%\label{eq:def:asmpptoticreliability}
\notag
\lim\nolimits_{\blx \to \infty} \Pem{av}^{(\blx)}
&=0.
\end{align}
A sequence of codes \(\{(\enc^{(\blx)},\dec^{(\blx)})\}_{\blx\in\integers{+}}\) is of rate \(\rate\) iff
\begin{align}
%\label{eq:def:rate}
\notag
\liminf\nolimits_{\blx \to \infty}   \tfrac{1}{\tin^{(\blx)} }\ln \tfrac{M^{(\blx)}}{L^{(\blx)}}
&=\rate.
\end{align}
A sequence of codes \(\{(\enc^{(\blx)},\dec^{(\blx)})\}_{\blx\in\integers{+}}\) is of list size \(L\)
iff \(L^{(\blx)}=L\) for all \(\blx\in\integers{+}\).

\emph{The channel capacity} for the pair \(\{(\Wm^{(\blx)},\tin^{(\blx)})\}_{\blx\in\integers{+}}\),
denoted by \(\SC{\Wm^{(\blx)},\tin^{(\blx)}}\), is the supremum of the rates of  
list size one reliable sequences of codes.
\end{definition}

The proof of the existence of reliable sequences for rates less than 
\(\SC{\Wm^{(\blx)},\tin^{(\blx)}}\) is usually called the direct part. 
The proof of the non-existence of reliable sequences for rates greater than  
\(\SC{\Wm^{(\blx)},\tin^{(\blx)}}\) is usually called the converse part.
For certain sequences of channels, one can strengthen the converse part 
by proving that \(\Pem{av}^{(\blx)}\) converges to one for all rate \(\rate\) 
sequences of codes for any \(\rate\) greater than \(\SC{\Wm^{(\blx)},\tin^{(\blx)}}\). 
These results are called the strong converses as opposed to the weak converses, which only establish that \(\Pem{av}^{(\blx)}\) 
is bounded away from zero.

\begin{theorem}\label{thm:capacity}
Let \(\{\Wm^{(\blx)}\}_{\blx\in\integers{+}}\) be a sequence of channels 
and  \(\{\tin^{(\blx)}\}_{\blx\in\integers{+}}\) be a sequence of scaling factors such that 
\(\lim_{\blx\to\infty}\tin^{(\blx)}=\infty\).
\begin{enumerate}[(a)]
\item\label{capacity:weak} If there exists an \(\varepsilon>0\) and a lower semicontinuous  function \(\rens\) such 
that
\begin{align}
\label{eq:lem:capacity:weak}
\lim\nolimits_{\blx \to \infty} \tfrac{\RC{\rno}{\Wm^{(\blx)}}}{\tin^{(\blx)}}
&=\rens(\rno)
&
&\forall  \rno \in [1-\varepsilon,1]
\end{align}
then \(\SC{\Wm^{(\blx)},\tin^{(\blx)}}=\rens(1)\).
\item\label{capacity:strong} If there exists an \(\varepsilon>0\) and an upper semicontinuous \(\rens\) 
satisfying \(\rens(1+\varepsilon)<\infty\) and
\begin{align}
\label{eq:lem:capacity:strong}
\limsup\nolimits_{\blx \to \infty} \tfrac{\RC{\rno}{\Wm^{(\blx)}}}{\tin^{(\blx)}}
&=\rens(\rno)
&
&\forall  \rno \in [1,1+\varepsilon]
\end{align}
then \(\lim_{\blx\to \infty} \Pem{av}^{(\blx)}=1\) for any rate \(\rate\) sequence of codes for \(\rate>\rens(1)\).
\end{enumerate}
\end{theorem}
\begin{remark}
We do not need the limits given in \eqref{eq:lem:capacity:weak} 
to exist for part (\ref{capacity:weak}). 
We only need 
\(\limsup_{\blx \to \infty} \sfrac{\RC{1}{\Wm^{(\blx)}}}{\tin^{(\blx)}}=\rens(1)\)
and
\(\liminf_{\blx \to \infty} \sfrac{\RC{\rno}{\Wm^{(\blx)}}}{\tin^{(\blx)}}=\rens(\rno)\)
for all \(\rno\in (1-\varepsilon,1)\).
\end{remark}

Theorem \ref{thm:capacity}, which is essentially a corollary of Lemmas \ref{lem:gallager} and \ref{lem:arimoto},
establishes the equality of the channel capacity to the scaled 
order one \renyi capacity for a class of channels much larger than the DSPCs
and  provides a sufficient condition for the existence of a strong converse.
Theorem \ref{thm:capacity}, however, does not claim 
either that the condition given in part (\ref{capacity:weak}) 
is necessary for the equality of the channel capacity to the scaled order one \renyi capacity, 
or that the condition given in part (\ref{capacity:strong}) is necessary 
for the existence of a strong  converse.

The hypotheses of Theorem \ref{thm:capacity} are relatively easy to check if \(\Wm^{(\blx)}\)'s are product channels
or product channels with feedback because of the additivity of the \renyi capacity established in 
Lemmas \ref{lem:capacityProduct} and \ref{lem:capacityFproduct}.
In particular, let \(\{\Wmn{\ind}\}_{\ind\in\integers{+}}\) be  a sequence of channels indexed by 
the positive integers such that
\begin{align}
\notag
\lim\nolimits_{\blx \to \infty}\tfrac{1}{\blx}
\sum\nolimits_{\ind=1}^{\blx} \RC{\rno}{\Wmn{\ind} }  
&\!=\!\rens(\rno)
&
&\forall \rno\in[1-\varepsilon,1+\varepsilon].
\end{align}
If \(\rens\) is a lower semicontinuous function on \([1-\varepsilon,1]\), 
then 
\(\SC{\Wmn{[1,\blx]},\blx}\!=\!\SC{\Wmn{\vec{[1,\blx]}},\blx}\!=\!\rens(1)\),
by Theorem \ref{thm:capacity}-(\ref{capacity:weak}), Lemma \ref{lem:capacityProduct},
 and  Lemma \ref{lem:capacityFproduct}. 
Furthermore, both \(\Wmn{[1,\blx]}\) and \(\Wmn{\vec{[1,\blx]}}\) have 
the strong converse property
provided that \(\rens\) is upper semicontinuous on \([1,1+\varepsilon]\) and finite at \((1+\varepsilon)\).

If \(\Wmn{\ind}\!=\!\Wm\) for all \(\ind\) for some \(\Wm\), 
i.e. if the product channel is stationary, 
then \eqref{eq:lem:capacity:weak} holds for \(\rens(\!\rno\!)\!=\!\RC{\rno}{\Wm}\) 
and the continuity of \(\rens\) on \((0,1]\) 
follows from  Lemma \ref{lem:capacityO}-(\ref{capacityO-zo}). Thus 
\(\SC{\Wmn{[1,\blx]},\blx}\!=\!\SC{\Wmn{\vec{[1,\blx]}},\blx}\!=\!\RC{1}{\Wm}\).
Furthermore, if \(\RC{\rnf}{\Wm}\!<\!\infty\) for a \(\rnf\!>\!1\), 
then \(\RC{\rno}{\Wm}\) is continuous on \((0,\rnf]\) by 
Lemma \ref{lem:capacityO}-(\ref{capacityO-continuity})
and the strong converse holds for both  
\(\Wmn{[1,\blx]}\)  and \(\Wmn{\vec{[1,\blx]}}\).
Recall that there are  SPCs that do not 
have  the strong converse property,
as demonstrated by Augustin \cite[pp. 57,58]{augustin66},
\cite[pp. 79,80]{augustin78}.
However,  \(\liminf_{\blx\to \infty} \Pem{av}^{(\blx)}\geq\varepsilon_{0}\)
for any rate \(\rate\) sequence of codes for \(\rate>\RC{1}{\Wm}\) on a SPC
where \(\varepsilon_{0}>0\) is a universal constant that does
not depend on the channel \(\Wm\), as a consequence of \cite{beckC78}
by Beck and \csiszar\!\!.

\begin{remark}
Augustin works with \(\Pem{\max}\DEF\max_{\dmes} \Pem{\dmes}\), rather than 
\(\Pem{av}\), and provides a necessary and sufficient condition for the strong converse 
property for the SPCs \cite[Cor. 10.3.1]{augustin66}, \cite[Cor. 13.4]{augustin78}.
The SPC described in \cite[pp. 57]{augustin66}, \cite[pp. 79]{augustin78} 
does not satisfy this condition. Hence, the strong converse does not hold for
\(\Pem{\max}\). The strong converse does not hold for \(\Pem{av}\)
either, because \(\Pem{\max(\blx)}\!>\!\Pem{av(\blx)}\).
%In \cite[\S10]{augustin66} and \cite[\S13]{augustin78}, Augustin discusses the 
%existence of the strong converse for the product channels more generally.
\end{remark}

\begin{remark}
To be precise, Beck and \csiszar \cite{beckC78} showed 
\(\liminf_{\blx\to \infty} \Pem{\max}^{(\blx)}\geq\epsilon_{0}\)
for a universal constant \(\epsilon_{0}\)
but this implies 
\(\liminf_{\blx\to \infty} \Pem{av}^{(\blx)}\geq \sfrac{\epsilon_{0}}{2}\)
by the standard application of Markov's inequality to the channel codes.
\end{remark}

Theorem \ref{thm:capacity} applies to certain sequences of 
memoryless channels that are not product channels, as well.
For example, the Poisson channels 
\(\Pcha{\tlx,\mA,\mB,\varrho}\), \(\Pcha{\tlx,\mA,\mB,\leq \varrho}\), 
\(\Pcha{\tlx,\mA,\mB,\geq\varrho}\),
whose input sets are described in \eqref{eq:def:poissonchannel},
satisfy the hypotheses for both parts of Theorem \ref{thm:capacity}
provided that \(\tlx^{(\blx)}\!=\!\tin^{(\blx)}\!=\!\blx\) and \(\mB\) is finite.
Thus \(\SC{\Pcha{\blx},\blx}=\sfrac{\RC{1}{\Pcha{\tlx}}}{\tlx}\) for any \(\tlx>0\)
and the strong converse property holds
for each one of these  Poisson channels,
see \cite[\S\ref*{A-sec:examples-poisson}]{nakiboglu19A}
for closed form expressions.
The closed form expressions for \(\SC{\Pcha{\blx,\mA,\mB},\blx}\) and \(\SC{\Pcha{\blx,\mA,\mB,\leq \varrho},\blx}\) 
	were determined by Kabanov \cite{kabanov78} and Davis \cite{davis80}, but only with a weak converse.
	All of the previous strong converses for the Poisson channels are for zero dark current cases:
	for \(\Pcha{\blx,0,\mB,\leq \varrho}\) by Burnashev and Kutoyants \cite{burnashevK99}, 
	for \(\Pcha{\blx,0,\mB}\) by Wagner and Anantharam \cite{wagnerA04}.
Hypotheses of Theorem \ref{thm:capacity} can be confirmed for certain channels with memory, as well. 
But formal statement and confirmation of those claims are beyond the scope of the current article. 

Lemma \ref{lem:arimoto}
imply \(\Pem{av}\!\geq\!1-e^{-\sce{\rate,\Wm}}\) for \(\rate\!\geq\!\RC{1}{\Wm}\) 
for  the Poisson channels described in \eqref{eq:def:poissonchannel}.
This was reported before by Wagner and Anantharam \cite{wagnerA04}, 
but only for \(\Pcha{\blx,0,\mB}\).

Before presenting the proof of Theorem \ref{thm:capacity}, 
let  us point out a sequence 
violating both hypotheses of Theorem \ref{thm:capacity}. 
Consider \(\cha\)'s described in \cite[Example \ref*{A-eg:erasure}]{nakiboglu19A}
for \(\tin^{(\blx)}=\ln \blx\). Then, 
\begin{align}
\notag
\lim\limits_{\blx\to\infty}\tfrac{\RC{\rno}{\Wm^{(\blx)}}}{\tin^{(\blx)}}
&=\begin{cases}
0
&\rno\in(0,1)
\\
1-\gamma
&\rno=1
\\
1
&\rno\in(1,\infty) 
\end{cases}.
\end{align}
Theorem \ref{thm:capacity} is silent about the channel capacity 
of this sequence of channels.
But \(\lim_{\blx \to \infty} M^{(\blx)}\!=\!1\) for any list size one reliable sequence of codes
because  \(\Pem{av}\!\geq\!\tfrac{M-1}{M}\gamma\) for any \((M,1)\) channel code on \(\Wm^{(\blx)}\).
Thus \(\SC{\Wm^{(\blx)},\tin^{(\blx)}}\!=\!0\) for any  \(\tin^{(\blx)}\) such that 
\(\lim_{\blx \to \infty} \tin^{(\blx)}\!>\!0\).
Previously, this sequence of channels is used by Verd\'{u} and Han \cite{verduH94} 
to motivate their investigation 
of a general expression for the channel capacity.

\begin{proof}[Proof of Theorem \ref{thm:capacity}]~
\begin{enumerate}
\item[(\ref{capacity:weak})]
Let us first prove the direct part.
For any \(\rate\!<\!\rens(1)\) there exists an \(\rno_{0}\)  such that 
\(\rens(\rno)\!>\!\rate\) for all \(\rno\!\geq\!\rno_{0}\) because 
\(\rens\) is lower  semicontinuous.
Thus, there exists an \(\rno_{1}\!\in\!(1\!-\!\varepsilon,1)\) such that 
\(\rens(\rno_{1})\!>\!\rate\).
Furthermore, there exists an \(\blx_{0}\) such that 
\(\sfrac{\RC{\rno_{1}}{\Wm^{(\blx)}}}{\tin^{(\blx)}}\!>\!\tfrac{\rens(\rno_{1})+\rate}{2}\) for all 
\(\blx\!\geq\!\blx_{0}\)
because \(\liminf_{\blx\to\infty}\!\sfrac{\RC{\rno_{1}}{\Wm^{(\blx)}}}{\tin^{(\blx)}}\!=\!\rens(\rno_{1})\).
Consequently, for each \(\blx\!\geq\!\blx_{0}\) there exists a \(\mP^{(\blx)}\) in
\(\pdis{\inpS^{(\blx)}}\) satisfying 
\(\sfrac{\RMI{\rno_{1}}{\mP^{(\blx)}}{\Wm^{(\blx)}}}{\tin^{(\blx)}}\!>\!\tfrac{\rens(\rno_{1})+\rate}{2}\) 
by the definition of \renyi capacity.
Then for each \(\blx\!>\!\blx_{0}\) there exists a 
\((\lfloor e^{\rate \tin^{(\blx)}}\rfloor,1)\) channel code on  \(\Wm^{(\blx)}\) 
satisfying
\begin{align}
\notag
\Pem{av}^{(\blx)}
&\leq e^{\frac{\rno_{1}-1}{\rno_{1}} \frac{\rens(\rno_{1})-\rate}{2} \tin^{(\blx)}}
\end{align}
by  Lemma \ref{lem:gallager}.
Thus there exists a rate \(\rate\) list size one reliable sequence
for any \(\rate\!<\!\rens(1)\).

If \(\rens(1)\) is infinite, then the converse is trivial.
Thus we assume  \(\rens(1)\) to be finite. 
For any \(\rate\!>\!\rens(1)\),  there exists an \(\blx_{0}\) such that \(\sfrac{\RC{1}{\Wm^{(\blx)}}}{\tin^{(\blx)}}\!<\!\tfrac{\rate+2\rens(1)}{3}\) 
for all \(\blx\!\geq\!\blx_{0}\) because  
\(\limsup_{\blx\to\infty}\!\sfrac{\RC{1}{\Wm^{(\blx)}}}{\tin^{(\blx)}}\!=\!\rens(1)\).
Furthermore, for any rate \(\rate\) sequence there exists an \(\blx_{1}\) such that 
\(\sfrac{\ln M^{(\blx)}}{\tin^{(\blx)}}>\tfrac{2\rate+\rens(1)}{3}\) for all \(\blx\geq\blx_{1}\)
because \(\liminf_{\blx\to\infty} \sfrac{\ln M^{(\blx)}}{\tin^{(\blx)}}=\rate\).
On the other hand, \((1\!-\!\Pem{av}^{(\blx)})\ln M^{(\blx)}\!-\!\ln 2\leq \RC{1}{\Wm^{(\blx)}}\) 
by  \eqref{eq:lem:arimoto}  for \(\rno=1\),
i.e. by Fano's inequality. Thus
\begin{align}
\notag
\Pem{av}^{(\blx)}
&\geq \tfrac{\rate-\rens(1)}{2\rate+\rens(1)}-\tfrac{3 \ln 2}{(2\rate+\rens(1))\tin^{(\blx)}}
&
&\forall \blx>(\blx_{0}\vee\blx_{1}).
\end{align}
Consequently,
there does not exist a rate \(\rate\) reliable sequence of codes for \(\rate\!>\!\rens(1)\).

\item[(\ref{capacity:strong})] 
For any \(\rate\!>\!\rens(1)\) there exists an \(\rno_{0}\) such that 
\(\rens(\rno)\!<\!\rate\) for all \(\rno\!\leq\!\rno_{0}\) because 
\(\rens\) is upper semicontinuous by the hypothesis.
Thus there exists an \(\rno_{1}\!\in\!(1,1\!+\!\varepsilon)\) 
such that \(\rens(\rno_{1})\!<\!\rate\).
Hence, there exists an \(\blx_{0}\) satisfying the inequality
\(\sfrac{\RC{\rno_{1}}{\Wm^{(\blx)}}}{\tin^{(\blx)}}\!<\!\tfrac{\rate+2\rens(\rno_{1})}{3}\) 
for all \(\blx\!\geq\!\blx_{0}\)
because \(\limsup_{\blx\to\infty} \sfrac{\RC{\rno_{1}}{\Wm^{(\blx)}}}{\tin^{(\blx)}}\!=\!\rens(\rno_{1})\).
Furthermore, for any rate \(\rate\) sequence of codes there exists an \(\blx_{1}\) satisfying 
the inequality
\(\sfrac{\ln M^{(\blx)}}{\tin^{(\blx)}}\!>\!\tfrac{2\rate+\rens(\rno_{1})}{3}\)
for all \(\blx\!\geq\!\blx_{1}\) because 
\(\liminf_{\blx\to\infty}\sfrac{\ln M^{(\blx)}}{\tin^{(\blx)}}=\rate\). 
Thus \eqref{eq:lem:arimoto} implies that
\begin{align}
\notag
\Pem{av}^{(\blx)}
&\geq 1-e^{\frac{\rno_{1}-1}{\rno_{1}} \frac{\rens(\rno_{1})-\rate}{3} \tin^{(\blx)}}
&
&\forall \blx>(\blx_{0}\vee\blx_{1}).
\end{align}
Thus \(\lim_{\blx\to\infty} \Pem{av}^{(\blx)}=1\) for any rate \(\rate\) sequence of codes for \(\rate>\rens(1)\).
\end{enumerate}
\end{proof}

\section*{Acknowledgment}
The author would like to thank Fatma Nakibo\u{g}lu and Mehmet Nakibo\u{g}lu for their hospitality; 
this work simply would not have been possible without it.  
The author would like to thank 
Imre Csisz\'{a}r for pointing out Agustin's work at Austin in 2010 ISIT,
Harikrishna R. Palaiyanur for sending him Augustin's manuscript \cite{augustin78},
Vincent Tan for pointing out \cite[Thm. 35-(b), (469)]{sasonV16},
Reviewer II for pointing out \cite{augustin66} and \cite{beckC78},
Reviewer III for pointing out \cite{dalaiW17} and \cite{nagaoka01},
and
G\"{u}ne\c{s} Nakibo\u{g}lu and the reviewers for their suggestions on the manuscript.
\bibliographystyle{plain} 
\bibliography{main}

\newcommand{\noopsort}[1]{} \newcommand{\printfirst}[2]{#1}
  \newcommand{\singleletter}[1]{#1} \newcommand{\switchargs}[2]{#2#1}
\begin{thebibliography}{10}

\bibitem{altugW11}
Y.~Altu{\u{g}} and A.~B. Wagner.
\newblock Refinement of the sphere packing bound for symmetric channels.
\newblock In {\em 49th Annual Allerton Conference on Communication, Control,
  and Computing}, pages 30--37, Sept 2011.

\bibitem{altugW12B}
Y.~Altu{\u{g}} and A.~B. Wagner.
\newblock A refinement of the random coding bound.
\newblock In {\em 50th Annual Allerton Conference on Communication, Control,
  and Computing}, pages 663--670, Oct 2012.

\bibitem{altugW14C}
Y.~Altu{\u{g}} and A.~B. Wagner.
\newblock Moderate deviations in channel coding.
\newblock {\em IEEE Transactions on Information Theory}, 60(8):4417--4426, Aug
  2014.

\bibitem{altugW14D}
Y.~Altu{\u{g}} and A.~B. Wagner.
\newblock Refinement of the random coding bound.
\newblock {\em IEEE Transactions on Information Theory}, 60(10):6005--6023, Oct
  2014.

\bibitem{altugW14A}
Y.~Altu{\u{g}} and A.~B. Wagner.
\newblock Refinement of the sphere-packing bound: Asymmetric channels.
\newblock {\em IEEE Transactions on Information Theory}, 60(3):1592--1614,
  March 2014.

\bibitem{arimoto73}
S.~Arimoto.
\newblock On the converse to the coding theorem for discrete memoryless
  channels (corresp.).
\newblock {\em IEEE Transactions on Information Theory}, 19(3):357--359, May
  1973.

\bibitem{augustin66}
U.~Augustin.
\newblock Ged{\"a}chtnisfreie {K}an{\"a}le f{\"u}r diskrete {Z}eit.
\newblock {\em Zeitschrift f{\"u}r {W}ahrscheinlichkeitstheorie und {V}erwandte
  {G}ebiete}, 6(1):10--61, Mar 1966.

\bibitem{augustin69}
U.~Augustin.
\newblock Error estimates for low rate codes.
\newblock {\em Zeitschrift f{\"u}r {W}ahrscheinlichkeitstheorie und {V}erwandte
  {G}ebiete}, 14(1):61--88, 1969.

\bibitem{augustin78}
Udo Augustin.
\newblock {\em Noisy Channels}.
\newblock Habilitation thesis, Universit\"{a}t Erlangen-N\"{u}rnberg, 1978.
\newblock (\url{http://bit.ly/2ID8h7m}).

\bibitem{beckC78}
J.~Beck and I.~Csisz{\'a}r.
\newblock Medium converse for memoryless channels with arbitrary alphabets.
\newblock {\em Problems of control and information theory}, 7(3):199--202,
  1978.

\bibitem{berry41}
A.~C. Berry.
\newblock The accuracy of the gaussian approximation to the sum of independent
  variates.
\newblock {\em Transactions of the American Mathematical Society},
  49(1):122--136, 1941.

\bibitem{bogachev}
Vladimir~I. Bogachev.
\newblock {\em Measure Theory}.
\newblock Springer-Verlag, Berlin Heidelberg, 2007.

\bibitem{burnashevK99}
M.~V. Burnashev and Yu.~A. Kutoyants.
\newblock On the sphere-packing bound, capacity, and similar results for
  {P}oisson channels.
\newblock {\em Problems of Information Transmission}, 35(2):95--111, 1999.

\bibitem{comoN10}
G.~Como and B.~Nakibo{\u{g}}lu.
\newblock Sphere-packing bound for block-codes with feedback and finite memory.
\newblock In {\em 2010 IEEE International Symposium on Information Theory},
  pages 251--255, June 2010.

\bibitem{coverthomas}
Thomas~M. Cover and Joy~A. Thomas.
\newblock {\em Elements of information theory}.
\newblock Wiley-Interscience, New York, NY, 2 edition, 2006.

\bibitem{csiszar95}
I.~Csisz{\'a}r.
\newblock Generalized cutoff rates and {R}{\'e}nyi's information measures.
\newblock {\em IEEE Transactions on Information Theory}, 41(1):26--34, Jan
  1995.

\bibitem{csiszar98}
I.~Csisz{\'a}r.
\newblock The method of types [information theory].
\newblock {\em IEEE Transactions on Information Theory}, 44(6):2505--2523, Oct
  1998.

\bibitem{csiszarK82}
I.~Csisz{\'a}r and J.~K{\"o}rner.
\newblock Feedback does not affect the reliability function of a {DMC} at rates
  above capacity (corresp.).
\newblock {\em IEEE Transactions on Information Theory}, 28(1):92--93, Jan
  1982.

\bibitem{csiszarkorner}
Imre Csisz{\'a}r and J{\'a}nos K{\"o}rner.
\newblock {\em Information theory: coding theorems for discrete memoryless
  systems}.
\newblock Cambridge University Press, Cambridge, UK, 2011.

\bibitem{dalai13A}
M.~Dalai.
\newblock Lower bounds on the probability of error for classical and
  classical-quantum channels.
\newblock {\em IEEE Transactions on Information Theory}, 59(12):8027--8056, Dec
  2013.

\bibitem{dalai17}
M.~Dalai.
\newblock Some remarks on classical and classical-quantum sphere packing
  bounds: {R}{\'e}nyi vs. {K}ullback-{L}eibler.
\newblock {\em Entropy}, 19(7):355, 2017.

\bibitem{dalaiW17}
M.~Dalai and A.~Winter.
\newblock Constant compositions in the sphere packing bound for
  classical-quantum channels.
\newblock {\em IEEE Transactions on Information Theory}, 63(9):5603--5617, Sept
  2017.

\bibitem{davis80}
M.~Davis.
\newblock Capacity and cutoff rate for {P}oisson-type channels.
\newblock {\em IEEE Transactions on Information Theory}, 26(6):710--715, Nov
  1980.

\bibitem{dobrushin62A}
R.~L. Dobrushin.
\newblock An asymptotic bound for the probability error of information
  transmission through a channel without memory using the feedback.
\newblock {\em Problemy Kibernetiki}, (8):161--168, 1962.

\bibitem{dudley}
Richard~M. Dudley.
\newblock {\em Real analysis and probability}, volume~74.
\newblock Cambridge University Press, New York, NY, 2002.

\bibitem{dueckK79}
G.~Dueck and J.~Korner.
\newblock Reliability function of a discrete memoryless channel at rates above
  capacity (corresp.).
\newblock {\em IEEE Transactions on Information Theory}, 25(1):82--85, Jan
  1979.

\bibitem{ebert66}
Paul~Michael Ebert.
\newblock {\em Error Bounds For Parallel Communication Channels}.
\newblock Technical report 448, Research Laboratory of Electronics at
  Massachusetts Institute of Technology, Cambridge, MA, 1966.
\newblock (\url{http://hdl.handle.net/1721.1/4295}).

\bibitem{elias57}
P.~Elias.
\newblock {\em List decoding for noisy channels}.
\newblock Technical report 335, Research Laboratory of Electronics at
  Massachusetts Institute of Technology, Cambridge, MA, 1957.
\newblock (\url{http://hdl.handle.net/1721.1/4484}).

\bibitem{ervenH14}
T.~van Erven and P.~Harremo\"{e}s.
\newblock {R}{\'e}nyi divergence and {Kullback-Leibler} divergence.
\newblock {\em IEEE Transactions on Information Theory}, 60(7):3797--3820, July
  2014.

\bibitem{esseen42}
C.~G. Esseen.
\newblock On the {L}iapunoff limit of error in the theory of probability.
\newblock {\em {A}rkiv f\"{o}r {M}atematik, {A}stronomi {O}ch {F}ysik},
  A28:1--19, 1942.

\bibitem{fongT16}
S.~L. Fong and V.~Y.~F. Tan.
\newblock Strong converse theorems for classes of multimessage multicast
  networks: A {R}{\'e}nyi divergence approach.
\newblock {\em IEEE Transactions on Information Theory}, 62(9):4953--4967, Sept
  2016.

\bibitem{gallager65}
R.~G. Gallager.
\newblock A simple derivation of the coding theorem and some applications.
\newblock {\em IEEE Transactions on Information Theory}, 11(1):3--18, Jan.
  1965.

\bibitem{gallager}
Robert~G. Gallager.
\newblock {\em Information theory and reliable communication}.
\newblock John Wiley \& Sons, Inc., New York, NY, 1968.

\bibitem{guntuboyinaSS14}
A.~Guntuboyina, S.~Saha, and G.~Schiebinger.
\newblock Sharp inequalities for $f$-divergences.
\newblock {\em IEEE Transactions on Information Theory}, 60(1):104--121, Jan
  2014.

\bibitem{han}
Te~Sun Han.
\newblock {\em Information-Spectrum Methods in Information Theory}.
\newblock Stochastic Modelling and Applied Probability (Book 50). Springer
  Berlin Heidelberg, 2003.

\bibitem{haroutunian68}
E.~A. Haroutunian.
\newblock Estimates of the error probability exponent for a semicontinuous
  memoryless channel.
\newblock {\em Problems of Information Transmission}, 4(4):37--48, 1968.

\bibitem{haroutunian77}
E.~A. Haroutunian.
\newblock Lower bound for error probability in channels with feedback.
\newblock {\em Problems of Information Transmission}, 13(2):36--44, 1977.

\bibitem{haroutunianHH07}
Evgueni~A. Haroutunian, Mariam~E. Haroutunian, and Ashot~N. Harutyunyan.
\newblock Reliability criteria in information theory and in statistical
  hypothesis testing.
\newblock {\em Foundations and Trends in Communications and Information
  Theory}, 4(2-3):97--263, 2007.

\bibitem{hayashi09B}
M.~Hayashi.
\newblock Information spectrum approach to second-order coding rate in channel
  coding.
\newblock {\em IEEE Transactions on Information Theory}, 55(11):4947--4966, Nov
  2009.

\bibitem{honda15}
J.~Honda.
\newblock Exact asymptotics for the random coding error probability.
\newblock In {\em 2015 IEEE International Symposium on Information Theory
  (ISIT)}, pages 91--95, June 2015.

\bibitem{kabanov78}
Yu.~M. Kabanov.
\newblock The capacity of a channel of the {P}oisson type.
\newblock {\em Theory of Probability \& Its Applications}, 23(1):143--147,
  1978.

\bibitem{moulin17}
P.~Moulin.
\newblock The log-volume of optimal codes for memoryless channels,
  asymptotically within a few nats.
\newblock {\em IEEE Transactions on Information Theory}, 63(4):2278--2313,
  April 2017.

\bibitem{munkres}
James~R. Munkres.
\newblock {\em Topology}.
\newblock Prentice Hall Inc., Upper Saddle River, NJ 07458, 2000.

\bibitem{nagaoka01}
H.~Nagaoka.
\newblock Strong converse theorems in quantum information theory.
\newblock In {\em Proceedings of the ERATO Conference on Quantum Information
  Science (EQIS)}, volume~33, 2001.
\newblock (also appeared as Chap. 4 of Asymptotic Theory of Quantum Statistical
  Inference: Selected Papers, ed. by M. Hayashi).

\bibitem{nakiboglu17}
B.~Nakibo{\u{g}}lu.
\newblock The {A}ugustin center and the sphere packing bound for memoryless
  channels.
\newblock In {\em 2017 IEEE International Symposium on Information Theory
  (ISIT)}, pages \href{http://arxiv.org/abs/1701.06610}{1401--1405}, Aachen,
  Germany, June 2017.

\bibitem{nakiboglu18C}
B.~Nakibo{\u{g}}lu.
\newblock The {A}ugustin capacity and center.
\newblock {\em \href{http://arxiv.org/abs/1803.07937v3}{arXiv:1803.07937}
  [cs.IT]}, 2018.

\bibitem{nakiboglu18D}
B.~Nakibo{\u{g}}lu.
\newblock The sphere packing bound for memoryless channels.
\newblock {\em \href{http://arxiv.org/abs/1804.06372v2}{arXiv:1804.06372}
  [cs.IT]}, 2018.

\bibitem{nakiboglu19A}
B.~Nakibo{\u{g}}lu.
\newblock The {R\'{e}nyi} capacity and center.
\newblock {\em IEEE Transactions on Information Theory}, 65(2):841--860, Feb
  2019.
\newblock (\href{http://arxiv.org/abs/1608.02424}{arXiv:1608.02424} [cs.IT]).

\bibitem{nakiboglu19E}
B.~Nakibo{\u{g}}lu.
\newblock The sphere packing bound for {DSPCs} with feedback \`{a} la
  {A}ugustin.
\newblock {\em IEEE Transactions on Communications}, 2019.
\newblock
  \href{https://doi.org/10.1109/TCOMM.2019.2931302}{DOI:10.1109/TCOMM.2019.2931302},
  (\href{http://arxiv.org/abs/1806.11531}{arXiv:1806.11531} [cs.IT]).

\bibitem{nakibogluZ12}
B.~Nakibo{\u{g}}lu and L.~Zheng.
\newblock Errors-and-erasures decoding for block codes with feedback.
\newblock {\em IEEE Transactions on Information Theory}, 58(1):24--49, Jan
  2012.

\bibitem{nakiboglu11}
Bar\i{\c{s}} Nakibo{\u{g}}lu.
\newblock {\em Exponential Bounds On Error Probability With Feedback}.
\newblock {Ph.D. Thesis}, Massachusetts Institute of Technology. Department of
  Electrical Engineering and Computer Science, February 2011.
\newblock (\url{http://hdl.handle.net/1721.1/64485}).

\bibitem{omura75}
J.~K. Omura.
\newblock A lower bounding method for channel and source coding probabilities.
\newblock {\em Information and Control}, 27(2):148 -- 177, 1975.

\bibitem{palaiyanurS10}
H.~Palaiyanur and A.~Sahai.
\newblock An upper bound for the block coding error exponent with delayed
  feedback.
\newblock In {\em 2010 IEEE International Symposium on Information Theory},
  pages 246--250, June 2010.

\bibitem{palaiyanurS15}
H.~Palaiyanur and A.~Sahai.
\newblock On {H}aroutunian's exponent for parallel channels and an application
  to fixed-delay codes without feedback.
\newblock {\em IEEE Transactions on Information Theory}, 61(3):1298--1308,
  March 2015.

\bibitem{palaiyanurthesis}
Harikrishna~R. Palaiyanur.
\newblock {\em The impact of causality on information-theoretic source and
  channel coding problems}.
\newblock {Ph.D. Thesis}, EECS Department, University of California, Berkeley,
  May 2011.
\newblock
  (\href{http://www.eecs.berkeley.edu/Pubs/TechRpts/2011/EECS-2011-55.html}{http://escholarship.org/uc/item/1m29g0bp}).

\bibitem{polyanskiyPV10}
Y.~Polyanskiy, H.~V. Poor, and S.~Verd{\'u}.
\newblock Channel coding rate in the finite blocklength regime.
\newblock {\em IEEE Transactions on Information Theory}, 56(5):2307--2359, May
  2010.

\bibitem{polyanskiyV10}
Y.~Polyanskiy and S.~Verd{\'u}.
\newblock Arimoto channel coding converse and {R}{\'e}nyi divergence.
\newblock In {\em Communication, Control, and Computing (Allerton), 2010 48th
  Annual Allerton Conference on}, pages 1327 --1333, Oct 2010.

\bibitem{richters67}
John~Stephen Richters.
\newblock {\em Communication over fading dispersive channels}.
\newblock Technical report 464, Research Laboratory of Electronics at
  Massachusetts Institute of Technology, Cambridge, MA, 1967.
\newblock (\url{http://hdl.handle.net/1721.1/4279}).

\bibitem{rudin}
Walter Rudin.
\newblock {\em Principles of Mathematical Analysis}.
\newblock McGraw-Hill, New York, NY, 1976.

\bibitem{sasonV16}
I.~Sason and S.~Verd{\'u}.
\newblock $f$-divergence inequalities.
\newblock {\em IEEE Transactions on Information Theory}, 62(11):5973--6006, Nov
  2016.

\bibitem{scarlettMF13}
J.~Scarlett, A.~Martinez, and A.~G.~i Fabregas.
\newblock A derivation of the asymptotic random-coding prefactor.
\newblock In {\em Communication, Control, and Computing (Allerton), 2013 51st
  Annual Allerton Conference on}, pages 956--961, Oct 2013.

\bibitem{scarlettMF14}
J.~Scarlett, A.~Martinez, and A.~G.~i Fabregas.
\newblock Mismatched decoding: Error exponents, second-order rates and
  saddlepoint approximations.
\newblock {\em IEEE Transactions on Information Theory}, 60(5):2647--2666, May
  2014.

\bibitem{shannon48}
C.~E. Shannon.
\newblock A mathematical theory of communication.
\newblock {\em Bell System Technical Journal, The}, 27(3 and 4):379--423 and
  623--656, July and October 1948.

\bibitem{shannon59}
C.~E. Shannon.
\newblock Probability of error for optimal codes in a {Gaussian} channel.
\newblock {\em The Bell System Technical Journal}, 38(3):611--656, May 1959.

\bibitem{shannonGB67A}
C.~E. Shannon, R.~G. Gallager, and E.~R. Berlekamp.
\newblock Lower bounds to error probability for coding on discrete memoryless
  channels. {I}.
\newblock {\em Information and Control}, 10(1):65--103, 1967.

\bibitem{sheverdyaev82}
A.~Yu Sheverdyaev.
\newblock Lower bound for error probability in a discrete memoryless channel
  with feedback.
\newblock {\em Problems of Information Transmission}, 18(4):5--15, 1982.

\bibitem{shevtsova10}
I.~G. Shevtsova.
\newblock An improvement of convergence rate estimates in the {L}yapunov
  theorem.
\newblock {\em Doklady Mathematics}, 82(3):862--864, 2010.

\bibitem{shiryaev}
Albert~N. Shiryaev.
\newblock {\em Probability}.
\newblock Springer-Verlag, New York, NY, 1995.

\bibitem{strassen62}
Volker Strassen.
\newblock Asymptotische absch{\"a}tzungen in {S}hannons {I}nformationstheorie.
\newblock In {\em Trans. Third Prague Conf. Inf. Theory}, pages 689--723, 1962.
\newblock (\url{http://www.math.cornell.edu/~pmlut/strassen.pdf}).

\bibitem{tomamichelT13}
M.~Tomamichel and V.~Y.~F. Tan.
\newblock A tight upper bound for the third-order asymptotics for most discrete
  memoryless channels.
\newblock {\em IEEE Transactions on Information Theory}, 59(11):7041--7051, Nov
  2013.

\bibitem{valemboisF04}
A.~Valembois and M.~P.~C. Fossorier.
\newblock Sphere-packing bounds revisited for moderate block lengths.
\newblock {\em IEEE Transactions on Information Theory}, 50(12):2998--3014, Dec
  2004.

\bibitem{verduH94}
S.~Verd{\'u} and Te~Sun Han.
\newblock A general formula for channel capacity.
\newblock {\em IEEE Transactions on Information Theory}, 40(4):1147--1157, Jul
  1994.

\bibitem{wagnerA04}
A.~B. Wagner and V.~Anantharam.
\newblock Feedback, queueing, and the reliability of the ideal {P}oisson
  channel above capacity.
\newblock In {\em Information Theory, 2004. ISIT 2004. Proceedings.
  International Symposium on}, page 447, June 2004.

\bibitem{wiechmanS08}
G.~Wiechman and I.~Sason.
\newblock An improved sphere-packing bound for finite-length codes over
  symmetric memoryless channels.
\newblock {\em IEEE Transactions on Information Theory}, 54(5):1962--1990, May
  2008.

\bibitem{wyner88-b}
A.~D. Wyner.
\newblock Capacity and error exponent for the direct detection photon channel.
  {II}.
\newblock {\em IEEE Transactions on Information Theory}, 34(6):1462--1471, Nov
  1988.

\end{thebibliography}
\end{document}